\documentclass[reprint,
aps,superscriptaddress,twocolumn,showpacs,longbibliography]{revtex4-2}

\usepackage{graphicx,color}
\usepackage{amsmath,amsfonts,enumerate,amsthm,amssymb,bbm}
\usepackage[colorlinks=true,citecolor=blue,linkcolor=magenta]{hyperref}
\usepackage[all]{hypcap}
\usepackage{multirow}
\usepackage{braket}
\usepackage{soul}
\usepackage{float}
\usepackage[caption = false]{subfig}
\usepackage[normalem]{ulem}
\usepackage{scrextend}
\usepackage[usenames,dvipsnames]{xcolor}
\usepackage{dsfont}
\usepackage{nicefrac}
\usepackage[linesnumbered,ruled,vlined]{algorithm2e}
\usepackage{quantikz}
\usepackage{apptools}

\makeatother
\newcommand{\suppsection}[1]{\section{#1}\addcontentsline{stoc}{section}{\protect\numberline{\thesection}#1}}

\newcommand{\suppsubsection}[1]{\subsection{#1}\addcontentsline{stoc}{subsection}{\protect\numberline{\thesubsection}#1}}

\usepackage{paper-style-cleaned}

\mathchardef\ordinarycolon\mathcode`\:
\mathcode`\:=\string"8000
\begingroup \catcode`\:=\active
  \gdef:{\mathrel{\mathop\ordinarycolon}}
\endgroup

\theoremstyle{plain}
\newtheorem{thm}{Theorem}

\newtheorem{corol}{Corollary}
\newtheorem{propos}{Proposition}

\AtAppendix{\counterwithin{thm}{section}}
\AtAppendix{\counterwithin{corol}{section}}
\AtAppendix{\counterwithin{lem}{section}}
\AtAppendix{\counterwithin{propos}{section}}
\AtAppendix{\counterwithin{defn}{section}}
\AtAppendix{\counterwithin{rmk}{section}}
\AtAppendix{\counterwithin{ex}{section}}

\def\<{\langle}

\def\>{\rangle}
\def\<{\langle}
 
\DeclareMathOperator{\poly}{poly}

\renewcommand{\ket}[1]{|#1\rangle}

\renewcommand{\bra}[1]{\langle #1|}

\renewcommand{\vec}[1]{\boldsymbol{#1}}

\begin{document}

\title{Universal Bounds for Out-of-Distribution Unitary Learning}

\author{Joachim Favre}
\affiliation{Institute of Physics, Ecole Polytechnique F\'{e}d\'{e}rale de Lausanne (EPFL), CH-1015 Lausanne, Switzerland.}
\affiliation{Centre for Quantum Science and Engineering, Ecole Polytechnique F\'{e}d\'{e}rale de Lausanne (EPFL),  Lausanne, Switzerland}
\affiliation{Department of Computer Science, Princeton University, Princeton, NJ, USA}

\author{Armando Angrisani}
\affiliation{Institute of Physics, Ecole Polytechnique F\'{e}d\'{e}rale de Lausanne (EPFL), CH-1015 Lausanne, Switzerland.}
\affiliation{Centre for Quantum Science and Engineering, Ecole Polytechnique F\'{e}d\'{e}rale de Lausanne (EPFL),  Lausanne, Switzerland}
\affiliation{Department of Mathematical Sciences, University of Copenhagen, Universitetsparken 5, 2100 Denmark}

\author{Zo\"{e} Holmes}
\affiliation{Institute of Physics, Ecole Polytechnique F\'{e}d\'{e}rale de Lausanne (EPFL), CH-1015 Lausanne, Switzerland.}
\affiliation{Centre for Quantum Science and Engineering, Ecole Polytechnique F\'{e}d\'{e}rale de Lausanne (EPFL),  Lausanne, Switzerland}

\date{\today}

\begin{abstract}
How much can the action of an unknown unitary dynamics on one family of states reveal about its action elsewhere?
We develop a general framework for out-of-distribution quantum dynamics learning based on the first and second moments of the input distribution. These define a notion of the bias and expressivity of an ensemble, which together control generalization between arbitrary mixed-state training and testing distributions.
Our framework shows that only studying the action of a unitary on random real states can learn its action on random complex states, that tensor products of unbiased pure-state ensembles preserve expressivity, and that single-qubit risks are independent of bias. It also recovers earlier locally scrambling results as loose special cases and yields tight product-to-Haar bounds. We conclude with a series of no-go results and limited generalization bounds for learning certain channels, which further serve to underscore the distinctive role of unitarity in enabling out-of-distribution generalization. 
\end{abstract}

\maketitle

\section{Introduction}

Quantum processes are defined by their action on the full operator space, yet in general we only access, or seek to predict, their action on restricted families of states. While complete process tomography reconstructs this action everywhere~\cite{nielsen2000quantum,mohseni2008quantum, gutoski2014process, Zhou2014Process, roncallo2023Pauli, torlai2023quantum, wadhwa2024agnostic}, learning from input-output examples asks how many states, and what kinds of states, are sufficient to predict the process on new inputs~\cite{bisio2010optimal, arrazola2019machine, sharma2020reformulation, chung2021sample, caro2022learning, huang2023learning, gibbs2022dynamical, caro2021generalization,caro2022outofdistribution,angrisani2023learning, kunjummen2023shadow, wadhwa2023learning, jerbi2023power, chen2024predicting, fanizza2024learning}. The first question concerns in-distribution generalization, where the training and testing states are drawn from the same ensemble~\cite{bousquet2003introduction, mohri2018foundations, shalev2014understanding}. The second concerns out-of-distribution generalization, where they are drawn from different ensembles~\cite{quinonero2008dataset,shen2021towards}. 

While in-distribution generalization has been studied extensively in quantum learning~\cite{caro2020pseudo, abbas2020power, sharma2020reformulation, huang2021power, bu2021onthestatistical, banchi2021generalization, du2021efficient, gyurik2021structural, caro2021encodingdependent, caro2021generalization, chen2021expressibility, popescu2021learning, cai2022sample, Volkoff2021Universal}, much less is known about out-of-distribution generalization between different ensembles of quantum states.
Previous work has established that out-of-distribution generalization is possible in certain cases  for learning unknown unitary dynamics from pure-state inputs~\cite{caro2022outofdistribution, gibbs2022dynamical, pereira2025out}. In particular, learning a unitary on random  product states can be sufficient to control performance on globally Haar-random  states, and more broadly generalization is possible between any so-called locally scrambling ensembles~\cite{kuo2020markovian,hu2021classical}. Analogous results have been obtained also for channels for weaker learning models, where the goal is to predict local observables rather than approximate the output state~\cite{huang2023learning}. These results open up the possibility of learning quantum dynamics and compiling quantum circuits using only easy to prepare or simulate unentangled states~\cite{danna2025circuit}.
However, these bounds are limited to highly random and unbiased ensembles. 

In this work, we develop a general framework for out-of-distribution unitary learning. Our starting point is that the dependence of the risk on the input ensemble can be encoded in a moment matrix. In the Pauli representation, this matrix packages together the first and second moments of the distribution and naturally identifies two quantities that control generalization. The first is the \emph{bias}, which measures how far the average state of a pure state ensemble lies from the maximally mixed state. The second is a measure of \emph{expressivity} that captures how evenly the ensemble probes different directions in operator space. We use these quantities to derive out-of-distribution bounds between arbitrary pure-state distributions. Broadly, the expressivity measure of the training ensemble controls a multiplicative learning error and the biases of the training and testing ensembles determine additive errors. 

For unitary learning from pure-state inputs, this geometric viewpoint leads to several potentially surprising conclusions. For one, learning only on random real states is sufficient to control performance on arbitrary complex Haar-random states, with an overhead that approaches one as the number of qubits increases. Secondly, for pure-state product distributions, tensoring an unbiased ensemble preserves its expressivity exactly, but the behavior of tensor products of biased ensembles is more subtle. 
In the single-qubit setting, we then find that a biased ensemble of three suitably chosen states has exactly the same risk as Haar-random inputs, while training on just two real states controls the Haar risk up to a constant factor.
These (non-exhaustive) examples show that the apparent complexity of the training states is not, by itself, what determines their usefulness for learning.
Our framework also recovers earlier guarantees for locally scrambling pure-state ensembles as a special case~\cite{caro2022outofdistribution}, while showing that the corresponding previous bounds are loose and extending the comparison to products of Haar-random blocks.

We also derive general out-of-distribution bounds for unitary learning from mixed-state inputs. Here, a large mismatch between the average purities of the training and testing ensembles can force a large multiplicative overhead when the additive error is small. This is consistent with the intuition that training on noisier inputs to predict outputs for purer ones is more demanding: mixing an input with the maximally mixed state suppresses differences between unitaries, making those outputs easier to predict but less informative for learning.
Our bounds further show that out-of-distribution generalization is robust to noise as long as the rate is sufficiently weak such that the output state purities decay at worst inversely polynomially in system size.

For unitary learning, we derive nontrivial generic generalization bounds that hold for arbitrary unitaries. However, such guarantees do not extend to general channels. We present several no-go results showing that out-of-distribution generalization can be impossible for channels even in simple settings. For example, unlike in the unitary case, transferring a guarantee from globally Haar-random pure states to product-Haar states can require an exponential overhead, even for an efficiently implementable measure-and-prepare channel. This contrast underscores the distinctive role of unitarity in enabling out-of-distribution generalization.

However, these no-go results do not rule out useful guarantees for restricted channel classes. For sparse Pauli channels, training on globally Haar-random pure states controls performance on arbitrary testing distributions, including distributions over mixed states, with an overhead linear in the sparsity. Training on product-Haar states gives a similar guarantee with a quadratic overhead. Thus, Pauli sparsity provides a setting in which polynomial out-of-distribution generalization remains possible beyond unitary learning. These bounds may be useful for noise learning when the noise process is well approximated by a sparse Pauli channel~\cite{flammia2020efficient,harper2020efficient,harper2021fast,flammia2021pauli,chen2022quantum}.

\section{Background}

In this work, we study how to learn an unknown $n$-qubit quantum channel $\mathcal{E}$ from input-output examples drawn from one distribution of states, with the aim of predicting its action on states drawn from a potentially different distribution. We will focus primarily on unitary channels but frame our setting more generally here for completeness.

More precisely, we assume access to a training data set
\begin{equation}
    \mathcal{D}_{\rm train}(N)
    =
    \left\{
    \left(
        \rho^{(j)}, \mathcal{E}\left(\rho^{(j)}\right)
    \right)
    \right\}_{j=1}^{N},
\end{equation}
where the input states $\rho^{(j)}$ are sampled independently from a training distribution $\mathcal{P}_{\rm train}$.
The task is to use this data to train the parameters $\boldsymbol{\alpha}$ of a model channel $\mathcal{E}_{\boldsymbol{\alpha}}$. The learned channel should then approximate $\mathcal{E}$ not only on states typical of the training distribution $\mathcal{P}_{\rm train}$, but also on states drawn from a testing distribution $\mathcal{P}_{\rm test}$, which need not coincide with $\mathcal{P}_{\rm train}$.

We quantify performance using a loss function $\ell$ that compares the true output state $\mathcal{E}(\rho)$ with the predicted output state $\mathcal{E}_{\boldsymbol{\alpha}}(\rho)$. The quantity that we ultimately care about is the average loss on the testing distribution $\mathcal{P}_{\rm test}$,
\begin{equation}\label{eq:test-cost-channel}
    R_{\mathcal{P}_{\rm test}}(\boldsymbol{\alpha})
    =
    \mathbb{E}_{\rho \sim \mathcal{P}_{\rm test}}
    \left[
        \ell\left(
            \mathcal{E}(\rho),
            \mathcal{E}_{\boldsymbol{\alpha}}(\rho)
        \right)
    \right].
\end{equation}
We refer to this quantity as the expected risk. It measures how well the candidate channel $\mathcal{E}_{\boldsymbol{\alpha}}$ predicts the action of the target channel on the states on which we want the learned model to generalize.

Although learning is possible for a variety of losses, in what follows we will primarily use the squared Schatten $2$-norm loss
\begin{equation}\label{eq:two-norm-loss-channel}
\ell\left(
    \mathcal{E}(\rho),
    \mathcal{E}_{\boldsymbol{\alpha}}(\rho)
\right)
=
\frac{1}{2}
\left\|
    \mathcal{E}(\rho)
    -
    \mathcal{E}_{\boldsymbol{\alpha}}(\rho)
\right\|_2^2 .
\end{equation}
Here $\|A\|_2 = \sqrt{\operatorname{Tr}(A^\dagger A)}$ denotes the Hilbert-Schmidt norm. The factor of $1/2$ ensures that the loss takes values in $[0,1]$.
In the case where both the target and model channels are unitary,
\begin{equation}
    \mathcal{E}(\rho)
    =
    U \rho U^\dagger,
    \qquad
    \mathcal{E}_{\boldsymbol{\alpha}}(\rho)
    =
    V(\boldsymbol{\alpha}) \rho V(\boldsymbol{\alpha})^\dagger,
\end{equation}
and where the input state is pure, $\rho = \ket{\Psi}\bra{\Psi}$, this loss reduces to the usual pure-state infidelity, i.e.,
\begin{align}\label{eq:two-norm-loss-unitary}
&\ell\left(
    U\ket{\Psi}\bra{\Psi}U^\dagger,
    V(\boldsymbol{\alpha})\ket{\Psi}\bra{\Psi}V(\boldsymbol{\alpha})^\dagger
\right)
\\&=
1
-
\left|
    \bra{\Psi}
    V(\boldsymbol{\alpha})^\dagger U
    \ket{\Psi}
\right|^2 .
\end{align}

In practice, the learner does not have access to the testing distribution $\mathcal{P}_{\rm test}$, and therefore cannot evaluate the expected risk in Eq.~\eqref{eq:test-cost-channel}. Instead, training is performed using the finite sample $\mathcal{D}_{\rm train}(N)$. The corresponding empirical loss is
\begin{equation}\label{eq:empirical-loss-channel}
    L_{\rm train}(\boldsymbol{\alpha}, N)
    =
    \frac{1}{N}
    \sum_{j=1}^{N}
    \ell\left(
        \mathcal{E}\left(\rho^{(j)}\right),
        \mathcal{E}_{\boldsymbol{\alpha}}\left(\rho^{(j)}\right)
    \right).
\end{equation}
{\color{ black}Training amounts to finding parameters $\boldsymbol{\alpha}_{\rm opt}$ that make the empirical loss $L_{\rm train}(\boldsymbol{\alpha}_{\rm opt},N)$ small.
However, this is only a proxy for the quantity of interest which is the expected risk 
$R_{\mathcal{P}_{\rm test}}(\boldsymbol{\alpha}_{\rm opt})$ on the testing distribution. 

The generalization error, defined as
\begin{equation}\label{eq:gen-error-channel}
    \operatorname{gen}
    :=
    R_{\mathcal{P}_{\rm test}}
    -
    L_{\rm train}
    \left( N
    \right) 
\end{equation}
(where we drop the $\boldsymbol{\alpha}$ dependence for compactness), 
can be decomposed into in-distribution and out-of-distribution contributions:
\begin{equation}\label{eq:gen-error-decomposition-channel}
\begin{aligned}
    \operatorname{gen}
    =
    &
    \underbrace{
    R_{\mathcal{P}_{\rm train}}
    -
    L_{\rm train}
    \left(N\right)
    }_{\text{in-distribution generalization error}}
    +
    \underbrace{
    R_{\mathcal{P}_{\rm test}}
    -
    R_{\mathcal{P}_{\rm train}}
    }_{\text{out-of-distribution error}} .
\end{aligned}
\end{equation}
The first term is the usual finite-sample contribution, i.e., it measures how well the empirical training loss estimates the expected risk on the same distribution. In the in-distribution setting, where $\mathcal{P}_{\rm test}=\mathcal{P}_{\rm train}$, the second term vanishes and the problem reduces to controlling this standard generalization error. This contribution has already been studied in the context of learning quantum processes. In particular, for models with independently trainable local gates, the empirical loss is known to be a good proxy for $R_{\mathcal{P}_{\rm train}}(\boldsymbol{\alpha}_{\rm opt})$ once the number of training samples is comparable to the number of trainable parameters~\cite{caro2021generalization}.
In this paper, we focus on the out-of-distribution term, which measures how the expected risk of the learned model changes between the training and testing input-state distributions.

Concretely, in this work we ask when good performance on the training distribution $\mathcal{P}_{\rm train}$ guarantees good performance on the testing distribution $\mathcal{P}_{\rm test}$.
At the level of expected risks, this amounts to finding constants $A,B\geq 0$ such that
\begin{equation}\label{eq:ood-risk-comparison-channel}
    R_{\mathcal{P}_{\rm test}}
    \leq
    A R_{\mathcal{P}_{\rm train}}
    +
    B
\end{equation}
for all model and target processes in the class under consideration. The useful regime is one in which $A$ is not too large and $B$ is small, or ideally zero. In that case, any model with small risk on $\mathcal{P}_{\rm train}$ is guaranteed to also have small risk on $\mathcal{P}_{\rm test}$. Combining this risk comparison with standard in-distribution generalization bounds~\cite{caro2020pseudo, abbas2020power, sharma2020reformulation, huang2021power, bu2021onthestatistical, banchi2021generalization, du2021efficient, gyurik2021structural, caro2021encodingdependent, caro2021generalization, chen2021expressibility, popescu2021learning, cai2022sample, Volkoff2021Universal} then yields a corresponding guarantee for models trained on finitely many samples.

\section{Preliminaries}

Before stating our main results, we rewrite the risk to identify two properties of the input distribution that control our \textit{unitary} learning analysis: its bias and expressivity. (This analysis does not extend directly to general channels, so we will use different arguments for our channel guarantees and no-go results.)  \\

\paragraph*{A convenient expression for the risk.}
The first step of our analysis is to rewrite the risk in a form that makes its dependence on the distribution of states easier to study.
Expanding the squared Hilbert--Schmidt loss, the unitary-learning risk can be written as
\begin{equation}\label{eq:unitary-loss-trace-main-text}
\begin{aligned}
R_{\mathcal{P}}\left(U,V\right)
=
\bar{\gamma}_{\mathcal{P}}
-
\mathbb{E}_{\rho\sim\mathcal{P}}
\operatorname{Tr}
\left(
U^\dagger V\rho V^\dagger U\rho
\right),
\end{aligned}
\end{equation}
where $\bar{\gamma}_{\mathcal{P}}$ denotes the average purity of the states in $\mathcal{P}$,
\begin{equation}\label{eq:average-purity-main-text}
\bar{\gamma}_{\mathcal{P}}
:=
\mathbb{E}_{\rho\sim\mathcal{P}}
\left[
\operatorname{Tr}\left(\rho^2\right)
\right],
\end{equation}
and, for pure inputs, $\bar{\gamma}_{\mathcal{P}}=1$.

The second term in Eq.~\eqref{eq:unitary-loss-trace-main-text} can be conveniently rewritten using vectorization notation. For any operator $A=\sum_{i,j}A_{ij}\ket{i}\bra{j}$, we define
\begin{equation}\label{eq:vectorization-main-text}
\kket{A}
:=
\operatorname{vec}(A)
=
\sum_{i,j}
A_{ij}\ket{i}\otimes\ket{j}.
\end{equation}
With this convention, the vectorized inner product is the Hilbert--Schmidt inner product,
$\bbrakket{A}{B}=\operatorname{Tr}(A^\dagger B)$, and the standard ricochet identity reads
$\kket{ABC}=(A\otimes C^{\mathsf T})\kket{B}$, where $^{\mathsf T}$ denotes the transpose in the computational basis. 
We can therefore write
\begin{equation}
\begin{aligned}
\operatorname{Tr}
\left(
U^\dagger V\rho V^\dagger U\rho
\right)
&=
\bbrakket{V^\dagger U}{\rho V^\dagger U\rho}
\\
&=
\bbra{V^\dagger U}
\left(
\rho\otimes\rho^{\mathsf T}
\right)
\kket{V^\dagger U}.
\end{aligned}
\end{equation}

The risk for unitary learning can therefore be written as
\begin{align}
R_{\mathcal{P}}
&=
\bar{\gamma}_{\mathcal{P}}
-
\bbra{w}
M(\mathcal{P})
\kket{w}
\label{eq:risk-rayleigh-main-text}
\end{align}
where we have defined the ensemble-dependent matrix
\begin{align}
M(\mathcal{P})
&:=
2^n
\mathbb{E}_{\rho\sim\mathcal{P}}
\left[
\rho\otimes\rho^{\mathsf T}
\right]
\label{eq:fsm-matrix-main-text}
\end{align}
and the normalized vector
\begin{equation}\label{eq:w-vector-main-text}
\kket{w}
:=
\frac{1}{\sqrt{2^n}}
\kket{V^\dagger U}.
\end{equation}
The same ensemble-dependent
moment matrix, $M(\mathcal{P})$, also appears in Ref.~\cite{chen2024predicting} for single-qubit pure inputs,
which studies in-distribution generalization for observable
prediction with arbitrary quantum channels and product input
distributions.

This expression for the risk, Eq.~\eqref{eq:risk-rayleigh-main-text}, is the core mathematical structure we use to derive our results (Lemma~\ref{lma:risk-rayleigh} in the Appendices).
The vector $\kket{w}$ depends only on the mismatch between the target unitary $U$ and the model unitary $V$, while the matrix $M(\mathcal{P})$ depends only on the input distribution $\mathcal{P}$. We further note that  $\bar{\gamma}_{\mathcal{P}}=2^{-n}\bbra{I}M(\mathcal{P})\kket{I}$ and so even in the case of mixed inputs, all dependence on the input distribution can be studied by looking at $M(\mathcal{P})$. In this manner, comparing risks for different input distributions reduces to comparing their moment matrices and it is possible to derive bounds for arbitrary target and model unitaries. \\

In the case of channel learning, the risk admits several analogous decompositions that separate its dependence on the input ensemble from its dependence on the channels. However, none appears to give useful general bounds. We therefore work directly with the risk for specific channels and restricted channel families, obtaining both no-go theorems and specialized positive results.

\medskip

\paragraph*{The moment matrix.}

 The matrix $M(\mathcal{P})$ contains information about the first and second moments of the distribution $\mathcal{P}$ and so we will call it the \textit{moment matrix} for brevity.

Let $\widetilde{I}:=I/\sqrt{2^n}$ denote the normalized identity operator and let
\begin{equation}\label{eq:nonidentity-projector-main-text}
{\color{ black}
\Pi_{\rm ni}
=
I \otimes I
-
\kket{\widetilde{I}}\bbra{\widetilde{I}}
}
\end{equation}
be the projector onto the subspace spanned by the \textit{non-identity} Pauli strings.
For a state $\rho$, we define its non-identity Pauli vector as
\begin{equation}\label{eq:nonidentity-pauli-vector-main-text}
\kket{r(\rho)}
=
\sqrt{2^n}
\Pi_{\rm ni}
\kket{\rho}.
\end{equation}
For any state $\rho$, the vectorized density matrix decomposes as
\begin{equation}\label{eq:rho-vector-nonidentity-decomposition-main-text}
{\color{ black}
\kket{\rho}
=
\frac{1}{\sqrt{2^n}}
\left(
\kket{\widetilde{I}}
+
\kket{r(\rho)}
\right).
}
\end{equation}
Thus, $\kket{r(\rho)}$ is the part of the vectorized state supported on non-identity Pauli strings, rescaled so that its entries are the usual non-identity Pauli expectation values.

{\color{ black}
The moment matrix of any ensemble can now be written as
}
\begin{equation}\label{eq:fsm-decomposition-main-text}
{\color{ black}
M(\mathcal{P})
=
\begin{pmatrix}
\bar{\gamma}_{\mathcal{P}} & \bbra{\mu} \\
\kket{\mu} & \mathcal{S}
\end{pmatrix}.
}
\end{equation}

{\color{ black}
Here the off-diagonal block is
}
\begin{equation}\label{eq:mu-definition-main-text}
{\color{ black}
\kket{\mu}
:=
\Pi_{\rm ni}
M(\mathcal{P})
\kket{\widetilde{I}}
=
\sqrt{2^n}\,
\Pi_{\rm ni}
\mathbb{E}_{\rho\sim\mathcal{P}}
\left[
\kket{\rho^2}
\right],
}
\end{equation}
{\color{ black}
while the non-identity block is
}
\begin{equation}\label{eq:S-definition-main-text}
{\color{ black}
\mathcal{S}
:=
\Pi_{\rm ni}
M(\mathcal{P})
\Pi_{\rm ni}.
}
\end{equation}
Thus, for mixed inputs, $\kket{\mu}$ contains the non-identity Pauli components of $\mathbb{E}_{\rho\sim\mathcal{P}}[\rho^2]$, while $\mathcal{S}$ contains the remaining quadratic moments $\mathbb{E}_{\rho\sim\mathcal{P}}[\operatorname{Tr}(P_j\rho P_k\rho)]$.
For pure inputs, $\rho^2=\rho$ and $\rho P_k\rho=\operatorname{Tr}(P_k\rho)\rho$. Consequently, we have 
\begin{equation}\label{eq:pure-state-moment-blocks-main-text}
\begin{aligned}
\kket{\mu}
&=
\mathbb{E}_{\rho\sim\mathcal{P}}
\left[
\kket{r(\rho)}
\right],
\\
\mathcal{S}
&=
\mathbb{E}_{\rho\sim\mathcal{P}}
\left[
\kket{r(\rho)}
\bbra{r(\rho)}
\right].
\end{aligned}
\end{equation}
{\color{ black}
In both cases, $M(\mathcal{P})$ packages together the usual first and second Pauli moments for pure-state ensembles. We therefore also call $\kket{\mu}$ the first moment and $\mathcal{S}$ the second moment of $\mathcal{P}$.
}

\medskip

\paragraph*{Bias.}

The first quantity that enters our bounds is the bias $\nu := \sqrt{\bbrakket{\mu}{\mu}}$, the norm of the vector $\kket{\mu}$.
By definition, $\kket{\mu}$ contains the non-identity Pauli components of $\mathbb{E}_{\rho\sim\mathcal{P}}[\rho^2]$. Its norm therefore admits the closed-form expression
\begin{equation}\label{eq:bias-average-state-main-text}
\nu
=
\sqrt{
2^n
\operatorname{Tr}
\left[
\left(
\mathbb{E}_{\rho\sim\mathcal{P}}
\left[\rho^2\right]
\right)^2
\right]
-
\bar{\gamma}_{\mathcal{P}}^2
}.
\end{equation}
For pure-state ensembles, $\rho^2=\rho$ and $\bar{\gamma}_{\mathcal{P}}=1$, so this reduces to the original interpretation: the bias measures how far the average state lies from the maximally mixed state. More generally, $\nu=0$ precisely when $\mathbb{E}_{\rho\sim\mathcal{P}}[\rho^2]=\bar{\gamma}_{\mathcal{P}}I/2^n$.
At the other extreme, a distribution supported on a single pure state has maximal bias, with $\nu=\sqrt{2^n-1}$ (see Appendix~\ref{sec:general-framework} for more details)}.

For pure-state ensembles, small bias is a non-trivial condition.
In particular, one can show that requiring $\nu<1$ forces the distribution to have exponentially many states in its support (see Theorem~\ref{thm:bound-support-size-mu-eta}).
However, small bias alone does not imply that a distribution probes a broad range of states or directions.
For example, the uniform distribution on $\left\{\ket{0},\ket{1}\right\}^{\otimes n}$ is unbiased, since it averages to the maximally mixed state, but it remains highly structured and therefore lowly expressive in a sense that we will now introduce.\\

\paragraph*{Expressivity.}
The second quantity that enters our bounds is the expressivity measure,
\begin{equation}\label{eq:expressivity-main-text}
{\color{ black}
\eta
=
\bar{\gamma}_{\mathcal{P}}
-
\lambda_{\max}
\left(
\mathcal{S}
\right).
}
\end{equation}
This gives the purity-dependent bound
\begin{equation}\label{eq:eta-range-main-text}
2\bar{\gamma}_{\mathcal{P}}-2^n
\leq
\eta
\leq
\frac{4^n}{4^n-1}
\left(
\bar{\gamma}_{\mathcal{P}}-\frac{1}{2^n}
\right) , 
\end{equation}
which for pure-state ensembles reduces to  $ 2-2^n\leq\eta\leq1-1/(2^n+1)$. The lower extreme is attained by a distribution concentrated on a single pure state, for which the second moment is supported entirely along one direction. The upper extreme is attained by the Haar distribution, whose second moment is isotropic on the non-identity Pauli subspace.
Thus, positive $\eta$ rules out distributions that are too concentrated along any single Pauli direction. At fixed average purity, larger values of $\eta$ indicate that the distribution spreads more evenly over the non-identity Pauli degrees of freedom.

\medskip
\paragraph*{Single qubit illustrative example.}
These definitions of the bias and expressivity are particularly intuitive for single qubit pure state ensembles. Every pure state can be written in standard Bloch form as
\begin{equation}\label{eq:single-qubit-bloch-state}
\rho
=
\frac{1}{2}
\left(
I
+
r_x X
+
r_y Y
+
r_z Z
\right),
\end{equation}
where $\boldsymbol{r}=(r_x,r_y,r_z)^{\mathsf{T}}$ is the corresponding Bloch vector. In our notation, the non-identity Pauli vector $\kket{r(\rho)}$ is precisely this Bloch vector. Thus, for a single-qubit distribution $\mathcal{P}$,
\begin{equation}\label{eq:single-qubit-mu-S}
\kket{\mu}
=
\begin{pmatrix}
\mathbb{E}_{\rho\sim \mathcal{P}}[r_x]\\
\mathbb{E}_{\rho\sim \mathcal{P}}[r_y]\\
\mathbb{E}_{\rho\sim \mathcal{P}}[r_z]
\end{pmatrix},
\qquad
\mathcal{S}
=
\mathbb{E}_{\rho\sim \mathcal{P}}
\left[
\boldsymbol{r}\boldsymbol{r}^{\mathsf{T}}
\right].
\end{equation}
Equivalently, writing $\mu_j=\mathbb{E}_{\rho\sim \mathcal{P}}[r_j]$ and $\mathcal{S}_{jk}=\mathbb{E}_{\rho\sim \mathcal{P}}[r_jr_k]$ for $j,k\in\{x,y,z\}$, the moment matrix is
\begin{equation}\label{eq:single-qubit-moment-matrix-explicit}
M(\mathcal{P})
=
\begin{pmatrix}
1 & \mu_x & \mu_y & \mu_z\\
\mu_x & \mathcal{S}_{xx} & \mathcal{S}_{xy} & \mathcal{S}_{xz}\\
\mu_y & \mathcal{S}_{yx} & \mathcal{S}_{yy} & \mathcal{S}_{yz}\\
\mu_z & \mathcal{S}_{zx} & \mathcal{S}_{zy} & \mathcal{S}_{zz}
\end{pmatrix}.
\end{equation}
In this case the bias $\nu$ is simply the length of the average Bloch vector, while the expressivity
\begin{equation}\label{eq:single-qubit-expressivity}
\eta
=
1
-
\lambda_{\max}\left(\mathcal{S}\right)
\end{equation}
measures how much the distribution spreads over different Bloch-sphere directions. A formal visualization of $\eta$ on the Bloch sphere can be found in Theorem~\ref{thm:eta-visualization} and Figure~\ref{fig:eta-visualization} in the appendices, but one should really just think of $\eta$ as a measure of spread. Both intuitions in fact carry over to mixed inputs (see Lemma~\ref{lma:mean-for-single-qubit} in the appendices for a proof of this fact for the bias).

This is the same single-qubit bias and expressivity measure used in Ref.~\cite{chen2024predicting}. More precisely, their second moment matrix for a one-qubit input distribution coincides with our $\mathcal{S}$, and their expressivity parameter coincides with our $\eta=1-\lambda_{\max}(\mathcal{S})$. Our bias is likewise the norm of the corresponding mean Bloch vector. Thus, on a single qubit, our definitions reduce to the quantities used in Ref.~\cite{chen2024predicting} but here we also use the same moment-matrix language for arbitrary $n$-qubit input distributions over mixed states.

\section{Unitary Learning}

In this section we present our main results for the unitary learning setting. 
All results in this section hold for any two unitaries $U$ and $V$ and so we do not explicitly show their dependence on them, i.e, we write $R \equiv R(U, V)$.
We first give an aesthetic general comparison theorem for arbitrary mixed-state distributions before introducing a second theorem for producing tight comparison bounds.
We then specialize this theorem to Haar-random states, random real states, product distributions and locally scrambling distributions. For mixed-state inputs, we establish a necessary purity penalty and show that some guarantees survive local depolarizing noise, although with a rapidly growing overhead that is expected due to the purity penalty.

\subsection{General expressivity and bias dependent bound}

Here we present a general out-of-distribution generalization bound, which relates the risk on a testing ensemble $\mathcal{P}_{\rm test}$ to that on a training ensemble $\mathcal{P}_{\rm train}$. 
\begin{thm}[General unitary learning guarantee]\label{thm-main-text:general-unitary-comparison}
Let $\mathcal{P}_{\rm train}$ and $\mathcal{P}_{\rm test}$ be arbitrary mixed-state distributions, with biases $\nu_{\rm train}$ and $\nu_{\rm test}$, respectively.
Let $\eta_{\rm train}$ be the expressivity measure of $\mathcal{P}_{\rm train}$.
If $\eta_{\rm train} > 0$, then, for all unitaries $U,V$,
\begin{equation}\label{eq:general-unitary-comparison-main-text}
R_{\mathcal{P}_{\rm test}}
\leq
\frac{
R_{\mathcal{P}_{\rm train}}
+
\nu_{\rm train}
}{
\eta_{\rm train}
}
+
\nu_{\rm test}.
\end{equation}
\end{thm}

This is a general out-of-distribution statement in the sense that no structural relation between $\mathcal{P}_{\rm train}$ and $\mathcal{P}_{\rm test}$ is assumed. The theorem tells us that small risk on $\mathcal{P}_{\rm train}$ transfers to small risk on $\mathcal{P}_{\rm test}$ whenever both distributions have small bias and the expressivity of the training ensemble is not too small. In particular, when both distributions are unbiased, it reduces to a purely multiplicative comparison with overhead $1/\eta_{\rm train}$. If the bias is too large the bound becomes trivial and if the expressivity too small learning only transfers in the limit of perfect learning. This result is what justifies the definition of bias and expressivity, and should be kept in mind whenever thinking about one of these two objects.

The theorem as stated is asymmetric: only the expressivity of the distribution used to control the risk, here $\mathcal{P}_{\rm train}$, appears in the denominator. This is natural, since an overly concentrated training ensemble cannot certify the behavior of $U$ and $V$ along directions that it does not probe. However, the testing ensemble also affects the cost of transfer, since some ensembles are more sensitive to differences between the two unitaries. 
Part~4 of Theorem~\ref{thm:inequalities} in the appendix captures this dependence through one more multiplicative factor, $1-\lambda_{\min}(\mathcal{S}_{\rm test})$ for pure-state distributions, which is bounded by one in the simpler result above.

In the following subsections we specialize this general comparison to more structured ensembles.
These examples show how the abstract quantities $\nu$ and $\eta$ behave for Haar-random states, random real states, product distributions, locally scrambling distributions, and Haar-random states under depolarizing noise.

\subsection{Tight bounds}

Before turning to applications, we discuss how we can derive tight bounds. Readers primarily interested in the applications may skip this subsection.

Theorem~\ref{thm-main-text:general-unitary-comparison} only depends on values that can be independently computed for the training and test distributions. This is both its strength and its weakness: it is simple to use, but we cannot expect it to yield tight inequalities in general. Using a similar proof technique as the one used for Theorem~\ref{thm-main-text:general-unitary-comparison}, we come up with the following theorem, which has stronger hypotheses but is able to produce tight inequalities.
\begin{thm}[Tight unitary learning guarantee]\label{thm-main-text:tight-unitary-comparison}
    Let $\mathcal{P}_{\rm train}$ and $\mathcal{P}_{\rm test}$ be arbitrary pure-state distributions, with expressivity $\eta_{\rm train}, \eta_{\rm test}$, respectively. Let $\mathcal{S}_{\rm train}$ and $\mathcal{S}_{\rm test}$ be their second moments, as defined from their moment matrices in Eq.~\ref{eq:fsm-decomposition-main-text}. If both distributions are unbiased, $\eta_{\rm train} > 0$ and $\mathcal{S}_{\rm train}, \mathcal{S}_{\rm test}$ are both diagonalizable in the Pauli basis, then, for all unitaries $U, V$,
    \begin{equation}\label{eq:tight-inequality-equation-theorem}
        R_{\mathcal{P}_{\rm test}} \leq \left(\max_{\substack{P\ \mathrm{Pauli} \\ P \neq I}} \frac{1 - \lambda_P(\mathcal{S}_{\rm test})}{1 - \lambda_P(\mathcal{S}_{\rm train})}\right)R_{\mathcal{P}_{\rm train}},
    \end{equation}
    where $\lambda_P(\mathcal{S})$ is the eigenvalue associated to the eigenvector $\kket{P}$ for $\mathcal{S}$.
    
    Moreover, Eq.~\ref{eq:tight-inequality-equation-theorem} is reached with equality for any unitaries $U, V$ such that $V^\dagger U = P^*$, where $P^*$ is a non-identity Pauli that maximizes the ratio above.
\end{thm}

There are two hypotheses in this theorem that are not in Theorem~\ref{thm-main-text:general-unitary-comparison}. The first one is that both distributions are unbiased (although Theorem~\ref{thm-main-text:general-unitary-comparison} already required that both biases are small to produce non-vacuous inequalities). The second one, much more restrictive, is that $\mathcal{S}_{\rm train}$ and $\mathcal{S}_{\rm test}$ are both diagonalizable in the Pauli basis. Instead of only looking at the maximum spread over the different Bloch-sphere distributions through expressivity measures $\eta$, this hypothesis makes it possible to jointly analyze the spread of both distributions. Interestingly, all the distributions considered in this paper have a second moment $\mathcal{S}$ that is diagonalizable in the Pauli basis (i.e., it is the case for $\Haar_n$, $\Haar_n(O)$, $\Haar_n^{\otimes k}$, $\Haar_n^{\otimes k}(O)$ and $D_{\delta}^{\otimes nk}(\Haar_n^{\otimes k})$ that we will consider below). This means that even if this hypothesis is restrictive, this theorem can be used to make tight generalization inequalities between many distributions of interest, such as any pair of distributions mentioned above. We will apply this theorem to some of these cases below.

The last hypothesis can be weakened to only requiring that $\mathcal{S}_{\rm train}, \mathcal{S}_{\rm test}$ are diagonalizable in the same basis, and that there exists an eigenvector $\kket{A}$ that has eigenvalues that maximize the ratio and such that $A$ is a unitary matrix.  Moreover, this theorem can easily be generalized to mixed-state inputs, by replacing the $1$ in the numerator and the $1$ in the denominator by the average purities $\bar{\gamma}_{\mathcal{P}_{\rm test}}$ and $\bar{\gamma}_{\mathcal{P}_{\rm train}}$, respectively. See Theorem~\ref{thm:inequalities} in the appendices for a proof of this more general version of Theorem~\ref{thm-main-text:tight-unitary-comparison}.

\subsection{Haar random states}

Theorem~\ref{thm-main-text:general-unitary-comparison} can be specialized to $\Haar_n$, the distribution of $n$-qubit pure states obtained by applying a Haar-random unitary from $U(2^n)$ to $\ket{0}^{\otimes n}$.
This gives a particularly clean comparison between an arbitrary distribution and the maximally uniform ensemble. As the risk only depends on the first and second moments, here and later our results equally apply to 2-designs, e.g., pseudorandom ensembles which agree with Haar up to the second moment~\cite{ambainis2007quantum}. See Corollary~\ref{crl:haar-zero-mean-generalisation} in the Appendices for a proof.

\begin{corol}[Haar to general guarantee]\label{corol-main-text:haar-general-comparison}
Let $\mathcal{P}$ be an arbitrary distribution with bias $\nu$ and expressivity measure $\eta$. Then, for all unitaries $U,V$, the risk on $\mathcal{P}$ is controlled by the risk on Haar-random states as
\begin{equation}\label{eq:haar-to-D-comparison-main-text}
R_{\mathcal{P}} 
\leq
\left(1 + \frac{1}{2^n}\right)
R_{\Haar_n} 
+
\nu .
\end{equation}
If moreover $\eta > 0$, then the converse comparison also holds:
\begin{equation}\label{eq:D-to-haar-comparison-main-text}
R_{\Haar_n} 
\leq
\frac{
R_{\mathcal{P}} 
+
\nu
}{
\eta
}.
\end{equation}
\end{corol}

The first inequality shows that good performance on Haar-random states implies good performance on any distribution $\mathcal{P}$ with sufficiently small bias.
Importantly, it does not depend on the expressivity measure $\eta$ of $\mathcal{P}$ but is only non-trivial if the bias is sufficiently small.
The second inequality gives the converse type of statement.
It quantifies the price paid when one trains using a distribution $\mathcal{P}$ rather than the Haar measure itself.
Here the expressivity measure $\eta$ appears in the denominator, reflecting the fact that a poorly spread distribution may fail to probe states that are relevant for Haar-random inputs.

\subsection{Random real states}\label{sec:randreal}

We next consider ensembles of random real states. We write $\Haar_n(O)$ for the distribution obtained by applying a Haar-random orthogonal matrix $O\in O\left(2^n\right)$ to $\ket{0}^{\otimes n}$. Thus, $\Haar_n(O)$ is supported only on pure states with real amplitudes in the computational basis.

Since real states form only a restricted subset of all pure states, one might expect that learning a unitary on random real inputs provides rather limited information about its action on general complex inputs. Perhaps surprisingly, this is not the case. We find that the risk on random real states controls the risk on any input distribution with sufficiently small bias.
As proven in Corollary~\ref{crl:orthogonal-haar-zero-mean-generalisation} in the Appendices, we obtain the following result.

\begin{corol}[Random real to general guarantee]\label{corol-main-text:random-real}
Let $\mathcal{P}$ be any pure-state distribution on $n$ qubits with bias $\nu$. Then, for all unitaries $U,V$,
\begin{equation}\label{eq:random-real-generalisation-main-text}
R_{\mathcal{P}} 
\leq
\left(1+\frac{1}{2^{n-1}}\right)
R_{\Haar_n(O)} 
+
\nu.
\end{equation}
\end{corol}

This result should be compared with the corresponding guarantee for complex Haar-random training states in Corollary~\ref*{corol-main-text:haar-general-comparison}. The two bounds have essentially the same form, differing only slightly in their multiplicative prefactors. Random real states therefore provide nearly the same out-of-distribution guarantees as fully complex Haar-random states, despite probing a substantially more restricted family of inputs.

In particular, a direct consequence of this result is that
\begin{equation}\label{eq:real-to-complex-haar-main-text}
R_{\Haar_n} 
\leq
\left(1+\frac{1}{2^{n-1}}\right)
R_{\Haar_n(O)} .
\end{equation}
Thus, at the level of risks, learning on random real states is sufficient to guarantee good performance on arbitrary complex Haar-random states. Moreover, the multiplicative overhead approaches one exponentially quickly with the number of qubits. Note that this multiplicative factor can be improved to the tight $1 + 1/(2^n + 1)$ using Theorem~\ref{thm-main-text:tight-unitary-comparison}.

The single-qubit case provides a simple geometric illustration. A real single-qubit state can be written as $\ket{\psi(\theta)}
=
\cos(\theta)\ket{0}
+
\sin(\theta)\ket{1}$.
These states lie on a single orbit of the Bloch sphere, whereas complex pure states span the full Bloch sphere. Nevertheless, Eq.~\eqref{eq:real-to-complex-haar-main-text} implies $R_{\Haar_1} 
\leq
2 R_{\Haar_1(O)}$  (which can be improved to $R_{\Haar_1} \leq \frac{4}{3} R_{\Haar_1(O)}$ using theorem~\ref{thm-main-text:tight-unitary-comparison}). 
Hence, controlling the average error on only this orbit is already sufficient to control the average error over the entire Bloch sphere. This is an average-risk statement rather than a pointwise one; nevertheless, it shows that restricting the learning ensemble to real states loses surprisingly little information about the action of the unitary on complex states. Indeed, if we learn close-to-perfectly on random real states ($R_{\Haar_1(O)} \rightarrow 0 $) we have also learnt close-to-perfectly on random complex states ($R_{\Haar_1} \rightarrow 0 $). Intriguingly the ability to transfer learning from real states to complex states does not carry over to channel learning as discussed in Section~\ref{sec:learningchannels}.

\subsection{Single-qubit distributions}

In general, the risk may depend on both the first and second moments of the input distribution. For a single qubit, however, the first-moment contribution vanishes exactly: the risk depends only on the second moment $\mathcal{S}$ and is completely independent of the bias.

\medskip
To see this, let $\mathcal{P}$ be a pure-state distribution on one qubit, and identify the non-identity Pauli vector $\kket{r(\rho)}$ with the corresponding Bloch vector $\boldsymbol{r}(\rho)\in\mathbb{R}^3$. Up to a global phase, the relative unitary can be written as
\begin{equation}
V^\dagger U
=
\cos(\theta)I
-
i\sin(\theta)\,
\hat{\boldsymbol{n}}\cdot\boldsymbol{\sigma},
\end{equation}
where $\hat{\boldsymbol{n}}\in\mathbb{R}^3$ is a unit vector and $\boldsymbol{\sigma}=(X,Y,Z)$. For an input state with Bloch vector $\boldsymbol{r}(\rho)$,
\begin{equation}
\left|
\operatorname{Tr}
\left(
\rho V^\dagger U
\right)
\right|^2
=
\cos^2(\theta)
+
\sin^2(\theta)
\left(
\hat{\boldsymbol{n}}\cdot\boldsymbol{r}(\rho)
\right)^2.
\end{equation}
Averaging over $\mathcal{P}$ therefore gives
\begin{equation}
R_{\mathcal{P}}
=
\sin^2(\theta)
\left(
1
-
\hat{\boldsymbol{n}}^{\mathsf{T}}
\mathcal{S}
\hat{\boldsymbol{n}}
\right).
\end{equation}
Note that the mean vector $\kket{\mu}$ simply does not appear in this equation. Consequently, two single-qubit distributions with the same second moment give exactly the same risk for every pair of unitaries, even if their biases are different (see Theorem~\ref{thm:same-moment-matrix-same-risk} in the appendices for a more formal proof of this fact).

In particular, any single-qubit distribution can be replaced by an unbiased distribution without changing its risk. To see this, one may sample a state from $\mathcal{P}$ and then, with equal probability, replace it by its antipodal state on the Bloch sphere (e.g. if $|0\rangle$ was drawn, with probability 1/2 it is replaced by $|1\rangle$). This symmetrization sets the mean to zero while leaving the second moment, and therefore the risk, unchanged.

\medskip
The exact form of the risk also gives a particularly clean out-of-distribution comparison. Since
\begin{equation}
\eta_{\mathcal{P}}\sin^2(\theta)
\leq
R_{\mathcal{P}}
\leq
\sin^2(\theta),
\end{equation}
the risk on any sufficiently expressive training distribution controls the risk on every testing distribution. 

\begin{corol}[Single-qubit guarantee]\label{corol-main-text:single-qubit-guarantee}
Let $\mathcal{P}_{\rm train}$ and $\mathcal{P}_{\rm test}$ be arbitrary pure-state distributions on one qubit, with expressivity measures $\eta_{\rm train}$ and $\eta_{\rm test}$, respectively. If $\eta_{\rm train}>0$, then, for all unitaries $U,V$,
\begin{equation}
\eta_{\rm test}
R_{\mathcal{P}_{\rm train}}
\leq
R_{\mathcal{P}_{\rm test}}
\leq
\frac{
R_{\mathcal{P}_{\rm train}}
}{
\eta_{\rm train}
}.
\end{equation}
\end{corol}

Thus, unlike the general multi-qubit comparison, the single-qubit bound is purely multiplicative and contains no bias-dependent correction. Moreover, only the training distribution must have positive expressivity in order to obtain an upper bound on the testing risk.

\medskip
For the Haar distribution on one-qubit pure states,
\begin{equation}
\mathcal{S}_{\rm Haar}
=
\frac{1}{3}I_3.
\end{equation}
Hence
\begin{equation}
R_{\Haar_1}
=
\frac{2}{3}
\sin^2(\theta).
\end{equation}
Combining this with the previous bounds gives the following sharper comparison.

\begin{corol}[Single-qubit Haar guarantee]\label{corol-main-text:single-qubit-haar-guarantee}
Let $\mathcal{P}$ be any pure-state distribution on one qubit with expressivity measure $\eta_{\mathcal{P}}$. Then, for all unitaries $U,V$,
\begin{equation}
\frac{3}{2}
\eta_{\mathcal{P}}
R_{\Haar_1}
\leq
R_{\mathcal{P}}
\leq
\frac{3}{2}
R_{\Haar_1}.
\end{equation}
If $\eta_{\mathcal{P}}>0$, this equivalently implies
\begin{equation}
R_{\Haar_1}
\leq
\frac{2}{3\eta_{\mathcal{P}}}
R_{\mathcal{P}}.
\end{equation}
\end{corol}

Hence, training on Haar-random single-qubit states controls the risk on any single-qubit testing distribution with a universal factor of $3/2$. Conversely, training on an arbitrary distribution $\mathcal{P}$ controls the Haar risk with a factor determined by its expressivity. This factor $3/2$ is also tight (it is saturated whenever $\mathcal{S}_{\mathcal{P}}$ has a zero eigenvalue and the relative unitary rotates about the corresponding Bloch-sphere direction).

\medskip
Consider the distribution
\begin{equation}
\mathcal{P}_+
=
\operatorname{Uniform}
\left\{
\ket{0},
\ket{+},
\ket{y+}
\right\},
\end{equation}
where
\begin{equation}
\ket{y+}
=
\frac{\ket{0}+i\ket{1}}{\sqrt{2}}.
\end{equation}
The corresponding Bloch vectors point along the positive $z$-, $x$-, and $y$-axes. Consequently,
\begin{equation}
\kket{\mu_+}
=
\frac{1}{3}
\begin{pmatrix}
1\\
1\\
1
\end{pmatrix},
\qquad
\mathcal{S}_+
=
\frac{1}{3}I_3.
\end{equation}
The distribution is therefore biased, but its second moment is identical to that of the Haar distribution. It follows that
\begin{equation}
R_{\mathcal{P}_+}
=
R_{\Haar_1}
\end{equation}
for every pair of single-qubit unitaries $U,V$.

This example is stronger than simply replacing the Haar distribution by a state $2$-design. The distribution $\mathcal{P}_+$ is not itself a $2$-design, since its mean is nonzero, yet it gives exactly the same risk as the Haar distribution for the unitary-learning task. Thus, for a single qubit, learning on only three suitably chosen states can reproduce the average learning performance obtained from the full Haar ensemble.

For the uniform ensemble on the two real states
$\{\ket{0},\ket{+}\}$, we instead have
$\mathcal{S}=\operatorname{diag}(1/2,0,1/2)$ and $\eta=1/2$.
The single-qubit Haar guarantee therefore gives
\[
R_{\Haar_1}
\leq \frac{4}{3}
R_{\operatorname{Uniform}\{\ket{0},\ket{+}\}}.
\]
Thus, a small risk on these two real states guarantees a small
Haar risk, and perfect learning on this pair implies perfect
learning on the full Haar ensemble.

\subsection{Random product states}

Product distributions are important because they are the natural way to build many-qubit input ensembles from smaller ones. They include the practically relevant setting where training data consists of independently sampled, unentangled input states and further give us a means of interpolating between the single qubit case and $n$-qubit Haar case by exploring the tensor product of $n/k$ random $k$-qubit states.

A key structural observation is that the moment matrix behaves simply under tensor products. Namely, if $\mathcal{Q}=\mathcal{P}^{\left(1\right)}\otimes\mathcal{P}^{\left(2\right)}$ where $\mathcal{P}^{\left(1\right)}$ and $\mathcal{P}^{\left(2\right)}$ are independent pure-state distributions, then $M(\mathcal{Q})$ and $M(\mathcal{P}^{\left(1\right)})\otimes M(\mathcal{P}^{\left(2\right)})$ are related via a similarity transform (Theorem~\ref{thm:m-product-distribution} in the Appendices).
Thus, once we understand a base distribution $\mathcal{P}$, we can also analyze the product distribution $\mathcal{Q}=\mathcal{P}^{\otimes k}$, where $\mathcal{Q}$ samples $k$ independent pure states from $\mathcal{P}$ and outputs their tensor product.

This gives a direct way to determine how the quantities entering our bounds behave under tensor products. The following corollary collects the resulting expressions for the bias and expressivity of a product distribution (Lemma~\ref{lma:mu-k-tensor-product-scaling} and Corollary~\ref{crl:s-for-n-product-pure} in the Appendices).

\begin{corol}[Tensor-product bias and expressivity scaling]\label{corol-main-text:product-bias-expressivity}
Let $\mathcal{P}$ be a pure-state distribution with bias $\nu$ and non-negative expressivity measure $\eta \geq 0$. Let $\mathcal{Q}=\mathcal{P}^{\otimes k}$, and denote its bias and expressivity measure by $\nu_k$ and $\eta_k$, respectively. Then
\begin{equation}\label{eq:product-bias-scaling-main-text}
\nu_k
=
\sqrt{\left(1+\nu^2\right)^k-1}.
\end{equation}
If $\mathcal{P}$ is unbiased, then $\mathcal{Q}$ is also unbiased and
\begin{equation}\label{eq:product-expressivity-scaling-main-text}
\eta_k=\eta.
\end{equation}
If $\mathcal{P}$ is biased, there is not in general a compact expression for relating  $\eta_k$ and $\eta$.
\end{corol}

The corollary highlights a key difference between biased and unbiased product distributions. If $\nu>0$, then the bias is amplified as the number of tensor factors increases. In particular, in the small-bias regime $0<\nu \ll 1$, there exists a $k\in\Theta\left(1/\nu^2\right)$ for which $\nu_k\geq 1$. Since $\nu_k<1$ is necessary for $\eta_k>0$ (Theorem~\ref{thm:mu-eta-bound} in the Appendices), even a small one-copy bias can eventually prevent the resulting product distribution from satisfying the conditions of our general comparison theorem.
By contrast, tensor products of an unbiased distribution with non-negative expressivity remain unbiased and preserve their expressivity exactly, independently of the number of tensor factors. Consequently, if $\mathcal{P}$ is unbiased and has positive expressivity, then $\mathcal{P}^{\otimes k}$ remains unbiased and positively expressive for arbitrary $k$.
For a biased base distribution, there is generally no corresponding expression for $\eta_k$ in terms of only $\eta$ and $\nu$: its behaviour depends on the fuller moment structure of $\mathcal{P}$.

In addition to this, one can prove that if the second moment $\mathcal{S}$ of an unbiased distribution $\mathcal{P}$ is diagonalizable in the Pauli basis, then so is the second moment of $\mathcal{P}^{\otimes k}$ (see Corollary~\ref{crl:s-for-n-product-mixed} in the appendices). In other words, if a pure-state distribution $\mathcal{P}$ fulfills the hypotheses of Theorem~\ref{thm-main-text:tight-unitary-comparison} as the training or testing distribution, then so does $\mathcal{P}^{\otimes k}$ for the same role. The same observation can be made for Theorem~\ref{thm-main-text:general-unitary-comparison} if $\mathcal{P}$ is unbiased.

\subsection{Locally scrambling distributions}

We now combine the preceding product-distribution results with our general comparison bound (Theorem~\ref{thm-main-text:general-unitary-comparison} with the extra $1 - \lambda_{\min}(\mathcal{S}_{\rm test})$ factor) for Haar-random states. This allows us to compare a broad family of locally scrambling ensembles that interpolate between random product states and globally Haar-random states.
The following corollary (Theorem~\ref{thm:haar-general-comparison} in the Appendices) compares two product Haar ensembles supported on the same total number of qubits.

\begin{corol}[Locally scrambling guarantees]\label{corol-main-text:locally-scrambling-comparison}
Let
$\mathcal{P}_i=\Haar_{n_i}^{\otimes k_i}$
for $i\in\{1,2\}$, where $n_1k_1=n_2k_2$, and define
$R_{n_i,k_i}:=R_{\mathcal{P}_i}$
and
$q_{n,k}:=1-(2^n+1)^{-k}$.
Then, for all unitaries $U,V$,
\begin{equation}\label{eq:locally-scrambling-comparison}
\frac{q_{n_1,1}}{q_{n_2,k_2}}R_{n_2,k_2}
\leq
R_{n_1,k_1}
\leq
\frac{q_{n_1,k_1}}{q_{n_2,1}}R_{n_2,k_2}.
\end{equation}
\end{corol}
This corollary shows that changing how a fixed number of qubits is divided into independently Haar-random blocks changes the unitary-learning risk by at most a factor of $3/2$. The $i$th ensemble consists of $k_i$ independently sampled states on $n_i$ qubits each, so the two ensembles can have very different entanglement structure. If both block sizes grow, the comparison factors approach one and the risks become nearly equal. However, adding more blocks of a fixed small size need not remove the constant overhead. Even so, training on a product of small blocks can control performance on inputs entangled across larger blocks.

A natural special case is the comparison between globally Haar-random states and tensor products of independently Haar-random single-qubit states. In this case, this even coincidentally coincides with the inequality produced by Theorem~\ref{thm-main-text:tight-unitary-comparison}, giving a tight guarantee. 
\begin{corol}[Global versus product Haar guarantee]\label{corol-main-text:haar-1-n-generalisation}
For all unitaries $U,V$,
\begin{equation}\label{eq:haar-1-n-generalisation}
\begin{split}
a_n
R_{\Haar_1^{\otimes n}}
\leq
R_{\Haar_n}
\leq
b_n
R_{\Haar_1^{\otimes n}}
\end{split}
\end{equation}
where
\begin{equation}
a_n:= \frac{3^n}{3^n-1}
\frac{2^n}{2^n+1} \, , \quad b_n:=\frac{3}{2} 
\frac{2^n}{2^n+1} \, .
\end{equation}

This is moreover tight: there exist unitaries $U \neq V$ such that the first inequality holds with equality, and similarly for the second inequality.
\end{corol}

Corollary~\ref{corol-main-text:haar-1-n-generalisation} (see Theorem~\ref{thm:haar-tight-comparison} in the Appendices for a proof) establishes a constant-factor equivalence between product-Haar and global-Haar risk. That is, although the product-state ensemble contains no entangled states, good average performance on this ensemble is both necessary and sufficient to guarantee good average performance on an ensemble dominated by highly entangled states.
The upper comparison factor approaches $3/2$, improving the asymptotic factor of $2$ obtained in Ref.~\cite{caro2022outofdistribution}. Figure~\ref*{fig:risk-graph-non-zero-mean} illustrates that the improved bounds are tight.

\begin{figure}[]
    \centering
    \includegraphics[width=\columnwidth]{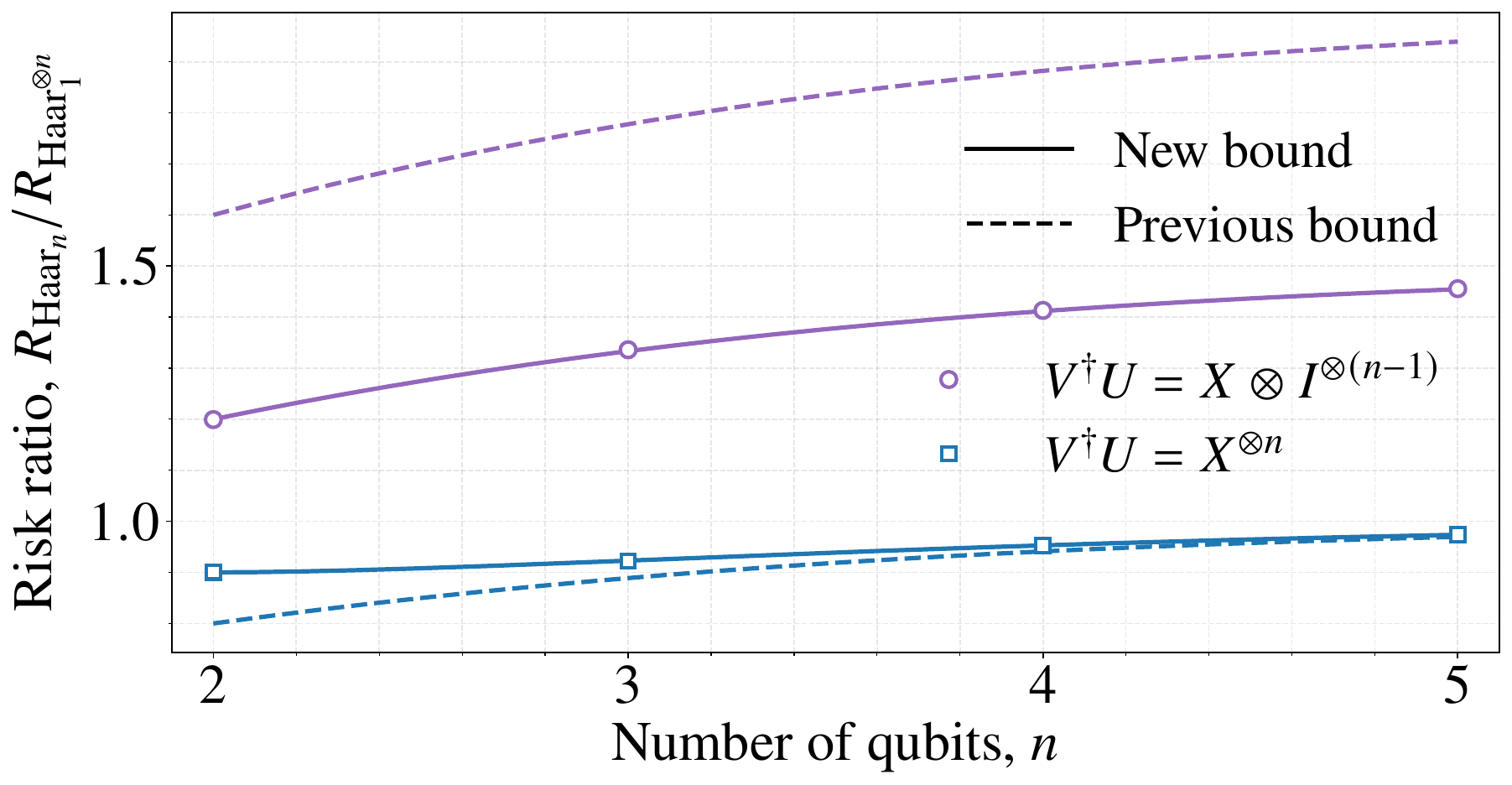}
    \caption{\textbf{Tight product-to-entangled generalization bounds.} Here we study the ratio between $R_{\Haar_n}$ and $R_{\Haar_1^{\otimes n}}$. The solid lines show our theoretical lower bound and upper bound (Eq.~\ref*{eq:haar-1-n-generalisation}) and the dashed lines the old looser bounds of Ref.~\cite{caro2022outofdistribution}. The blue and purple markers are unitaries $U, V$ that saturate the bound (concretely, $V^{\dagger} U = X^{\otimes n}$ and $V^{\dagger} U = X \otimes I^{\otimes n-1}$ respectively).}
    \label{fig:risk-graph-non-zero-mean}
\end{figure}

\subsection{Mixed-state inputs}

Mixed inputs can make two unitaries appear artificially similar, since mixing an input with the maximally mixed state suppresses the distinguishability of their outputs. We first show that this introduces an unavoidable purity penalty in any non-trivial out-of-distribution bound (see Theorem~\ref{thm:purity-penalty} in the Appendices for the proof).

\begin{propos}[Mixing penalty]\label{prop-main-text:purity-penalty}
Let $\mathcal{P}_{\rm train}$ and $\mathcal{P}_{\rm test}$ be distributions over mixed states with average purities $\bar{\gamma}_{\rm train}$ and $\bar{\gamma}_{\rm test}$, respectively. Suppose that $\alpha,\beta\geq 0$ satisfy
\begin{equation}
    R_{\mathcal{P}_{\rm test}}
    \leq
    \alpha
    R_{\mathcal{P}_{\rm train}}
    +
    \beta
\end{equation}
for all unitaries $U,V$. If
\begin{equation}
    \beta
    \leq
    \frac{
    \bar{\gamma}_{\rm test}-2^{-n}
    }{2},
\end{equation}
then
\begin{equation}\label{eq:purity-penalty-main-text}
    \alpha
    \geq
    \frac{
    \bar{\gamma}_{\rm test}-2^{-n}
    }{
    2\bar{\gamma}_{\rm train}
    }.
\end{equation}
\end{propos}

The restriction on $\beta$ ensures that the comparison cannot be satisfied by a trivial additive bound. Indeed, the risk on any distribution is bounded above by its average purity, $ R_{\mathcal{P}_{\rm test}}
    \leq
    \bar{\gamma}_{\rm test}$
so taking $\alpha=0$ and $\beta=\bar{\gamma}_{\rm test}$ always gives a valid but uninformative comparison (Corollary~\ref{crl:trivial-odg-bound-purity} in the Appendices). 
Under this condition, Proposition~\ref{prop-main-text:purity-penalty} gives a necessary condition for out-of-distribution generalization rather than a positive guarantee. If the testing ensemble is bounded away from the maximally mixed state, the multiplicative overhead must be at least proportional to the ratio between the testing and training purities. This asymmetry is natural. Highly mixed states are easier to predict, but consequently provide less information about the difference between two unitaries.

We next show that this purity penalty remains manageable under mild local depolarizing noise. Let
\begin{equation}
    \mathcal{D}_{\delta}(\rho)
    =
    (1-\delta)\rho+\delta\frac{I}{2}
\end{equation}
denote the single-qubit depolarizing channel~\cite{nielsen2000quantum}, where $\delta\in[0,1)$ is the local noise rate. Applying it independently to every qubit, we define
\begin{equation}
\begin{aligned}
\mathcal{P}_{\rm global}^{(\delta)}
&:=\mathcal{D}_{\delta}^{\otimes n}(\Haar_n),\\
\mathcal{P}_{\rm product}^{(\delta)}
&:=\mathcal{D}_{\delta}^{\otimes n}(\Haar_1^{\otimes n}).
\end{aligned}
\end{equation}
That is, we first sample a pure state from the corresponding Haar ensemble and then apply local depolarizing noise. In Appendix~\ref{sec:noiseproofs} (Theorems~\ref{thm:depolarised-haar-tight-comparison} and~\ref{thm:depolarised-product-haar-tight-comparison}) we prove the following tight bounds using the mixed-state generalization of Theorem~\ref{thm-main-text:tight-unitary-comparison}.

\begin{propos}[Noise tolerance]\label{thm-main-text:locally-depolarized-inputs}
Let $q_{\delta}:=(1-\delta)^2$. Then, for all unitaries $U,V$,
\begin{equation}\label{eq:locally-depolarized-inputs-main-text}
\begin{aligned}
R_{\mathcal{P}_{\rm global}^{(0)}}
&\leq
\frac{1}{q_{\delta}}
\left(\frac{4}{1+3q_{\delta}}\right)^{n-1}
R_{\mathcal{P}_{\rm global}^{(\delta)}},\\
R_{\mathcal{P}_{\rm product}^{(0)}}
&\leq
\frac{1}{q_{\delta}}
\left(\frac{2}{1+q_{\delta}}\right)^{n-1}
R_{\mathcal{P}_{\rm product}^{(\delta)}}.
\end{aligned}
\end{equation}

Both bounds are tight: there exist unitaries $U \neq V$ such that they are reached with equality.
\end{propos}

Combining this result with Corollary~\ref{corol-main-text:haar-1-n-generalisation} shows that training on either noisy ensemble controls the risk on either clean Haar ensemble. In particular, if $\delta=O(1/n)$, the resulting multiplicative overhead is independent of $n$. More generally, if $\delta=O(\log(n)/n)$, the overhead is polynomial in $n$. However, for any fixed local noise rate $\delta>0$, the tight comparison factors grow exponentially with the number of qubits, as expected per Proposition~\ref{prop-main-text:purity-penalty}. 
Thus, out-of-distribution generalization remains possible with mixed-state inputs, but its quality deteriorates rapidly as local noise accumulates across the system. 

\section{Channel Learning}\label{sec:learningchannels}

The preceding results show that out-of-distribution generalization for arbitrary unitary processes persists for mixed-state inputs, subject to restrictions on how much mixing can be tolerated. It is therefore natural to ask whether similarly general guarantees, i.e., ones that hold for all unitaries, extend to general quantum channels. We show that they do not. \\

\paragraph*{An exponential separation for general channels.}
We first ask whether performance on globally Haar-random states can control performance on random product states for arbitrary quantum channels. For unitary learning, the corresponding risks are equivalent up to constant factors. The following result shows that this equivalence can fail exponentially for general channels. 
\begin{propos}[Global-to-product separation]\label{thm-main-text:general-channel-no-go}
Any bound of the form
\begin{equation}\label{eq:general-channel-no-go-main-text}
\begin{aligned}
    &R_{\Haar_1^{\otimes n}}
    \left(\mathcal{E}_A,\mathcal{E}_B\right)\leq
    \alpha R_{\Haar_n}
    \left(\mathcal{E}_A,\mathcal{E}_B\right)
    +
    \frac{1}{4}
\end{aligned}
\end{equation}
that holds for all pairs of $n$-qubit efficiently implementable quantum channels must satisfy
\begin{equation}\label{eq:general-channel-alpha-lower-bound-main-text}
    \alpha
    \geq
    \frac{2^n+1}{12}.
\end{equation}
\end{propos}

See Theorem~\ref{thm:general-channel-no-go} in the Appendices for the proof. It uses a pair of simple measure-and-prepare channels that measure the first input qubit and encode the outcome in all $n$ output qubits. For these channels, the global-Haar risk decreases exponentially with $n$, while the product-Haar risk remains constant. Thus, even an additive term of $1/4$ leaves an exponentially large multiplicative overhead. This separation occurs for pure-state input ensembles, showing that the obstruction comes from the generality of the process class rather than from the mixedness of the inputs. \\

\paragraph*{An exponential separation for real to complex ensembles for Pauli channels.}
We next ask whether the obstruction disappears when the process class is restricted to Pauli channels, i.e.\ channels of the form
\[\mathcal{E}(\rho) = \sum_{j=0}^{4^n-1} a_j P_j\rho P_j\]
for some probability vector $\boldsymbol{a}$, and show that it does not. Concretely, we show that Pauli channels with constant support can be indistinguishable on random real product states while remaining separated on random complex product states.  This result directly contrasts with unitary learning, where performance on random real states controls performance on random complex states with a constant multiplicative overhead. See Theorem~\ref{thm:pauli-channel-real-complex-no-go} in the Appendices for a proof.

\begin{propos}[Real-to-complex Pauli-channel separation]\label{thm-main-text:pauli-channel-no-go}
Any bound of the form
\begin{equation}\label{eq:pauli-channel-no-go-main-text}
\begin{aligned}
    &R_{\Haar_1^{\otimes n}}
    \left(\mathcal{E}_A,\mathcal{E}_B\right) \leq
    \alpha
    R_{\Haar_1^{\otimes n}(O)}
    \left(\mathcal{E}_A,\mathcal{E}_B\right)
    +
    \beta
\end{aligned}
\end{equation}
that holds for all efficiently implementable Pauli channels must satisfy
\begin{equation}
    \beta
    \geq
    \frac{1}{3},
\end{equation}
independently of $\alpha$.
\end{propos}
 
The counterexample in Proposition~\ref{thm-main-text:pauli-channel-no-go} works because a channel can erase the information carried by real inputs. A real single-qubit state has no $Y$ component. Complete dephasing in the $Y$ basis removes its $X$ and $Z$ components, leaving $I/2$ regardless of which real state was supplied. The channels in the construction therefore agree on real inputs. Complex states can have a $Y$ component that survives dephasing, allowing a constant average error on complex Haar-random inputs.

A unitary cannot erase information in this way. On a single qubit, unitary conjugation rotates the Bloch sphere. If two unitary channels agree on real states, they agree on the $X$ and $Z$ directions, which also fixes their action on $Y$. This explains why an exact separation is impossible, but it does not make the approximate result obvious. Real single-qubit inputs never directly probe the $Y$ direction, and exact agreement alone does not tell us what happens when the average error on real inputs is merely small. Corollary~\ref{corol-main-text:random-real} in Section~\ref{sec:randreal} showed that the complex-Haar risk is nevertheless controlled by the real-Haar risk, with an overhead approaching one as $n$ grows.

\medskip

\paragraph*{Out-of-distribution guarantees for sparse Pauli channels.}

The preceding results show that out-of-distribution generalization can fail both for general channels and within the class of Pauli channels. However, useful guarantees can be recovered when the target and model Pauli channels are sparse (or, more generally, when their coefficient difference is sparse) and training is performed on either Haar-random product states or globally Haar-random states. See Theorems~\ref{thm:sparse-pauli-global-haar} and~\ref{thm:pauli-channel-bound-any-distrib-from-product-haar} in the Appendices for their proofs.
\begin{propos}[Sparse Pauli-channel guarantee]
\label{thm-main-text:sparse-pauli-channel-generalization}
Let
\begin{equation}
\begin{aligned}
    \mathcal{E}_A(\rho)
    &=
    \sum_{j=0}^{4^n-1}
    a_j P_j\rho P_j,
    \\
    \mathcal{E}_B(\rho)
    &=
    \sum_{j=0}^{4^n-1}
    b_j P_j\rho P_j
\end{aligned}
\end{equation}
be two $n$-qubit Pauli channels, where $\boldsymbol{a}$ and $\boldsymbol{b}$ are probability vectors. Assume that their combined support contains at most $s=\operatorname{poly}(n)$ Pauli strings. Let $\mathcal{P}$ be an arbitrary distribution over $n$-qubit states, possibly mixed. Then
\begin{equation}\label{eq:sparse-pauli-channel-generalization-main-text}
\begin{aligned}
     R_{\mathcal{P}}
    \left(\mathcal{E}_A,\mathcal{E}_B\right) &\leq
    s\left(1+2^{-n}\right)
    R_{\Haar_n}
    \left(\mathcal{E}_A,\mathcal{E}_B\right),
    \\
    R_{\mathcal{P}}
    \left(\mathcal{E}_A,\mathcal{E}_B\right) &\leq
    s^2
    R_{\Haar_1^{\otimes n}}
    \left(\mathcal{E}_A,\mathcal{E}_B\right).
\end{aligned}
\end{equation}
Since these bounds hold for \emph{every} distribution $\mathcal{P}$, they also provide worst-case guarantees: in both inequalities, $R_{\mathcal{P}}(\mathcal{E}_A,\mathcal{E}_B)$ can be replaced by
\[
    \max_{\rho}\frac{1}{2}
    \left\|\mathcal{E}_A(\rho)-\mathcal{E}_B(\rho)\right\|_2^2,
\]
where the maximum is over all $n$-qubit states.
\end{propos}

The first inequality shows that training on globally Haar-random states controls performance on an arbitrary testing distribution with an overhead linear in the sparsity. The second gives the corresponding guarantee for product-Haar training, with a quadratic overhead. Notably, neither result requires the testing distribution to be unbiased, expressive, or supported on pure states. Thus, whenever the Pauli channels are polynomially sparse, both training ensembles give polynomial average and worst case out-of-distribution guarantees. These results show that useful generalization beyond unitary learning can be recovered by imposing structure on the processes being compared.

\section{Discussion}
Whether we can generalize from one input ensemble to another depends on the unitarity of the processes we aim to learn. We see this when comparing global- and product-Haar states. For unitaries, their risks control one another up to constant factors, whereas for general channels we construct a pair whose risk is exponentially small on global-Haar inputs but remains constant on product inputs. The real-to-complex comparison sharpens this contrast. Globally Haar-random real states control the complex-Haar risk for unitary learning with an overhead tending to one, yet two-term Pauli channels can agree on real product states and have risk $1/3$ on complex product states. Thus, bias and expressivity explain how an ensemble probes differences between unitaries, but cannot alone tell us whether it reveals differences between more general processes. Moreover, mixedness is not what separates the positive and negative results. Both channel counterexamples use pure inputs, while we still obtain general unitary comparisons when Haar inputs are locally depolarized. 

An instructive comparison is provided by Ref.~\cite{huang2023learning}, which proves a positive out-of-distribution result for arbitrary quantum channels. Training on random Pauli product states suffices to predict bounded-degree observable expectation values on any locally flat testing distribution, including distributions of entangled states. This does not conflict with our channel no-go results. Their learner predicts selected properties of the output using a particular informative training ensemble, whereas our risk compares the full output states across arbitrary training and testing distributions. Thus, restricted properties of an arbitrary channel may remain predictable out of distribution even when broad out-of-distribution control of the channel itself is impossible.

It is perhaps worth highlighting that, with the exception of Proposition~\ref{thm-main-text:sparse-pauli-channel-generalization}, the guarantees in this work concern average risk and should not be interpreted pointwise. Small $R_{\mathcal{P}}$ implies only that inputs with large loss carry little total probability, not that the learned process is accurate everywhere. This matters because any unbiased pure-state distribution with finite support contains at least $2^n$ states (Theorem~\ref{thm:bound-support-size-mu-eta}). A model may therefore fail on exponentially many individually unlikely states while retaining small average risk. Obtaining pointwise guarantees beyond sparse Pauli channels and identifying broader channel classes that permit polynomial transfer remain important open problems.

While we have focussed on out-of-distribution generalization as a problem of fundamental conceptual interest, these bounds may also have practical applications in terms of reducing the cost of learning and compiling quantum dynamics. Product-to-Haar guarantees allow a global unitary objective to be replaced by a loss evaluated on product states~\cite{caro2022outofdistribution,gibbs2022dynamical}. For short-time local dynamics, these states can be handled efficiently using tensor networks, enabling classical compilation in one-dimensional and quasi-one-dimensional systems~\cite{zhang2024scalable, kukliansky2024leveraging, gibbs2024deep}. Combined with Pauli propagation, the same principle has been applied to large two-dimensional systems~\cite{danna2025circuit}. Our framework may similarly enable the development of efficient methods for circuit compilation and noise learning by identifying training ensembles that are easy to simulate or prepare, yet sufficiently informative to control performance on the target ensemble.

\begin{acknowledgments} 
\vspace{-10pt}
JF acknowledges Daniel Stilck França, Sacha Lerch and Laura Lewis for helpful discussions.
AA and ZH acknowledge support from the Sandoz Family Foundation-Monique de Meuron program for Academic Promotion.
AA acknowledges support from the European Union’s Horizon Europe research and innovation programme under the Marie Skłodowska-Curie grant agreement No. 101282902 (PROQSY).
\end{acknowledgments}

\clearpage

\section*{AI Statement}
Most results (all on learning unitaries from pure states) were derived and written in the AI-dark ages of 2024, when a genuinely useful research contribution from AI was an achievement. We then sat on it for two years, partly through routine time-management incompetence and partly because we struggled to extend to channels. In the past month, ChatGPT helped us find no-go results to rule out the hoped-for general channel extension. It also helped with the final writing sprint.

\bibliography{quantum}
\makeatletter
\newcommand{\supplementtableofcontents}{\section*{Contents}\@starttoc{stoc}}
\makeatother

\appendix
\onecolumngrid

\begin{center}
    {\Large\bfseries Appendix\par}
\end{center}

\vspace{1em}
\setcounter{tocdepth}{0}
\supplementtableofcontents
\clearpage

\newpage

\numberwithin{equation}{\eqAndThmNumberedWithin}

\begin{table}[t]
    \centering
    \renewcommand{\arraystretch}{2}
    \setlength{\tabcolsep}{20pt}

    \begin{tabular}{|c|c|c|}
        \hline
        \textbf{Symbol} & \textbf{Name} & \textbf{Definition} \\
        \hline

        $\displaystyle \kket{A}$, $\displaystyle \bbra{A}$ & Vectorised notation & \ref{def:vectorised-notation} \\

        $\displaystyle R_D\left(U, V\right) = \frac{1}{2}\exval_{\rho \followsdistr D}\left[\left\|U \rho U^{\dagger} - V \rho V^{\dagger}\right\|_2^2\right]$ & Risk & \ref{def:risk} \\

        $\displaystyle \mathbb{P}_n = \left\{I, X, Y, Z\right\}^{\otimes n}$ & Pauli set & \ref{def:pauli-set} \\

        $\displaystyle \widetilde{P} = \frac{1}{\sqrt{2^n}} P$ & Normalised Pauli & \ref{def:normalisation-factors} \\

        $\displaystyle \Pi \ket{y} = \ket{\pi\left(y\right)} = \ket{y} - \bbraket{\widetilde{I}^{\otimes n}}{y} \kket{\widetilde{I}^{\otimes n}}$ & $I$-truncation for vectors & \ref{def:truncation-for-vectors} \\

        $\displaystyle \pi\left(A\right) = A - \frac{I^{\otimes n}}{2^n} \Tr\left(A\right)$ & $I$-truncation for operators & \ref{def:truncation-for-operators} \\

        $\displaystyle \bar{\gamma} = \exval_{\rho \followsdistr D}\left[\Tr\left(\rho^2\right)\right]$ & Average purity & \ref{def:average-purity} \\

        $\displaystyle \mu = \sqrt{2^n} \pi\left(\exval_{\rho \followsdistr D}\left[\rho^2\right]\right)$ & Mean & \ref{def:mean} \\

        $\displaystyle \nu = \|\mu\|_2$ & Bias & \ref{def:bias} \\

        $\displaystyle M = 2^n \exval_{\rho \followsdistr D}\left[\rho \otimes \rho^T\right]$ & Moment matrix & \ref{def:moment-matrix} \\

        $\displaystyle M_{\mathrm{core}} = \bar{\gamma} \oplus \mathcal{S}$, $M_{\mathrm{mean}}$, $\displaystyle \mathcal{S} = \Pi M \Pi$ & Moment decomposition & \ref{def:moment-decomposition} \\

        $\displaystyle M_{mean}\left(\ket{m}\right)$ & Generalised mean matrix & \ref{def:generalised-mean-matrix} \\

        $\displaystyle \eta = \bar{\gamma} - \lambda_{max}\left(\mathcal{S}\right)$ & Expressivity measure & \ref{def:eta}  \\

        $\displaystyle \ket{w\left(U, V\right)} = \frac{1}{\sqrt{2^n}} \kket{V^{\dagger} U}$ & $w$ vector & \ref{def:w-vector}  \\

        $\displaystyle r_D\left(\ket{x}\right) = \bar{\gamma} - \bra{x} M \ket{x}$ & Pseudo-risk & \ref{def:pseudo-risk}  \\
        
        \hline
    \end{tabular}
    \caption{Notations used throughout the paper.}
    \label{tab:notation}
\end{table}

\suppsection{Preliminaries}\label{sec:preliminaries}

We introduce two standard tools from linear algebra used to prove our risk bounds. We use vectorization to write the risk as a quadratic form, and we employ the Rayleigh quotient to bound this quadratic form through its eigenvalues. We collect the required identities here.

\begin{definition}[Vectorised notation]\label{def:vectorised-notation}
    Similarly to Ref.~\cite{mele2023introduction}, we consider the linear operator $\text{vec}: \mathbb{C}^{2^n \times 2^n} \mapsto \mathbb{C}^{2^n} \otimes \mathbb{C}^{2^n} \simeq \mathbb{C}^{4^n}$ defined by: 
    \[\text{vec}\left(\ket{i}\bra{j}\right) = \ket{i} \otimes \ket{j}.\]
    
    We then note: 
    \[\kket{A} = \text{vec}\left(A\right) = \text{vec}\left(\sum_{i, j} \bra{i} A \ket{j} \ket{i}\bra{j}\right) = \sum_{i, j} \bra{i} A \ket{j} \ket{i} \otimes \ket{j}.\]
    
    We also note $\bbra{A} = \kket{A}^{\dagger}$. This allows us to make a bijection between operators and vectors, so that we can exploit operators as if they were column vectors.
    
    As known in the literature, this has the following properties:
    \begin{enumerate}
    \item $\kket{\alpha A + \beta B} = \alpha \kket{A} + \beta \kket{B}$ for any $\alpha, \beta \in \mathbb{C}$.
    \item $\bbrakket{A}{B} = \Tr\left(A^{\dagger} B\right)$.
    \item If $\left\{T_i\right\}_{i \in \mathcal{I}}$ is a Hilbert-Schmidt orthonormal basis, i.e.\ $\Tr\left(T_i^{\dagger} T_j\right) = \delta_{ij}$, then $\left\{\kket{T_i}\right\}_{i \in \mathcal{I}}$ is an orthonormal basis for $\mathbb{C}^{2^n} \otimes \mathbb{C}^{2^n}$.
    \item $\kket{A} = A \otimes I \kket{I} = I \otimes A^T \kket{I}.$
    \item $\kket{ABC} = A \otimes C^T \kket{B}$.
    \item $\bbra{ABC} = \bbra{B} A^{\dagger} \otimes C^*.$
    \end{enumerate}
    
    For completeness, we will show those equalities.
\end{definition}
\begin{proof}[Proof (1)]
    vec$(\cdot)$ is a linear operator by definition.
\end{proof}
\begin{proof}[Proof (2)]
    By definition: 
    \begin{equation}\begin{aligned}\bbrakket{A}{B} &= \sum_{i,j,k,\ell} \left(\bra{i}A\ket{j}\right)^* \left(\bra{i} \otimes \bra{j}\right) \bra{k}B \ket{\ell} \left(\ket{k} \otimes \ket{\ell}\right) \\ &= \sum_{i,j,k,\ell} \bra{j} A^{\dagger} \ket{i} \bra{k}B \ket{\ell}\ \delta_{ik} \delta_{j\ell} \\ &= \sum_{i, j} \bra{j}A^{\dagger} \ket{i} \bra{i} B \ket{j} \\ &= \sum_{j} \bra{j} A^{\dagger} B \ket{j} \\ &= \sum_{j} \left(A^{\dagger}B\right)_{jj} \\ &= \Tr\left(A^{\dagger} B\right).\end{aligned}\end{equation}
\end{proof}
\begin{proof}[Proof (3)]
    From (2), we see that $\left\{T_i\right\}$ being a Hilbert-Schmidt orthonormal basis means that $\kket{T_i}$ are orthonormal. However, since we have $4^n$ orthonormal vectors, they necessarily form a basis.
\end{proof}
\begin{proof}[Proof (4)]
    We have: 
    \[\kket{A} = \sum_{i, j} \bra{i} A \ket{j} \ket{i} \otimes \ket{j} = \sum_{i, j} \ket{i} \bra{i} A \ket{j} \otimes \ket{j} = \sum_{j} A \ket{j} \otimes \ket{j},\]
    showing indeed that $\kket{A} = A \otimes I \kket{I}$. Similarly: 
    \[\kket{A} = \sum_{i, j} \bra{i} A \ket{j} \ket{i} \otimes \ket{j} = \sum_{i, j} \ket{i} \otimes \ket{j} \bra{j} A^T \ket{i} = \sum_{i} \ket{i} \otimes A^T \ket{i},\]
    implying that $\kket{A} = I \otimes A^T \kket{I}$, as required.
\end{proof}
\begin{proof}[Proof (5)]
    We notice that, by (4):
    \[C \otimes I \kket{I} = \kket{C} = I \otimes C^T \kket{I}.\]
    
    This means that, still by (4): 
    \[\kket{ABC} = \left(A \otimes I\right)\left(B \otimes I\right) \left(C \otimes I\right) \kket{I} = \left(A \otimes C^T\right) \left(B \otimes I\right) \kket{I},\]
    indeed implying that $\kket{ABC} = A \otimes C^T \kket{B}$.
\end{proof}
\begin{proof}[Proof (6)]
    By definition and by (5): 
    \[\bbra{ABC} = \kket{ABC}^{\dagger} = \left(A \otimes C^T \kket{B}\right)^{\dagger} = \bbra{B} A^{\dagger} \otimes C^*.\]
\end{proof}

The Rayleigh quotient allows us to bound a quadratic form without choosing a particular input vector. Its smallest and largest possible values are determined by the extremal eigenvalues of the matrix.

\begin{definition}[Rayleigh quotient]\label{def:rayleigh-quotient}
    Let $M$ be some Hermitian matrix, and $\ket{x}$ be some nonzero vector. 
    
    We recall the definition of a Rayleigh quotient: 
    \[q\left(M, \ket{x}\right) = \frac{\bra{x} M \ket{x}}{\braket{x}{x}}.\]
    
    When the context is clear, we will note $q\left(M, \ket{x}\right) = q$.
\end{definition}
\begin{property*}
    The main property that we will exploit is that, for all $\ket{x} \neq 0$: 
    \[\lambda_{min}\left(M\right) \leq \frac{\bra{x} M \ket{x}}{\braket{x}{x}} \leq \lambda_{max}\left(M\right),\]
    with equality when $\ket{x}$ is an eigenket associated to $\lambda_{min}\left(M\right)$ or $\lambda_{max}\left(M\right)$.
\end{property*}

\suppsection{General framework}\label{sec:general-framework}

We now introduce our framework for unitary learning, which associates a moment matrix with each distribution. Separating out the identity component gives the out-of-distribution generalization bounds in Theorem~\ref{thm:inequalities}. These bounds have two parts: a multiplicative factor controlled by the spectrum, and an additive term that depends on the means of the distributions.

\suppsubsection{Definitions and properties}

\subsubsection{General definitions}

\begin{definition}[Average purity]\label{def:average-purity}
    Let $D$ be an arbitrary mixed-state distribution. Its average purity, written $\bar{\gamma}$, is defined as: 
    \[\bar{\gamma} = \exval_{\rho \followsdistr D}\left[\Tr\left(\rho^2\right)\right].\]
\end{definition}

\begin{definition}[Risk]\label{def:risk}
    Let $D$ be an arbitrary quantum state distribution, and $U, V$ be unitary matrices.
    
    We define the risk (or error) to be: 
    \[R_D\left(U, V\right) = \frac{1}{2}\exval_{\rho \followsdistr D}\left(\left\|U \rho U^{\dagger} - V \rho V^{\dagger}\right\|_2^2\right).\]
\end{definition}

\begin{theorem}[Risk equivalent formulation]\label{thm:risk-equivalent-formulation}
    Let $D$ be a distribution over mixed states of average purity $\bar{\gamma}$. We have that, for any unitary matrices $U, V$:
    \[R_D\left(U, V\right) = \bar{\gamma} - \exval\left(\Tr\left(U \rho U^{\dagger} V \rho V^{\dagger}\right)\right).\]
    
    Our definition of risk is thus equivalent to reference~\cite{caro2022outofdistribution}'s, for pure states.
\end{theorem}
\begin{remark*}
    In particular, this proves that for pure-state distributions, inserting $\rho = \ket{\psi}\bra{\psi}$:
    \[R_D\left(U, V\right) = 1 - \exval_{\ket{\psi} \followsdistr D}\left(\left|\bra{\psi} U^{\dagger} V \ket{\psi}\right|^2\right).\]
\end{remark*}
\begin{proof}[Proof]
    This is a direct equality: 
    \begin{equation}\begin{aligned}R_D\left(U, V\right) &= \frac{1}{2} \exval\left(\left\|U \rho U^{\dagger} - V \rho V^{\dagger}\right\|_2^2\right) \\ &= \frac{1}{2} \exval\left(\Tr\left(\left(U \rho U^{\dagger} - V \rho V^{\dagger}\right)^{\dagger} \left(U \rho U^{\dagger} - V \rho V^{\dagger}\right)\right)\right) \\ &= \frac{1}{2} \exval\left(\Tr\left(\left(U \rho U^{\dagger} - V \rho V^{\dagger}\right)^2\right)\right) \\ &= \frac{1}{2} \exval\left(\Tr\left(U \rho^2 U^{\dagger}\right) + \Tr\left(V \rho^2 V^{\dagger}\right) - 2\Tr\left(U \rho U^{\dagger} V \rho V^{\dagger}\right)\right)\end{aligned}\end{equation}
    
    We now exploit that $U$ and $V$ are unitary, that $\exval\left[\Tr\left(\rho^2\right)\right] = \bar{\gamma}$ by definition of average purity:
    \[R_D\left(U, V\right) = \frac{1}{2} \cdot 2 \exval\left(\Tr\left(\rho^2\right)\right) - \frac{1}{2} \cdot  2 \exval\left(\Tr\left(U \rho U^{\dagger} V \rho V^{\dagger}\right)\right) = \bar{\gamma} - \exval\left(\Tr\left(U \rho U^{\dagger} V \rho V^{\dagger}\right)\right) .\]
\end{proof}

\begin{definition}[Pauli set]\label{def:pauli-set}
    We note $\mathbb{P}_n$ to be the $n$-qubit Pauli set:
    \[\mathbb{P}_n = \left\{I, X, Y, Z\right\}^{\otimes n}.\]
\end{definition}

\begin{definition}[Normalisation factors]\label{def:normalisation-factors}
    Let $P \in \mathbb{P}_n$ be arbitrary. We note $\widetilde{P}$ to be a normalised Pauli, i.e.
    \[\widetilde{P} = \frac{P}{\sqrt{2^n}}.\]
\end{definition}
\begin{example*}
    For instance, $\widetilde{I}^{\otimes n} = I^{\otimes n} / \sqrt{2^n}$.
\end{example*}
\begin{remark*}
    We say those Paulis are normalised since $\Tr\left(\widetilde{P}^2\right) = 1$. This allows us to define the following orthonormal basis for $\mathbb{C}^{4^n}$, which we will take as our canonical basis: 
    \[\kket{\widetilde{P}}, \mathspace P \in \mathbb{P}_n.\]
\end{remark*}

\subsubsection{$I$-truncation}

To separate the identity contribution from the rest of the risk, we project onto the traceless subspace. We give two equivalent ways to do this, acting either on operators or on their vectorized form.

\begin{definition}[$I$-truncation for vectors]\label{def:truncation-for-vectors}
    We define the operator $\Pi \in \mathbb{C}^{4^n \times 4^n}$ as: 
    \[\Pi = I^{\otimes n} \otimes I^{\otimes n} - \kket{\widetilde{I}^{\otimes n}}\bbra{\widetilde{I}^{\otimes n}}.\]

    This way, given any $\ket{y} \in \mathbb{C}^{4^n}$, we can define its $I$-truncation as: 
    \[\Pi \ket{y} = \sum_{P \neq I^{\otimes n}} \bbraket{\widetilde{P}}{y} \kket{\widetilde{P}} = \ket{y} - \bbraket{\widetilde{I}^{\otimes n}}{y} \kket{\widetilde{I}^{\otimes n}} \in \mathbb{C}^{4^n}.\]

    We may moreover write $\ket{\pi\left(y\right)} := \Pi \ket{y}$.
\end{definition}
\begin{remark*}
    Note that the way $\Pi \ket{y}$ is defined, it lives in the linear space defined by: 
    \[\mathcal{H} = \left\{\ket{z} \in \mathbb{C}^{4^n} \suchthat \bbraket{\widetilde{I}^{\otimes n}}{z} = 0\right\}.\]

    Now, this space is isomorphic to $\mathbb{C}^{4^n - 1}$, so we can equivalently consider $\Pi \ket{y} \in \mathbb{C}^{4^n - 1}$. Thanks to the isomorphism, we will alternatively consider the case that's simplest to us.
    
    Moreover, note that $\Pi \ket{y}$ is the projection of $\ket{y}$ onto $\mathcal{H}$, hence the use of the character $\Pi$.
\end{remark*}
\begin{remark}
    We wrote $\Pi_{ni} := \Pi$ and $\ket{r(y)} := \sqrt{2^n} \ket{\pi(y)}$ in the main text.
\end{remark}

\begin{definition}[$I$-truncation for operators]\label{def:truncation-for-operators}
    Let $A \in \mathbb{C}^{2^n \times 2^n}$ be an arbitrary operator. We define its $I$-truncation as: 
    \[\pi\left(A\right) = A - \frac{I^{\otimes n}}{2^n} \Tr\left(A\right) .\]
\end{definition}
\begin{remark*}
    $\pi$ is a linear super-operator.
\end{remark*}
\begin{remark*}
    Considering some arbitrary operator $A \in \mathbb{C}^{2^n \times 2^n}$, the vectorisation of $\pi\left(A\right)$ is: 
    \[\begin{split}
        \kket{\pi\left(A\right)} & = \kket{A - \frac{I^{\otimes n}}{2^n} \Tr\left(A\right)} = \kket{A} - \frac{1}{2^n} \Tr\left(A\right) \kket{I^{\otimes n}} = \kket{A} - \Tr\left(\widetilde{I}^{\otimes n} A\right) \kket{\widetilde{I}^{\otimes n}} \\
         & = \kket{A} - \bbrakket{\widetilde{I}^{\otimes n}}{A} \kket{\widetilde{I}^{\otimes n}} = \Pi \kket{A}.
    \end{split}\]
    
    In other words, the $I$-truncation of an operator and of a vector are exactly the same thing, up to vectorisation. More precisely, operator $\Pi$ is the vectorisation of the super-operator $\pi$. This also means that the notation $\kket{\pi\left(A\right)} = \Pi \kket{A}$ and $\ket{\pi\left(y\right)} := \Pi \ket{y}$ are not ambiguous (in case the reader is confused by how these definitions are not circular we stress that the former expression uses $\kket{\dots}$ whereas the latter $\ket{\dots}$).
\end{remark*}

\subsubsection{Mean and bias}

For mixed-state distributions, our bounds involve the average of $\rho^2$, rather than the average state alone. This quantity determines how the identity and traceless parts of the moment matrix are coupled, and the bias measures the strength of that coupling.

\begin{definition}[Mean]\label{def:mean}
    Let $D$ be an $n$-qubit distribution over mixed states. We define its mean (or first moment) $\mu \in \mathbb{C}^{2^n \times 2^n}$ as: 
    \[\mu = \sqrt{2^n} \pi\left(\exval_{\rho \followsdistr D}\left[\rho^2\right]\right).\]

    We will mostly work with the vectorised version, $\kket{\mu}$.
\end{definition}
\begin{remark*}
    Note that $\left(\mu, \bar{\gamma}\right)$ is in bijection with $\exval_{\rho \followsdistr D}\left[\rho^2\right]$, i.e.~they both contain the same amount of information about the distribution. Indeed, since the expected value is linear and hence commutes with $\pi$:
    \begin{equation}\label{eq:equivalent-expression-mu}
        \mu = \sqrt{2^n} \exval_{\rho \followsdistr D}\left[\pi\left(\rho^2\right)\right] = \sqrt{2^n} \exval_{\rho \followsdistr D}\left[\rho^2 - \frac{I^{\otimes n}}{2^n} \cdot \Tr\left(\rho^2\right)\right] = \sqrt{2^n}\left(\exval_{\rho \followsdistr D}\left[\rho^2\right] - \bar{\gamma} \frac{I^{\otimes n}}{2^n}\right).
    \end{equation}
\end{remark*}
\begin{remark*}
    For pure-state distributions, we have $\rho^2 = \rho$. Therefore, we can define $\bar{\rho} = \exval_{\rho \followsdistr D}\left[\rho\right]$ to get: 
    \[\mu = \sqrt{2^n} \pi\left(\bar{\rho}\right).\]

    This is a justification of the name ``mean'' for this object. Now, this is not the only case this identity holds, as shown by the following lemma.
\end{remark*}

\begin{lemma}[Mean for $n = 1$]\label{lma:mean-for-single-qubit}
    Let $D$ be a mixed-state distribution of mean $\mu$. If it is over $n = 1$ qubit states then, leaving $\bar{\rho} = \exval_{\rho \followsdistr D}\left[\rho\right]$, we have 
    \[\mu = \sqrt{2^n} \pi\left(\bar{\rho}\right).\]
\end{lemma}
\begin{remark*}
    This shows that, when $n = 1$, $\kket{\mu} \in \mathbb{R}^3$ is the mean position in the Bloch ball. Indeed, for $P \in \left\{X, Y, Z\right\}$: 
    \[\bbrakket{\widetilde{P}}{\mu} = \sqrt{2^n} \bbrakket{\widetilde{P}}{\bar{\rho}} = \bbrakket{P}{\bar{\rho}} = \exval_{\rho \followsdistr D}\left[\bbrakket{P}{\rho}\right] = \exval_{\rho \followsdistr D}\left[\Tr\left(P \rho\right)\right].\]

    This is another justification of the name ``mean'' for this object.
\end{remark*}
\begin{proof}
    Since $n = 1$, we can write any state $\rho$ in its Bloch form: 
    \[\rho = \frac{1}{2}\left(I + x X + y Y + z Z\right).\]
    
    Now, note that by anti-commutation properties, $xy XY + yx YX = 0$. But then, its square is: 
    \begin{equation}\label{eq:single-qubit-rho-squared}\begin{aligned}\rho^2 &= \frac{1}{4}\left(I + x X + y Y + z Z \right)\left(I + x X + y Y + z Z \right) \\ &= \frac{1}{4} \left(I^2 + 2 I \left(x X + y Y + z Z\right) + \left(x X\right)^2 + \left(y Y\right)^2 + \left(z Z\right)^2\right) \\ &= \frac{1}{4}\left(I \left(1 + x^2 + y^2 + z^2\right) + 2 x X + 2 y Y + 2 z Z\right).\end{aligned}\end{equation}
    
    Therefore, applying an $I$-truncation: 
    \[\pi\left(\rho^2\right) = \frac{1}{2}\left(x X + y Y + z Z\right) = \pi\left(\rho\right).\]

    This immediately gives our result.
\end{proof}

\begin{definition}[Bias]\label{def:bias}
    Let $D$ be an arbitrary mixed-state distribution of mean $\mu$. Its bias is defined as $\left\|\mu\right\|_2$, where $\left\|.\right\|_2$ is the Schatten 2-norm. We may also write $\nu = \left\|\mu\right\|_2$.
\end{definition}
\begin{remark*}
     Recall that $\left\|\mu\right\|_2 = \left\|\kket{\mu}\right\|_2$ by properties of the vectorised notation (definition~\ref{def:vectorised-notation}), i.e.~analysing the operator $\mu$ or the vector $\kket{\mu}$ is equivalent for the bias.
\end{remark*}

\begin{lemma}\label{lma:mu-properties}
    Let $D$ be an $n$-qubit distribution over mixed states, of average purity $\bar{\gamma}$ and mean $\kket{\mu}$. It has the following properties:
    \begin{enumerate}
    \item In the Pauli basis, $\kket{\mu} \in \mathbb{R}^{4^n}$ is a real vector. In other words, for any Pauli $P \in \mathbb{P}_n$: 
    \[\bbrakket{\widetilde{P}}{\mu} \in \mathbb{R}.\]
    \item $\left\|\mu\right\|_2 \leq \bar{\gamma}\sqrt{2^n - 1} \leq \sqrt{2^n - 1}$.
    \end{enumerate}
\end{lemma}
\begin{proof}[Proof 1]
    Note that this claim is trivial for $P = I^{\otimes n}$ since $\bbrakket{\widetilde{I}^{\otimes n}}{\mu} = 0$ by the definition of $I$-truncation. For any other Pauli $P$:
    \[\bbrakket{\widetilde{P}}{\mu} = \sqrt{2^n} \exval_{\rho \followsdistr D}\left[\bbrakket{\widetilde{P}}{\rho^2}\right].\]

    This is real since $\rho^2$ is a Hermitian matrix. More formally, we prove that $\bbrakket{\widetilde{P}}{\rho^2}$ is equal to its complex conjugate:
    \[\bbrakket{\widetilde{P}}{\rho^2}^* = \bbrakket{\rho^2}{\widetilde{P}} = \Tr\left(\left(\rho^{\dagger}\right)^2 \widetilde{P}\right) = \Tr\left(\widetilde{P} \rho^2\right) = \bbrakket{\widetilde{P}}{\rho^2}\]
\end{proof}
\begin{proof}[Proof 2]

To bound the bias, we separate the identity contribution to the squared norm of the average of $\rho^2$. We then control the remaining norm using positivity via its trace, which is the average purity.

    By definition (more precisely, by equation~\ref{eq:equivalent-expression-mu}), we have: 
    \[\mu = \sqrt{2^n} \left(\exval_{\rho \followsdistr D}\left[\rho^2\right] - \bar{\gamma} \frac{I^{\otimes n}}{2^n}\right).\]

    Therefore: 
    \begin{equation}\begin{aligned}\left\|\mu\right\|_2^2 &= 2^n \Tr\left[\left(\exval_{\rho \followsdistr D}\left[\rho^2\right] - \bar{\gamma} \frac{I^{\otimes n}}{2^n}\right)^2\right] \\ &= 2^n \left(\Tr\left(\exval_{\rho \followsdistr D}\left[\rho^2\right]^2\right) - \frac{2}{2^n} \bar{\gamma} \Tr\left(\exval_{\rho \followsdistr D}\left[\rho^2\right]\right) + \frac{1}{4^n}\bar{\gamma}^2 \Tr\left(I^{\otimes n}\right)\right) \\ &= 2^n \left(\left\|\exval_{\rho \followsdistr D}\left[\rho^2\right]\right\|_2^2 - \frac{1}{2^n}\bar{\gamma}^2\right).\end{aligned}\end{equation}
    
    Now, by usual $\ell_1$-$\ell_2$ bounds for operators in $\mathbb{C}^{2^n \times 2^n}$, we know that:
    \begin{equation}\label{eq:untight-inequality-in-bias-upper-bound}
        \frac{1}{2^n} \left\|\exval_{\rho \followsdistr D}\left[\rho^2\right]\right\|_1^2 \leq \left\|\exval_{\rho \followsdistr D}\left[\rho^2\right]\right\|_2^2 \leq \left\|\exval_{\rho \followsdistr D}\left[\rho^2\right]\right\|_1^2
    \end{equation}
 
    Note moreover that, since $\exval_{\rho \followsdistr D}\left[\rho^2\right]$ is positive semi-definite, its Schatten 1-norm and its trace are equal: 
    \[\left\|\exval_{\rho \followsdistr D}\left[\rho^2\right]\right\|_1 = \Tr\left(\exval_{\rho \followsdistr D}\left[\rho^2\right]\right) = \bar{\gamma}.\]
    
    Therefore, we find precisely our two bounds. Focusing only on the upper-bound:
    \[0 \leq \left\|\mu\right\|_2^2 \leq 2^n \left(\bar{\gamma}^2 - \frac{1}{2^n} \bar{\gamma}^2\right) = \bar{\gamma}^2 \left(2^n - 1\right).\]

    Note that the upper-bound is only tight for $\bar{\gamma} = 1$ because of equation~\ref{eq:untight-inequality-in-bias-upper-bound}: $\left\|A\right\|_1 = \left\|A\right\|_2$ if and only if $A$ has rank $1$, which, for $A = \exval_{\rho \followsdistr D}\left[\rho^2\right]$, requires that $D$ is a pure-state distribution.
\end{proof}

\subsubsection{Moment matrix and its decomposition}

The moment matrix contains all the distributional information needed to evaluate the risk. We give its operator and Pauli-basis representations before separating the contributions from the average purity, the mean, and the second moment.

\begin{definition}[Moment matrix]\label{def:moment-matrix}
    Let $D$ be an $n$-qubit distribution over mixed states. We define its moment matrix $M \in \mathbb{C}^{4^n \times 4^n}$ by: 
    \[M = 2^n \exval_{\rho \followsdistr D}\left[\rho \otimes \rho^T\right].\]
\end{definition}
\begin{intuition*}
    This matrix appears naturally in the construction of lemma~\ref{lma:risk-rayleigh}. Its only purpose is this lemma.
\end{intuition*}
\begin{remark*}
    We may only know $M' = \exval\left(\rho^{\otimes 2}\right)$ for some distribution. It is possible to compute $M$ from this $M'$ by applying the partial transpose. That is, consider the linear operator $T_2: \mathbb{C}^{4^n \times 4^n} \mapsto \mathbb{C}^{4^n \times 4^n}$ defined by: 
    \[T_2\left(\ket{i}\bra{j} \otimes \ket{k}\bra{\ell}\right) = \ket{i}\bra{j} \otimes \ket{\ell}\bra{k}.\]
    Since this is a linear operator, we thus have: 
    \[M = 2^n T_2\left(\exval\left(\rho^{\otimes 2}\right)\right) = 2^n T_2\left(M'\right).\]
\end{remark*}

\begin{lemma}[Other representation of the moment matrix]\label{lma:moment-matrix-other-representation}
    Let $D$ be an $n$-qubit distribution over mixed states of moment matrix $M$. Then: 
    \[M = \sum_{P, Q} \exval_{\rho \followsdistr D}\left[\Tr\left(P \rho Q \rho\right)\right] \kket{\widetilde{P}} \bbra{\widetilde{Q}}\]
\end{lemma}
\begin{proof}
    This is direct using properties of the vectorisation: 
    \[\bbra{\widetilde{P}} M \kket{\widetilde{Q}} = \frac{\bbra{P}}{\sqrt{2^n}} 2^n \exval_{\rho \followsdistr D}\left[\rho \otimes \rho^T\right] \frac{\kket{Q}}{\sqrt{2^n}} = \exval_{\rho \followsdistr D}\left[\bbrakket{P}{\rho Q \rho}\right] = \exval_{\rho \followsdistr D}\left[\Tr\left(P \rho Q \rho\right)\right].\]
\end{proof}

\begin{lemma}[Moment matrix for pure states]\label{lma:moment-matrix-equivalence-pure-states}
    Let $D$ be an $n$-qubit distribution, of moment matrix $M$. If $D$ is a pure-state distribution, then:
    \[M = 2^n \exval_{\rho \followsdistr D}\left(\kket{\rho} \bbra{\rho}\right) = \sum_{P, Q \in \mathbb{P}_n} \exval_{\rho \followsdistr D}\left(\Tr\left(P \rho\right) \Tr\left(Q \rho\right)\right) \kket{\widetilde{P}} \bbra{\widetilde{Q}}.\]
\end{lemma}
\begin{remark*}
    Note that this gives us another way to compute $M$ from $M' = \exval\left(\rho^{\otimes 2}\right)$, for pure state distributions. Using the fact $\Tr\left(A\right)\Tr\left(B\right) = \Tr\left(A \otimes B\right)$:
    \autoeq{M = \sum_{P, Q \in \mathbb{P}_n} \exval\left(\Tr\left(P \rho\right) \Tr\left(Q \rho\right)\right) \kket{\widetilde{P}} \bbra{\widetilde{Q}} = \sum_{P, Q \in \mathbb{P}_n} \exval\left(\Tr\left(\left(P \otimes Q\right) \left(\rho \otimes \rho\right)\right)\right) \kket{\widetilde{P}} \bbra{\widetilde{Q}} = \sum_{P, Q \in \mathbb{P}_n} \Tr\left(\left(P \otimes Q\right) M'\right) \kket{\widetilde{P}} \bbra{\widetilde{Q}}.}
\end{remark*}
\begin{proof}[Proof 1]
    Let $\ket{\psi} = \sum_{i} \psi_i \ket{i}$ be an arbitrary pure state, of density matrix $\rho = \ket{\psi}\bra{\psi}$. Using the definition of vectorisation, we notice:
    \autoeq{\kket{\rho} = \text{vec}\left(\ket{\psi}\bra{\psi}\right) = \sum_{i,j} \psi_i \psi_j^* \text{vec}\left(\ket{i} \bra{j}\right) = \sum_{i, j} \psi_i \psi_j^* \ket{i} \otimes \ket{j} = \ket{\psi} \otimes \ket{\psi}^*.}
    
    This implies from the definition of moment matrix: 
    \[M  = 2^n \exval_{\rho \followsdistr D}\left(\rho \otimes \rho^T\right) = 2^n \exval_{\rho \followsdistr D}\left(\rho \otimes \rho^*\right) = 2^n \exval_{\rho \followsdistr D}[\kket{\rho}\bbra{\rho}],\]
    where we used the fact that $\rho^{\dagger} = \rho$ and hence $\rho^T = \rho^*$.
\end{proof}
\begin{proof}[Proof 2]
    Using completeness relations $I = \sum_{P} \kket{\widetilde{P}}\bbra{\widetilde{P}}$ and the normalisation $\widetilde{P} = \frac{1}{\sqrt{2^n}} P$, we find: 
    \[2^n \kket{\rho}\bbra{\rho} = 2^n \sum_{P, Q} \kket{\widetilde{P}}\bbrakket{\widetilde{P}}{\rho} \bbrakket{\rho}{\widetilde{Q}} \bbra{\widetilde{Q}} = \sum_{P, Q \in\mathbb{P}_n} \bbrakket{P}{\rho} \bbrakket{\rho}{Q} \kket{\widetilde{P}} \bbra{\widetilde{Q}} = \sum_{P, Q \in \mathbb{P}_n} \Tr\left(P \rho\right) \Tr\left(Q \rho\right) \kket{\widetilde{P}} \bbra{\widetilde{Q}}.\]
    
    So, by point (1):
    \[M = 2^n \exval\left(\kket{\rho} \bbra{\rho}\right) = \sum_{P, Q \in\mathbb{P}_n} \exval\left(\Tr\left(P \rho\right) \Tr\left(Q \rho\right)\right) \kket{\widetilde{P}} \bbra{\widetilde{Q}}. \]
\end{proof}

The traceless block and its coupling to the identity play different roles in the bounds. We control the first through its eigenvalues and the second through the mean. The following decomposition makes this distinction explicit.

\begin{definition}[Moment decomposition]\label{def:moment-decomposition}
    Let $D$ be an $n$-qubit distribution over mixed states, of moment matrix $M$. We decompose $M$ into a core part $M_{core}$, and a mean part $M_{mean}$. In other words, we let $M = M_{core} + M_{mean}$, where:
    \[M_{core} = \bbra{\widetilde{I}^{\otimes n}} M \kket{\widetilde{I}^{\otimes n}} \kket{\widetilde{I}^{\otimes n}}\bbra{\widetilde{I}^{\otimes n}} + \mathcal{S}, \mathspace \text{where } \mathcal{S} = \sum_{P, Q \neq I^{\otimes n}} \bbra{\widetilde{P}} M \kket{\widetilde{Q}} \kket{\widetilde{P}} \bbra{\widetilde{Q}} \in \mathbb{C}^{4^n \times 4^n},\]
    \[M_{mean} = \sum_{P \neq I^{\otimes n}} \bbra{\widetilde{I}^{\otimes n}} M \kket{\widetilde{P}} \kket{\widetilde{I}^{\otimes n}} \bbra{\widetilde{P}} + \sum_{P \neq I^{\otimes n}} \bbra{\widetilde{P}} M \kket{\widetilde{I}^{\otimes n}} \kket{\widetilde{P}} \bbra{\widetilde{I}^{\otimes n}}.\]
    
    We call $\{ \mu , M_{mean}\} $ the first moment of $D$, and $\{ M_{core},\mathcal{S} \}$ the second moment of $D$.
\end{definition}
\begin{remark*}
    On the one hand, $M_{core}$ is constructed by taking all the coordinates that are either both identities, or where neither is identity. On the other hand, $M_{mean}$ takes all the coordinates which have exactly one identity.
\end{remark*}
\begin{intuition*}
    Let us consider those matrices for pure states when $n = 1$. In particular, they simplify nicely thanks to lemma~\ref{lma:moment-matrix-equivalence-pure-states}. To simplify notation, let for $P, Q \in \left\{X, Y, Z\right\}$:
    \[\mu_P = \bbra{\widetilde{P}} M \kket{\widetilde{I}} = \exval_{\rho \followsdistr D}\left[\Tr\left(P \rho\right)\right], \mathspace \mathcal{S}_{PQ} = \bbra{\widetilde{P}} \mathcal{S} \kket{\widetilde{Q}} = \exval_{\rho \followsdistr D}\left[\Tr\left(P \rho\right) \Tr\left(Q \rho\right)\right].\]
    
    Recall that $\mu_P$ is the $P$th component of the average position on the Bloch sphere when $n = 1$. Similarly, for pure states, we observe that $\mathcal{S}$ can be interpreted as the uncentred second moment matrix of the position on the Bloch sphere. This justifies why we refer to $\ket{\mu}$ as the first moment and $\mathcal{S}$ as the second moment of $D$, and proves in some way that $M$ contains both the first and the second moment of the distribution.
    
    Then, leaving $\bar{\gamma}$ to be the average purity of the distribution, our decomposition becomes:
    \autoeq{M = \begin{pmatrix} \bar{\gamma} & \mu_X & \mu_Y & \mu_Z \\ \mu_X & \mathcal{S}_{XX} & \mathcal{S}_{XY} & \mathcal{S}_{XZ} \\ \mu_Y & \mathcal{S}_{YX} & \mathcal{S}_{YY} & \mathcal{S}_{YZ} \\ \mu_Z & \mathcal{S}_{ZX} & \mathcal{S}_{ZY} & \mathcal{S}_{ZZ} \end{pmatrix} = \begin{pmatrix} \bar{\gamma} & 0 & 0 & 0 \\ 0 & \mathcal{S}_{XX} & \mathcal{S}_{XY} & \mathcal{S}_{XZ} \\ 0 & \mathcal{S}_{YX} & \mathcal{S}_{YY} & \mathcal{S}_{YZ} \\ 0 & \mathcal{S}_{ZX} & \mathcal{S}_{ZY} & \mathcal{S}_{ZZ} \end{pmatrix} + \begin{pmatrix} 0 & \mu_X & \mu_Y & \mu_Z \\ \mu_X & 0 & 0 & 0 \\ \mu_Y & 0 & 0 & 0 \\ \mu_Z & 0 & 0 & 0 \end{pmatrix} = \underbrace{\left(\begin{array}{c|ccc} \bar{\gamma} & & 0 & \\ \hline & & & \\ 0 & & \mathcal{S} & \\ & & & \end{array}\right)}_{M_{core}} + \underbrace{\left(\begin{array}{c|ccc} 0 & & \bra{\mu} & \\ \hline & & & \\ \ket{\mu} & & 0 & \\ & & & \end{array}\right)}_{M_{mean}}.}
    
    The decomposition is analogous for any $n$.
\end{intuition*}
\begin{remark*}
    When $n = 1$ and states are pure, $\mathcal{S}$ matches reference~\cite{chen2024predicting}'s definition 2.1.
\end{remark*}
\begin{remark*}
    This decomposition depends on the choice of the basis, in the only sense that it must contain an identity and be orthonormal. We chose here the Pauli basis since it is easy to work with, but another orthonormal basis containing identity would have worked just the same. The reason we want the basis to contain the identity will become clearer in lemma~\ref{lma:risk-without-identity}.
\end{remark*}

\begin{lemma}[Other formulation for the decomposition]\label{lma:other-formulation-decomposition}
    Let $D$ be a mixed-state distribution of mean $\mu$, average purity $\bar{\gamma}$ and moment matrix $M = M_{core} + M_{mean}$ as decomposed in definition~\ref{def:moment-decomposition}. Then, we can write: 
    \[M_{core} = \bar{\gamma} \oplus \mathcal{S}, \mathspace \text{where } \mathcal{S} = \sum_{P, Q \neq I^{\otimes n}} \exval_{\rho \followsdistr D}\left[\Tr\left(P \rho Q \rho\right)\right] \kket{\widetilde{P}}\bbra{\widetilde{Q}},\]
    \[M_{mean} = \sum_{P \neq I^{\otimes n}} \bbrakket{\widetilde{P}}{\mu} \left(\kket{\widetilde{I}^{\otimes n}}\bbra{\widetilde{P}} + \kket{\widetilde{P}} \bbra{\widetilde{I}^{\otimes n}}\right).\]

    We will therefore call $M = \bar{\gamma} \oplus \mathcal{S} + M_{mean}$ the decomposition of $M$.
\end{lemma}
\begin{proof}
    We notice that:
    \[\bbra{\widetilde{I}^{\otimes n}} M \kket{\widetilde{I}^{\otimes n}} = \bbra{I^{\otimes n}} \exval_{\rho \followsdistr D}\left[\rho \otimes \rho^T\right] \kket{I^{\otimes n}} = \exval_{\rho \followsdistr D}\left[\bbrakket{I^{\otimes n}}{\rho^2}\right] = \exval_{\rho \followsdistr D}\left[\Tr\left(\rho^2\right)\right] = \bar{\gamma}.\]

    Therefore, we directly find from lemma~\ref{lma:moment-matrix-other-representation} that:
    \[M_{core} = \bar{\gamma} \oplus \mathcal{S}.\]

    Note that the direct sum here is a slight abuse of notation since $\mathcal{S} \in \mathbb{C}^{4^n \times 4^n}$. This is an example of an instance where we use the fact that we can consider $\mathcal{S} \in \mathbb{C}^{\left(4^n -1\right) \times \left(4^n - 1\right)}$ as well since its first row and column are zero.

    Similarly for $M_{mean}$, we find using lemma~\ref{lma:moment-matrix-other-representation} that for any $P \neq I$: 
    \[\bbra{\widetilde{I}^{\otimes n}} M \kket{\widetilde{P}} = \bbra{\widetilde{P}} M \kket{\widetilde{I}^{\otimes n}} = \exval_{\rho \followsdistr D}\left[\Tr\left(P \rho^2\right)\right] = \sqrt{2^n} \bbra{\widetilde{P}} \exval_{\rho \followsdistr D}\left[\kket{\rho^2}\right] = \sqrt{2^n} \bbra{\widetilde{P}} \exval_{\rho \followsdistr D}\left[\kket{\pi\left(\rho^2\right)}\right] = \bbrakket{\widetilde{P}}{\mu},\]
    where we used the fact that $\Pi\left(\kket{A}\right) = \kket{\pi\left(A\right)}$, i.e.\ that vector $I$-truncation and operator $I$-truncation are the same up to vectorisation (see definition~\ref{def:truncation-for-operators}).
\end{proof}

\begin{lemma}[$\mathcal{S}$ as $I$-truncation]\label{lma:s-as-truncation}
    Let $D$ be an $n$-qubit distribution over mixed states of moment matrix $M = \bar{\gamma} \oplus \mathcal{S} + M_{mean}$. Then, one can write: 
    \[\mathcal{S} = \Pi M \Pi.\]
\end{lemma}
\begin{proof}
    We find precisely that, for any $P, Q \in \mathbb{P}_n$: 
    \[\bbra{\widetilde{P}} \Pi M \Pi \kket{\widetilde{Q}} = \begin{cases} \bbra{\widetilde{P}} M \kket{\widetilde{Q}}, & \text{if $P \neq I^{\otimes n}$ and $Q \neq I^{\otimes n}$} \\ 0, & \text{otherwise} \end{cases} = \bbra{\widetilde{P}} \mathcal{S} \kket{\widetilde{Q}}.\]
\end{proof}

\begin{lemma}\label{lma:s-properties}
    Let $D$ be an $n$-qubit distribution over mixed states of moment matrix $M = \bar{\gamma} \oplus \mathcal{S} + M_{mean}$. $\mathcal{S}$ has the following properties:
    \begin{enumerate}
    \item $\mathcal{S}$ is symmetric with real coefficients in the Pauli basis.
    \item $\mathcal{S}$ is positive semi-definite.
    \item $\Tr\left(\mathcal{S}\right) = 2^{n} - \bar{\gamma}$.
    \item $\displaystyle  \frac{1}{2^n + 1} \leq \frac{2^n - \bar{\gamma}}{4^n - 1} \leq \left\|\mathcal{S}\right\|_{op} = \lambda_{max}\left(\mathcal{S}\right) \leq 2^n - \bar{\gamma}$.
    \item $\displaystyle 0 \leq \lambda_{min}\left(\mathcal{S}\right) \leq \frac{2^n - \bar{\gamma}}{4^n - 1} \leq \bar{\gamma}$.
    \end{enumerate}
\end{lemma}
\begin{proof}[Proof]
    Recall that lemma~\ref{lma:s-as-truncation} tells us we can write:
    \[\mathcal{S} = \Pi M \Pi = 2^n \exval_{\rho \followsdistr D}\left[\Pi \left(\rho \otimes \rho^T\right) \Pi\right].\]

    We will use this heavily in the proof of our lemma.
    
    \begin{enumerate}[left=0pt]
    \item First, we know that $\mathcal{S}$ has real coefficients since $\Pi\left(\rho \otimes \rho^T\right) \Pi$ has real coefficients in the Pauli basis. Indeed, $\Pi$ trivially only has real coefficients in the Pauli basis, and so does $\rho \otimes \rho^T$: 
    \[\left(\bbra{\widetilde{P}} \rho \otimes \rho^T \kket{\widetilde{Q}}\right)^* = \bbra{\widetilde{Q}} \rho \otimes \rho^T \kket{\widetilde{P}} = \bbrakket{\widetilde{Q}}{\rho \widetilde{P} \rho} = \Tr\left(\widetilde{Q} \rho \widetilde{P} \rho\right) = \Tr\left(\widetilde{P} \rho \widetilde{Q} \rho\right) = \bbrakket{\widetilde{P}}{\rho \widetilde{Q} \rho} = \bbra{\widetilde{P}} \rho \otimes \rho^T \kket{\widetilde{Q}}.\]

        This in fact also proves that $\bbra{\widetilde{Q}} \rho \otimes \rho^T \kket{\widetilde{P}} =\bbra{\widetilde{P}} \rho \otimes \rho^T \kket{\widetilde{Q}}$ and hence that $\mathcal{S}$ is symmetric.
    
    \item Second, notice that, since a density matrix $\rho$ is positive semi-definite, we can consider its positive semi-definite square root $\rho^{1/2}$ such that $\rho^{1/2} \rho^{1/2} = \rho$. This gives for any $\kket{A} \in \mathbb{C}^{4^n}$, using the definition of $I$-truncation:
        \begin{equation}\begin{aligned}\bbra{A} \mathcal{S} \kket{A} &= 2^n \exval\left[\bbra{A}\Pi \left(\rho \otimes \rho^T\right) \Pi \kket{A}\right] \\ &= 2^n \exval\left[\bbra{\pi\left(A\right)} \left(\rho^{1/2} \otimes \left(\rho^{1/2}\right)^T\right)\left(\rho^{1/2} \otimes \left(\rho^{1/2}\right)^T\right) \kket{\pi\left(A\right)}\right] \\ &= 2^n \exval\left[\bbrakket{\rho^{1/2} \pi\left(A\right) \rho^{1/2}}{\rho^{1/2} \pi\left(A\right) \rho^{1/2}}\right] \\ &= 2^n \exval\left[\left\|\rho^{1/2} \pi\left(A\right) \rho^{1/2}\right\|_2^2\right] \\ &\geq 0,\end{aligned}\end{equation}
        indeed showing positive semi-definiteness.
    
    \item Third, for the trace, using the easy-to-prove fact $\Pi^2 = \Pi = I^{\otimes n} \otimes I^{\otimes n} - \kket{\widetilde{I}^{\otimes n}}\bbra{\widetilde{I}^{\otimes n}}$:
        \[\Tr\left(\mathcal{S}\right) = 2^n \exval\left[\Tr\left(\Pi \left(\rho \otimes \rho^T\right)\right)\right] = 2^n \exval\left[\Tr\left(\rho \otimes \rho^T\right) - \bbra{\widetilde{I}^{\otimes n}} \left(\rho \otimes \rho^T\right) \kket{\widetilde{I}^{\otimes n}}\right].\]

        We can consider both terms independently: 
        \[\exval\left[\Tr\left(\rho \otimes \rho^T\right)\right] = \exval\left[\Tr\left(\rho\right)\Tr\left(\rho^T\right)\right] = \exval\left[1\cdot 1\right] = 1,\]
        \[\exval\left[\bbra{\widetilde{I}^{\otimes n}} \left(\rho \otimes \rho^T\right) \kket{\widetilde{I}^{\otimes n}}\right] = \frac{1}{2^n} \exval\left[\bbrakket{I^{\otimes n}}{\rho^2} \right] = \frac{1}{2^n} \exval\left[\Tr\left(\rho^2\right)\right] = \frac{1}{2^n} \bar{\gamma},\]
        by definition of $\bar{\gamma}$. This immediately yields: 
        \[\Tr\left(\mathcal{S}\right) = 2^n \left(1 - \frac{\bar{\gamma}}{2^n}\right) = 2^n - \bar{\gamma}.\]

        This result makes complete intuitive sense by looking at the decomposition $M = \bar{\gamma} \oplus \mathcal{S} + M_{mean}$ in array form (since $\Tr\left(M\right) = 2^n$ using a completely similar reasoning).
    
    \item  Fourth, let us consider operator norm. Since $\mathcal{S}$ is symmetric positive semi-definite, we have $\left\|\mathcal{S}\right\|_{op} = \lambda_{max}\left(\mathcal{S}\right)$. This also means $\lambda_i\left(\mathcal{S}\right) \geq 0$ for all $i$, and thus, since $\mathcal{S} \in \mathbb{C}^{\left(4^n - 1\right) \times \left(4^n - 1\right)}$: 
    \[2^n - \bar{\gamma} = \Tr\left(\mathcal{S}\right) = \sum_{i=1}^{4^n - 1} \lambda_i\left(\mathcal{S}\right) \geq \lambda_{max}\left(\mathcal{S}\right).\]
    
    For the other direction, we exploit that $\lambda_i\left(\mathcal{S}\right) \leq \lambda_{max}\left(\mathcal{S}\right)$ for all $i$:
    \[2^n - \bar{\gamma} = \Tr\left(\mathcal{S}\right) = \sum_{i=1}^{4^n - 1} \lambda_i\left(\mathcal{S}\right) \leq \left(4^n - 1\right)\lambda_{max}\left(\mathcal{S}\right).\]
    
    Therefore: 
    \[\lambda_{max}\left(\mathcal{S}\right) \geq \frac{2^n - \bar{\gamma}}{4^n - 1} \geq \frac{2^n - 1}{4^n - 1} = \frac{2^n - 1}{\left(2^n - 1\right)\left(2^n + 1\right)} = \frac{1}{2^n +  1}.\]
    
    \item For $\lambda_{min}\left(\mathcal{S}\right)$, we use again the same method: 
    \[2^n - \bar{\gamma} = \Tr\left(\mathcal{S}\right) = \sum_{i=1}^{4^n - 1} \lambda_i\left(\mathcal{S}\right) \geq \left(4^n - 1\right) \lambda_{min}\left(\mathcal{S}\right) \iff \lambda_{min}\left(\mathcal{S}\right) \leq \frac{2^n - \bar{\gamma}}{4^n - 1}.\]

    We can further upper-bound this by using the fact that $\bar{\gamma} \geq \frac{1}{2^n}$ since it is a purity and hence: 
    \[\frac{2^n - \bar{\gamma}}{4^n - 1} - \bar{\gamma} = \frac{2^n - \bar{\gamma} - 4^n \bar{\gamma} + \bar{\gamma}}{4^n - 1} = \frac{2^n - 4^n\bar{\gamma}}{4^n - 1} \leq \frac{2^n - 4^n \cdot \frac{1}{2^n}}{4^n -1} = 0 \implies \frac{2^n - \bar{\gamma}}{4^n - 1} \leq \bar{\gamma}.\]
    
    The fact that $\lambda_{min}\left(\mathcal{S}\right) \geq 0$ directly comes from the fact $\mathcal{S}$ is positive semi-definite.
    \end{enumerate}
\end{proof}

\subsubsection{Expressivity measure}

Our spectral bounds (detailed in theorems~\ref{thm-main-text:general-unitary-comparison} and~\ref{thm-main-text:tight-unitary-comparison} in the main text and derived in theorem~\ref{thm:inequalities} in the appendices) require the largest eigenvalue of $\mathcal{S}$ to lie below the average purity. We quantify this difference by the expressivity measure. However, we note that distributions whose states share an eigenbasis cannot satisfy this condition, which already identifies a limitation of the bounds.

\begin{definition}[Expressivity measure]\label{def:eta}
    Let $D$ be an $n$-qubit distribution over mixed states, of moment matrix $M = \bar{\gamma} \oplus \mathcal{S} + M_{mean}$. We define its expressivity measure $\eta \in \mathbb{R}$ as: 
    \[\eta = \bar{\gamma} - \left\|\mathcal{S}\right\|_{op} = \bar{\gamma} - \lambda_{max}\left(\mathcal{S}\right).\]
\end{definition}
\begin{remark*}
    Just like $\mathcal{S}$, when $n = 1$ and the distribution is pure, $\eta$ matches reference~\cite{chen2024predicting}'s definition from section 4.
\end{remark*}

\begin{lemma}\label{lma:eta-properties}
    Let $D$ be an $n$-qubit distribution over mixed states of moment matrix $M = \bar{\gamma} \oplus \mathcal{S} + M_{mean}$. Its expressivity measure $\eta$ has the following properties:
    \begin{enumerate}
    \item $\displaystyle -2^n \leq 2 \bar{\gamma} - 2^n \leq \eta \leq \frac{4^n}{4^n - 1} \left(\bar{\gamma} - 2^{-n}\right) \leq \frac{2^n}{2^n + 1}$.
    \item Let $U \in U\left(2^n\right)$ be an arbitrary unitary matrix, and let $D'$ be the distribution defined so that, when $D$ outputs $\rho$, $D'$ outputs $U\rho U^{\dagger}$. Moreover, let $\eta'$ be the expressivity measure of $D'$. Then, $\eta = \eta'$.
    \item If the support of $D$ contains states that are all diagonalisable in the same eigenbasis (i.e.\ $D$ is a classical distribution of probability distributions over bitstrings), then $\eta \leq 0$.
    \end{enumerate}
\end{lemma}
\begin{intuition*}
    The third property implies some form of link between this object and a classicality measure. Indeed, a classical distribution is a distribution of support $\left\{\ket{0}, \ket{1}\right\}^{\otimes n}$. This means that a classical distribution necessarily has $\eta \leq 0$. We will later see theorems that require $\eta > 0$ (or, equivalently, $\lambda_{max}\left(\mathcal{S}\right) < \bar{\gamma}$), which can be interpreted as a signature of quantumness.
    
    Moreover, note that, when $n = 1$ and the distribution is pure, then $\eta \leq 0 \iff \eta = 0$ by the first property. Reference~\cite{chen2024predicting} observed in fact that when $n = 1$ and states are pure, then $\eta = 0$ if and only if $D$ has an orthogonal support (i.e.\ it is a classical distribution of bitstrings).

\end{intuition*}
\begin{remark*}
    For our theorems to apply in a reasonable way, we will often require $\eta \geq \frac{1}{\poly\left(n\right)}$. Property (1) tells us that, for this to have a chance to apply, we in particular need average purity to be such that $\bar{\gamma} \geq \frac{1}{\poly\left(n\right)}$.
\end{remark*}
\begin{proof}[Proof 1]
    By lemma~\ref{lma:s-properties}, we know that: 
    \[\left\|\mathcal{S}\right\|_{op} \in \left[\frac{2^n - \bar{\gamma}}{4^n - 1}, 2^n - \bar{\gamma}\right].\]
    
Using $2^{-n}\leq \bar{\gamma}\leq 1$, we obtain
\[
\eta
= \bar{\gamma}-\left\|\mathcal{S}\right\|_{op}
\geq \bar{\gamma}-\left(2^n-\bar{\gamma}\right)
= 2\bar{\gamma}-2^n
\geq 2^{1-n}-2^n
\geq -2^n,
\]
and
\[
\eta
\leq \bar{\gamma}-\frac{2^n-\bar{\gamma}}{4^n-1}
= \frac{4^n}{4^n-1}\left(\bar{\gamma}-2^{-n}\right)
\leq \frac{4^n}{4^n-1}\left(1-2^{-n}\right)
= \frac{2^n}{2^n+1}.
\]
\end{proof}
\begin{proof}[Proof 2]

Conjugating every state by the same unitary leaves the identity unchanged and maps traceless operators to traceless operators. We use this to show that $\mathcal{S}$ and $\mathcal{S}'$ have the same eigenvalues. Since conjugation also preserves purity, the expressivity is unchanged.

    Let $M = \bar{\gamma} \oplus \mathcal{S} + M_{mean}$ be the moment matrix of $D$, and $M' = \bar{\gamma}' \oplus \mathcal{S}' + M_{mean}'$ the one of $D'$. We first immediately notice that: 
    \[\bar{\gamma}' = \exval_{\rho' \followsdistr D'}\left[\Tr\left({\rho'}^2\right)\right] = \exval_{\rho \followsdistr D}\left[\Tr\left(U \rho U^{\dagger} U \rho U^{\dagger}\right)\right] = \exval_{\rho \followsdistr D}\left[\Tr\left(\rho^2\right)\right] = \bar{\gamma},\]
    since $U^{\dagger} U = I^{\otimes n}$ by definition of a unitary matrix. Proving that their purities are equal is the first step to proving that their expressivity measures are the same, we are only left with proving that they have the same largest eigenvalue.

    By lemma~\ref{lma:s-as-truncation}, we know we can write:
    \begin{equation}\begin{aligned}\mathcal{S}' &= \Pi M' \Pi \\ &= 2^n \Pi \exval_{\rho' \followsdistr D'}\left[\left(\rho' \otimes {\rho'}^T\right)\right] \Pi \\ &= 2^n \Pi \exval_{\rho \followsdistr D}\left[\left(U \rho U^{\dagger}\right) \otimes \left(U \rho U^{\dagger}\right)^T\right] \Pi \\ &= 2^n \Pi \exval_{\rho \followsdistr D}\left[\left(U \rho U^{\dagger}\right) \otimes \left(U^* \rho^T U^T\right)\right] \Pi \\ &= 2^n \Pi \left(U \otimes U^*\right) \exval_{\rho \followsdistr D}\left[\rho \otimes \rho^T\right] \left(U^{\dagger} \otimes U^T\right) \Pi \\ &= 2^n \Pi V \exval_{\rho \followsdistr D}\left[\rho \otimes \rho^T\right] V^{\dagger} \Pi,\end{aligned}\end{equation}
    where we left $V = U \otimes U^*$. We aim to show that $\Pi$ and $V$ commute. Therefore, we evaluate:
    \[\Pi V = \Pi\left(U \otimes U^*\right) = \left(I^{\otimes n} \otimes I^{\otimes n} - \kket{\widetilde{I}^{\otimes n}}\bbra{\widetilde{I}^{\otimes n}}\right) \left(U \otimes U^*\right) = U \otimes U^* - \frac{1}{2^n} \kket{I^{\otimes n}}\bbra{U^{\dagger} U}.\]

    Now, $U^{\dagger} U = U U^{\dagger} = I^{\otimes n}$ since $U$ is a unitary, so we can indeed verify $\Pi$ and $V$ commute.
    \autoeq{\Pi V = U \otimes U^* - \frac{1}{2^n} \kket{U U^{\dagger}}\bbra{I^{\otimes n}} = U \otimes U^* - \frac{1}{2^n} \left(U \otimes U^*\right) \kket{I^{\otimes n}}\bbra{I^{\otimes n}} = \left(U \otimes U^*\right)\left(I^{\otimes n} \otimes I^{\otimes n} - \kket{\widetilde{I}^{\otimes n}}\bbra{\widetilde{I}^{\otimes n}}\right) = V\Pi.}

    In particular, this also means that $\Pi^{\dagger} = \Pi$ and $V^{\dagger} = U^{\dagger} \otimes U^T$ also commute. Therefore, we have found so far that:
    \[\mathcal{S}' = 2^n \Pi V \exval_{\rho \followsdistr D}\left[\rho \otimes \rho^T\right] V^{\dagger} \Pi = 2^n V \Pi \exval_{\rho \followsdistr D}\left[\rho \otimes \rho^T\right] \Pi V^{\dagger} = V \mathcal{S} V^{\dagger}.\]
    
    It is trivial to check that $V$ is a unitary matrix, and hence that $V^{-1} = V^{\dagger}$. Therefore, we have really been able to write $\mathcal{S}' = V \mathcal{S} V^{-1}$. This shows that $\mathcal{S}'$ and $\mathcal{S}$ are similar matrices by definition, and hence that they have the exact same eigenvalues. This does imply that they have the same largest eigenvalue in particular, and hence: 
    \[\eta = \bar{\gamma} - \lambda_{max}\left(\mathcal{S}\right) = \bar{\gamma}' - \lambda_{max}\left(\mathcal{S'}\right) = \eta'.\]
\end{proof}
\begin{proof}[Proof 3]

We first consider states that are diagonal in the computational basis. The block of $\mathcal{S}$ corresponding to nonidentity strings of $I$ and $Z$ has average eigenvalue $\bar{\gamma}$, so the largest eigenvalue of $\mathcal{S}$ is at least $\bar{\gamma}$ and hence $\eta \leq 0$. By unitary invariance, the same conclusion holds whenever the states share a common eigenbasis.

    Let $\Lambda = \left\{\ket{\lambda_1}, \ldots, \ket{\lambda_{2^n}}\right\}$ be the common eigenbasis of all the mixed states of $D$ which exists by hypothesis.
    \begin{itemize}[left=0pt]
    \item We first consider that $\Lambda = \left\{\ket{0}, \ket{1}\right\}^{\otimes n}$.
  
    We define the following two sets: 
    \[\mathbb{P}_{IZ} = \left\{I, Z\right\}^{\otimes n} \setminus \left\{I^{\otimes n}\right\}, \mathspace \mathbb{P}_{other} = \mathbb{P}_n \setminus \mathbb{P}_{IZ} \setminus \left\{I^{\otimes n}\right\}.\]

    Note that we ignore $I^{\otimes n}$ since $\mathcal{S} \in \mathbb{R}^{\left(4^n - 1\right) \times \left(4^n - 1\right)}$ does not use it in its representation. Observe that any $\rho \followsdistr \supp D$ commutes with any $P_{IZ} \in \mathbb{P}_{IZ}$, since they are all diagonal in the $\Lambda$ eigenbasis. Moreover, for any $Q_{other} \in \mathbb{P}_{other}$, by lemma~\ref{lma:other-formulation-decomposition}: 
    \[\bbra{\widetilde{P}_{IZ}} \mathcal{S} \kket{\widetilde{Q}_{other}} = \exval_{\rho \followsdistr D}\left[\Tr\left(P_{IZ} \rho Q_{other} \rho\right)\right] = \exval_{\rho \followsdistr D}\left[\Tr\left(\rho P_{IZ} Q_{other} \rho\right)\right] = \exval_{\rho \followsdistr D}\left[\Tr\left(\left(P_{IZ} Q_{other}\right) \rho^2\right)\right] = \exval_{\rho \followsdistr D}\left[0\right] = 0,\]
    since $\left(P_{IZ} Q_{other}\right)$ contains at least a $X$ or $Y$ Pauli in its tensor product expansion, and for all $i \in \left\{0, 1\right\}$:
    \[\bra{i} X \ket{i} = \bra{i} Y \ket{i} = 0.\]
    
    Consider the ordering of the basis $\mathbb{P}_n \setminus \left\{I^{\otimes n}\right\}$ where we put first $\mathbb{P}_{IZ}$ and then $\mathbb{P}_{other}$. In this ordering, we have just proven that $\mathcal{S}$ can be written as: 
    \[\mathcal{S} = \begin{pmatrix} \mathcal{T}_{IZ} & & 0 & \\ & & & \\ 0 & & \mathcal{T}_{other} & \\ & & & \end{pmatrix},\]
    where $\mathcal{T}_{IZ} \in \mathbb{R}^{\left(2^n - 1\right) \times \left(2^n - 1\right)}$ since $\left|\mathbb{P}_{IZ}\right| = 2^n -1$. We observe that $\mathcal{T}_{IZ}$ and $\mathcal{T}_{other}$ are both positive semi-definite since they are both Hermitian (by Hermitianness of $\mathcal{S}$) and have non-negative eigenvalues (since any eigenvector for either is also an eigenvector for $\mathcal{S}$, with the same eigenvalue). Let us consider the trace of $\mathcal{T}_{IZ}$, using again the fact that any $P_{IZ} \in \mathbb{P}_{IZ}$ commutes with any $\rho \followsdistr D$: 
    \begin{equation}\begin{aligned}\Tr\left(\mathcal{T}_{IZ}\right) &= \sum_{P_{IZ} \in \mathbb{P}_{IZ}} \bbra{\widetilde{P}_{IZ}} \mathcal{T}_{IZ} \kket{\widetilde{P}_{IZ}} \\ &{\ =\ } \sum_{P_{IZ} \in \mathbb{P}_{IZ}} \bbra{\widetilde{P}_{IZ}} \mathcal{S} \kket{\widetilde{P}_{IZ}} \\ &{\ =\ } \sum_{P_{IZ} \in \mathbb{P}_{IZ}} \exval_{\rho \followsdistr D}\left[\Tr\left(P_{IZ} \rho P_{IZ} \rho\right)\right] \\ &= \sum_{P_{IZ} \in \mathbb{P}_{IZ}} \exval_{\rho \followsdistr D}\left[\Tr\left(P_{IZ}^2 \rho^2\right)\right] \\ &{\ =\ } \sum_{P_{IZ} \in \mathbb{P}_{IZ}} \exval_{\rho \followsdistr D}\left[\Tr\left(\rho^2\right)\right] \\ &{\ =\ } \left|\mathbb{P}_{IZ}\right| \bar{\gamma} \\ &{\ =\ } \left(2^n - 1\right)\bar{\gamma}.\end{aligned}\end{equation}
    
    Therefore, by positive semi-definiteness of $\mathcal{T}_{IZ}$, this immediately gives us that: 
    \[\bar{\gamma} \left(2^n - 1\right) = \Tr\left(\mathcal{T}_{IZ}\right) = \sum_{i=1}^{2^n-1}  \lambda_i\left(\mathcal{T}_{IZ}\right) \leq \left(2^n - 1\right) \lambda_{max}\left(\mathcal{T}_{IZ}\right) \iff \lambda_{max}\left(\mathcal{T}_{IZ}\right) \geq \bar{\gamma}.\]
    
    Observe again that any eigenvector of $\mathcal{T}_{IZ}$ is also an eigenvector of $\mathcal{S}$, with the same eigenvalue. Therefore: 
    \[\lambda_{max}\left(\mathcal{S}\right) \geq \lambda_{max}\left(\mathcal{T}_{IZ}\right) \geq \bar{\gamma}.\]

    This does give that $\eta = \bar{\gamma} - \lambda_{max}\left(\mathcal{S}\right) \leq 0$.
    
    \item Now, let us assume that the common eigenbasis of all the mixed states $\Gamma = \left\{\ket{v_1}, \ldots, \ket{v_m}\right\}$ is an arbitrary set of orthogonal eigenvectors. We can assume without loss of generality that $m = 2^n$ since, otherwise, we can complete $\left\{\ket{v_1}, \ldots, \ket{v_m}\right\}$ with some more orthogonal vectors using the Gram-Schmidt algorithm. Consider the matrix $U$ so that, for all $i \in \left\{1, \ldots, 2^n\right\}$: 
    \[U \ket{v_i} = \ket{i}.\]
    
    Notice that this $U$ is unitary since, for all $i \in \left\{1, \ldots, 2^n\right\}$, using the orthogonality of $\left\{\ket{v_1}, \ldots, \ket{v_m}\right\}$: 
    \[\bra{v_i}I\ket{v_j} = \delta_{ij} = \braket{i}{j} = \bra{v_i}U^{\dagger} U \ket{v_j}.\]
    
    This does imply that $U^{\dagger} U = I$. But then, considering the distribution $D'$ that outputs $U \rho U^{\dagger}$ when $D$ outputs $\rho$, we can use the second property, to know that: 
    \[\eta = \eta'.\]

    However, when $D$ outputs a state $\rho = \sum_{i} p_i \ket{v_i}\bra{v_i}$, then $D'$ outputs a state: 
    \[\rho' = U \rho U^{\dagger} = \sum_{i} p_i U \ket{v_i} \left(U \ket{v_i}\right)^{\dagger} = \sum_{i} p_i \ket{i}\bra{i}.\]
    
    We are therefore back to the first case, telling us $\eta' \leq 0$. Hence, $\eta = \eta' \leq 0$ as well.
    \end{itemize}
\end{proof}

Before moving on, we prove one more property of the expressivity measure, which gives a natural way to visualize it on the Bloch sphere.

\begin{theorem}\label{thm:eta-visualization}
    Let $D$ be a mixed-state distribution over a single qubit. Then, leaving $\ket{r(\rho)} = \sqrt{2} \kket{\pi(\rho)} \in \mathbb{R}^3$ to be the Bloch vector of a state $\rho$: 
    \begin{equation}\label{eq:eta-visualization}
        \eta = \min_{\substack{\ket{u} \in \mathbb{R}^3 \\ \braket{u}{u} = 1}} \exval_{\rho \followsdistr D}\left[\left\|\ket{r(\rho)} - \projection_{\ket{u}} \ket{r(\rho)}\right\|_2^2\right],
    \end{equation}
    where $\projection_{\ket{u}} \ket{r(\rho)} = \braket{u}{r(\rho)} \cdot \ket{u}$ is the projection of the vector $\ket{r(\rho)}$ on the line spanned by $\ket{u}$.
\end{theorem}
\begin{figure}[tb]
    \centering
    \subfloat{
        \includegraphics[width=0.48\linewidth]{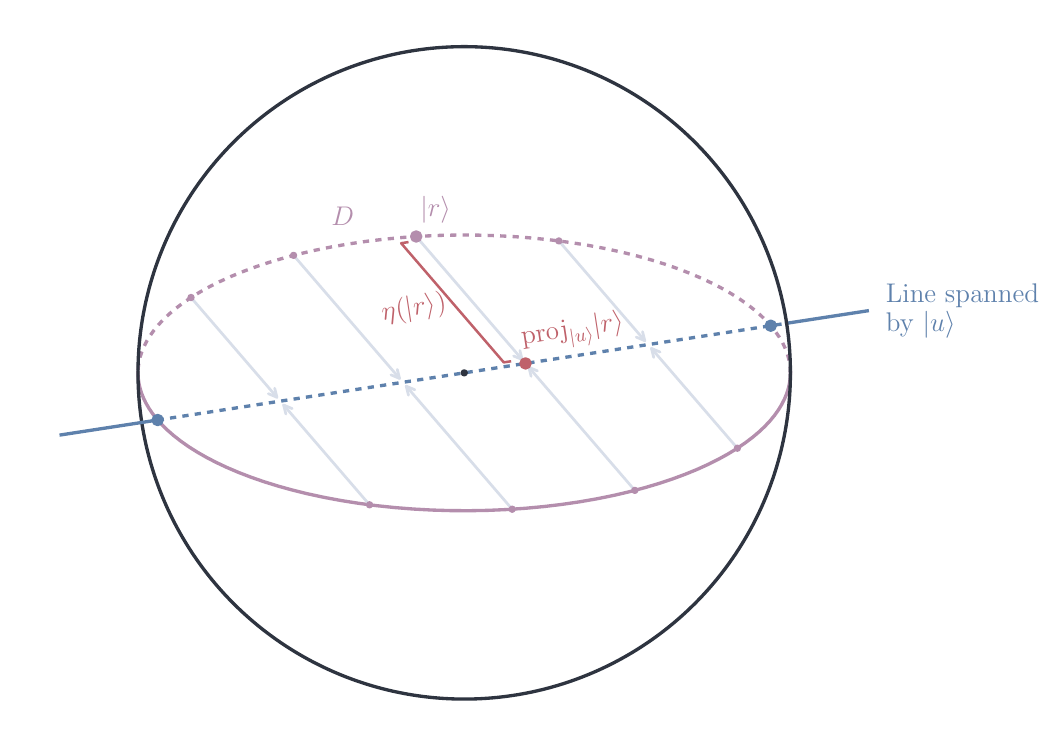}
    }
    \hfill
    \subfloat{
        \includegraphics[width=0.48\linewidth]{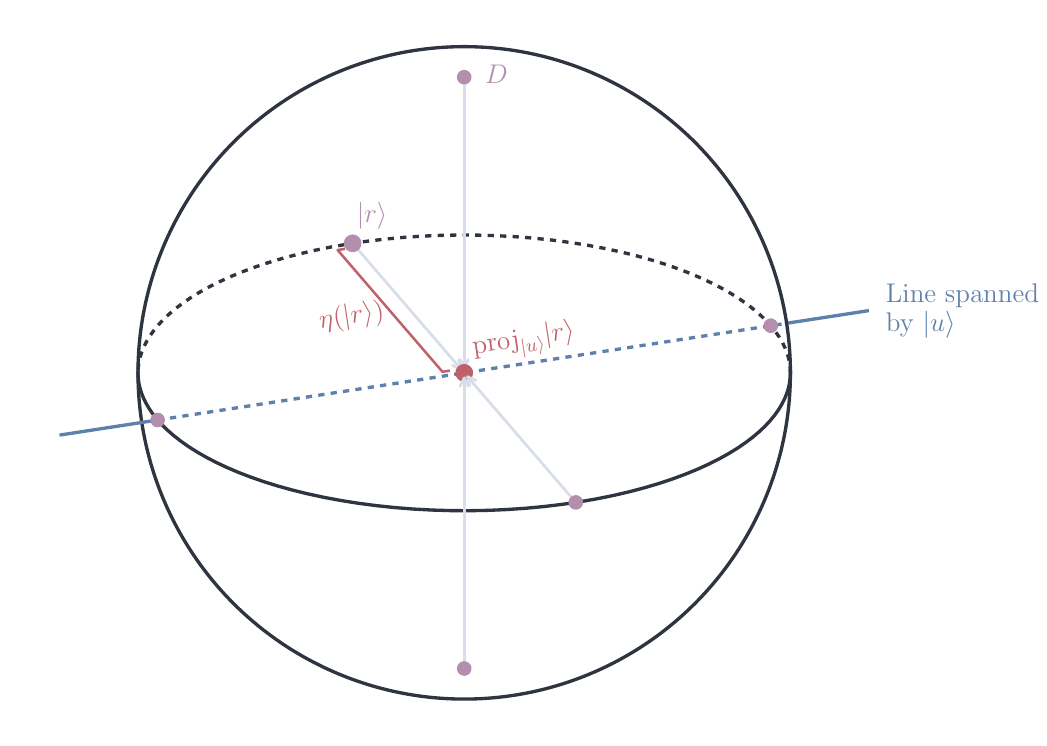}
    }
    \caption{Visualization of theorem~\ref{thm:eta-visualization}, for two different distributions: (left) $D = \Haar_1(O)$ (see definition~\ref{def:orthgonal-haar}) and (right) $D = {\rm Unif}\{\ket{0}, \ket{1}, \ket{\pm}, \ket{y\pm}\}$. We draw a line spanned by a $\ket{u}$ that minimizes equation~\ref{eq:eta-visualization}, and write $\ket{r} := \ket{r(\rho)}$ to be the Bloch vectors of our random samples and $\eta(\ket{r}) := \|\ket{r(\rho)} - \projection_{\ket{u}} \ket{r(\rho)}\|_2^2$. Under this notation, theorem~\ref{thm:eta-visualization} tells us $\eta = \exval_{\rho \followsdistr D}\left[\eta(\ket{r})\right]$. In other words, this theorem states that the expressivity $\eta$ is the average squared distance between points $\ket{r}$ and their projection $\projection_{\ket{u}} \ket{r}$, i.e.\ the average squared length of the gray arrows. It can be proven the minimum in equation~\ref{eq:eta-visualization} is reached for (left) any $\ket{u}$ on the equator and (right) any $\ket{u}$. This moreover yields expressivity measures (left) $\eta = 1/2$ and (right) $\eta = 2/3$.}
    \label{fig:eta-visualization}
\end{figure}
\begin{intuition*}
    Observe that $\|\ket{r(\rho)} - \projection_{\ket{u}} \ket{r(\rho)}\|_2^2$ is the squared distance between $\ket{r(\rho)}$ and the line spanned by $\ket{u}$. This result therefore states that, for single-qubit distributions, $\eta$ is the minimum, over all lines going through the origin, of the average squared distance between the Bloch vectors and this line. This is visualized for two natural distributions on figure~\ref{fig:eta-visualization}, where one should really think of the line as rotating about the origin until the minimum is reached.
    
    This also explains why a classical distribution, i.e.\ a distribution entirely supported on a single line, has zero expressivity: $\ket{u}$ can be taken along this line, since all the Bloch vectors are at a zero distance from it. Similarly, this result intuitively yields that any distribution with small average purity must have small expressivity: $\ket{u}$ can be taken to be any direction, since most Bloch vectors will be close to the origin anyway and hence to any line. 
\end{intuition*}
\begin{proof}
Let us first analyze the eigenvalue term of $\eta = \bar{\gamma} - \lambda_{max}(\mathcal{S})$. We observe the well-known fact that
\begin{equation}\label{eq:largest-eigenvalue-s-as-maximisation-problem}
    \lambda_{max}(\mathcal{S}) = \max_{\substack{\ket{u} \in \mathbb{R}^3 \\ \braket{u}{u} = 1}} \bra{u} \mathcal{S} \ket{u}.
\end{equation}

Indeed, $\bra{u} \mathcal{S} \ket{u} \leq \lambda_{max}(\mathcal{S})$ by properties of Rayleigh quotient (definition~\ref{def:rayleigh-quotient}), and equality is reached when $\ket{u}$ is an eigenvector associated to the largest eigenvalue, which can indeed be taken to be real by the spectral theorem.

To simplify this, let us evaluate $\mathcal{S}$. By lemma~\ref{lma:other-formulation-decomposition}, its components are given, for $P, Q \in \left\{X, Y, Z\right\}$ by:
\[\bbra{\widetilde{P}} \mathcal{S} \kket{\widetilde{Q}} = \exval_{\rho \followsdistr D}\left[\Tr\left(P \rho Q \rho\right)\right].\]

We try to simplify this trace. Let $\rho$ be arbitrary, and for any $P \in \left\{X, Y, Z\right\}$, let $r_{P} = \braket{\widetilde{P}}{r(\rho)} = \bbrakket{P}{\pi\left(\rho\right)} = \Tr(P \rho)$ be the $P$ component of its Bloch vector. Moreover, let $\|r\|_2^2 = \braket{r(\rho)}{r(\rho)}$. Observe that: 
\[\begin{split}
    \rho Q =\ & \frac{1}{2}\left(I + \sum_{P} r_P P\right) Q = Q \frac{1}{2}\left(I - \sum_{P \neq Q} r_P P + r_Q Q\right) = Q \frac{1}{2}\left(-I - \sum_{P} r_P P + 2 I + 2 r_Q Q\right) \\
    =\ & -Q \rho + Q\cdot \frac{1}{2}\cdot \left(2I + 2 r_Q Q\right) = -Q \rho + Q + r_Q I.
\end{split}\]

In addition to this, we can simplify equation~\ref{eq:single-qubit-rho-squared} to get: 
\[\rho^2 = \frac{1}{2}\left(\frac{1 + r_X^2 + r_Y^2 + r_Z^2}{2}I + \sum_{P} r_P P\right) = \rho + \frac{1}{2} I\left(\frac{1 + r_X^2 + r_Y^2 + r_Z^2}{2} - 1\right) = \rho + \frac{\left\|r\right\|_2^2 - 1}{4} I = \rho + \frac{\Tr\left(\rho^2\right) - 1}{2} I,\]
where we used that $\Tr\left(\rho^2\right) = (1 + \|r\|_2^2)/2 \iff \|r\|_2^2 = 2 \Tr(\rho^2) - 1$ also by equation~\ref{eq:single-qubit-rho-squared}. Combining the two, we find that: 
\[\rho Q \rho = \left(-Q \rho + Q + r_Q I\right) \rho = -Q \rho^2 + Q \rho + r_Q \rho = -Q \left(\rho + \frac{\Tr(\rho^2) - 1}{2} I\right) + Q \rho + r_Q \rho = \frac{1 - \Tr(\rho^2)}{2} Q + r_Q \rho.\]

We are therefore able to simplify our trace: 
\[\begin{split}
    \Tr\left(P \rho Q \rho\right) =\ & \frac{1 - \Tr(\rho^2)}{2} \Tr\left(P Q\right) + r_Q \Tr\left(P \rho\right) = \left(1 - \Tr(\rho^2) \right)\delta_{PQ} + r_Q r_P \\
    =\ & \left(1 - \Tr(\rho^2)\right) \bbrakket{\widetilde{P}}{\widetilde{Q}} + \bbraket{\widetilde{P}}{r\left(\rho\right)} \brakket{r\left(\rho\right)}{\widetilde{Q}} = \bbra{\widetilde{P}}\left(\left(1 - \Tr(\rho^2)\right) I + \ket{r\left(\rho\right)}\bra{r\left(\rho\right)}\right) \kket{\widetilde{Q}}.
\end{split}\]

Plugging it back in the expression of $\mathcal{S}$, this reads: 
\[\bbra{\widetilde{P}} \mathcal{S} \kket{\widetilde{Q}} = \exval_{\rho \followsdistr D}\left[\Tr\left(P \rho Q \rho\right)\right] =  \bbra{\widetilde{P}} \exval_{\rho\followsdistr D}\left[\left(1 - \Tr(\rho^2)\right)I + \ket{r\left(\rho\right)}\bra{r\left(\rho\right)}\right] \kket{\widetilde{Q}} \implies \mathcal{S} = \left(1 - \bar{\gamma}\right) I + \exval_{\rho\followsdistr D}\left[\ket{r(\rho)}\bra{r(\rho)}\right].\]

    We are now ready to analyze $\eta$. By equation~\ref{eq:largest-eigenvalue-s-as-maximisation-problem}, and using again that $\|r\|_2^2 = 2 \Tr(\rho^2) - 1$ as observed above: 
    \[\begin{split}
        \eta =\ & \bar{\gamma} - \max_{\substack{\ket{u} \in \mathbb{R}^3 \\ \braket{u}{u} = 1}} \bra{u} \mathcal{S} \ket{u} = 2\bar{\gamma} - 1 - \max_{\substack{\ket{u} \in \mathbb{R}^3 \\ \braket{u}{u} = 1}} \exval_{\rho \followsdistr D}\left[\braket{u}{r(\rho)}^2\right] \\
        =\ & \exval_{\rho \followsdistr D}\left[\braket{r(\rho)}{r(\rho)}\right] - \max_{\substack{\ket{u} \in \mathbb{R}^3 \\ \braket{u}{u} = 1}} \exval_{\rho \followsdistr D}\left[\braket{r(\rho)}{u} \braket{u}{r(\rho)}\right] 
        = \min_{\substack{\ket{u} \in \mathbb{R}^3 \\ \braket{u}{u} = 1}} \exval_{\rho \followsdistr D}\left[\bra{r(\rho)}\left(I - \ket{u}\bra{u}\right)\ket{r(\rho)}\right].
    \end{split}\]

    Observe finally that $I - \ket{u}\bra{u} = \left(I - \ket{u}\bra{u}\right)^2$ since it is just a projector on the space orthogonal to $\ket{u}$, giving precisely our result:
    \[\begin{split}
        \eta =\ & \min_{\substack{\ket{u} \in \mathbb{R}^3 \\ \braket{u}{u} = 1}} \exval_{\rho \followsdistr D}\left[\bra{r(\rho)}\left(I - \ket{u}\bra{u}\right)^2\ket{r(\rho)}\right] = \min_{\substack{\ket{u} \in \mathbb{R}^3 \\ \braket{u}{u} = 1}} \exval_{\rho \followsdistr D}\left[\left\|(I - \ket{u}\bra{u}) \ket{r(\rho)}\right\|_2^2\right] \\
        =\ & \min_{\substack{\ket{u} \in \mathbb{R}^3 \\ \braket{u}{u} = 1}} \exval_{\rho \followsdistr D}\left[\left\|\ket{r(\rho)} - \braket{u}{r(\rho)}\ket{u}\right\|_2^2\right].
    \end{split}\]
\end{proof}

\subsubsection{Generalised mean matrix}

\begin{definition}[Generalised mean matrix]\label{def:generalised-mean-matrix}
    Let $\ket{m} \in \mathbb{R}^{4^n - 1}$ be an arbitrary real vector in the Pauli basis (i.e.\ $\bbraket{\widetilde{P}}{m} \in \mathbb{R}$ for all $P \in \mathbb{P}_n$). We define the generalised mean matrix of $\ket{m}$, $M_{mean}\left(\ket{m}\right)$, as: 
    \[M_{mean}\left(\ket{m}\right) = \sum_{P \neq I^{\otimes n}} \bbraket{\widetilde{P}}{m} \left(\kket{\widetilde{I}^{\otimes n}}\bbra{\widetilde{P}} + \kket{\widetilde{P}} \bbra{\widetilde{I}^{\otimes n}}\right)\]
\end{definition}
\begin{intuition*}
    Note that, given a distribution $D$ of moment matrix $M = M_{core} + M_{mean}$ and first moment $\kket{\mu}$, then $M_{mean}\left(\kket{\mu}\right) = M_{mean}$ (and yes our notation is a little overloaded but hopefully the reader agrees that it is harmless given this equivalence).
    
    The idea is just that this matrix has the shape of $M_{mean}$, while allowing for an arbitrary real vector:
    \[M_{mean}\left(\ket{m}\right) = \left(\begin{array}{c|ccc} 0 & & \bra{m} & \\ \hline & & & \\ \ket{m} & & 0 & \\ & & & \end{array}\right).\]
    
    Having this more general matrix will simplify some proofs for $M_{mean}$
\end{intuition*}

\begin{lemma}\label{lma:generalised-mean-linear}
    Let $a, b \in \mathbb{R}$ be scalars, and $\ket{m}, \ket{p} \in \mathbb{R}^{4^n - 1}$ be real vectors (in the Pauli basis). Then: 
    \[M_{mean}\left(a \ket{m} + b \ket{p}\right) = a M_{mean}\left(\ket{m}\right) + b M_{mean}\left(\ket{p}\right)\]
    
    In other words, the generalised mean matrix is linear.
\end{lemma}
\begin{proof}[Proof]
    Note that this proposition is trivial in array form. Let us however do the proof formally: 
    \begin{equation}\begin{aligned}M_{mean}\left(a \ket{m} + b \ket{p}\right) &= \sum_{P \neq I^{\otimes n}} \bbra{\widetilde{P}}\left(a\ket{m} + b \ket{p}\right) \left(\kket{\widetilde{I}^{\otimes n}}\bbra{\widetilde{P}} + \kket{\widetilde{P}} \bbra{\widetilde{I}^{\otimes n}}\right) \\ &= a \sum_{P \neq I^{\otimes n}} \bbraket{\widetilde{P}}{m} \left(\kket{\widetilde{I}^{\otimes n}}\bbra{\widetilde{P}} + \kket{\widetilde{P}} \bbra{\widetilde{I}^{\otimes n}}\right) + b \sum_{P \neq I^{\otimes n}} \bbraket{\widetilde{P}}{p} \left(\kket{\widetilde{I}^{\otimes n}}\bbra{\widetilde{P}} + \kket{\widetilde{P}} \bbra{\widetilde{I}^{\otimes n}}\right) \\ &= a M_{mean}\left(\ket{m}\right) + b M_{mean}\left(\ket{p}\right).\end{aligned}\end{equation}
\end{proof}

\subsubsection{$w$ vector}

The risk depends on $U$ and $V$ only through the relative unitary $V^{\dagger}U$. Its normalized vectorization lets us express the risk using the moment matrix. We remark that not every normalized vector represents a unitary, and this restriction matters when checking whether the bounds are tight.

\begin{definition}[$w$ vector]\label{def:w-vector}
    Let $U, V$ be unitary matrices. Then, we let $\ket{w\left(U, V\right)} \in \mathbb{C}^{4^n}$ be:
    \[\ket{w\left(U, V\right)} = \frac{1}{\sqrt{2^n}} \kket{V^{\dagger} U}.\]
\end{definition}
\begin{intuition*}
    Just like for the moment matrix, this definition naturally appears in the construction of lemma~\ref{lma:risk-rayleigh}.
\end{intuition*}

\begin{lemma}\label{lma:w-normalised}
    Let $U, V$ be unitary matrices. Then, $\ket{w\left(U, V\right)}$ is normalised:
    \[\braket{w\left(U, V\right)}{w\left(U, V\right)} = 1.\]
\end{lemma}
\begin{proof}[Proof]
    This is direct, since $U$ and $V$ are unitary:
    \autoeq{\braket{w\left(U, V\right)}{w\left(U, V\right)} = \frac{1}{2^n} \bbrakket{V^{\dagger} U}{V^{\dagger} U} = \frac{1}{2^n} \Tr\left(U^{\dagger} V V^{\dagger} U\right) = \frac{1}{2^n}\Tr\left(I\right) = 1.}
\end{proof}

\suppsubsection{Main results}

\subsubsection{General bounds}

We separate out the identity component and bound the contributions from the second moment and the mean. We first do this for any normalized vector, then specialize to vectors representing unitaries. The resulting out-of-distribution generalisation bounds have a multiplicative factor controlled by the second-moment spectrum and an additive term determined by the means.

\begin{remark}
    The core of the reasoning takes its inspiration from reference~\cite{chen2024predicting}'s proof of their lemma~4.1: we will similarly represent our matrices in the Pauli basis and express risks as Rayleigh quotients, which will allow us to use eigenvalues to compare them. 
\end{remark}

\begin{lemma}\label{lma:risk-rayleigh}
    Let $D$ be an $n$-qubit distribution over mixed states of moment matrix $M$ and purity $\bar{\gamma}$, and let $U, V$ be unitary matrices. Then:
    \[R_D\left(U, V\right) = \bar{\gamma} - \frac{\bbra{V^{\dagger} U}}{\sqrt{2^n}}\cdot  2^n\exval_{\rho \followsdistr D}\left[\rho \otimes \rho^T\right] \cdot \frac{\kket{V^{\dagger} U}}{\sqrt{2^n}} = \bar{\gamma} - \bra{w\left(U, V\right)} M \ket{w\left(U, V\right)}.\]
\end{lemma}
\begin{proof}[Proof]
    This is direct using theorem~\ref{thm:risk-equivalent-formulation}, the cyclicity of the trace and the vectorised notation: 
    \begin{equation}\begin{aligned}R_D\left(U, V\right) &= \bar{\gamma} - \exval\left(\Tr\left(U \rho U^{\dagger} V \rho V^{\dagger}\right)\right) \\ &= \bar{\gamma} - \exval_{\rho \followsdistr D}\left[\Tr\left(U^{\dagger} V \rho V^{\dagger} U \rho\right)\right] \\ &= \bar{\gamma} - \exval_{\rho \followsdistr D}\left[\bbrakket{V^{\dagger}U}{\rho V^{\dagger} U \rho}\right] \\ &= \bar{\gamma} - \frac{\bbra{V^{\dagger} U}}{\sqrt{2^n}}\cdot  2^n\exval_{\rho \followsdistr D}\left[\rho \otimes \rho^T\right]\cdot \frac{\kket{V^{\dagger} U}}{\sqrt{2^n}} \\ &= \bar{\gamma} - \bra{w\left(U, V\right)} M \ket{w\left(U, V\right)}.\end{aligned}\end{equation}
\end{proof}
\begin{intuition*}
    Since $\braket{w\left(U, V\right)}{w\left(U, V\right)} = 1$, we recognise a Rayleigh quotient (Definition~\ref{def:rayleigh-quotient}), getting closer to our goal. We next bound this Rayleigh quotient using eigenvalues. However, there is the issue that: 
    \[\lambda_{max}\left(M\right) \geq \bbra{\widetilde{I}^{\otimes n}} M \kket{\widetilde{I}^{\otimes n}} = \bar{\gamma}.\]
    
    In other words, eigenvalue bounds do not beat the trivial bound $R_D\left(U, V\right) \geq 0$, which will not be enough for an out-of-distribution generalisation bound.
    The current issue in our reasoning comes from the fact that $R_D\left(U, U\right) = 0$ since $\ket{w\left(U, U\right)} = \kket{\widetilde{I}^{\otimes n}}$. However, a useful property of unitaries is that $U = V \iff \ket{w\left(U, V\right)} = \kket{\widetilde{I}^{\otimes n}}$. This issue is therefore entirely packaged in the identity component of $\ket{w\left(U, V\right)}$. The idea will thus be to consider separately the identity component of $\ket{w\left(U, V\right)}$ from its other components.
\end{intuition*}

\begin{definition}[Pseudo-risk]\label{def:pseudo-risk}
    Let $D$ be a mixed-state distribution over $n$ qubits of average purity $\bar{\gamma}$ and moment matrix $M$. For any normalised vector $\ket{x} \in \mathbb{C}^{4^n}$, we define its pseudo-risk as: 
    \[r_D\left(\ket{x}\right) = \bar{\gamma} - \bra{x} M \ket{x}.\]
\end{definition}
\begin{remark*}
    Observe that, thanks to lemma~\ref{lma:risk-rayleigh}, for unitaries $U, V$: 
    \[r_D\left(\ket{w\left(U, V\right)}\right) = R_D\left(U, V\right).\]
    
    The idea will be to make bounds $r_{D_1}\left(\ket{x}\right) \leq \alpha r_{D_2}\left(\ket{x}\right) + \beta$ for all $\ket{x} \in \mathbb{C}^{4^n}$. This however loses some of the structure that $\ket{w\left(U, V\right)}$ is a vectorised unitary, and hence does not guarantee that the bounds are tight.
\end{remark*}

\begin{lemma}\label{lma:risk-without-identity}
    Let $D$ be an $n$-qubit distribution over mixed states of moment matrix $M = \bar{\gamma} \oplus \mathcal{S} + M_{mean}$, and let $\ket{x} \in \mathbb{C}^{4^n}$ be normalised. If $\ket{\pi\left(x\right)} := \Pi \ket{x} \neq 0$, then:
    \[r_D\left(\ket{x}\right) = \left\|\ket{\pi\left(x\right)}\right\|_2^2 \left(\bar{\gamma} - \frac{\bra{\pi\left(x\right)} \mathcal{S} \ket{\pi\left(x\right)}}{\braket{\pi\left(x\right)}{\pi\left(x\right)}}\right) - \bra{x} M_{mean} \ket{x}.\]
\end{lemma}
\begin{remark*}
    In particular, for unitaries $U, V$ that do not differ only by a global phase, writing $\ket{\pi\left(w\right)} := \Pi \ket{w\left(U, V\right)}$ and $\ket{w} := \ket{w\left(U, V\right)}$, this yields: 
    \[R_D\left(U, V\right) = \left\|\ket{\pi\left(w\right)}\right\|_2^2 \left(\bar{\gamma} - \frac{\bra{\pi\left(w\right)} \mathcal{S} \ket{\pi\left(w\right)}}{\braket{\pi\left(w\right)}{\pi\left(w\right)}}\right) - \bra{w} M_{mean} \ket{w}.\]
\end{remark*}

\begin{proof}
    For the sake of notation, let us write:
    \[\ket{\bar{\pi}\left(x\right)} = \frac{\ket{\pi\left(x\right)}}{\left\|\ket{\pi\left(x\right)}\right\|_2}.\]

    Moreover, observe that $I^{\otimes 2n} = \Pi + \kket{\widetilde{I}^{\otimes n}}\bbra{\widetilde{I}^{\otimes n}}$ by definition of $\Pi$, so:
    \[\ket{x} = \kket{\widetilde{I}^{\otimes n}}\bbraket{\widetilde{I}^{\otimes n}}{x} + \Pi \ket{x} = c \kket{\widetilde{I}^{\otimes n}} + \left\|\ket{\pi\left(x\right)}\right\|_2 \frac{\ket{\pi\left(x\right)}}{\left\|\ket{\pi\left(x\right)}\right\|_2} = c \kket{\widetilde{I}^{\otimes n}} + d \ket{\bar{\pi}\left(x\right)},\]
    where we left $c = \bbraket{\widetilde{I}^{\otimes n}}{x}$ and $d = \left\|\ket{\pi\left(x\right)}\right\|_2$. Observe moreover that since $\ket{x}$ is normalised and since $\ket{\bar{\pi}\left(x\right)}$ and $\kket{\widetilde{I}^{\otimes n}}$ are orthonormal: 
    \[1 = \braket{x}{x} = \left|c\right|^2 + \left|d\right|^2.\]
    
    Using the decomposition $M = \bar{\gamma} \oplus \mathcal{S} + M_{mean}$, pseudo-risk becomes: 
    \begin{equation}\begin{aligned}r_D\left(\ket{x}\right) &= \bar{\gamma} - \bra{x} M \ket{x} \\ &= \bar{\gamma} \left(\left|c\right|^2 + \left|d\right|^2\right) - \left(\bbra{\widetilde{I}^{\otimes n}} c^* + \bra{\bar{\pi}\left(x\right)} d^*\right) \left(\bar{\gamma} \oplus \mathcal{S}\right) \left(c\kket{\widetilde{I}^{\otimes n}} + d \ket{\bar{\pi}\left(x\right)} \right) - \bra{x} M_{mean} \ket{x} \\ &= \bar{\gamma} \left(\left|c\right|^2 + \left|d\right|^2\right) - \left|c\right|^2 \bar{\gamma} - \left|d\right|^2 \bra{\bar{\pi}\left(x\right)} \mathcal{S} \ket{\bar{\pi}\left(x\right)} - \bra{x} M_{mean} \ket{x} \\ &= \left|d\right|^2 \left(\bar{\gamma} - \bra{\bar{\pi}\left(x\right)} \mathcal{S} \ket{\bar{\pi}\left(x\right)}\right) - \bra{x} M_{mean} \ket{x},\end{aligned}\end{equation}
    which is precisely what we wanted to prove.
\end{proof}

\begin{remark*}
    We use the expression for the risk in Lemma~\ref{lma:risk-without-identity} to derive
    our out-of-distribution generalisation bounds. Consider again the special case where $M_{mean} = 0$, i.e.\ that distributions are unbiased (definition~\ref{def:bias}):
    \[R_D\left(U, V\right) = \left\|\ket{\pi\left(w\right)}\right\|_2^2 \left(\bar{\gamma} - \frac{\bra{\pi\left(w\right)} \mathcal{S} \ket{\pi\left(w\right)}}{\braket{\pi\left(w\right)}{\pi\left(w\right)}}\right).\]

    We make two observations.
    \begin{enumerate}
        \item We claim that there exist unbiased distributions $D$ whose expressivity measure $\eta = \bar{\gamma} - \lambda_{max}\left(\mathcal{S}\right)$ is positive and sufficiently bounded away from $0$ (i.e.\ $\lambda_{max}\left(\mathcal{S}\right) \ll \bar{\gamma}$). This fact will be proven later, in appendix~\ref{sec:applications}. Recall that this fact is not true if we replace $\mathcal{S}$ by $M$, as observed above in lemma~\ref{lma:risk-rayleigh}. 

        Considering such a distribution $D$, then we have by properties of Rayleigh quotients (definition~\ref{def:rayleigh-quotient}), positive semi-definiteness of $S$ and definition of expressivity measure $\eta = \bar{\gamma} - \lambda_{max}\left(\mathcal{S}\right)$: 
            \[\left\|\ket{\pi\left(w\right)}\right\|_2^2 \eta \leq R_D\left(U, V\right) \leq \left\|\ket{\pi\left(w\right)}\right\|_2^2 \bar{\gamma}.\]

            In other words, the norm of $\ket{\pi\left(w\right)}$ is precisely what defines the scale of the risk for distributions with $\eta = \Theta\left(1\right)$. Recall again that $U = V \iff \ket{w\left(U, V\right)} = \kket{\widetilde{I}^{\otimes n}}$, so this makes a lot of intuitive sense: the risk is zero for the identity component, and then the eigenvalues of $\mathcal{S}$ define how much to penalise the other directions. Moreover, $\eta > 0$ makes sure that all the other directions are penalised, i.e.\ that none is ignored by the distribution.

        \item Consider some other arbitrary unbiased distribution $D'$. Then, since $\mathcal{S}$ is positive semi-definite: 
        \[R_{D'}\left(U, V\right) \leq \left\|\ket{\pi\left(w\right)}\right\|_2^2 \bar{\gamma}'.\]

        Taking a similar intuition as the one above, for a general (unbiased) distribution, $\mathcal{S}$ is not guaranteed to penalise all directions. We know it has to be small if $\ket{w}$ is strongly aligned with $\kket{\widetilde{I}^{\otimes n}}$, but it may also be small in other cases. 

        However, since we know that $R_D \sim \left\|\ket{\pi\left(w\right)}\right\|_2^2$, this immediately gives us an out-of-distribution generalization bound:
        \[R_{D'}\left(U, V\right) \leq \left\|\ket{\pi\left(w\right)}\right\|_2^2 \bar{\gamma}' = \left\|\ket{\pi\left(w\right)}\right\|_2^2 \eta \frac{\bar{\gamma}'}{\eta} \leq \frac{\bar{\gamma}'}{\eta} R_D\left(U, V\right).\]

        This is precisely the reasoning we will use below to make our out-of-distribution generalisation bounds.
    \end{enumerate}
\end{remark*}

Nonzero means contribute an additional term to the risk comparison. We expand this term and bound its absolute value for any normalized vector. Applying the result to the difference between the rescaled means gives the additive error in our bounds.

\begin{lemma}\label{lma:m-mean-expansion}
    Let $\ket{m} \in \mathbb{R}^{4^n - 1}$ be an arbitrary real vector in the Pauli basis, $M_{mean}\left(\ket{m}\right) \in \mathbb{R}^{4^n \times 4^n}$ be its corresponding generalised mean matrix (definition~\ref{def:generalised-mean-matrix}), and $\ket{x} \in \mathbb{C}^{4^n}$ be arbitrary. Then: 
    \[\bra{x} M_{mean}\left(\ket{m}\right) \ket{x} = 2 \sum_{P \neq I^{\otimes n}} \bbraket{\widetilde{P}}{m} \cre\left(\brakket{x}{\widetilde{I}^{\otimes n}} \bbraket{\widetilde{P}}{x}\right).\]
\end{lemma}
\begin{remark*}
    Note formally that, to be able to write $\bbraket{\widetilde{P}}{m}$, we are again using the isomorphism between $\ket{m} \in \mathbb{R}^{4^n - 1}$ and $\ket{m} \in \mathbb{R}^{4^n}$ such that $\bbraket{\widetilde{I}^{\otimes n}}{m} = 0$.
\end{remark*}
\begin{proof}[Proof]
    By definition of $M_{mean}\left(\ket{m}\right)$:
    \[\bra{x} M_{mean}\left(\ket{m}\right) \ket{x} = \sum_{P \neq I^{\otimes n}} \bbraket{\widetilde{P}}{m} \left(\brakket{x}{\widetilde{I}^{\otimes n}}\bbraket{\widetilde{P}}{x} + \brakket{x}{\widetilde{P}} \bbraket{\widetilde{I}^{\otimes n}}{x}\right).\]
    
    However, note that: 
    \[\left(\brakket{x}{\widetilde{P}} \bbraket{\widetilde{I}^{\otimes n}}{x}\right)^* = \brakket{x}{\widetilde{P}}^* \bbraket{\widetilde{I}^{\otimes n}}{x}^* = \bbraket{\widetilde{P}}{x}\brakket{x}{\widetilde{I}^{\otimes n}}.\]
    
    We know that for any $z \in \mathbb{C}$, $z + z^* = 2\cre\left(z\right)$, giving our result.
\end{proof}

\begin{lemma}\label{lma:m-mean-bound}
    Let $\ket{m} \in \mathbb{R}^{4^n - 1}$ be an arbitrary real vector in the Pauli basis, and let $\ket{x} \in \mathbb{C}^{4^n}$ be an arbitrary normalised vector. Then: 
    \[-\left\|\ket{m}\right\|_2 \leq \bra{x} M_{mean}\left(\ket{m}\right) \ket{x} \leq \left\|\ket{m}\right\|_2.\]
    
    Moreover, when $\ket{m} = 0$, both inequalities are always reached with equality. Otherwise, the first inequality is reached for $\ket{x_-}$ and the second for $\ket{x_+}$, where: 
    \[\ket{x_{\pm}} = \pm\frac{1}{\sqrt{2}} \kket{\widetilde{I}^{\otimes n}} + \frac{1}{\sqrt{2} \left\|\ket{m}\right\|_2} \sum_{P \neq I^{\otimes n}} \bbraket{\widetilde{P}}{m} \kket{\widetilde{P}} = \frac{1}{\sqrt{2}} \begin{pmatrix} \pm1 \\ \ket{m} / \left\|\ket{m}\right\|_2 \end{pmatrix} \in \mathbb{R}^{4^n}.\]
\end{lemma}
\begin{proof}[Proof 1]
    By lemma~\ref{lma:m-mean-expansion}: 
    \[\bra{x} M_{mean}\left(\ket{m}\right) \ket{x} = 2 \sum_{P \neq I^{\otimes n}} \bbraket{\widetilde{P}}{m} \cre\left(\brakket{x}{\widetilde{I}^{\otimes n}} \bbraket{\widetilde{P}}{x}\right).\]
    
    Note that, since $\bbraket{\widetilde{P}}{m} \in \mathbb{R}$ for all $P \in \mathbb{P}_n$ by hypothesis, we notably have $\bra{x} M_{mean}\left(\ket{m}\right) \ket{x} \in \mathbb{R}$.
    
    Now using the triangle inequality and the fact that $\left|\cre\left(z\right)\right| \leq \left|z\right|$ for all $z \in \mathbb{C}$: 
    \autoeq{\left|\bra{x} M_{mean}\left(\ket{m}\right) \ket{x}\right| \leq 2 \sum_{P \neq I^{\otimes n}} \left|\bbraket{\widetilde{P}}{m}\right| \left|\cre\left(\brakket{x}{\widetilde{I}^{\otimes n}} \bbraket{\widetilde{P}}{x}\right)\right| \leq 2 \left|\brakket{x}{\widetilde{I}^{\otimes n}}\right| \sum_{P \neq I^{\otimes n}} \left|\bbraket{\widetilde{P}}{m}\right| \left|\bbraket{\widetilde{P}}{x}\right|.}
    
    The Cauchy-Schwarz inequality tells us that, for $a, b \in \mathbb{R}^{m}$, then $\sum_{i} a_i b_i \leq \sqrt{\sum_{i} a_i^2 \sum_{j} b_j^2}$. So: 
    \begin{equation}\begin{aligned}\left|\bra{x} M_{mean}\left(\ket{m}\right) \ket{x}\right| &\leq 2 \left|\brakket{x}{\widetilde{I}^{\otimes n}}\right| \sqrt{\sum_{P \neq I^{\otimes n}} \left|\bbraket{\widetilde{P}}{m}\right|^2 \sum_{Q \neq I^{\otimes n}} \left|\bbraket{\widetilde{Q}}{x}\right|^2} \\ &= 2 \left|\brakket{x}{\widetilde{I}^{\otimes n}}\right| \left\|\ket{m}\right\|_2 \sqrt{\sum_{Q \neq I^{\otimes n}} \left|\bbraket{\widetilde{Q}}{x}\right|^2} \\ &= 2 \left|\brakket{x}{\widetilde{I}^{\otimes n}}\right| \left\|\ket{m}\right\|_2 \sqrt{\sum_{Q \in \mathbb{P}_n} \left|\brakket{x}{\widetilde{Q}}\right|^2 - \left|\brakket{x}{\widetilde{I}^{\otimes n}}\right|^2}.\end{aligned}\end{equation}
    
    Since finally $\braket{x}{x} = 1$, we get: 
    \[\left|\bra{x} M_{mean}\left(\ket{m}\right) \ket{x}\right| \leq 2 \left\|\ket{m}\right\|_2 \left|\brakket{x}{\widetilde{I}^{\otimes n}}\right| \sqrt{1 - \left|\brakket{x}{\widetilde{I}^{\otimes n}}\right|^2}.\]
    
    Using usual analysis tools, we notice that the function $f: \left[0, 1\right] \mapsto \mathbb{R}$ defined by $f\left(t\right) = t\sqrt{1 - t^2}$ has image $\left[0, \frac{1}{2}\right]$, i.e.\ $f\left(t\right) \leq \frac{1}{2}$ for all $t \in \left[0, 1\right]$. Since $\left|\bra{x} M_{mean}\left(\ket{m}\right) \ket{x}\right| \leq 2 \left\|\ket{m}\right\|_2 f\left(\left|\brakket{x}{\widetilde{I}^{\otimes n}}\right|\right)$, and $\bra{x} M_{mean}\left(\ket{m}\right) \ket{x} \in \mathbb{R}$ as mentioned earlier, this yields our result.
\end{proof}
\begin{proof}[Proof 2]
    Let us first verify that both $\ket{x_{\pm}}$ are normalised: 
    \[\braket{x_{\pm}}{x_{\pm}} = \left|\frac{\pm 1}{\sqrt{2}}\right|^2 + \left|\frac{1}{\sqrt{2}}\right|^2 \left\|\frac{\ket{m}}{\left\|\ket{m}\right\|_2}\right\|_2^2 = \frac{1}{2} + \frac{1}{2} = 1.\]
    
    We now verify they indeed saturate the inequalities. Recall that $\ket{m}$ has real coefficients in the Pauli basis by hypothesis. Hence, by lemma~\ref{lma:m-mean-expansion}:  
    \begin{equation}\begin{aligned}\bra{x_{\pm}} M_{mean}\left(\ket{m}\right) \ket{x_{\pm}} &= 2 \sum_{P \neq I^{\otimes n}} \bbraket{\widetilde{P}}{m} \cre\left(\brakket{x_{\pm}}{\widetilde{I}^{\otimes n}} \bbraket{\widetilde{P}}{x_{\pm}}\right) \\ &= 2 \sum_{P \neq I^{\otimes n}} \bbraket{\widetilde{P}}{m} \left(\pm \frac{1}{\sqrt{2}}\right) \left(\frac{1}{\sqrt{2} \left\|\ket{m}\right\|_2}\bbraket{\widetilde{P}}{m}\right) \\ &= \pm \frac{1}{\left\|\ket{m}\right\|_2} \sum_{P \neq I^{\otimes n}} \brakket{m}{\widetilde{P}}\bbraket{\widetilde{P}}{m} \\ &= \pm \frac{1}{\left\|\ket{m}\right\|_2} \sum_{P \in \mathbb{P}_n} \brakket{m}{\widetilde{P}}\bbraket{\widetilde{P}}{m} \\ &= \pm\frac{\braket{m}{m}}{\left\|m\right\|_2} \\ &= \pm\left\|m\right\|_2,\end{aligned}\end{equation}
    where we used that $\bbraket{\widetilde{I}^{\otimes n}}{m} = 0$.
\end{proof}
\begin{intuition*}
    The vectors for the second proof were found by trying to reach all inequalities of the first proof with equality. To do so, we used the fact that $\left|\cre\left(z\right)\right| = \left|z\right| \iff z \in \mathbb{R}$, that the triangle inequality $\left|\sum_{i} a_i\right| \leq \sum_{i} \left|a_i\right|$ is reached with equality when all the terms of the sum have the same sign, that the Cauchy-Schwarz inequality $a \dotprod b \leq \left\|a\right\|_2\cdot \left\|b\right\|_2$ is reached with equality when $a = \gamma b$ for some $\gamma \in \mathbb{R}_{\geq 0}$, and that $f\left(t\right) = t \sqrt{1 - t^2}$ is maximised for $t = \frac{1}{\sqrt{2}}$.
    
    In other words, we wished for all components of $\ket{x}$ to be real, that $\ket{\pi\left(x\right)} = \gamma \ket{m}$ for some $\gamma \in \mathbb{R}$ that preserves $\braket{x}{x} = 1$, and that $\left|\brakket{x}{\widetilde{I}^{\otimes n}}\right| = \frac{1}{\sqrt{2}}$.
\end{intuition*}

We can compare the risks at several levels of detail. Full moment matrices give the most direct condition, while a common eigenbasis allows a direction-by-direction comparison. When this information is unavailable, we can still obtain bounds from the extremal eigenvalues and the means.

\begin{lemma}[Out-of-distribution  bounds for pseudo-risks] \label{lma:x-inequalities}
    Let $D^{\left(1\right)}$ and $D^{\left(2\right)}$ be $n$ qubit distributions over mixed states of moment matrix $M^{\left(i\right)} = \bar{\gamma}^{\left(i\right)} \oplus \mathcal{S}^{\left(i\right)} + M_{mean}^{\left(i\right)}$, mean $\mu^{\left(i\right)}$, and expressivity measure $\eta^{\left(i\right)}$, for $i \in \left\{1, 2\right\}$. Moreover, let $\ket{x} \in \mathbb{C}^{4^n}$ be an arbitrary normalised vector. We write $r_i := r_{D^{\left(i\right)}}$.
    
    Let $\alpha, \alpha', \beta \in \mathbb{R}_{\geq 0}$. All the following propositions hold.
    \begin{enumerate}
    \item We always have that:
    \[\lambda_{min}\left(M^{\left(1\right)} - \alpha M^{\left(2\right)}\right) \geq \bar{\gamma}^{\left(1\right)} - \alpha\bar{\gamma}^{\left(2\right)} - \beta \implies r_1\left(\ket{x}\right) \leq \alpha r_2\left(\ket{x}\right) + \beta.\]

    If moreover the hypothesis is tight, $\lambda_{min}\left(M^{\left(1\right)} - \alpha M^{\left(2\right)}\right) = \bar{\gamma}^{\left(1\right)} - \alpha\bar{\gamma}^{\left(2\right)} - \beta$, then any eigenvector $\ket{x}$ of $M^{\left(1\right)} - \alpha M^{\left(2\right)}$ associated to its smallest eigenvalue is such that the Out-of-distribution  (out-of-distribution generalisation) bound is reached with equality, $r_1\left(\ket{x}\right) = \alpha r_2\left(\ket{x}\right) + \beta$.
    
    \item If $\mathcal{S}^{\left(1\right)}$ and $\mathcal{S}^{\left(2\right)}$ are diagonalisable in the same basis, $\mu^{\left(1\right)} = \mu^{\left(2\right)} = 0$ and $\eta^{\left(2\right)} > 0$:
    \[\alpha \geq \frac{\bar{\gamma}^{\left(1\right)} - \lambda_j\left(\mathcal{S}^{\left(1\right)}\right) - \beta}{\bar{\gamma}^{\left(2\right)} - \lambda_j\left(\mathcal{S}^{\left(2\right)}\right)}, \forall j \implies r_1\left(\ket{x}\right) \leq \alpha r_2\left(\ket{x}\right) + \beta.\]

    If moreover the hypothesis is tight, $\alpha = \max_j \left(\bar{\gamma}^{\left(1\right)} - \lambda_j\left(\mathcal{S}^{\left(1\right)}\right) - \beta\right)/\left(\bar{\gamma}^{\left(2\right)} - \lambda_j\left(\mathcal{S}^{\left(2\right)}\right)\right)$, then any common eigenvector $\ket{x}$ of $\mathcal{S}^{\left(1\right)}$ and $\mathcal{S}^{\left(2\right)}$ associated to the eigenvalues that maximise $\left(\bar{\gamma}^{\left(1\right)} - \lambda_j\left(\mathcal{S}^{\left(1\right)}\right) - \beta\right)/\left(\bar{\gamma}^{\left(2\right)} - \lambda_j\left(\mathcal{S}^{\left(2\right)}\right)\right)$ is such that the Out-of-distribution  bound is reached with equality, $r_1\left(\ket{x}\right) = \alpha r_2\left(\ket{x}\right) + \beta$.
    
    \item If $\eta^{\left(2\right)} > 0$ and $r_2\left(\ket{x}\right) \geq 0$: 
        \[\alpha \geq \alpha' \geq \frac{\bar{\gamma}^{\left(1\right)} - \lambda_{min}\left(\mathcal{S}^{\left(1\right)}\right)}{\eta^{\left(2\right)}} \implies r_1\left(\ket{x}\right) \leq \alpha r_2\left(\ket{x}\right) + \left\|\mu^{\left(1\right)} - \alpha' \mu^{\left(2\right)}\right\|_2.\]
    
    \item If $\eta^{\left(2\right)} > 0$: 
    \[\alpha \geq \frac{\bar{\gamma}^{\left(1\right)} - \lambda_{min}\left(\mathcal{S}^{\left(1\right)}\right)}{\eta^{\left(2\right)}} \implies r_1\left(\ket{x}\right) \leq \alpha\left(r_2\left(\ket{x}\right) + \left\|\mu^{\left(2\right)}\right\|_2\right) + \left\|\mu^{\left(1\right)}\right\|_2.\]
    \end{enumerate}
\end{lemma}
\begin{remark*}
    Note that, by properties of the second moment (lemma~\ref{lma:s-properties}), when both distributions are pure and $\eta^{\left(2\right)} > 0$: 
    \[\frac{1 - \lambda_{min}\left(\mathcal{S}^{\left(1\right)}\right)}{\eta^{\left(2\right)}} = \frac{1 - \lambda_{min}\left(\mathcal{S}^{\left(1\right)}\right)}{1 - \lambda_{max}\left(\mathcal{S}^{\left(2\right)}\right)} \geq \frac{1 - \frac{1}{2^n + 1}}{1 - \frac{1}{2^n + 1}} = 1.\]
    
    Therefore, if both distributions are pure-state distributions, then we can never take $\alpha < 1$ or $\alpha' < 1$ for the third and fourth inequalities.
\end{remark*}
\begin{proof}[Proof 1]
    We suppose by hypothesis that $\lambda_{min}\left(M^{\left(1\right)} - \alpha M^{\left(2\right)}\right) \geq \bar{\gamma}^{\left(1\right)} - \alpha\bar{\gamma}^{\left(2\right)} - \beta$. We know by properties of Rayleigh quotients (definition~\ref{def:rayleigh-quotient}) that: 
    \[\bra{x} \left(M^{\left(1\right)} - \alpha M^{\left(2\right)}\right)\ket{x} \geq \lambda_{min}\left(M^{\left(1\right)} - \alpha M^{\left(2\right)}\right).\]
    
    Hence, by hypothesis:
    \autoeq{\bra{x} \left(M^{\left(1\right)} - \alpha M^{\left(2\right)}\right)\ket{x} \geq \bar{\gamma}^{\left(1\right)} - \alpha\bar{\gamma}^{\left(2\right)} - \beta \iff \alpha\bar{\gamma}^{\left(2\right)} - \alpha \bra{x}M^{\left(2\right)}\ket{x} + \beta \geq \bar{\gamma}^{\left(1\right)} - \bra{x} M^{\left(1\right)} \ket{x} \iff \alpha r_2\left(\ket{x}\right) + \beta \geq r_1\left(\ket{x}\right).}

    All inequalities become equalities if $\lambda_{min}\left(M^{\left(1\right)} - \alpha M^{\left(2\right)}\right) = \bar{\gamma}^{\left(1\right)} - \alpha\bar{\gamma}^{\left(2\right)} - \beta$ and $\ket{x}$ is an eigenvector of $M^{\left(1\right)} - \alpha M^{\left(2\right)}$ associated to its smallest eigenvalue.
\end{proof}
\begin{proof}[Proof 2]
    Note that, since $\mu^{\left(1\right)} = \mu^{\left(2\right)} = 0$, we have $M^{\left(i\right)} = M_{core}^{\left(i\right)} = \bar{\gamma}^{\left(i\right)} \oplus \mathcal{S}^{\left(i\right)}$. Since $\mathcal{S}^{\left(1\right)}$ and $\mathcal{S}^{\left(2\right)}$ are diagonalisable in the same basis, this tells us $M^{\left(1\right)}$ and $M^{\left(2\right)}$ are also diagonalisable in the same basis, of eigenvalues $\lambda_0\left(M^{\left(i\right)}\right) = \bar{\gamma}^{\left(i\right)}$, and $\lambda_j\left(M^{\left(i\right)}\right) = \lambda_j\left(\mathcal{S}^{\left(i\right)}\right)$ for $j \geq 1$.
    
    By hypothesis, we know that $\eta^{\left(2\right)} > 0 \iff \lambda_j\left(M^{\left(2\right)}\right) \leq \lambda_{max}\left(\mathcal{S}^{\left(2\right)}\right) < \bar{\gamma}^{\left(2\right)}$, and hence $\bar{\gamma}^{\left(2\right)} - \lambda_j\left(M^{\left(2\right)}\right) > 0$ for all $j \geq 1$. Also by hypothesis, we get that for all $j \geq 1$: 
        \begin{equation}\begin{aligned}&\alpha \geq \frac{\bar{\gamma}^{\left(1\right)} - \lambda_j\left(M^{\left(1\right)}\right) - \beta}{\bar{\gamma}^{\left(2\right)} - \lambda_j\left(M^{\left(2\right)}\right)} \\ &\iff \alpha \bar{\gamma}^{\left(2\right)} - \alpha \lambda_j\left(M^{\left(2\right)}\right) \geq \bar{\gamma}^{\left(1\right)} - \lambda_j\left(M^{\left(1\right)}\right) - \beta \\ &\iff \lambda_j\left(M^{\left(1\right)}\right) - \alpha \lambda_j\left(M^{\left(2\right)}\right) \geq \bar{\gamma}^{\left(1\right)} - \alpha \bar{\gamma}^{\left(2\right)} - \beta \\ &\iff \lambda_j\left(M^{\left(1\right)} - \alpha M^{\left(2\right)}\right) \geq \bar{\gamma}^{\left(1\right)} - \alpha \bar{\gamma}^{\left(2\right)} - \beta,\end{aligned}\end{equation}
    where, for the last step, we used the fact that $M^{\left(1\right)}$ and $M^{\left(2\right)}$ are diagonalisable in the same basis by hypothesis; and hence, considering the basis where they are both diagonal, the $j$\Th diagonal component of $M^{\left(1\right)} - \alpha M^{\left(2\right)}$ (i.e.\ its $j$\Th eigenvalue) can be computed using the $j$\Th diagonal component of $M^{\left(1\right)}$ and $M^{\left(2\right)}$.
    
    Now, for $j = 0$, we also have: 
    \[\lambda_0\left(M^{\left(1\right)} - \alpha M^{\left(2\right)}\right) = \lambda_0\left(M^{\left(1\right)}\right) - \alpha \lambda_0\left(M^{\left(2\right)}\right) = \bar{\gamma}^{\left(1\right)} - \alpha\cdot \bar{\gamma}^{\left(2\right)} \geq \bar{\gamma}^{\left(1\right)} - \alpha \bar{\gamma}^{\left(2\right)} - \beta,\]
    since $\beta \geq 0$ by hypothesis. 

    Since it holds for all the eigenvalues, we have in particular shown $\lambda_{min}\left(M^{\left(1\right)} - \alpha M^{\left(2\right)}\right) \geq \bar{\gamma}^{\left(1\right)} - \alpha\bar{\gamma}^{\left(2\right)} - \beta$. This shows that we have the hypotheses to use the first inequality, giving the result.

    Moreover, when the hypothesis is tight, one can easily check that the index that maximises $\left(\bar{\gamma}^{\left(1\right)} - \lambda_j\left(M^{\left(1\right)}\right) - \beta\right)/\left(\bar{\gamma}^{\left(2\right)} - \lambda_j\left(M^{\left(2\right)}\right)\right)$ is the same as the one that minimises $\lambda_j\left(M^{\left(1\right)} - \alpha M^{\left(2\right)}\right)$. Picking the eigenvector associated to these eigenvalues thus turns all inequalities into equalities by point (1).
\end{proof}
\begin{proof}[Proof 3]
    If $\Pi\ket{x} = 0$, then $r_1(\ket{x}) = r_2(\ket{x}) = 0$ and the bound is immediate. Otherwise, by lemma~\ref{lma:risk-without-identity}, recall that we can write: 
    \[r_i\left(\ket{x}\right) = \left\|\ket{\pi\left(x\right)}\right\|_2^2 \left(\bar{\gamma}^{\left(i\right)} - \frac{\bra{\pi\left(x\right)} \mathcal{S}^{\left(i\right)} \ket{\pi\left(x\right)}}{\braket{\pi\left(x\right)}{\pi\left(x\right)}}\right) - \bra{x} M_{mean}^{\left(i\right)} \ket{x} = \left\|\ket{\pi\left(x\right)}\right\|_2^2 \left(\bar{\gamma}^{\left(i\right)} - q^{\left(i\right)}\left(\ket{x}\right)\right) - \bra{x} M_{mean}^{\left(i\right)} \ket{x},\]
    where we left $q^{\left(i\right)} = \left(\bra{\pi\left(x\right)} \mathcal{S}^{\left(i\right)} \ket{\pi\left(x\right)}\right)/\left(\braket{\pi\left(x\right)}{\pi\left(x\right)}\right)$, which we recognise to be a Rayleigh quotient. We bound each term independently. 
    \begin{enumerate}[left=0pt]
        \item We consider the first term. By properties of Rayleigh quotients (definition~\ref{def:rayleigh-quotient}), we know: 
        \[\lambda_{min}\left(\mathcal{S}^{\left(i\right)}\right) \leq q^{\left(i\right)}\left(\ket{x}\right) \leq \lambda_{max}\left(\mathcal{S}^{\left(i\right)}\right).\]
         
        Therefore, using the fact $\eta^{\left(2\right)} > 0 \iff \bar{\gamma}^{\left(2\right)} - \lambda_{max}\left(\mathcal{S}^{\left(2\right)}\right) > 0$ by hypothesis and the fact $\lambda_{min}\left(\mathcal{S}^{\left(1\right)}\right) \leq \bar{\gamma}^{\left(1\right)} \iff \bar{\gamma}^{\left(1\right)} - \lambda_{min}\left(\mathcal{S}^{\left(1\right)}\right) \geq 0$ by lemma~\ref{lma:s-properties}: 
        \begin{equation}\begin{aligned}\left\|\ket{\pi\left(x\right)}\right\|_2^2 \left(\bar{\gamma}^{\left(1\right)} - q^{\left(1\right)}\left(\ket{x}\right)\right) &\leq \left\|\ket{\pi\left(x\right)}\right\|_2^2 \left(\bar{\gamma}^{\left(1\right)} - \lambda_{min}\left(\mathcal{S}^{\left(1\right)}\right)\right) \\ &= \left\|\ket{\pi\left(x\right)}\right\|_2^2 \left(\bar{\gamma}^{\left(2\right)} - \lambda_{max}\left(\mathcal{S}^{\left(2\right)}\right)\right)\cdot \frac{\bar{\gamma}^{\left(1\right)} - \lambda_{min}\left(\mathcal{S}^{\left(1\right)}\right)}{\bar{\gamma}^{\left(2\right)} - \lambda_{max}\left(\mathcal{S}^{\left(2\right)}\right)} \\ &\leq \alpha' \left\|\ket{\pi\left(x\right)}\right\|_2^2 \left(\bar{\gamma}^{\left(2\right)} - \lambda_{max}\left(\mathcal{S}^{\left(2\right)}\right)\right).\end{aligned}\end{equation}
        \item We now consider the second term. By definition of generalised mean matrix, we know that: 
        \[M_{mean}^{\left(i\right)} = M_{mean}\left(\kket{\mu^{\left(i\right)}}\right).\]
        
        By linearity (lemma~\ref{lma:generalised-mean-linear}), this means that: 
        \[M_{mean}^{\left(1\right)} - \alpha' M_{mean}^{\left(2\right)} = M_{mean}\left(\kket{\mu^{\left(1\right)}} - \alpha' \kket{\mu^{\left(2\right)}}\right) = M_{mean}\left(\kket{\mu^{\left(1\right)} - \alpha'\mu^{\left(2\right)}}\right).\]
        
        By lemma~\ref{lma:m-mean-bound}, and since $\left\|\kket{A}\right\|_2 = \left\|A\right\|_2$ by properties of vectorisation, this thus implies: 
        \autoeq{-\left\|\mu^{\left(1\right)} - \alpha' \mu^{\left(2\right)}\right\|_2 \leq \bra{x} \left(M_{mean}^{\left(1\right)} - \alpha' M_{mean}^{\left(2\right)}\right) \ket{x} \leq \left\|\mu^{\left(1\right)} - \alpha' \mu^{\left(2\right)}\right\|_2.}
        
        This means that:
        \autoeq{- \bra{x} M_{mean}^{\left(1\right)} \ket{x} = -\alpha' \bra{x} M_{mean}^{\left(2\right)} \ket{x} + \bra{x}\left(\alpha' M_{mean}^{\left(2\right)} - M_{mean}^{\left(1\right)}\right) \ket{x} \leq -\alpha' \bra{x} M_{mean}^{\left(2\right)} \ket{x} + \left\|\alpha' \mu^{\left(2\right)} - \mu^{\left(1\right)}\right\|_2.}
    \end{enumerate}
    
    Putting both together, we get, using the fact $\alpha \geq \alpha'$ and that $r_2\left(\ket{x}\right) \geq 0$ by hypothesis:
    \begin{equation}\begin{aligned}\label{eq:odg-bound-for-pseudo-risk-with-alpha-alphaprime} r_1\left(\ket{x}\right) &= \left\|\ket{\pi\left(x\right)}\right\|_2^2 \left(\bar{\gamma}^{\left(1\right)} - q^{\left(1\right)}\left(\ket{x}\right)\right) - \bra{x} M_{mean}^{\left(1\right)} \ket{x} \\ &\leq \alpha' \left\|\ket{\pi\left(x\right)}\right\|_2^2 \left(\bar{\gamma}^{\left(2\right)} - q^{\left(2\right)}\left(\ket{x}\right)\right) - \alpha' \bra{x} M_{mean}^{\left(2\right)} \ket{x} + \left\|\alpha' \mu^{\left(2\right)} - \mu^{\left(1\right)}\right\|_2 \\ &= \alpha' r_2\left(\ket{x}\right) + \left\|\alpha' \mu^{\left(2\right)} - \mu^{\left(1\right)}\right\|_2 \\ &\leq \alpha r_2\left(\ket{x}\right) + \left\|\alpha' \mu^{\left(2\right)} - \mu^{\left(1\right)}\right\|_2.\end{aligned}\end{equation}
    
    This gives our result, since $\left\|\ket{y}\right\|_2 = \left\|-\ket{y}\right\|_2$ for any $\ket{y} \in \mathbb{C}^{m}$.
\end{proof}
\begin{proof}[Proof 4]
    Terminating equation~\ref{eq:odg-bound-for-pseudo-risk-with-alpha-alphaprime} early (getting rid of the hypothesis that $r_2\left(\ket{x}\right) \geq 0$), we know:
    \[r_1\left(\ket{x}\right) \leq \alpha r_2\left(\ket{x}\right) + \left\|\mu^{\left(1\right)} - \alpha \mu^{\left(2\right)}\right\|.\]
    
    But then, using the triangle inequality: 
    \autoeq{r_1\left(\ket{x}\right) \leq \alpha r_2\left(\ket{x}\right) + \left\|\mu^{\left(1\right)}\right\| + \alpha \left\|\mu^{\left(2\right)}\right\| = \alpha \left(r_2\left(\ket{x}\right) + \left\|\mu^{\left(2\right)}\right\|\right) + \left\|\mu^{\left(1\right)}\right\|.}
\end{proof}

The pseudo-risk bounds give bounds on the unitary-learning risk by restricting to normalized vectorized unitaries. However, this restriction can prevent the bounds from being tight. A vector that saturates a spectral bound must also represent a unitary for the corresponding equality to be attained by a pair of unitaries.

\begin{theorem}[Out-of-distribution generalization bounds] \label{thm:inequalities}
    Let $D^{\left(1\right)}$ and $D^{\left(2\right)}$ be $n$ qubit distributions over mixed states of moment matrix $M^{\left(i\right)} = \bar{\gamma}^{\left(i\right)} \oplus \mathcal{S}^{\left(i\right)} + M_{mean}^{\left(i\right)}$, mean $\mu^{\left(i\right)}$, and expressivity measure $\eta^{\left(i\right)}$, for $i \in \left\{1, 2\right\}$. Moreover, let $U, V$ be arbitrary unitaries. We write $R_i := R_{D^{\left(i\right)}}$.
    
    Let $\alpha, \alpha', \beta \in \mathbb{R}_{\geq 0}$. All the following propositions hold.
    \begin{enumerate}
    \item We always have that:
    \[\lambda_{min}\left(M^{\left(1\right)} - \alpha M^{\left(2\right)}\right) \geq \bar{\gamma}^{\left(1\right)} - \alpha\bar{\gamma}^{\left(2\right)} - \beta \implies R_1\left(U, V\right) \leq \alpha R_2\left(U, V\right) + \beta.\]

    If moreover the hypothesis is tight, $\lambda_{min}\left(M^{\left(1\right)} - \alpha M^{\left(2\right)}\right) = \bar{\gamma}^{\left(1\right)} - \alpha\bar{\gamma}^{\left(2\right)} - \beta$, and if there is an eigenvector $\ket{x}$ of $M^{\left(1\right)} - \alpha M^{\left(2\right)}$ associated to its smallest eigenvalue such that $\ket{x} = \frac{1}{\sqrt{2^n}} \kket{W}$ for some unitary $W$, then there exist unitaries $U, V$ such that the out-of-distribution  generalisation bound is reached with equality, $R_1\left(U, V\right) = \alpha R_2\left(U, V\right) + \beta$.
    
    \item If $\mathcal{S}^{\left(1\right)}$ and $\mathcal{S}^{\left(2\right)}$ are diagonalisable in the same basis, $\mu^{\left(1\right)} = \mu^{\left(2\right)} = 0$ and $\eta^{\left(2\right)} > 0$:
    \[\alpha \geq \frac{\bar{\gamma}^{\left(1\right)} - \lambda_j\left(\mathcal{S}^{\left(1\right)}\right) - \beta}{\bar{\gamma}^{\left(2\right)} - \lambda_j\left(\mathcal{S}^{\left(2\right)}\right)}, \forall j \implies R_1\left(U, V\right) \leq \alpha R_2\left(U, V\right) + \beta.\]

    If moreover the hypothesis is tight, $\alpha = \max_j \left(\bar{\gamma}^{\left(1\right)} - \lambda_j\left(\mathcal{S}^{\left(1\right)}\right) - \beta\right)/\left(\bar{\gamma}^{\left(2\right)} - \lambda_j\left(\mathcal{S}^{\left(2\right)}\right)\right)$, and if there is an eigenvector $\ket{x}$ of $\mathcal{S}^{\left(1\right)}$ and $\mathcal{S}^{\left(2\right)}$ associated to the eigenvalues that maximise $\left(\bar{\gamma}^{\left(1\right)} - \lambda_j\left(\mathcal{S}^{\left(1\right)}\right) - \beta\right)/\left(\bar{\gamma}^{\left(2\right)} - \lambda_j\left(\mathcal{S}^{\left(2\right)}\right)\right)$ such that $\ket{x} = \frac{1}{\sqrt{2^n}}\kket{W}$ for some unitary $W$, then there exist unitaries $U, V$ such that the out-of-distribution generalization bound is reached with equality, $R_1\left(U, V\right) = \alpha R_2\left(U, V\right) + \beta$.
    
    \item If $\eta^{\left(2\right)} > 0$: 
        \[\alpha \geq \alpha' \geq \frac{\bar{\gamma}^{\left(1\right)} - \lambda_{min}\left(\mathcal{S}^{\left(1\right)}\right)}{\eta^{\left(2\right)}} \implies R_1\left(U, V\right) \leq \alpha R_2\left(U, V\right) + \left\|\mu^{\left(1\right)} - \alpha' \mu^{\left(2\right)}\right\|_2.\]
    
    \item If $\eta^{\left(2\right)} > 0$: 
    \[\alpha \geq \frac{\bar{\gamma}^{\left(1\right)} - \lambda_{min}\left(\mathcal{S}^{\left(1\right)}\right)}{\eta^{\left(2\right)}} \implies R_1\left(U, V\right) \leq \alpha\left(R_2\left(U, V\right) + \left\|\mu^{\left(2\right)}\right\|_2\right) + \left\|\mu^{\left(1\right)}\right\|_2.\]
    \end{enumerate}
\end{theorem}
\begin{remark*}
    The first inequality is very hard to use. The second one is harder to use than the third and the fourth, and may not even be tighter, but has the great advantage of giving a candidate for verifying the tightness of the resulting inequality.
\end{remark*}
\begin{remark*}
    The fourth inequality yields theorem~\ref{thm-main-text:general-unitary-comparison} in the main text, by letting $\beta = 0$ and by upper-bounding $\bar{\gamma}^{\left(1\right)} - \lambda_{min}\left(\mathcal{S}^{\left(1\right)}\right) \leq \bar{\gamma}^{\left(1\right)} \leq 1$. 
    Similarly, theorem~\ref{thm-main-text:tight-unitary-comparison} in the main text is a special case of the second inequality. Indeed, the main-text theorem requires that both $\mathcal{S}^{\left(1\right)}$ and $\mathcal{S}^{\left(2\right)}$ are diagonalisable in the Pauli basis. This implies that they admit a common eigenbasis of vectorised unitary matrices, which precisely fulfils the hypothesis of the appendix theorem that some specific eigenvector is the vectorisation of some unitary matrix.
\end{remark*}
\begin{proof}[Proof]
    This theorem is a direct consequence of lemma~\ref{lma:x-inequalities}. Indeed, recall that, by lemma~\ref{lma:risk-rayleigh}: 
    \[R_i\left(U, V\right) = \bar{\gamma}^{\left(i\right)} - \bra{w\left(U, V\right)} M^{\left(i\right)} \ket{w\left(U, V\right)} = r_i\left(\ket{w\left(U, V\right)}\right).\]
    
    But then, our theorem is just a special case of this lemma, with some specific $\ket{x} = \ket{w\left(U, V\right)}$. Note moreover that, in this case, we always have the hypothesis that $r_i\left(\ket{w\left(U, V\right)}\right) = R_i\left(U, V\right) \geq 0$ for the third inequality, by definition of risk (definition~\ref{def:risk}). 

    Finally, note that if there exists a $\ket{x}$ such that $\ket{x} = \frac{1}{\sqrt{2^n}} \kket{W}$ for some unitary $W$ and $r_1\left(\ket{x}\right) = \alpha r_2\left(\ket{x}\right) + \beta$, then we can pick any unitaries such that $\ket{w\left(U, V\right)} = \frac{1}{\sqrt{2^n}} \kket{W} \iff V^{\dagger} U = W$ to get $R_1\left(U, V\right) = \alpha R_2\left(U, V\right) + \beta$. An example is $U = W$ and $V = I$. However, if there does not exist such a unitary $W$ then, in general, we are not guaranteed tightness.
\end{proof}

A positive expressivity measure gives a comparison even when the target spectrum is unknown. We use only that its smallest eigenvalue is nonnegative. The resulting bound depends on the two biases and can be sharpened if the target's average purity is known.

\begin{corollary}\label{crl:zero-mean-generalisation}
    Let $D$ be some $n$-qubit distribution over mixed states, of bias $\left\|\mu\right\|_2$, expressivity measure $\eta$. If $\eta > 0$, then for all $n$-qubit distributions over mixed states $D'$ that have bias $\left\|\mu'\right\|_2$ and for all unitaries $U, V$: 
    \[R_{D'}\left(U, V\right) \leq \frac{1}{\eta}  R_D\left(U, V\right) + \left(\frac{1}{\eta} \left\|\mu\right\|_2 + \left\|\mu'\right\|_2\right).\]
    
    If moreover $D'$ is known to have average purity $\bar{\gamma}'$, then:
    \[R_{D'}\left(U, V\right) \leq \frac{\bar{\gamma}'}{\eta}  R_D\left(U, V\right) + \left(\frac{\bar{\gamma}'}{\eta} \left\|\mu\right\|_2 + \left\|\mu'\right\|_2\right).\]
\end{corollary}
\begin{intuition*}
    This theorem has a big interest if $D$ is unbiased. The theorem then states that, for any $D'$ of bias $\left\|\mu'\right\|_2$: 
    \[R_{D'}\left(U, V\right) \leq \frac{1}{\eta} R_D\left(U, V\right) + \left\|\mu'\right\|_2.\]

     In the particular case where $D'$ is also unbiased, this states that, if $D$ yields a small error on some $U, V$, then any zero-mean distribution $D'$ will also yield a small error on $U, V$. Reciprocally, if there is some zero-mean distribution $D'$ that yields a big risk on some $U, V$, then $D$ will also yield a big risk on those $U, V$.
    
    As later shown in corollary~\ref{crl:haar-zero-mean-generalisation}, the Haar distribution is zero mean and has expressivity measure $\eta \geq \frac{2}{3}$. Thus, having a small error on the Haar distribution means having a small error on any zero-mean distribution.
    
    This is nuanced by section~\ref{sec:bound-support-size}, where we show that asking for an unbiased pure-state distribution is a big requirement.
\end{intuition*}
\begin{remark*}
    Recall that, for a classical distribution $D$, we necessarily have $\eta \leq 0$ (lemma~\ref{lma:eta-properties}). $D$ thus has to be a quantum distribution.
\end{remark*}
\begin{remark*}
    This is precisely theorem~\ref{thm-main-text:general-unitary-comparison} in the main text.
\end{remark*}
\begin{proof}[Proof]
    We relabel our distributions, $D^{\left(2\right)} = D$ and $D^{\left(1\right)} = D'$, for clarity. Let $\kket{\mu^{\left(i\right)}}$ be their respective first moments, $\mathcal{S}^{\left(i\right)}$ be their respective second moments, $\eta^{\left(i\right)}$ be their respective expressivity measures, and $\bar{\gamma}^{\left(i\right)}$ be their respective average purities.
    
    Let $\alpha = \bar{\gamma}^{\left(1\right)}/\eta^{\left(2\right)}$. Since $\eta^{(2)}>0$ and
$\lambda_{\min}(\mathcal{S}^{(1)})\geq 0$, we have
\[
\alpha
= \frac{\bar{\gamma}^{(1)}}{\eta^{(2)}}
= \frac{\bar{\gamma}^{(1)}}
       {\bar{\gamma}^{(2)}-\lambda_{\max}(\mathcal{S}^{(2)})}
\geq
\frac{\bar{\gamma}^{(1)}-\lambda_{\min}(\mathcal{S}^{(1)})}
     {\bar{\gamma}^{(2)}-\lambda_{\max}(\mathcal{S}^{(2)})}
=
\frac{\bar{\gamma}^{(1)}-\lambda_{\min}(\mathcal{S}^{(1)})}
     {\eta^{(2)}}.
\]
    
    We then immediately find thanks to the fourth inequality of theorem~\ref{thm:inequalities}:
    \[R_{D^{\left(1\right)}}\left(U, V\right) \leq \frac{\bar{\gamma}^{\left(1\right)}}{\eta^{\left(2\right)}}  R_{D^{\left(2\right)}}\left(U, V\right) + \left(\frac{\bar{\gamma}^{\left(1\right)}}{\eta^{\left(2\right)}} \left\|\mu^{\left(2\right)}\right\|_2 + \left\|\kket{\mu}^{\left(1\right)}\right\|_2\right).\]

    Our other inequality directly comes from the fact $\bar{\gamma}^{\left(1\right)} \leq 1$.
\end{proof}

The bound can remain purely multiplicative even when the means are nonzero, provided they cancel under the rescaling used to compare the risks. The following corollary shows this for pure-state distributions and gives the spectral condition that the rescaling must satisfy.

\begin{corollary}\label{crl:collinear-generalisation}
    Let $D^{\left(1\right)}$ and $D^{\left(2\right)}$ be $n$ qubit distributions over pure states of first moment $\mu^{\left(i\right)}$, second moment $\mathcal{S}^{\left(i\right)}$, average purity $\bar{\gamma}^{\left(i\right)}$ and expressivity measure $\eta^{\left(i\right)}$. Also, let $U, V$ be unitary operators.
    
    If $\eta^{\left(2\right)} > 0$ and there exists some $\xi > 0$ such that $\mu^{\left(1\right)} = \xi \mu^{\left(2\right)}$ (i.e.\ they are positively collinear) and such that:
    \[\xi \geq \frac{\bar{\gamma}^{\left(1\right)} - \lambda_{min}\left(\mathcal{S}^{\left(1\right)}\right)}{\eta^{\left(2\right)}}.\]
    
    Then:
    \[R_1\left(U, V\right) \leq \xi R_2\left(U, V\right).\]
\end{corollary}
\begin{intuition*}
    In other words, we suppose that $\kket{\mu^{\left(1\right)}}$ and $\kket{\mu^{\left(2\right)}}$ are positively collinear, i.e.\ that they lie on the same line. The hypothesis on $\xi$ requires that $\kket{\mu^{\left(1\right)}}$ is farther from the origin than $\kket{\mu^{\left(2\right)}}$ (since it notably requires $\xi \geq 1$, as mentioned in the remark of lemma~\ref{lma:x-inequalities}), by a factor sufficiently big. In that case, if $\xi$ is reasonably small, having a small error on $D^{\left(2\right)}$ yields a small error on $D^{\left(1\right)}$.
\end{intuition*}
\begin{proof}[Proof]
    Let $\alpha = \alpha' = \xi$. We have the hypotheses to use the third inequality of theorem~\ref{thm:inequalities}, yielding: 
    \[R_1\left(U, V\right) \leq \xi R_2\left(U, V\right) + \left\|\mu^{\left(1\right)} - \xi \mu^{\left(2\right)}\right\|_2.\]
    
    However, by hypothesis: 
    \[\mu^{\left(1\right)} - \xi \mu^{\left(2\right)} = \xi \mu^{\left(2\right)} - \xi \mu^{\left(2\right)} = 0.\]
    
    This yields exactly our result.
\end{proof}

\subsubsection{The bias does not contribute for a single qubit}

Here we show that if two distributions have the same first and second moment, they yield the same risk. Moreover, if $n = 1$, it suffices that their second moments are equal for them to yield the same risk. The first result is reasonably easy, the second one may however be surprising.

\begin{theorem}\label{thm:same-moment-matrix-same-risk}
    Let $D^{\left(1\right)}$ and $D^{\left(2\right)}$ be $n$-qubit distributions over pure states, of first moment $\mu^{\left(i\right)}$ and second moment $\mathcal{S}^{\left(i\right)}$, respectively.
    
    The two following propositions are true:
    \begin{enumerate}
    \item If $\mu^{\left(1\right)} = \mu^{\left(2\right)}$ and $\mathcal{S}^{\left(1\right)} = \mathcal{S}^{\left(2\right)}$, then for all unitaries $U, V$: 
    \[R_1\left(U, V\right) = R_2\left(U, V\right).\]
    \item If $n = 1$ and $\mathcal{S}^{\left(1\right)} = \mathcal{S}^{\left(2\right)}$, then for all unitaries $U, V$:  
    \[R_1\left(U, V\right) = R_2\left(U, V\right).\]
    \end{enumerate}
\end{theorem}
\begin{remark*}
    As explained in the main text, this result implies corollary~\ref{corol-main-text:single-qubit-guarantee} (and hence corollary~\ref{corol-main-text:single-qubit-haar-guarantee}) by a symmetrisation argument: given a distribution $D$, we can construct a distribution $D'$ that samples from $D$ and flips the Bloch vector with probability $1/2$. This $D'$ can be shown to have the same expressivity as $D$, but zero bias. Our theorem however states that they have the same risk. Hence, combining it with corollary~\ref{crl:zero-mean-generalisation} yields corollary~\ref{corol-main-text:single-qubit-guarantee} in the main text.
\end{remark*}
\begin{proof}[Proof 1]
    Let $U, V$ be arbitrary unitaries. Let $\bar{\gamma}^{\left(i\right)}$ be the average purity of distribution $D^{\left(i\right)}$. Observe that, by lemma~\ref{lma:s-properties}, since $\mathcal{S}^{\left(1\right)} = \mathcal{S}^{\left(2\right)}$ by hypothesis: 
    \[\bar{\gamma}^{\left(1\right)} = 2^n - \Tr\left(\mathcal{S}^{\left(1\right)}\right) = 2^n - \Tr\left(\mathcal{S}^{\left(2\right)}\right) = \bar{\gamma}^{\left(2\right)}.\]
    
        Since $\bar{\gamma}^{\left(i\right)}, \mathcal{S}^{\left(i\right)}, \mu^{\left(i\right)}$ are the same for $i = 1$ and $i = 2$, we know that $M^{\left(1\right)} = M^{\left(2\right)}$, where $M^{\left(i\right)}$ is the moment matrix of $D^{\left(i\right)}$. Therefore, for all $\ket{x} \in \mathbb{C}^{4^n}$: 
    \[\bar{\gamma}^{\left(1\right)} - \bra{x} M^{\left(1\right)} \ket{x} = \bar{\gamma}^{\left(2\right)} - \bra{x} M^{\left(2\right)} \ket{x}.\]
    
    This is in particular true for $\ket{x} = \ket{w\left(U, V\right)}$. But then, by lemma~\ref{lma:risk-rayleigh}:
    \[R_1\left(U, V\right) = \bar{\gamma}^{\left(1\right)} - \bra{w\left(U, V\right)} M^{\left(1\right)} \ket{w\left(U, V\right)} = R_2\left(U, V\right).\]
\end{proof}
\begin{proof}[Proof 2]
    Let $U, V$ be arbitrary unitaries, and let $M^{\left(i\right)} = M_{mean}^{\left(i\right)} + M_{core}^{\left(i\right)}$ be the moment decomposition of $M^{\left(i\right)}$. We aim to show that: 
    \[\bra{w\left(U, V\right)} M_{mean}^{\left(i\right)} \ket{w\left(U, V\right)} = 0.\]
    
    This will show that, by lemma~\ref{lma:risk-rayleigh}: 
    \autoeq{R_i\left(U, V\right) = \bar{\gamma}^{\left(i\right)} - \bra{w\left(U, V\right)} M^{\left(i\right)} \ket{w\left(U, V\right)} = \bar{\gamma}^{\left(i\right)} - \bra{w\left(U, V\right)} M_{core}^{\left(i\right)} \ket{w\left(U, V\right)},}
    which will immediately give our result, since it only depends on $\mathcal{S}^{\left(i\right)}$ but not on $\kket{\mu^{\left(i\right)}}$, and since $\bar{\gamma}^{\left(i\right)}$ can be computed from $\mathcal{S}^{\left(i\right)}$ as proven in the first point.
    
    Let $W = V^{\dagger} U$. Note that $W$ is unitary, $W^{\dagger} W = I$. By definition of the $w$ vector (definition~\ref{def:w-vector}), we know: 
    \[w\left(U, V\right) = \frac{1}{\sqrt{2^n}} \kket{V^{\dagger} U} = \frac{1}{\sqrt{2^n}} \kket{W}.\]
    
    Since $W$ is unitary, there exists a Hermitian matrix $H$ such that:
    \[W = \exp\left(i H\right).\]
    
    We can express this $H$ in the Pauli basis, using the hypothesis that $n = 1$: 
    \[H = a_I I + a_X X + a_Y Y + a_Z Z = \phi I + \theta\left(n_X X + n_Y Y + n_Z Z\right),\]
    where $a_I, a_X, a_Y, a_Z \in \mathbb{R}$ are arbitrary and thus we can choose parameters $\phi, \theta, n_X, n_Y, n_Z \in \mathbb{R}$ and $n_X^2 + n_Y^2 + n_Z^2 = 1$ such that the identity above holds. Therefore, using the generalised Euler identity for Pauli operators: 
    \autoeq{W = \exp\left(i H\right) = \exp\left(i \phi I\right) \exp\left(i\theta \left(n_X X + n_Y Y + n_Z Z\right)\right) = e^{i \phi} \left(I \cos\left(\theta\right) + i \left(n_X X + n_Y Y + n_Z Z\right) \sin\left(\theta\right)\right).}
    
    We know that $P = \sqrt{2}\cdot \widetilde{P}$ since $n = 1$. We can therefore directly read that, in the $\left\{\kket{\widetilde{P}} \suchthat P \in \mathbb{P}_1\right\}$ basis: 
    \[\kket{W} = \sqrt{2} e^{i\phi} \begin{pmatrix} \cos\left(\theta\right) \\ i n_X \sin\left(\theta\right) \\ i n_Y \sin\left(\theta\right) \\ i n_Z \sin\left(\theta\right) \end{pmatrix} \implies \ket{w\left(U, V\right)} = e^{i \phi} \begin{pmatrix} \cos\left(\theta\right) \\ i n_X \sin\left(\theta\right) \\ i n_Y \sin\left(\theta\right) \\ i n_Z \sin\left(\theta\right) \end{pmatrix} .\]
    
    This yields that, for any $P \in \mathbb{P}_1 \setminus \left\{I\right\}$: 
    \autoeq{\brakket{w\left(U, V\right)}{\widetilde{I}} \bbraket{\widetilde{P}}{w\left(U, V\right)} = \left(e^{-i\phi} \cos\left(\theta\right)\right) \left(e^{i \phi} i n_P \sin\left(\theta\right)\right) = i n_P \sin\left(\theta\right) \cos\left(\theta\right).}
    
    However, this shows that $\cre\left(\brakket{w\left(U, V\right)}{\widetilde{I}} \bbraket{\widetilde{P}}{w\left(U, V\right)}\right) = 0$, and thus, by lemma~\ref{lma:m-mean-expansion}: 
    \[\bra{w\left(U, V\right)} M_{mean}^{\left(i\right)} \ket{w\left(U, V\right)} = 0.\]
\end{proof}

Thus, a contribution that matters for general vectors can disappear when we restrict to unitaries. The result for single-qubit unitaries illustrates why bounds on pseudo-risks need not capture every property of the unitary-learning risk.

\suppsection{Product distributions}

Product structure lets us obtain the moment matrix of a large ensemble from those of its independently prepared subsystems. For zero-mean distributions, we can also determine the second-moment spectrum. We use these results in Section~\ref{sec:applications} to compare global and product ensembles.

\begin{theorem}\label{thm:m-product-distribution}
    Let $D^{\left(1\right)}$ and $D^{\left(2\right)}$ be independent mixed-state distributions, over $n_1$ and $n_2$ qubits respectively. Moreover, let $D = D^{\left(1\right)} \otimes D^{\left(2\right)}$ be a product distribution. Furthermore, let $M^{\left(1\right)}, M^{\left(2\right)}, M$ be the moment matrices of $D^{\left(1\right)}, D^{\left(2\right)}, D$, respectively.
    
    Finally, consider the permutation matrix
    $W \in \left\{0,1\right\}^{
    2^{2(n_1+n_2)} \times 2^{2(n_1+n_2)}}$ that fixes the first $n_1$ qubits, swaps the next $n_1$ qubits with the following $n_2$ qubits, and fixes the $n_2$ last qubits: 
    \[W \ket{ijk\ell} = \ket{ikj\ell}, \mathspace i, j \in \left\{0, 1\right\}^{n_1}, k, \ell \in \left\{0, 1\right\}^{n_2}.\]
    
    Then:
    \[M = W \left(M^{\left(1\right)} \otimes M^{\left(2\right)}\right) W^{-1}.\]
    
    In other words, $M$ and $M^{\left(1\right)} \otimes M^{\left(2\right)}$ are similar matrices. This notably implies that they have the same eigenvalues and that any eigenvector $\ket{v}$ of $M^{\left(1\right)} \otimes M^{\left(2\right)}$ yields an eigenvector $\ket{w} = W \ket{v}$ of $M$.
\end{theorem}
\begin{intuition*}
    The $W$ operator is a lot easier to understand with its action on vectorised operators. Let $A \in \mathbb{C}^{2^{n_1}\times 2^{n_1}}$
and $B \in \mathbb{C}^{2^{n_2}\times 2^{n_2}}$. Then: 
    \[W \left(\kket{A} \otimes \kket{B}\right) = \kket{A \otimes B}.\]
    
    In particular, this means that, if we have $P \in \mathbb{P}_{n_1}$ and $Q \in \mathbb{P}_{n_2}$ so that $\kket{\widetilde{P}}$ is an eigenvector of $M^{\left(1\right)}$ with eigenvalue $a$ and $\kket{\widetilde{Q}}$ is an eigenvector of $M^{\left(2\right)}$ with eigenvalue $b$, then $\kket{\widetilde{P} \otimes \widetilde{Q}} = W\left(\kket{\widetilde{P}} \otimes \kket{\widetilde{Q}}\right)$ is an eigenvector of $M$ with eigenvalue $ab$ ($\ket{v}$ is an eigenvector of $M^{\left(1\right)} \otimes M^{\left(2\right)}$ implies that $W \ket{v}$ is an eigenvector $M$ associated to the same eigenvalue, as mentioned above), as one may expect.
    
    Let us prove this affirmation. We express our operators in the computation basis: 
    \[A = \sum_{i, j} A_{ij} \ket{i}\bra{j}, \mathspace B = \sum_{k, \ell} B_{k \ell} \ket{k} \bra{\ell}.\]
    
    This implies that: 
    \[A \otimes B = \sum_{i, j, k, \ell} A_{ij} B_{k \ell} \ket{i k}\bra{j \ell} \implies \kket{A \otimes B} = \sum_{i, j, k, \ell} A_{ij} B_{k \ell} \ket{ikj\ell}.\]
    
    Moreover: 
    \[\kket{A} \otimes \kket{B} = \left(\sum_{i, j} A_{ij} \ket{ij}\right) \otimes \left(\sum_{k, \ell} B_{k\ell} \ket{k \ell}\right) = \sum_{i, j, k, \ell} A_{ij} B_{k \ell} \ket{ijk \ell}.\]
    
    We do clearly have $W\left(\kket{A} \otimes \kket{B}\right) = \kket{A \otimes B}$.
\end{intuition*}
\begin{proof}[Proof]
    By definition, we have: 
    \autoeq{M = 2^{n_1 + n_2} \exval_{\rho_1 \otimes \rho_2 \followsdistr D_1 \otimes D_2}\left(\left(\rho_1 \otimes \rho_2\right) \otimes \left(\rho_1 \otimes \rho_2\right)^T\right) = 2^{n_1 + n_2} \exval_{\substack{\rho_1 \followsdistr D_1 \\ \rho_2 \followsdistr D_2}}\left(\rho_1 \otimes \rho_2 \otimes \rho_1^T \otimes \rho_2^T\right),}
    \autoeq{M_1 \otimes M_2 = 2^{n_1} \exval_{\rho_1 \followsdistr D_1}\left(\rho_1 \otimes \rho_1^T\right) \otimes 2^{n_2}\exval_{\rho_2 \followsdistr D_2}\left(\rho_2 \otimes \rho_2^T\right) = 2^{n_1 + n_2} \exval_{\substack{\rho_1 \followsdistr D_1 \\ \rho_2 \followsdistr D_2}}\left(\rho_1 \otimes \rho_1^T \otimes \rho_2 \otimes \rho_2^T\right).}
    
    Now, notice that, by definition: 
    \[W \left(\ket{\psi_1} \otimes \ket{\psi_1}^* \otimes \ket{\psi_2} \otimes \ket{\psi_2}^*\right) = \ket{\psi_1} \otimes \ket{\psi_2} \otimes \ket{\psi_1}^* \otimes \ket{\psi_2}^*.\]
    
    This yields that: 
    \[W \left(\rho_1 \otimes \rho_1^* \otimes \rho_2  \otimes \rho_2^*\right) W^{\dagger} = \rho_1 \otimes \rho_2 \otimes \rho_1^* \otimes \rho_2^*.\]
    
    Since $\rho^* = \rho^T$ for any density matrix and $W^{-1} = W^T = W^{\dagger}$ since it is a permutation matrix, this yields our result: 
    \[M = W \left(M^{\left(1\right)} \otimes M^{\left(2\right)}\right) W^{-1}.\]
\end{proof}

When both means vanish, each subsystem’s moment matrix splits into an identity component and a traceless block. The second-moment spectrum of the product then comes from three combinations: a traceless factor on the first subsystem, on the second, or on both. Each identity factor contributes the average purity of its subsystem.

\begin{theorem}\label{thm:s-product-distribution}

    Let $D^{\left(1\right)}$ and $D^{\left(2\right)}$ be independent mixed-state distributions, over $n_1$ and $n_2$ qubits respectively. Moreover, let $D = D^{\left(1\right)} \otimes D^{\left(2\right)}$ be a product distribution. Furthermore, let $\bar{\gamma}^{\left(1\right)}, \bar{\gamma}^{\left(2\right)}, \bar{\gamma}$ be the average purities, $\mu^{\left(1\right)}, \mu^{\left(2\right)}, \mu$ be the first moments and $\mathcal{S}^{\left(1\right)}, \mathcal{S}^{\left(2\right)}, \mathcal{S}$ be the second moments of $D^{\left(1\right)}, D^{\left(2\right)}, D$, respectively.
    
    Finally, let $\left(\lambda_i^{\left(1\right)}, \kket{A_i}\right)$ be the eigenvalue-eigenvector pairs of $\mathcal{S}^{\left(1\right)}$, and let $\left(\lambda_i^{\left(2\right)}, \kket{B_i}\right)$ be the ones of $\mathcal{S}^{\left(2\right)}$. Notice that it is always possible to express these vectors in the vectorised notation, since $\sum_{i} c_i \kket{C_i} = \kket{\sum_{i} c_i C_i}$ is linear (definition~\ref{def:vectorised-notation}).
    
    If $\mu^{\left(1\right)} = \mu^{\left(2\right)} = 0$, then:
    \begin{enumerate}
    \item $\mu = 0$.
    \item $\bar{\gamma} = \bar{\gamma}^{\left(1\right)} \bar{\gamma}^{\left(2\right)}$.
    \item The following are eigenvalue-eigenvector tuples of $\mathcal{S}$:
    \[\left(\lambda_i^{\left(1\right)} \bar{\gamma}^{\left(2\right)}, \kket{A_i \otimes \widetilde{I}_{n_2}}\right), \mathspace \left(\bar{\gamma}^{\left(1\right)} \lambda_j^{\left(2\right)}, \kket{\widetilde{I}_{n_1} \otimes B_j}\right), \mathspace \left(\lambda_i^{\left(1\right)} \lambda_j^{\left(2\right)}, \kket{A_i \otimes B_j}\right),\]
    where $i \in \left\{1, \ldots, 4^{n_1} - 1\right\}$ and $j \in \left\{1, \ldots, 4^{n_2} - 1\right\}$. The set
    \[\left\{\kket{A_i \otimes \widetilde{I}_{n_2}}, \kket{\widetilde{I}_{n_1} \otimes B_j}, \kket{A_i \otimes B_j} \suchthat i \in \left\{1, \ldots, 4^{n_1} - 1\right\}, j \in \left\{1, \ldots, 4^{n_2} - 1\right\}\right\}\]
    is moreover an eigenbasis of $\mathcal{S}$.
    \item The minimum and maximum eigenvalues of $\mathcal{S}$ are:
        \[\lambda_{min}\left(\mathcal{S}\right) = \lambda_{min}^{\left(1\right)} \lambda_{min}^{\left(2\right)}, \mathspace \lambda_{max}\left(\mathcal{S}\right) = \max\left\{\lambda_{max}^{\left(1\right)} \bar{\gamma}^{\left(2\right)}, \bar{\gamma}^{\left(1\right)} \lambda_{max}^{\left(2\right)}, \lambda_{max}^{\left(1\right)} \lambda_{max}^{\left(2\right)}\right\}.\]
    \end{enumerate}
\end{theorem}
\begin{proof}[Proof]
    Let $M^{\left(1\right)}, M^{\left(2\right)}, M$ be the moment matrices of $D^{\left(1\right)}, D^{\left(2\right)}, D$, respectively. By theorem~\ref{thm:m-product-distribution}, we know that $M$ and $M^{\left(1\right)} \otimes M^{\left(2\right)}$ are similar (with similarity matrix $W$ defined in this theorem).
    
    Notice that, since $\mu^{\left(1\right)} = \mu^{\left(2\right)} = 0$ by hypothesis: 
    \autoeq{M^{\left(i\right)} = \begin{pmatrix} \bar{\gamma}^{\left(i\right)} & 0 \\ 0 & \mathcal{S}^{\left(i\right)} \end{pmatrix} \implies M^{\left(1\right)} \otimes M^{\left(2\right)} = \begin{pmatrix} \bar{\gamma}^{\left(1\right)} \bar{\gamma}^{\left(2\right)} & 0 & 0 & 0 \\ 0 & \bar{\gamma}^{\left(1\right)} \mathcal{S}^{\left(2\right)} & 0 & 0 \\ 0 & 0 & \bar{\gamma}^{\left(2\right)} \mathcal{S}^{\left(1\right)} & 0 \\ 0 & 0 & 0 & \mathcal{S}^{\left(1\right)} \otimes \mathcal{S}^{\left(2\right)} \end{pmatrix}.}
    
    We can prove our four affirmations from there.
    
    \begin{enumerate}[left=0pt]
    \item Since $M^{\left(1\right)} \otimes M^{\left(2\right)}$ is block diagonal, $M = W\left(M^{\left(1\right)} \otimes M^{\left(2\right)}\right) W^{-1}$ is also block diagonal. In particular, this means that its first row and first column are full of zeroes (except for the top-left corner). This yields that $\mu = 0$, by the structure of the moment matrix (definition~\ref{def:moment-matrix}).
    \item By the exact same reasoning, we find $\bar{\gamma} = \bar{\gamma}^{\left(1\right)} \bar{\gamma}^{\left(2\right)}$ is the value in the top-left corner.
    \item By a similar reasoning, this means that $\mathcal{S}$ is similar to: 
    \[\mathcal{T} = \begin{pmatrix} \bar{\gamma}^{\left(1\right)} \mathcal{S}^{\left(2\right)} & 0 & 0 \\ 0 & \bar{\gamma}^{\left(2\right)} \mathcal{S}^{\left(1\right)} & 0 \\ 0 & 0 & \mathcal{S}^{\left(1\right)} \otimes \mathcal{S}^{\left(2\right)} \end{pmatrix}.\]
    
    In other words, $\mathcal{S} = W_{\perp} \mathcal{T} W_{\perp}^{-1}$, where $W_{\perp}$ is the restriction of the similarity matrix $W$ defined in theorem~\ref{thm:m-product-distribution} to the subspaces orthogonal to the vectorised identity. Note that, by definition of tensor product, the first block is indexed by the sub-basis 
    \[\left\{\kket{\widetilde{I}_{n_1}} \otimes \kket{\widetilde{Q}} \suchthat Q \in \mathbb{P}_{n_2} \setminus \left\{I_{n_2}\right\}\right\},\]
    the second block is indexed by
    \[\left\{\kket{\widetilde{P}} \otimes \kket{\widetilde{I}_{n_2}} \suchthat P \in \mathbb{P}_{n_1} \setminus \left\{I_{n_1}\right\}\right\},\]
    and the third block is indexed by
    \[\left\{\kket{\widetilde{P}} \otimes \kket{\widetilde{Q}} \suchthat P \in \mathbb{P}_{n_1} \setminus \left\{I_{n_1}\right\}, Q \in \mathbb{P}_{n_2} \setminus \left\{I_{n_2}\right\}\right\}.\]
    
    Now, $\kket{A_i} \otimes \kket{B_j}$ is clearly an eigenvector of $\mathcal{S}^{\left(1\right)} \otimes \mathcal{S}^{\left(2\right)}$ associated to $\lambda_i^{\left(1\right)} \lambda_j^{\left(2\right)}$. Since $\mathcal{T}$ is a block diagonal matrix, and that both $\kket{A_i}$ and $\kket{B_j}$ are such that $\bbrakket{\widetilde{I}_{n_1}}{A_i} = \bbrakket{\widetilde{I}_{n_2}}{B_j} = 0$ (since the second moment lives in $\mathbb{R}^{\left(4^{n_i} - 1\right) \times \left(4^{n_i} - 1\right)}$, similarly to what was explained in the remark of definition~\ref{def:truncation-for-vectors}), it is fully contained in the third block and it is thus also an eigenvector of $\mathcal{T}$, associated to the same eigenvalue. 
    
    As explained in theorem~\ref{thm:m-product-distribution}, since $\mathcal{S} = W_{\perp} \mathcal{T} W_{\perp}^{-1}$, we know that this yields an eigenvector of $\mathcal{S}$,
    \[W \left(\kket{A_i} \otimes \kket{B_j}\right) = \kket{A_i \otimes B_j},\]
    associated to the same eigenvalue, $\lambda_i^{\left(1\right)} \lambda_j^{\left(2\right)}$. This is our first eigenvalue-eigenvector pair. We can use the exact same reasoning for the two other types of eigenvalue-eigenvector pairs.
    
    Choosing orthonormal eigenbases for $\mathcal{S}^{(1)}$ and
    $\mathcal{S}^{(2)}$, the three families above form an orthonormal
    set containing
    \autoeq[s]{
    \left(4^{n_1}-1\right)
    + \left(4^{n_2}-1\right)
    + \left(4^{n_1}-1\right)\left(4^{n_2}-1\right)
    = 4^{n_1+n_2}-1
    }
    eigenvectors. This equals the dimension of the traceless
    subspace on $n_1+n_2$ qubits, so these vectors form an
    eigenbasis of $\mathcal{S}$.
    
    \item The eigenvalues of $\mathcal{S}$ belong to the three
    families
    \[
    \bar{\gamma}^{(2)}\lambda_i^{(1)},\qquad
    \bar{\gamma}^{(1)}\lambda_j^{(2)},\qquad
    \lambda_i^{(1)}\lambda_j^{(2)}.
    \]

    By Lemma~\ref{lma:s-properties},
    \[
    0\leq \lambda_{\min}^{(i)}
    \leq \bar{\gamma}^{(i)}.
    \]
    Therefore, the smallest eigenvalue comes from the third
    family:
    \[
    \lambda_{\min}(\mathcal{S})
    = \lambda_{\min}^{(1)}\lambda_{\min}^{(2)}.
    \]

    For the largest eigenvalue, we take the maximum over the
    three families. Since all eigenvalues and average purities
    are nonnegative, this gives
    \[
    \lambda_{\max}(\mathcal{S})
    = \max\left\{
    \lambda_{\max}^{(1)}\bar{\gamma}^{(2)},
    \bar{\gamma}^{(1)}\lambda_{\max}^{(2)},
    \lambda_{\max}^{(1)}\lambda_{\max}^{(2)}
    \right\}.
    \]
    \end{enumerate}
\end{proof}

\begin{corollary}[Moments of a mixed distribution tensor product]\label{crl:s-for-n-product-mixed}
    Let $D_0$ be an $n_0$-qubit mixed-state distribution of first moment $\mu_0$, second moment $\mathcal{S}_0$, purity $\bar{\gamma}_0$ and expressivity measure $\eta_0$. Moreover, let $D = D_0^{\otimes k}$ be the distribution that takes a tensor product of $k$ copies of independent samples of $D_0$. Let $\mu, \mathcal{S}, \bar{\gamma}, \eta$ be its first moment, second moment, purity and expressivity measure.

    If $\mu_0 = 0$, then: 
    \[\mu = 0, \mathspace \bar{\gamma} = \bar{\gamma}_0^k, \mathspace \lambda_{min}\left(\mathcal{S}\right) = \lambda_{min}\left(\mathcal{S}_0\right)^k.\]
    
    If moreover $\eta_0 \geq 0$, then: 
    \[\lambda_{max}\left(\mathcal{S}\right) = \bar{\gamma}_0^{k-1} \lambda_{max}\left(\mathcal{S}_0\right), \mathspace \eta = \bar{\gamma}_0^{k-1} \eta_0.\]

    Finally, if $\kket{A_{i}}$ are the eigenvectors of $\mathcal{S}_0$, then $\lambda_{min}\left(\mathcal{S}\right)$ is associated to $\kket{A_{min}^{\otimes k}}$, and if $\eta_0 \geq 0$ then $\lambda_{max}\left(\mathcal{S}\right)$ is associated to many eigenvectors, including $\kket{A_{max} \otimes \widetilde{I}_{n_0}^{\otimes k - 1}}$ and $\kket{\widetilde{I}_{n_0}^{\otimes 3} \otimes A_{max} \otimes \widetilde{I}_{n_0}^{\otimes k - 4}}$ (the latter for $k \geq 4$).
\end{corollary}
\begin{proof}
    This result follows immediately by induction from theorem~\ref{thm:s-product-distribution}, and the following observation. When $\eta_0 \geq 0$, then: 
    \[\bar{\gamma}_0 - \lambda_{max}\left(\mathcal{S}_0\right) \geq 0 \implies \lambda_{max}\left(\mathcal{S}_0\right) \leq \bar{\gamma}_0.\]

    Therefore:
    \[\max\left\{\lambda_{max}\left(\mathcal{S}_0\right) \bar{\gamma}_0, \lambda_{max}\left(\mathcal{S}_0\right)^2\right\} = \bar{\gamma}_0 \lambda_{max}\left(\mathcal{S}_0\right).\]
\end{proof}

\begin{corollary}[Moments of a pure distribution tensor product] \label{crl:s-for-n-product-pure}
    Let $D_0$ be an $n_0$-qubit pure-state distribution of first moment $\mu_0$, second moment $\mathcal{S}_0$, purity $\bar{\gamma}_0 = 1$ and expressivity measure $\eta_0$. Moreover, let $D = D_0^{\otimes k}$ be the distribution that takes a tensor product of $k$ copies of independent samples of $D_0$. Let $\mu, \mathcal{S}, \bar{\gamma}, \eta$ be its first moment, second moment, purity and expressivity measure.

    If $\mu_0 = 0$, then $D$ is also pure (i.e.\ $\bar{\gamma} = 1$) and: 
    \[\mu = 0, \mathspace \lambda_{min}\left(\mathcal{S}\right) = \lambda_{min}\left(\mathcal{S}_0\right)^k.\]
    
    If moreover $\eta_0 \geq 0$, then: 
    \[\lambda_{max}\left(\mathcal{S}\right) = \lambda_{max}\left(\mathcal{S}_0\right), \mathspace \eta = \eta_0.\]

    Finally, if $\kket{A_{i}}$ are the eigenvectors of $\mathcal{S}_0$, then $\lambda_{min}\left(\mathcal{S}\right)$ is associated to $\kket{A_{min}^{\otimes k}}$, and if $\eta_0 \geq 0$ then $\lambda_{max}\left(\mathcal{S}\right)$ is associated to many eigenvectors, including $\kket{A_{max} \otimes \widetilde{I}_{n_0}^{\otimes k - 1}}$ and $\kket{\widetilde{I}_{n_0}^{\otimes 3} \otimes A_{max} \otimes \widetilde{I}_{n_0}^{\otimes k - 4}}$ (the latter for $k \geq 4$).
\end{corollary}
\begin{remark*}
    Observe that the second result of corollary~\ref{corol-main-text:product-bias-expressivity} from the main text comes from this result: taking tensor products of a zero-mean distribution preserves its expressivity, $\eta = \eta_0$.
\end{remark*}
\begin{proof}
    This is just corollary~\ref{crl:s-for-n-product-mixed} with $\bar{\gamma}_0 = 1$.
\end{proof}

Repeating a zero-mean pure-state distribution preserves its expressivity, provided $\eta_0$ is nonnegative. However, for mixed states the same expression includes a power of the local purity. As a consequence,  the bounds for noisy product ensembles will depend on purity.

\suppsection{Link between out-of-distribution generalization bounds and purity}\label{sec:purityproofs}

Even very different unitaries can have a small risk on highly mixed inputs, while behaving quite differently on purer states. This is because the risk is bounded above by the average input purity for any pair of unitaries. We use this observation to quantify the cost of generalizing from a distribution of lower-purity states when the allowed additive error is small.

\begin{theorem}\label{thm:universal-bound-risk-purity}
    Let $D$ be a distribution of average purity $\bar{\gamma}$. Then, for all unitaries $U, V$: 
    \[R_D\left(U, V\right) \leq \bar{\gamma}.\]
\end{theorem}
\begin{remark*}
    In particular, if the average purity of the distribution is exponentially small, then the risk between any unitaries $U, V$ is exponentially small.
\end{remark*}
\begin{proof}
    By lemma~\ref{lma:risk-rayleigh} and since the moment matrix $M$ of $D$ is positive semi-definite: 
    \[R_D\left(U, V\right) = \bar{\gamma} - \bra{w\left(U, V\right)} M \ket{w\left(U, V\right)} \leq \bar{\gamma}.\]
\end{proof}

\begin{corollary}\label{crl:trivial-odg-bound-purity}
    Let $D^{\left(1\right)}, D^{\left(2\right)}$ be mixed-state distributions, and $\bar{\gamma}^{\left(1\right)}$ be the average purity of $D^{\left(1\right)}$. Then for all unitaries $U, V$, writing $R_i := R_{D^{\left(i\right)}}$:
    \[R_1\left(U, V\right) \leq 0\cdot R_2\left(U, V\right) + \bar{\gamma}^{\left(1\right)}.\]
\end{corollary}
\begin{proof}
    This is just a rephrasing of theorem~\ref{thm:universal-bound-risk-purity} in an out-of-distribution generalization  form.
\end{proof}

We derive the purity penalty by applying the out-of-distribution generalization  inequality to Pauli unitaries and averaging over them. This turns the moment-matrix terms into traces, leaving a necessary condition that depends only on the two average purities and the allowed additive error.

\begin{theorem}[Purity penalty]\label{thm:purity-penalty}
    Let $D^{\left(1\right)}, D^{\left(2\right)}$ be mixed-state distributions, of average purities $\bar{\gamma}^{\left(1\right)}, \bar{\gamma}^{\left(2\right)}$, respectively. Moreover, let $\alpha, \beta \in \mathbb{R}_{\geq 0}$ such that $R_{D_1}\left(U, V\right) \leq \alpha R_{D_2}\left(U, V\right) + \beta$ for any unitaries $U, V$. Finally, assume that $\beta \leq \left(\bar{\gamma}_1 - 2^{-n}\right)/2$. Then: 
    \[\alpha \geq \frac{\bar{\gamma}^{\left(1\right)} - 2^{-n}}{2 \bar{\gamma}^{\left(2\right)}}.\]
\end{theorem}
\begin{intuition*}
    This result shows that, for ODG to be possible, then we need $\alpha \geq \Omega\left(\bar{\gamma}^{\left(1\right)} / \bar{\gamma}^{\left(2\right)}\right)$ (in the case $\bar{\gamma}^{\left(1\right)} \geq \left(1 + \Theta\left(1\right)\right) 2^{-n}$). In particular, if we want $\alpha$ to be polynomial, we know that $\bar{\gamma}^{\left(2\right)}$ must be at most polynomially smaller than the average purity of the distribution we are trying to generalise to $\bar{\gamma}^{\left(1\right)}$.
\end{intuition*}
\begin{remark*}
    The hypothesis $\beta \leq \left(\bar{\gamma}_1 - 2^{-n}\right)/2$ is pretty natural: if we are ready to take $\beta = \bar{\gamma}_1$, then (trivial/uninteresting) ``out-of-distribution generalization'' is always possible for any $\alpha \geq 0$ by corollary~\ref{crl:trivial-odg-bound-purity}.
\end{remark*}
\begin{remark*}
    This is precisely proposition~\ref{prop-main-text:purity-penalty} in the main text.
\end{remark*}
\begin{proof}
    We write $R_i := R_{D_i}$, and consider $M^{\left(i\right)}$ to be the moment matrix of $D^{\left(i\right)}$. Considering the unitaries $U = P \in \mathbb{P}_n$ and $V = I^{\otimes n}$, then $\ket{w\left(U, V\right)} = \frac{1}{\sqrt{2^n}} \kket{P} = \kket{\widetilde{P}}$. Therefore, by lemma~\ref{lma:risk-rayleigh}, we have:
    \[R_i\left(P, I\right) = \bar{\gamma}^{\left(i\right)} - \bbra{\widetilde{P}} M^{\left(i\right)} \kket{\widetilde{P}}.\]
    
    Summing on $P \in \mathbb{P}_n$, this naturally becomes a trace:
    \[\frac{1}{4^n} \sum_{P \in \mathbb{P}_n} R_i\left(P, I\right) = \bar{\gamma}^{\left(i\right)} - \frac{1}{4^n} \Tr\left(M^{\left(i\right)}\right) = \bar{\gamma}^{\left(i\right)} - \frac{1}{2^n},\]
    by lemma~\ref{lma:s-properties}, since $\Tr\left(M^{\left(i\right)}\right) = \Tr\left(\mathcal{S}^{\left(i\right)}\right) + \bar{\gamma}^{\left(i\right)} = 2^n$. But then, using our hypothesis that $R_1\left(U, V\right) \leq \alpha R_2\left(U, V\right) + \beta$ for all $U, V$, this requires: 
    \[\bar{\gamma}^{\left(1\right)} - \frac{1}{2^n} = \frac{1}{4^n} \sum_{P \in \mathbb{P}_n} R_1\left(P, I\right) \leq \alpha \frac{1}{4^n} \sum_{P \in \mathbb{P}_n} R_2\left(P, I\right) + \beta = \alpha \left(\bar{\gamma}^{\left(2\right)} - \frac{1}{2^n}\right) + \beta.\]
    
    Using the hypothesis that $\beta \leq \left(\bar{\gamma}_1 - 2^{-n}\right)/2$, this becomes: 
    \[\bar{\gamma}^{\left(1\right)} - \frac{1}{2^n} \leq \alpha \left(\bar{\gamma}^{\left(2\right)} - \frac{1}{2^n}\right) + \left(\frac{\bar{\gamma}^{\left(1\right)} - \frac{1}{2^n}}{2}\right) \implies \frac{\bar{\gamma}^{\left(1\right)} - \frac{1}{2^n}}{2} \leq \alpha \left(\bar{\gamma}^{\left(2\right)} - \frac{1}{2^n}\right) \leq \alpha \bar{\gamma}^{\left(2\right)} \implies \alpha \geq \frac{\bar{\gamma}^{\left(1\right)} - \frac{1}{2^n}}{2 \bar{\gamma}^{\left(2\right)}}.\]
\end{proof}

\suppsection{Applications}\label{sec:applications}

We will make the general bounds explicit for Haar ensembles. We begin by comparing global and product ensembles, then study how local depolarizing noise changes the comparison.
Finally, we derive generalisation bounds for ensembles of real states.

\suppsubsection{Haar measure on the unitary group}\label{sec:haar-meausre}

For the Haar ensemble, the second moment is proportional to the identity on the traceless subspace. This makes its spectrum particularly simple. The product results then give the spectra for independent Haar blocks, allowing us to compare different block sizes and to bound the risk on other distributions.

\begin{remark}
    By theorem~\ref{thm:same-moment-matrix-same-risk}, the following results are true for any 2-design, i.e.\ any distribution that has the same first and second moment as the Haar distribution on the unitary group.
\end{remark}

\begin{lemma}\label{lma:m-haar-distribution}
    Let $D = \Haar_n\left(U\right)$ be the $n$-qubit pure state Haar distribution on the unitary group, $M$ be its moment matrix, $\kket{\mu}$ be its first moment and $\mathcal{S}$ be its second moment. 
    
    Then, in the $\kket{\widetilde{P}}$ basis: 
    \[M = 1 \oplus \left(\frac{1}{2^n + 1} I_{4^n - 1}\right).\]
    
    In particular, this means that $\kket{\mu} = 0$ and $\mathcal{S} = \frac{1}{2^n + 1} I$.
\end{lemma}
\begin{intuition*}
    In the $\kket{\widetilde{P}}$ basis, this reads:
    \[M = \begin{pmatrix} 1 &  &  &  &  \\  & \frac{1}{2^n + 1} &  &  &  \\  &  & \frac{1}{2^n + 1} &  &  \\  &  &  & \ddots &  \\  &  &  &  &  \frac{1}{2^n+1} \end{pmatrix}.\]
    
    So, a big advantage of the $\kket{\widetilde{P}}$ basis is that, for Haar distributions, $M$ is diagonal in it. For instance, in the computation basis $\ket{ij}$, it is possible to instead show that, when $n = 1$:
    \[M = \frac{1}{3} \begin{pmatrix} 2 & 0 & 0 & 1 \\ 0 & 1 & 0 & 0 \\ 0 & 0 & 1 & 0 \\ 1 & 0 & 0 & 2 \end{pmatrix},\]
    which is not even block diagonal (with a non-trivial block).
\end{intuition*}
\begin{remark*}
    Since $M$ encapsulates the first and the second moment of $D$, this lemma is, in fact, valid for any 2-design.
\end{remark*}
\begin{proof}[Proof]
    Our goal is to  compute: 
    \[M = \sum_{P, Q \in \mathbb{P}_n} \exval_{\rho \followsdistr D}\left(\Tr\left(P \rho\right) \Tr\left(Q \rho\right)\right) \kket{\widetilde{P}} \bbra{\widetilde{Q}}.\]
    
    Reference~\cite{mele2023introduction}'s theorem 22 tells us that: 
    \[\exval\left(\rho^{\otimes 2}\right) = \frac{I \otimes I + SWAP}{2^n \left(2^n + 1\right)}.\]
    
    We could use the remark of definition~\ref{def:moment-matrix}, using the $T_2$ operator. However, expressing $\exval\left(\rho^{\otimes 2}\right)$ in the $\kket{\widetilde{P}}$ basis is already very non-trivial. Let us instead simply compute the coefficients $\exval\left(\Tr\left(P \rho\right) \Tr\left(Q \rho\right)\right)$  by Lemma~\ref{lma:moment-matrix-equivalence-pure-states}.
    
    Since $\Tr\left(A \otimes B\right) = \Tr\left(A\right)\Tr\left(B\right)$, and the expected value and the trace are linear operators: 
    \autoeq{\exval\left(\Tr\left(P \rho\right) \Tr\left(Q \rho\right)\right) = \exval\left(\Tr\left(\left(P \otimes Q\right) \left(\rho \otimes \rho\right)\right)\right) = \Tr\left(\left(P \otimes Q\right) \exval\left(\rho \otimes \rho\right)\right) = \Tr\left(\left(P \otimes Q\right) \frac{I \otimes I + SWAP}{2^n \left(2^n + 1\right)}\right).}
    
    Expressing any two operators $A, B$ in the computation basis, using the fact that $SWAP$ is a symmetric permutation matrix and thus Hermitian, we have the following property: 
    \begin{equation}\begin{aligned}\label{eq:trace-swap-product}\Tr\left(\left(A \otimes B\right) \cdot SWAP\right) &= \Tr\left(\sum_{i,j,k,\ell} \left(a_{ij} \ket{i}\bra{j} \otimes b_{k\ell} \ket{k}\bra{\ell}\right) \cdot SWAP\right) \\ &= \sum_{i, j, k, \ell} a_{ij} b_{k\ell} \Tr\left(\ket{i} \bra{\ell} \otimes \ket{k} \bra{j}\right) \\ &= \sum_{i, j, k, \ell} a_{ij} b_{k\ell} \braket{\ell}{i} \braket{j}{k} \\ &= \sum_{i, j} a_{ij} b_{ji} \\ &= \sum_{i} \left(AB\right)_{i, i} \\ &= \Tr\left(AB\right).\end{aligned}\end{equation}
    
    Since $\Tr\left(P\right) =  2^n\delta_{P I}$ for any Pauli $P \in \mathbb{P}_n$, this yields: 
    \autoeq{\bbra{\widetilde{P}} M \kket{\widetilde{Q}} = \exval\left(\Tr\left(P \rho\right)\Tr\left(Q \rho\right)\right) = \frac{\Tr\left(P\right)\Tr\left(Q\right) + \Tr\left(PQ\right)}{2^n \left(2^n + 1\right)} = \frac{2^n \delta_{P I} 2^n \delta_{Q I} + 2^n \delta_{PQ}}{2^n \left(2^n + 1\right)} = \delta_{PQ} \frac{\delta_{P I} 2^n + 1}{2^n + 1}.}
    
    The $\delta_{PQ}$ at the front means that $M$ is diagonal. This yields exactly our result.
\end{proof}

For independent Haar blocks, a Pauli's eigenvalue depends only on how many blocks it acts on nontrivially. We call this number the Pauli block weight.

\begin{definition}[Pauli block weight]\label{def:pauli-block-weight}
    Let $n, k$ be integers, and $P \in \mathbb{P}_{nk}$.
    
    We say that $P$ has $n$-block weight $j \in \left\{0, \ldots, k\right\}$ if there exist $k$ smaller Paulis $Q_{\ell} \in \mathbb{P}_n$ and a set $L \subseteq \left\{1, \ldots, k\right\}$ so that:
    \begin{itemize}
    \item $\displaystyle P = \bigotimes_{\ell = 1}^k Q_{\ell}$.
    \item $\displaystyle \ell \in L \iff Q_{\ell} \neq I^{\otimes n}$, for all $\ell \in \left\{1, \ldots, k\right\}$.
    \item $\displaystyle \left|L\right| = j$.
    \end{itemize}
    
    We moreover note $w_n\left(P\right) = j$.
\end{definition}
\begin{intuition*}
    Intuitively, this means that, splitting $P$ as a tensor product of Paulis $Q_{\ell}$ of size $n$, exactly $j$ of them are non-identity. 
    
    Note that this is just a generalisation of regular Pauli weight: when $n = 1$, this represents the number of non-identity terms in $P$.
\end{intuition*}
\begin{example*}
    Let us for instance consider: 
    \[P = IXIYIZ \in \mathbb{P}_6, \mathspace n = 2.\]
    
    Its three blocks are $IX$, $IY$, $IZ$. None of them is $II$, so $P$ has $2$-block weight $w_2\left(P\right) = 3$.
\end{example*}
\begin{property*}
    We notice that, if $P_1 \in \mathbb{P}_{n k_1}$ and $P_2 \in \mathbb{P}_{n k_2}$: 
    \[w_n\left(P_1 \otimes P_2\right) = w_n\left(P_1\right) + w_n\left(P_2\right).\]
\end{property*}

\begin{lemma}\label{lma:s-haar-distribution}
    Let $D = \Haar_n\left(U\right)^{\otimes k}$, $\kket{\mu}$ be its first moment, $\mathcal{S}$ be its second moment and $\eta$ be its expressivity measure.
    
    Then: 
    \[\kket{\mu} = 0, \mathspace \lambda_{min}\left(\mathcal{S}\right) = \left(\frac{1}{2^n + 1}\right)^k, \mathspace \lambda_{max}\left(\mathcal{S}\right) = \frac{1}{2^n + 1}, \mathspace \eta = \frac{2^n}{2^n + 1}.\]
    
    Moreover, an eigenbasis of $\mathcal{S}$ is 
    \[\left\{\kket{\widetilde{P}} \suchthat P \in \mathbb{P}_{kn} \setminus \left\{I^{\otimes nk}\right\}\right\}.\]
    
    Finally, each $\kket{\widetilde{P}}$ is associated to the eigenvalue 
    \[\left(\frac{1}{2^n + 1}\right)^{w_n\left(P\right)},\]
    where $w_n\left(P\right) \in \left\{1, \ldots, k\right\}$ is the $n$-block weight of $P$ (definition~\ref{def:pauli-block-weight}).
\end{lemma}
\begin{intuition*}
    In other words, any eigenvalue of $\mathcal{S}$ has the form $\left(\frac{1}{2^n + 1}\right)^j$ for some $j \in \left\{1, \ldots, k\right\}$, and a corresponding eigenvector is $\kket{\widetilde{P}}$ where $P \in \mathbb{P}_{kn}$ is any Pauli of $n$-block weight $w_n\left(P\right) = j$.
\end{intuition*}
\begin{proof}[Proof] 
    By lemma~\ref{lma:m-haar-distribution}, $\Haar_n\left(U\right)$ is a pure-state distribution of zero mean, whose second moment is $\mathcal{S} = \frac{1}{2^n + 1} I$ and whose expressivity measure is $\eta = \frac{2^n}{2^n + 1} > 0$. Therefore, by corollary~\ref{crl:s-for-n-product-pure}:
    \[\kket{\mu} = 0, \mathspace \lambda_{min}\left(\mathcal{S}\right) = \left(\frac{1}{2^n + 1}\right)^k, \mathspace \lambda_{max}\left(\mathcal{S}\right) = \frac{1}{2^n + 1}, \mathspace \eta = \frac{2^n}{2^n + 1}.\]

    Let us now consider an arbitrary $P \in \mathbb{P}_{nk}$. We can decompose:
    \[P = \bigotimes_{\ell = 1}^k Q_{\ell}, \mathspace Q_{\ell} \in \mathbb{P}_n.\]

    By lemma~\ref{lma:m-haar-distribution}, each $\kket{\widetilde{Q}_{\ell}}$ is an eigenvector of the single-block moment matrix, associated to the eigenvalue $1$ if $Q_{\ell} = I^{\otimes n}$, and $\frac{1}{2^n + 1}$ otherwise. Applying theorem~\ref{thm:m-product-distribution} repeatedly, this means that $\kket{\widetilde{P}}$ is an eigenvector of the moment matrix of $D$, associated to the product of these eigenvalues. Since exactly $w_n\left(P\right)$ blocks are non-identity, this eigenvalue is:
    \[\left(\frac{1}{2^n + 1}\right)^{w_n\left(P\right)}.\]

    Finally, since $\kket{\mu} = 0$ and $D$ is pure, its moment matrix reads $M = 1 \oplus \mathcal{S}$ in the $\kket{\widetilde{P}}$ basis by lemma~\ref{lma:other-formulation-decomposition}. Therefore, removing the identity direction from this eigenbasis of $M$ gives precisely the claimed eigenbasis of $\mathcal{S}$, with the corresponding eigenvalues.
\end{proof}

\begin{theorem}\label{thm:haar-general-comparison}
    Let $n_1, k_1, n_2, k_2$ be integers so that $n_1 k_1 = n_2 k_2$. Moreover, let $D_1 = \Haar_{n_1}^{\otimes k_1}\left(U\right)$ and $D_2 = \Haar_{n_2}^{\otimes k_2}\left(U\right)$.
    
    Then, for any unitaries $U, V$:
    \autoeq[s]{\left(1 + \frac{1}{\left(2^{n_2} + 1\right)^{k_2} - 1}\right)\left(1 - \frac{1}{2^{n_1} + 1}\right)R_2\left(U, V\right) \leq R_1\left(U, V\right) \leq \left(1 - \frac{1}{\left(2^{n_1} + 1\right)^{k_1}}\right)\left(1 + \frac{1}{2^{n_2}}\right)R_2\left(U, V\right).}
\end{theorem}
\begin{intuition*}
    The hypothesis that $n_1 k_1 = n_2 k_2$ just states that $D_1$ and $D_2$ are two distributions over the same number of qubits.
    
    Moreover, note that, for $n_1, k_1, n_2, k_2$ sufficiently large, this theorem means that $R_1\left(U, V\right) \approx R_2\left(U, V\right)$.
\end{intuition*}
\begin{remark*}
    This result is precisely corollary~\ref{corol-main-text:locally-scrambling-comparison} in the main text.
\end{remark*}
\begin{proof}[Proof]
    To simplify notations, we note $R_i\left(U, V\right) = R_i$.
    
    We will apply twice the fourth inequality of theorem~\ref{thm:inequalities}. First, using lemma~\ref{lma:s-haar-distribution}: 
    \autoeq{R_1 \leq \frac{1 - \lambda_{min}\left(\mathcal{S}^{\left(1\right)}\right)}{1 - \lambda_{max}\left(\mathcal{S}^{\left(2\right)}\right)} R_2 = \frac{1 - \left(\frac{1}{2^{n_1} + 1}\right)^{k_1}}{1 - \frac{1}{2^{n_2} + 1}} R_2 = \left(1 - \frac{1}{\left(2^{n_1} + 1\right)^{k_1}}\right) \frac{2^{n_2} + 1}{2^{n_2}} R_2 = \left(1 - \frac{1}{\left(2^{n_1} + 1\right)^{k_1}}\right) \left(1 + \frac{1}{2^{n_2}}\right) R_2.}
    
    This is our second inequality. Similarly, using again the fourth inequality of theorem~\ref{thm:inequalities} with lemma~\ref{lma:s-haar-distribution}:
    \autoeq{R_2 \leq \frac{1 - \lambda_{min}\left(\mathcal{S}^{\left(2\right)}\right)}{1 - \lambda_{max}\left(\mathcal{S}^{\left(1\right)}\right)} R_1 = \left(1 - \frac{1}{\left(2^{n_2} + 1\right)^{k_2}}\right) \left(1 + \frac{1}{2^{n_1}}\right) R_1 = \frac{\left(2^{n_2} + 1\right)^{k_2} - 1}{\left(2^{n_2} + 1\right)^{k_2}} \cdot \frac{2^{n_1} + 1}{2^{n_1}} R_1.}
    
    Reordering the terms, this yields:
    \autoeq{R_1 \geq \frac{\left(2^{n_2} + 1\right)^{k_2}}{\left(2^{n_2} + 1\right)^{k_2} - 1} \cdot \frac{2^{n_1}}{2^{n_1} + 1} R_2 = \left(1 + \frac{1}{\left(2^{n_2} + 1\right)^{k_2} - 1}\right)\left(1 - \frac{1}{2^{n_1} + 1}\right) R_2.}
    
    This is our first inequality.
\end{proof}
\begin{remark*}
    We used the fourth inequality of theorem~\ref{thm:inequalities}. It is possible---although harder---to use the second inequality instead, yielding a potentially tighter inequality. We will instead apply it in the following, less general, theorem.
\end{remark*}

The bounds are tight because the eigenvectors that attain them
are normalized vectorizations of Pauli operators, which are
unitary. We can therefore choose $V^\dagger U$ to realize these
cases: a Pauli of weight $1$ attains the upper bound, while a
Pauli of weight $n$ attains the lower bound.

\begin{theorem}\label{thm:haar-tight-comparison}
    Let $D_1 = \Haar_n\left(U\right)$ and $D_2 = \Haar_1^{\otimes n}\left(U\right)$. Then, for all unitaries $U, V$:
    \[\frac{3^n}{3^n - 1}\cdot \frac{2^n}{2^n + 1} R_2\left(U, V\right) \leq R_1\left(U, V\right) \leq \frac{3}{2}\cdot \frac{2^n}{2^n + 1} R_2\left(U, V\right).\]
    
    Let $W = V^{\dagger} U$. If $W \in \mathbb{P}_n$ and $w_1\left(W\right) = n$ has maximal Pauli weight (definition~\ref{def:pauli-block-weight}), then the first inequality is reached with equality. If $W \in \mathbb{P}_n$ and $w_1\left(W\right) = 1$, then the second inequality is reached with equality.
\end{theorem}
\begin{remark*}
    Note that this inequality is just the one given by theorem~\ref{thm:haar-general-comparison} for $n_1 = n, k_1 = 1, n_2 = 1, k_2 = n$. However, the interest here is that it is proven to be tight. This theorem is verified numerically and discussed in the main text. See for instance figure~\ref{fig:risk-graph-non-zero-mean}.     
    
    Moreover, note that this is very similar to the inequality found in reference~\cite[Appendix C.1]{caro2022outofdistribution}, except that we have an extra $\frac{3^n}{3^n - 1}$ factor for the first inequality, and a $\frac{3}{2}$ instead of a $2$ for the second inequality. These new inequalities are therefore tighter.
\end{remark*}
\begin{remark*}
    This is precisely corollary~\ref{corol-main-text:haar-1-n-generalisation} in the main text.
\end{remark*}
\begin{proof}[Proof]
    We want to use the second inequality of theorem~\ref{thm:inequalities}. Let $\kket{\mu^{\left(i\right)}}$ be the first moment and $\mathcal{S}^{\left(i\right)}$ be the second moment of $D^{\left(i\right)}$. By lemma~\ref{lma:s-haar-distribution}, both $\mathcal{S}^{\left(i\right)}$ are diagonalisable in the same basis, 
    \[\left\{\kket{\widetilde{P}} \suchthat P \in \mathbb{P}_n \setminus \left\{I^{\otimes n}\right\}\right\}.\]
    
    Moreover, we know that $\kket{\mu^{\left(1\right)}} = \kket{\mu^{\left(2\right)}} = 0$, $\lambda_{max}\left(\mathcal{S}^{\left(1\right)}\right) = \frac{1}{2^n + 1} < 1$ and $\lambda_{max}\left(\mathcal{S}^{\left(2\right)}\right) = \frac{1}{3} < 1$. We therefore have the hypothesis to use the second inequality of theorem~\ref{thm:inequalities}.
    
    Still by lemma~\ref{lma:s-haar-distribution}, any eigenvalue of $\mathcal{S}^{\left(1\right)}$ is $\frac{1}{2^n + 1}$, and any eigenvalue of $\mathcal{S}^{\left(2\right)}$ is of the form $\frac{1}{3^j}$ where $j = w_1\left(P\right)$  is the Pauli weight associated to the corresponding eigenvector $\kket{\widetilde{P}}$.
    
    We now apply twice the second inequality of theorem~\ref{thm:inequalities}, one for each inequality we try to prove.
    \begin{itemize}[left=0pt]
    \item We wish that, for all $k$, and thus for all $j$: 
    \[\alpha_{12} \geq \frac{1 - \lambda_k\left(\mathcal{S}^{\left(1\right)}\right)}{1 - \lambda_k\left(\mathcal{S}^{\left(2\right)}\right)} = \frac{1 - \frac{1}{2^n + 1}}{1 - \frac{1}{3^j}} = \frac{2^n}{2^n + 1} \frac{3^j}{3^j - 1}\]
    
    Using usual calculus tools, one can find that this is maximised for the smallest $j$, $j = 1$. We can thus pick: 
    \[\alpha_{12} = \frac{2^n}{2^n + 1}\cdot \frac{3}{2}.\]
    
    This yields that, by theorem~\ref{thm:inequalities}: 
    \[R_1\left(U, V\right) \leq \frac{3}{2} \cdot \frac{2^n}{2^n + 1} R_2\left(U, V\right).\]

    Since we used the second inequality of theorem~\ref{thm:inequalities}, and since $\kket{\widetilde{P}}$ for any $P \in \mathbb{P}_n$ with $w_1\left(P\right) = 1$ is a common eigenvector whose eigenvalues maximise the above ratio, we know that any unitaries $U, V$ such that $V^{\dagger} U = P$  make the bound tight by theorem~\ref{thm:inequalities}.
    
    \item We now wish that, for all $k$ and thus for all $j$: 
    \[\alpha_{21} \geq \frac{1 - \lambda_k\left(\mathcal{S}^{\left(2\right)}\right)}{1 - \lambda_k\left(\mathcal{S}^{\left(1\right)}\right)} = \frac{1 - \frac{1}{3^j}}{1 - \frac{1}{2^n + 1}} = \frac{3^j - 1}{3^j} \cdot \frac{2^n + 1}{2^n}.\]
    
    Again, using calculus tools, one can find that this is maximised for the largest $j$, $j = n$. We can thus pick: 
    \[\alpha_{21} = \frac{3^n - 1}{3^n} \cdot \frac{2^n + 1}{2^n}.\]
    
    This yields that, by theorem~\ref{thm:inequalities}:  
    \autoeq{R_2\left(U, V\right) \leq \frac{3^n - 1}{3^n}\cdot \frac{2^n + 1}{2^n} R_1\left(U, V\right) \iff R_1\left(U, V\right) \geq \frac{3^n}{3^n - 1}\cdot \frac{2^n}{2^n + 1} R_2\left(U, V\right).}
    
    This is our second inequality. Since we again used the second inequality of theorem~\ref{thm:inequalities}, we know equality is reached whenever $\ket{w\left(U, V\right)} = \kket{\widetilde{P}}$, where $w_1\left(P\right) = j = n$ has maximal Pauli weight. 
    \end{itemize}
\end{proof}

\begin{corollary}\label{crl:haar-zero-mean-generalisation}
    For all mixed-state distributions $D$ over $nk$ qubits of mean $\kket{\mu}$, and for all unitaries $U, V$: 
    \[R_D\left(U, V\right) \leq \frac{2^n + 1}{2^n} R_{\Haar_n^{\otimes k}\left(U\right)}\left(U, V\right) + \left\|\kket{\mu}\right\|.\]
\end{corollary}
\begin{intuition*}
    Intuitively, this means that having a small error on any 2-design means that we have a small error on any zero-mean distribution. Recall that theorem~\ref{thm:bound-support-size-mu-eta} states that asking for a pure-state distribution to be zero-mean is not a small requirement, and that learning well on average (like having a small error on the Haar distribution) does not necessarily mean learning well everywhere.
\end{intuition*}
\begin{remark*}
    Leaving $k = 1$, this is precisely the first inequality of corollary~\ref{corol-main-text:haar-general-comparison} in the main text. The second inequality of this corollary is also just a direct application of corollary~\ref{crl:zero-mean-generalisation}.
\end{remark*}
\begin{proof}[Proof]
    By lemma~\ref{lma:s-haar-distribution}, $\Haar_n^{\otimes k}\left(U\right)$ is zero mean and has an expressivity measure $\eta$ given by: 
    \[\kket{\mu} = 0, \mathspace \eta = \frac{2^n}{2^n + 1}.\]
    
    But then, by corollary~\ref{crl:zero-mean-generalisation}, given any mixed-state distribution $D$ over the same number of qubits as $\Haar_n^{\otimes k}\left(U\right)$, of first moment $\kket{\mu}$: 
    \autoeq{R_D\left(U, V\right) \leq \frac{1}{\eta} R_{\Haar_n^{\otimes k}\left(U\right)}\left(U, V\right) + \left\|\kket{\mu}\right\| = \frac{2^n + 1}{2^n} R_{\Haar_n^{\otimes k}\left(U\right)}\left(U, V\right) + \left\|\kket{\mu}\right\|.}
\end{proof}

\suppsubsection{Haar measure on the unitary group under depolarising noise}\label{sec:noiseproofs}

Local depolarizing noise changes both the purity and the second-moment spectrum of a Haar ensemble. We compute these changes to determine how a small noisy risk constrains the noiseless risk. The same moments also give bounds for generalizing from noisy Haar inputs to other target distributions.

\begin{definition}
    Given some $\delta \in \left[0, 1\right]$, the depolarising noise channel on one qubit is defined as:
    \[\mathcal{D}_{\delta}\left(\rho\right) = \left(1-\delta\right) \rho + \delta \frac{I}{2} \Tr\left(\rho\right).\]

    Given an $n$-qubit distribution $D$, we write $\mathcal{D}_{\delta}^{\otimes n}\left(D\right)$ for the distribution which outputs $\mathcal{D}_{\delta}^{\otimes n}\left(\rho\right)$ for $\rho \followsdistr D$.
\end{definition}

\begin{theorem}\label{thm:m-depolarised-haar-distribution}
    Let $\delta \in \left[0, 1\right]$ and $D = \mathcal{D}_{\delta}^{\otimes n}\left(\Haar_n\right)$. Moreover, let $M$ be its moment matrix, $\mu$ be its first moment, $\bar{\gamma}$ be its average purity, $\mathcal{S}$ be its second moment and $\eta$ be its expressivity measure.

    Then, for any Paulis $P, Q \in \mathbb{P}_n$:
    \[
        \bbra{\widetilde{P}} M \kket{\widetilde{Q}}
        =
        \delta_{PQ}
        \frac{
            1 +
            \left(\frac{1 + 3\left(1-\delta\right)^2}{2}\right)^{n-w_1\left(P\right)}
            \left(\frac{1-\left(1-\delta\right)^2}{2}\right)^{w_1\left(P\right)}
        }{2^n+1},
    \]
    where $w_1\left(P\right)$ is the Pauli weight of $P$ (definition~\ref{def:pauli-block-weight}), i.e.\ the number of elements that are not equal to the identity in its tensor product expansion.

    In particular:
    \[
        \mu = 0,
        \mathspace
        \bar{\gamma}
        =
        \frac{1+\left(\frac{1+3\left(1-\delta\right)^2}{2}\right)^n}{2^n+1}, \mathspace \eta
        =
        \frac{2\left(1-\delta\right)^2}{2^n+1}
        \left(\frac{1+3\left(1-\delta\right)^2}{2}\right)^{n-1}.
    \]

    Moreover, a basis of eigenvector-eigenvalue pairs of $\mathcal{S}$ is given by, for $P \in \mathbb{P}_n \setminus \left\{I^{\otimes n}\right\}$:
    \[\left(\kket{\widetilde{P}}, 
        \frac{
                    1 +
                    \left(\frac{1+3\left(1-\delta\right)^2}{2}\right)^{n-w_1\left(P\right)}
                    \left(\frac{1-\left(1-\delta\right)^2}{2}\right)^{w_1\left(P\right)}
                }{2^n+1}.
    \right)\]
\end{theorem}
\begin{proof}[Proof]

We compute the noisy moment matrix by applying the depolarizing channel to both copies of the Haar tensor moment and then taking a partial transpose. The resulting Pauli sum factors over qubits, so two single-qubit calculations determine its dependence on Pauli weight.

    We start from the second tensor moment of the pure-state Haar distribution. Reference~\cite{mele2023introduction}'s theorem 22 tells us that:
    \[
        \exval_{\rho \followsdistr \Haar_n}\left[\rho^{\otimes 2}\right]
        =
        \frac{I^{\otimes 2n}+SWAP}{2^n\left(2^n+1\right)},
    \]
    where $SWAP$ exchanges the two $n$-qubit registers. We first express $SWAP$ in the Pauli basis. Observe that $\Tr\left(\left(\widetilde{P} \otimes \widetilde{Q}\right) SWAP\right) = \delta_{PQ}$ as proven in equation~\ref{eq:trace-swap-product}. Since $\left\{\widetilde{P} \otimes \widetilde{Q} \suchthat P, Q \in \mathbb{P}_n\right\}$ is an orthonormal basis of $\mathbb{C}^{4^n \times 4^n}$, this tells us that:
    \[
        SWAP = \sum_{Q \in \mathbb{P}_n} \widetilde{Q} \otimes \widetilde{Q} = 2^{-n}\sum_{Q \in \mathbb{P}_n} Q \otimes Q.
    \]
    Now, by definition of the depolarising channel:
    \[
        \mathcal{D}_{\delta}\left(I\right) = I,
        \mathspace
        \mathcal{D}_{\delta}\left(Q\right) = \left(1-\delta\right)Q,
        \quad Q \in \left\{X, Y, Z\right\}.
    \]
    Therefore, for any $Q \in \mathbb{P}_n$:
    \[
        \mathcal{D}_{\delta}^{\otimes n}\left(Q\right)
        =
        \left(1-\delta\right)^{w_1\left(Q\right)}Q.
    \]

    Using the linearity of the expected value and of the channels, this yields:
    \[
    \begin{split}
        \exval_{\rho \followsdistr D}\left[\rho^{\otimes 2}\right]
        &= \exval_{\rho \followsdistr \Haar_n}\left[\mathcal{D}_{\delta}^{\otimes n}\left(\rho\right) \otimes \mathcal{D}_{\delta}^{\otimes n}\left(\rho\right)\right] \\
        &= \left(\mathcal{D}_{\delta}^{\otimes n} \otimes \mathcal{D}_{\delta}^{\otimes n}\right) \exval_{\rho \followsdistr \Haar_n}\left[\rho^{\otimes 2}\right]\\
        &=
        \left(\mathcal{D}_{\delta}^{\otimes n} \otimes \mathcal{D}_{\delta}^{\otimes n}\right)
        \left(
            \frac{I^{\otimes 2n}+SWAP}{2^n\left(2^n+1\right)}
        \right) \\
        &=
        \frac{1}{2^n\left(2^n+1\right)}
        \left[
            I^{\otimes 2n}
            +
            \frac{1}{2^n}
            \sum_{Q \in \mathbb{P}_n}
            \left(1-\delta\right)^{2w_1\left(Q\right)}
            Q \otimes Q
        \right].
    \end{split}
    \]
    We now apply the partial transpose $T_2$ on the second register, as in the remark of definition~\ref{def:moment-matrix}:
    \[
    \begin{split}
        M
        &=
        2^n T_2\left(\exval_{\rho \followsdistr D}\left[\rho^{\otimes 2}\right]\right) \\
        &=
        \frac{1}{2^n+1}
        \left[
            I^{\otimes 2n}
            +
            \frac{1}{2^n}
            \sum_{Q \in \mathbb{P}_n}
            \left(1-\delta\right)^{2w_1\left(Q\right)}
            Q \otimes Q^T
        \right].
    \end{split}
    \]

    It remains to express this matrix in the $\kket{\widetilde{P}}$ basis. Let $P = P_1 \otimes \ldots \otimes P_n \in \mathbb{P}_n$ be arbitrary. By properties of vectorisation (definition~\ref{def:vectorised-notation}):
    \[
        \frac{1}{2^n} \sum_{Q \in \mathbb{P}_n} \left(1 - \delta\right)^{2 w_1\left(Q\right)} \left(Q \otimes Q^T\right)\kket{P} = \kket{\frac{1}{2^n} \sum_{Q \in \mathbb{P}_n} \left(1 - \delta\right)^{2 w_1\left(Q\right)}QPQ}.
    \]
    We aim to simplify the content of the ket. Consider the single-qubit case first, we observe that for each $\ell \in \left\{1, \ldots, n\right\}$, we have:
    \[
        \sum_{Q_{\ell} \in \mathbb{P}_1}
        \left(1-\delta\right)^{2w_1\left(Q_{\ell}\right)}
        Q_{\ell} P_{\ell} Q_{\ell}
        =
        \begin{cases}
            \left(1+3\left(1-\delta\right)^2\right)P_{\ell},
            & \text{if } P_{\ell} = I, \\
            \left(1-\left(1-\delta\right)^2\right)P_{\ell},
            & \text{if } P_{\ell} \neq I.
        \end{cases}
    \]
    Indeed, when $P_{\ell} = I$, all three non-identity Paulis contribute positively. Otherwise, one non-identity Pauli commutes with $P_{\ell}$ (contributing positively), and the two others anti-commute with it (contributing negatively). The general sum on $n$ qubits simplifies nicely from this result:
\begin{equation}\begin{aligned}
\frac{1}{2^n} \sum_{Q \in \mathbb{P}_n}
\left(1-\delta\right)^{2w_1\left(Q\right)} QPQ
&=
\frac{1}{2^n} \bigotimes_{\ell=1}^{n}
\left(
\sum_{Q_{\ell}\in\mathbb{P}_1}
\left(1-\delta\right)^{2w_1\left(Q_{\ell}\right)}
Q_{\ell} P_{\ell} Q_{\ell}
\right)
\\ &=
\frac{1}{2^n}
\left(1+3\left(1-\delta\right)^2\right)^{n-w_1\left(P\right)}
\left(1-\left(1-\delta\right)^2\right)^{w_1\left(P\right)} P
\\ &=
\left(\frac{1+3\left(1-\delta\right)^2}{2}\right)^{n-w_1\left(P\right)}
\left(\frac{1-\left(1-\delta\right)^2}{2}\right)^{w_1\left(P\right)} P.
\end{aligned}\end{equation}

    We have found so far that: 
    \[M \kket{\widetilde{P}} = \frac{1}{2^n + 1} \kket{\widetilde{P}} + \frac{1}{2^n + 1} \left(\frac{1 + 3 \left(1 - \delta\right)^2}{2}\right)^{n - w_1\left(P\right)} \left(\frac{1 - \left(1 - \delta\right)^2}{2}\right)^{w_1\left(P\right)} \kket{\widetilde{P}}.\]
    
    Therefore, $\kket{\widetilde{P}}$ is an eigenvector of $M$ with the eigenvalue stated in the theorem. Since the normalised Paulis form an orthonormal basis, this gives the expression for all the entries of $M$.

    Diagonality implies $M_{mean} = 0$ and hence by lemma~\ref{lma:other-formulation-decomposition} that $\mu = 0$. By the same theorem, we know the eigenvalue associated to $\kket{\widetilde{I}^{\otimes n}}$ is $\bar{\gamma}$. Removing this eigenvector gives the claimed eigenbasis of $\mathcal{S}$.

    Finally, notice that:
    \[
        0
        \leq
        \frac{1-\left(1-\delta\right)^2}{2}
        \leq
        \frac{1+3\left(1-\delta\right)^2}{2}.
    \]
    Hence, the largest eigenvalue of $\mathcal{S}$ is attained for Pauli weight $1$:
    \[
        \lambda_{max}\left(\mathcal{S}\right)
        =
        \frac{
            1+
            \left(\frac{1+3\left(1-\delta\right)^2}{2}\right)^{n-1}
            \frac{1-\left(1-\delta\right)^2}{2}
        }{2^n+1}.
    \]
    Subtracting this expression from $\bar{\gamma}$ gives:
    \[
        \eta
        =
        \bar{\gamma}-\lambda_{max}\left(\mathcal{S}\right)
        =
        \left(\frac{1 + 3 \left(1 - \delta\right)^2}{2} - \frac{1 - \left(1- \delta\right)^2}{2}\right) \cdot \frac{1}{2^n + 1}\left(\frac{1 + 3\left(1- \delta\right)^2}{2}\right)^{n-1}
        =
        \frac{2\left(1-\delta\right)^2}{2^n+1}
        \left(\frac{1+3\left(1-\delta\right)^2}{2}\right)^{n-1},
    \]
    as required.
\end{proof}

Local noise leaves the single-qubit factors of a product Haar ensemble independent. We can therefore use the one-qubit result from Theorem~\ref{thm:m-depolarised-haar-distribution} and take tensor products. The Pauli weight counts how often each of the two one-qubit eigenvalues contributes.

\begin{theorem}\label{thm:m-depolarised-product-haar-distribution}
    Let $\delta \in \left[0, 1\right]$ and $D = \mathcal{D}_{\delta}^{\otimes n}\left(\Haar_1^{\otimes n}\right)$. Moreover, let $M$ be its moment matrix, $\mu$ be its first moment, $\bar{\gamma}$ be its average purity, $\mathcal{S}$ be its second moment and $\eta$ be its expressivity measure.

    Then, for any Paulis $P, Q \in \mathbb{P}_n$:
    \[
        \bbra{\widetilde{P}} M \kket{\widetilde{Q}}
        =
        \delta_{PQ}
        \left(\frac{1+\left(1-\delta\right)^2}{2}\right)^{n-w_1\left(P\right)}
        \left(\frac{3-\left(1-\delta\right)^2}{6}\right)^{w_1\left(P\right)},
    \]
    where $w_1\left(P\right)$ is the Pauli weight of $P$ (definition~\ref{def:pauli-block-weight}).

    In particular:
    \[\mu = 0, \mathspace \bar{\gamma} = \left(\frac{1+\left(1-\delta\right)^2}{2}\right)^n, \mathspace \eta = \frac{2\left(1-\delta\right)^2}{3} \left(\frac{1+\left(1-\delta\right)^2}{2}\right)^{n-1}.\]

    Moreover, a basis of eigenvector-eigenvalue pairs of $\mathcal{S}$ is given by, for $P \in \mathbb{P}_n \setminus \left\{I^{\otimes n}\right\}$:
    \[ 
        \left(\kket{\widetilde{P}}, \left(\frac{1+\left(1-\delta\right)^2}{2}\right)^{n-w_1\left(P\right)}
        \left(\frac{3-\left(1-\delta\right)^2}{6}\right)^{w_1\left(P\right)}\right).
    \]
\end{theorem}
\begin{proof}[Proof]
    We first apply theorem~\ref{thm:m-depolarised-haar-distribution} with $n = 1$. In the $\kket{\widetilde{P}}$ basis, the moment matrix of $\mathcal{D}_{\delta}\left(\Haar_1\right)$ is:
    \[
        \frac{1+\left(1-\delta\right)^2}{2}
        \oplus
        \left(
            \frac{3-\left(1-\delta\right)^2}{6}I_3
        \right).
    \]

    Now, observe that:
    \[
        \mathcal{D}_{\delta}^{\otimes n}\left(\Haar_1^{\otimes n}\right)
        =
        \left(\mathcal{D}_{\delta}\left(\Haar_1\right)\right)^{\otimes n}.
    \]
    By theorem~\ref{thm:m-product-distribution}, the normalised vectorised Paulis are thus eigenvectors of $M$, and their eigenvalues are products of the corresponding one-qubit eigenvalues. For any $P \in \mathbb{P}_n$, exactly $w_1\left(P\right)$ factors of $P$ are non-identity, and exactly $n - w_1\left(P\right)$ factors of $P$ are the identity. This gives:
    \[
        M\kket{\widetilde{P}}
        =
        \left(\frac{1+\left(1-\delta\right)^2}{2}\right)^{n-w_1\left(P\right)}
        \left(\frac{3-\left(1-\delta\right)^2}{6}\right)^{w_1\left(P\right)}
        \kket{\widetilde{P}}.
    \]

    As in the previous theorem, diagonality gives $\mu = 0$, the identity eigenvector gives $\bar{\gamma}$, and the other eigenvectors give the spectrum of $\mathcal{S}$.

    Finally, since:
    \[
        0
        \leq
        \frac{3-\left(1-\delta\right)^2}{6}
        \leq
        \frac{1+\left(1-\delta\right)^2}{2},
    \]
    the largest eigenvalue of $\mathcal{S}$ is attained for Pauli weight $1$. Therefore:
    \[
    \begin{split}
        \eta
        &= \bar{\gamma} - \lambda_{max}\left(\mathcal{S}\right)\\
        &=
        \left(\frac{1+\left(1-\delta\right)^2}{2}\right)^n
        -
        \left(\frac{1+\left(1-\delta\right)^2}{2}\right)^{n-1}
        \frac{3-\left(1-\delta\right)^2}{6} \\
        &=
        \frac{2\left(1-\delta\right)^2}{3}
        \left(\frac{1+\left(1-\delta\right)^2}{2}\right)^{n-1},
    \end{split}
    \]
    as required.
\end{proof}

The noisy and noiseless risks can be compared in the same Pauli basis. In each direction, the relevant coefficient is the average purity minus the second-moment eigenvalue. The largest ratio of these coefficients determines the comparison factor, and a Pauli attaining that ratio makes the bound tight.

\begin{theorem}\label{thm:depolarised-haar-tight-comparison}
    Let $\delta \in \left[0, 1\right)$. Then, for all unitaries $U, V$:
    \[
        R_{\Haar_n}\left(U, V\right)
        \leq
        \frac{1}{\left(1-\delta\right)^2}
        \left(\frac{4}{1+3\left(1-\delta\right)^2}\right)^{n-1}
        R_{\mathcal{D}_{\delta}^{\otimes n}\left(\Haar_n\right)}\left(U, V\right).
    \]

    If moreover $V^{\dagger} U \in \mathbb{P}_n$ and it has a single non-identity element in its tensor product expansion, then this inequality is reached with equality. That is, the bound is tight. 
\end{theorem}
\begin{remark*}
    Let $\alpha$ be the multiplicative factor in this theorem. Note that $\alpha$ is polynomially bounded in $n$ if and only if:
    \[
        \delta = O\left(\frac{\log n}{n}\right).
    \]

    In particular, $\delta = O\left(1/n\right)$ gives a bounded multiplicative factor, whereas any fixed $\delta \in \left(0, 1\right)$ gives an exponentially growing factor. This makes sense by theorem~\ref{thm:purity-penalty} since, whenever $\delta = o\left(\frac{\log n}{n}\right)$, then $\bar{\gamma} = \frac{1}{\poly\left(n\right)}$.
\end{remark*}
\begin{remark*}
    The restriction $\delta < 1$ is necessary. Indeed, when $\delta = 1$, the distribution outputs $I^{\otimes n}/2^n$ deterministically. Its risk is therefore zero for all unitaries $U, V$. However, the corresponding noiseless risks can be strictly positive, for instance when $V^{\dagger}U$ is a Pauli of weight $1$. Hence, no finite multiplicative factor can give a generalisation bound with $\beta = 0$.
\end{remark*}
\begin{remark*}
    This is the first result of proposition~\ref{thm-main-text:locally-depolarized-inputs} in the main text.
\end{remark*}
\begin{proof}[Proof]
    Let $D_1 = \Haar_n$ and $D_2 = \mathcal{D}_{\delta}^{\otimes n}\left(\Haar_n\right)$. We write $R_i = R_{D_i}$. We use the second inequality of theorem~\ref{thm:inequalities}, with $\beta = 0$. Let $\mu^{\left(i\right)}$, $\bar{\gamma}^{\left(i\right)}$, $\mathcal{S}^{\left(i\right)}$ and $\eta^{\left(i\right)}$ be the respective moments, average purity and expressivity measure of $D_i$.

    By lemma~\ref{lma:m-haar-distribution} and theorem~\ref{thm:m-depolarised-haar-distribution}, both distributions have zero mean and their second moments are diagonal in the same Pauli basis. Moreover, $\eta^{\left(2\right)} > 0$ since $\delta < 1$. Therefore, we can take:
    \[
        \alpha
        =
        \max_j
        \frac{
            \bar{\gamma}^{\left(1\right)}
            -
            \lambda_j\left(\mathcal{S}^{\left(1\right)}\right)
        }{
            \bar{\gamma}^{\left(2\right)}
            -
            \lambda_j\left(\mathcal{S}^{\left(2\right)}\right)
        }.
    \]

    Since $\bar{\gamma}^{\left(1\right)} = 1$ and $\mathcal{S}^{\left(1\right)} = \frac{1}{2^n+1}I$, the numerator does not depend on $j$. Hence, the maximum is attained at the largest eigenvalue of $\mathcal{S}^{\left(2\right)}$, which corresponds to Pauli weight $1$. Using theorem~\ref{thm:m-depolarised-haar-distribution}, this yields:
    \[
    \begin{split}
        \alpha
        &=
        \frac{1-\frac{1}{2^n+1}}{\eta^{\left(2\right)}} \\
        &=
        \frac{2^n}{2^n+1}
        \frac{2^n+1}{
            2\left(1-\delta\right)^2
            \left(\frac{1+3\left(1-\delta\right)^2}{2}\right)^{n-1}
        } \\
        &=
        \frac{1}{\left(1-\delta\right)^2}
        \left(\frac{4}{1+3\left(1-\delta\right)^2}\right)^{n-1}.
    \end{split}
    \]
    This proves the inequality.

    Finally, any $\kket{\widetilde{P}}$ with $w_1\left(P\right) = 1$ is a common eigenvector attaining the maximum above. Since $P$ is unitary, we can take $U = P$ and $V = I^{\otimes n}$, giving:
    \[
        \ket{w\left(U, V\right)} = \kket{\widetilde{P}}.
    \]
    More generally, any $U, V$ such that $V^{\dagger}U = P$ work. By the equality condition of theorem~\ref{thm:inequalities}, these choices make the inequality tight.
\end{proof}

\begin{theorem}\label{thm:depolarised-product-haar-tight-comparison}
    Let $\delta \in \left[0, 1\right)$. Then, for all unitaries $U, V$:
    \[
        R_{\Haar_1^{\otimes n}}\left(U, V\right)
        \leq
        \frac{1}{\left(1-\delta\right)^2}
        \left(\frac{2}{1+\left(1-\delta\right)^2}\right)^{n-1}
        R_{\mathcal{D}_{\delta}^{\otimes n}\left(\Haar_1^{\otimes n}\right)}\left(U, V\right).
    \]

    Let $W = V^{\dagger}U$. If $W \in \mathbb{P}_n$ and $w_1\left(W\right) = 1$, then this inequality is reached with equality. In particular, its multiplicative factor cannot be improved.
\end{theorem}
\begin{remark*}
    The same three remarks as theorem~\ref{thm:depolarised-haar-tight-comparison} apply. In particular, this is the second part of proposition~\ref{thm-main-text:locally-depolarized-inputs} in the main text.
\end{remark*}
\begin{proof}[Proof]
    Let $D_1 = \Haar_1^{\otimes n}$ and $D_2 = \mathcal{D}_{\delta}^{\otimes n}\left(\Haar_1^{\otimes n}\right)$. We write $R_i = R_{D_i}$. We again use the second inequality of theorem~\ref{thm:inequalities}, with $\beta = 0$. Let $\mu^{\left(i\right)}$, $\bar{\gamma}^{\left(i\right)}$, $\mathcal{S}^{\left(i\right)}$ and $\eta^{\left(i\right)}$ be the respective moments, average purity and expressivity measure of $D_i$.

    By lemma~\ref{lma:s-haar-distribution} and theorem~\ref{thm:m-depolarised-product-haar-distribution}, both distributions have zero mean and their second moments are diagonal in the same Pauli basis. Moreover, $\eta^{\left(2\right)} > 0$ since $\delta < 1$.

     For a common eigenvector $\kket{\widetilde{P}}$ of Pauli weight $j = w_1\left(P\right) \in \left\{1, \ldots, n\right\}$, the ratio appearing in theorem~\ref{thm:inequalities} is:
    \[\frac{1-3^{-j}}{\left(\frac{1+\left(1-\delta\right)^2}{2}\right)^n - \left(\frac{1+\left(1-\delta\right)^2}{2}\right)^{n-j}\left(\frac{3-\left(1-\delta\right)^2}{6}\right)^j}.\]
    $\alpha$ is the maximum of this expression over $j$. Our goal is to show that this ratio is maximal when $j = 1$. Factoring the denominator, we obtain:
    \autoeq{\frac{1-3^{-j}}{\left(\frac{1+\left(1-\delta\right)^2}{2}\right)^n - \left(\frac{1+\left(1-\delta\right)^2}{2}\right)^{n-j}\left(\frac{3-\left(1-\delta\right)^2}{6}\right)^j} = \left(\frac{2}{1+\left(1-\delta\right)^2}\right)^n \frac{1-3^{-j}}{1-\left(\frac{3-\left(1-\delta\right)^2}{3\left(1+\left(1-\delta\right)^2\right)}\right)^j}.}

    The first factor does not depend on $j$. For the second factor, using the geometric series identity $1-x^j = \left(1-x\right)\sum_{\ell=0}^{j-1}x^{\ell}$ gives:
    \autoeq{\frac{1-3^{-j}}{1-\left(\frac{3-\left(1-\delta\right)^2}{3\left(1+\left(1-\delta\right)^2\right)}\right)^j} = \frac{1-\frac{1}{3}}{1-\frac{3-\left(1-\delta\right)^2}{3\left(1+\left(1-\delta\right)^2\right)}} \cdot \frac{\sum_{\ell=0}^{j-1}3^{-\ell}}{\sum_{\ell=0}^{j-1}\left(\frac{3-\left(1-\delta\right)^2}{3\left(1+\left(1-\delta\right)^2\right)}\right)^{\ell}}.}

    Now, since $0 < \left(1-\delta\right)^2 \leq 1$:
    \[\frac{1}{3} \leq \frac{3-\left(1-\delta\right)^2}{3\left(1+\left(1-\delta\right)^2\right)} < 1.\]
    Therefore, every term in the numerator sum is at most the corresponding term in the denominator sum. Their ratio is thus at most $1$, with equality when $j = 1$. Hence, the maximum is attained when $j = 1$, giving:
    \autoeq{\alpha = \left(\frac{2}{1+\left(1-\delta\right)^2}\right)^n \frac{1-\frac{1}{3}}{1-\frac{3-\left(1-\delta\right)^2}{3\left(1+\left(1-\delta\right)^2\right)}} = \left(\frac{2}{1+\left(1-\delta\right)^2}\right)^n \frac{1+\left(1-\delta\right)^2}{2\left(1-\delta\right)^2} = \frac{1}{\left(1-\delta\right)^2}\left(\frac{2}{1+\left(1-\delta\right)^2}\right)^{n-1}.}

    Finally, the maximum is attained by any common eigenvector $\kket{\widetilde{P}}$ with $w_1\left(P\right) = 1$. As in the previous theorem, taking $V^{\dagger}U = P$ gives $\ket{w\left(U, V\right)} = \kket{\widetilde{P}}$. The equality condition of theorem~\ref{thm:inequalities} therefore proves tightness.
\end{proof}

These noisy training ensembles also allow comparisons with arbitrary target distributions. Substituting their expressivity measures into Corollary~\ref{crl:zero-mean-generalisation} gives the noise-dependent multiplicative factor. The target mean contributes the remaining additive term.

\begin{theorem}\label{thm:depolarised-haar-arbitrary-comparison}
    Let $\delta \in \left[0, 1\right)$. Then, for all distributions $D$ over mixed states of mean $\mu$, and for all unitaries $U, V$:
    \[
        R_D\left(U, V\right)
        \leq
        \frac{1+2^{-n}}{\left(1-\delta\right)^2}
        \left(\frac{4}{1+3\left(1-\delta\right)^2}\right)^{n-1}
        R_{\mathcal{D}_{\delta}^{\otimes n}\left(\Haar_n\right)}\left(U, V\right)
        +
        \left\|\mu\right\|_2.
    \]
\end{theorem}
\begin{remark*}
    In particular, generalisation to any zero-mean distribution is multiplicative. Moreover, the multiplicative factor can again be proven to be polynomially bounded in $n$ if and only if $\delta = O\left(\frac{\log n}{n}\right)$, and constant in $n$ if and only if $\delta = O\left(\frac{1}{n}\right)$. This again makes sense by theorem~\ref{thm:purity-penalty}, for the same reason as theorem~\ref{thm:depolarised-haar-tight-comparison}.
\end{remark*}
\begin{remark*}
    If the average purity $\bar{\gamma}$ of the target distribution $D$ is known, the bound can be improved by multiplying the multiplicative factor by $\bar{\gamma}$, per corollary~\ref{crl:zero-mean-generalisation}.
\end{remark*}
\begin{proof}[Proof]
    Let $D_2 = \mathcal{D}_{\delta}^{\otimes n}\left(\Haar_n\right)$. We write $R_2 = R_{D_2}$. Let $\mu^{\left(2\right)}$ and $\eta^{\left(2\right)}$ be the mean and expressivity measure of $D_2$.

   Theorem~\ref{thm:m-depolarised-haar-distribution} gives $\mu^{\left(2\right)} = 0$ and:
    \[
    \begin{split}
        \frac{1}{\eta^{\left(2\right)}}
        &=
        \frac{2^n+1}{2\left(1-\delta\right)^2}
        \left(\frac{2}{1+3\left(1-\delta\right)^2}\right)^{n-1} \\
        &=
        \frac{1+2^{-n}}{\left(1-\delta\right)^2}
        \left(\frac{4}{1+3\left(1-\delta\right)^2}\right)^{n-1},
    \end{split}
    \]
    Since $\eta^{\left(2\right)} > 0$, corollary~\ref{crl:zero-mean-generalisation} thus tells us that:
    \[
        R_D\left(U, V\right)
        \leq
        \frac{1}{\eta^{\left(2\right)}}R_2\left(U, V\right)
        +
        \left\|\mu\right\|_2,
    \]
    giving our result.
\end{proof}

\begin{theorem}\label{thm:depolarised-product-haar-arbitrary-comparison}
    Let $\delta \in \left[0, 1\right)$. For any distribution over mixed states $D$ of mean $\mu$, and unitaries $U, V$: 
    \[
        R_D\left(U, V\right)
        \leq
        \frac{3}{2\left(1-\delta\right)^2}
        \left(\frac{2}{1+\left(1-\delta\right)^2}\right)^{n-1}
        R_{\mathcal{D}_{\delta}^{\otimes n}\left(\Haar_1^{\otimes n}\right)}\left(U, V\right)
        +
        \left\|\mu\right\|_2.
    \]
\end{theorem}
\begin{remark*}
    The remarks of theorem~\ref{thm:depolarised-haar-arbitrary-comparison} also apply to this theorem.
\end{remark*}
\begin{proof}[Proof]
    Let $D_2 = \mathcal{D}_{\delta}^{\otimes n}\left(\Haar_1^{\otimes n}\right)$. We write $R_2 = R_{D_2}$, let $\mu^{\left(2\right)}$ and $\eta^{\left(2\right)}$ be the mean and expressivity measure of $D_2$. Theorem~\ref{thm:m-depolarised-product-haar-distribution} gives $\mu^{\left(2\right)} = 0$ and:
    \[\frac{1}{\eta^{\left(2\right)}} = \frac{3}{2\left(1-\delta\right)^2}\left(\frac{2}{1+\left(1-\delta\right)^2}\right)^{n-1}.\]

    Since $\eta^{\left(2\right)} > 0$, corollary~\ref{crl:zero-mean-generalisation} thus tells us that:
    \[R_D\left(U, V\right) \leq \frac{1}{\eta^{\left(2\right)}}R_2\left(U, V\right) + \left\|\mu\right\|_2,\]
    giving our result.
\end{proof}

\suppsubsection{Haar measure on the orthogonal group}

Real-state ensembles have vanishing second-moment eigenvalues in some Pauli directions. However, this does not prevent an out-of-distribution generalization bound. We compute the spectrum and compare global and product real-state ensembles to see how this restriction changes the risk comparison.

\begin{definition}[Haar measure on the orthogonal group]\label{def:orthgonal-haar}
    By $\Haar_n\left(O\right)$, we mean the distribution of Haar-random real states. In other words, given a Haar random orthogonal matrix $O$, $\Haar_n\left(O\right)$ outputs $\ket{\psi} = O \ket{0}^{\otimes n}$.
\end{definition}

\begin{definition}[Real matrix indicator function]\label{def:real-matrix-indicator-function}
    We define the real matrix indicator function $1_{\mathbb{R}}: \mathbb{C}^{2^n \times 2^n} \mapsto \left\{0, 1\right\}$, so that: 
    \[1_\mathbb{R}\left(M\right) = \begin{cases} 1, & \text{if $M \in \mathbb{R}^{2^n \times 2^n}$,} \\ 0, & \text{otherwise.} \end{cases}\]
\end{definition}
\begin{example*}
    For instance, for a Pauli $P \in \mathbb{P}_n$, this is $1$ if and only if it contains an even number of $Y$ in its tensor product expansion:
    \[1_{\mathbb{R}}\left(X\right) = 1_{\mathbb{R}}\left(XX\right) = 1_{\mathbb{R}}\left(YY\right) = 1,\]
    \[1_{\mathbb{R}}\left(Y\right) = 1_{\mathbb{R}}\left(XY\right) = 1_{\mathbb{R}}\left(YYY\right) = 0.\]
\end{example*}

\begin{lemma}\label{lma:m-orthogonal-haar-distribution}
    Let $D = \Haar_n\left(O\right)$ be the distribution of Haar random real pure states, $M$ be its moment matrix, $\kket{\mu}$ be its first moment and $\mathcal{S}$ be its second moment.
    
    Then, in the $\kket{\widetilde{P}}$ basis, for all $P, Q \in \mathbb{P}_n$: 
    \[\bbra{\widetilde{P}} M \kket{\widetilde{Q}} = \delta_{PQ} \frac{2^{n - 1} \delta_{PI} + 1_{\mathbb{R}}\left(P\right)}{2^{n - 1} + 1},\]
    where $1_{\mathbb{R}}$ is the real matrix indicator function (definition~\ref{def:real-matrix-indicator-function}).
    
    In particular, this means that $\kket{\mu} = 0$, and, for $P, Q \neq I^{\otimes n}$: 
    \[\bbra{\widetilde{P}} \mathcal{S} \kket{\widetilde{Q}} = \frac{1}{2^{n-1} + 1} \delta_{PQ} 1_{\mathbb{R}}\left(P\right).\]
\end{lemma}
\begin{remark*}
    Note that both $M$ and $\mathcal{S}$ are diagonal in the $\kket{\widetilde{P}}$ basis.
\end{remark*}
\begin{intuition*}
    Let us consider the case where $n = 1$ to build intuition. This lemma states that, in this case:
    \[M = \begin{pmatrix} 1 &  &  &  \\  & \frac{1}{2} &  &  \\  &  & 0 &  \\  &  &  & \frac{1}{2} \end{pmatrix}.\]
    
    Its shape is very similar to the one of Haar measure on the unitary group (lemma~\ref{lma:m-haar-distribution}), but with its eigenvalues associated to real Paulis ($X$ and $Z$) increased, and its eigenvalues associated to imaginary Paulis ($Y$) set to $0$. 
    
    Considering the representation of the vectors output by the distribution on the Bloch sphere is a good way to justify this intuitively: for $\ket{\psi} = \cos\left(\frac{\theta}{2}\right)\ket{0} + e^{i\phi} \sin\left(\frac{\theta}{2}\right) \ket{1}$ to be a real state, we may take $\phi \in \{0,\pi\}$. Therefore, this distribution simply outputs a state uniformly at random on the great circle where the Bloch $y$-component is zero. The $y$-component of any of those vectors is $0$, and the squared $x$ and $z$ components are both $\frac{1}{2}$ on average, giving exactly this second moment $\mathcal{S}$.
\end{intuition*}
\begin{proof}[Proof]

We evaluate the real Haar tensor moment by resolving its three pairings into the identity, the swap, and an outer product of the vectorized identity. The swap is the term that distinguishes real from imaginary traceless Pauli directions.

    Note that, by definition, states in our distribution are real. Therefore, by lemma~\ref{lma:moment-matrix-equivalence-pure-states}: 
    \[M = 2^n \exval\left(\rho \otimes \rho^*\right) = 2^n \exval\left(\rho^{\otimes 2}\right).\]
    
    This value is computed by reference~\cite{schatzki2024random}'s proposition 4. We will use their notations to build $M$ in our specific case. Here, we have $t = 2$ and $d = 2^n$, and we thus wish to compute $Z\left(2^n, 2\right)$ and $\mathcal{B}_2\left(2^n\right)$.
    
    To begin with, their lemma 3 directly gives: 
    \[Z\left(2^n, 2\right) = \prod_{j=0}^{2-1} \left(2^n + 2j\right) = 2^n \left(2^n + 2\right).\]
    
    Combining their figure 1 with their equation 7, we have three elements $\mathfrak{m}: \left\{1, \ldots, 2t\right\} \mapsto \left\{1, \ldots, 2t\right\}$ in $\mathcal{M}_2$: 
    \[\mathfrak{m}_1\left(1\right) = 1, \mathspace \mathfrak{m}_1\left(2\right) = 3, \mathspace \mathfrak{m}_1\left(3\right) = 2, \mathspace \mathfrak{m}_1\left(4\right) = 4;\]
    \[\mathfrak{m}_2\left(1\right) = 1, \mathspace \mathfrak{m}_2\left(2\right) = 4, \mathspace \mathfrak{m}_2\left(3\right) = 2, \mathspace \mathfrak{m}_2\left(4\right) = 3;\]
    \[\mathfrak{m}_3\left(1\right) = 1, \mathspace \mathfrak{m}_3\left(2\right) = 2, \mathspace \mathfrak{m}_3\left(3\right) = 3, \mathspace \mathfrak{m}_3\left(4\right) = 4.\]
    
    We can then use equation 8 to convert those into a representation of $\mathcal{B}_2\left(2^n\right)$: 
    \begin{equation}\begin{aligned}&\mathfrak{m}_1 \mapsto \sum_{i_1, i_2, i_3, i_4 = 1}^{2^n} \ket{i_1 i_2}\bra{i_3 i_4} \delta_{i_{\mathfrak{m}_1\left(1\right)} i_{\mathfrak{m}_1\left(2\right)}} \delta_{i_{\mathfrak{m}_1\left(3\right)} i_{\mathfrak{m}_1\left(4\right)}} \\ &\implies \mathfrak{m}_1 \mapsto \sum_{i_1, i_2, i_3, i_4 = 1}^{2^n} \ket{i_1 i_2}\bra{i_3 i_4} \delta_{i_1 i_3} \delta_{i_2 i_4} \\ &= \sum_{i, j} \ket{ij} \bra{ij} \\ &= I^{\otimes n} \otimes I^{\otimes n},\end{aligned}\end{equation}
    \[\mathfrak{m}_2 \mapsto \sum_{i_1, i_2, i_3, i_4=1}^{2^n} \ket{i_1 i_2} \bra{i_3 i_4} \delta_{i_1 i_4} \delta_{i_2 i_3} = \sum_{i, j} \ket{ij}\bra{ji} = SWAP,\]
    \[\mathfrak{m}_3 \mapsto \sum_{i_1, i_2, i_3, i_4 = 1}^{2^n} \ket{i_1 i_2} \bra{i_3 i_4} \delta_{i_1 i_2} \delta_{i_3 i_4} = \sum_{i, j} \ket{ii} \bra{jj} = \kket{I^{\otimes n}} \bbra{I^{\otimes n}}.\]
    
    We can then finally apply their proposition 4 to get that: 
    \autoeq{M = 2^n \exval\left(\rho^{\otimes 2}\right) = 2^n \frac{1}{Z\left(2^n, 2\right)} \sum_{\mathfrak{m} \in \mathcal{B}_2\left(2^n\right)} \mathfrak{m} = \frac{I^{\otimes n} \otimes I^{\otimes n} + SWAP + \kket{I^{\otimes n}} \bbra{I^{\otimes n}}}{2^n + 2}.}
    
    We want to express this in the $\kket{\widetilde{P}}$ basis. Since $I^{\otimes n} \otimes I^{\otimes n} = I^{\otimes 2n}$ and $\kket{I^{\otimes n}} \bbra{I^{\otimes n}} = 2^n \kket{\widetilde{I}^{\otimes n}} \bbra{\widetilde{I}^{\otimes n}}$, we only need to know how to represent $SWAP$ in this basis. So, let $P, Q \in \mathbb{P}_n$ be arbitrary. Expressing them in the computation basis, we notice that: 
    \[P = \sum_{i, j} P_{ij} \ket{i}\bra{j} \implies \kket{\widetilde{P}} = \frac{1}{\sqrt{2^n}} \sum_{i, j} P_{ij} \ket{ij}.\]
    
    Therefore: 
    \begin{equation}\begin{aligned}\bbra{\widetilde{P}} SWAP \kket{\widetilde{Q}} &= \frac{1}{2^n} \sum_{i, j, k, \ell} P_{ij}^* Q_{k\ell} \bra{ij} SWAP \ket{k\ell} \\ &= \frac{1}{2^n} \sum_{i, j, k, \ell} P_{ij}^* Q_{k\ell} \braket{ij}{\ell k} \\ &= \frac{1}{2^n} \sum_{i, j} P_{ij}^* Q_{ji} \\ &= \frac{1}{2^n} \sum_{i} \left(P^* Q\right)_{ii} \\ &= \frac{1}{2^n} \Tr\left(P^* Q\right).\end{aligned}\end{equation}
    
    Notice that there are two types of Paulis: the real ones, and the imaginary ones, $\mathbb{P}_n \subseteq \left\{0, 1, -1\right\}^{2^n \times 2^n} \cup \left\{0, i, -i\right\}^{2^n \times 2^n}$. The first ones are Paulis that have an even number of $Y$ in their tensor product expansion (in which case, $P^* = P$), and the second ones are Paulis that have an odd number of $Y$ in their tensor product expansion (in which case, $P^* = -P$). So, for all $P \in \mathbb{P}_n$, let $s_P \in \left\{-1, 1\right\}$ be so that $P^* = s_P P$. This yields: 
    \[\bbra{\widetilde{P}} SWAP \kket{\widetilde{Q}} = \frac{1}{2^n} s_P \Tr\left(PQ\right) = s_P \delta_{PQ}.\]
    
    Combining everything, this gives us:
    \[\bbra{\widetilde{P}} M \kket{\widetilde{Q}} = \frac{\bbrakket{\widetilde{P}}{\widetilde{Q}} + s_P \delta_{PQ} + 2^n \delta_{P I} \delta_{Q I}}{2^n + 2} = \delta_{PQ} \frac{\left(1 + s_P\right) + 2^n\delta_{PI}}{2^n + 2}.\]
    
    Finally, we notice that, if $P$ is real, then $s_P = 1$ and so $1 + s_P = 2$; but, if $P$ is purely imaginary, $s_P = -1$, and so $1 + s_P = 0$. In other words, $1 + s_P = 2\cdot 1_{\mathbb{R}}\left(P\right)$. This gives exactly our result: 
    \[\bbra{\widetilde{P}} M \kket{\widetilde{Q}} = \delta_{PQ} \frac{\delta_{PI} 2^n + 2\cdot 1_{\mathbb{R}}\left(P\right)}{2^n + 2} = \delta_{PQ} \frac{\delta_{PI} 2^{n-1} + 1_{\mathbb{R}}\left(P\right)}{2^{n-1} + 1}.\]
\end{proof}

\begin{lemma}\label{lma:s-orthogonal-haar-distribution}
    Let $D = \Haar_n\left(O\right)^{\otimes k}$, $\kket{\mu}$ be its first moment, $\mathcal{S}$ be its second moment and $\eta$ be its expressivity measure.
    
    Then: 
    \[\kket{\mu} = 0, \mathspace \lambda_{min}\left(\mathcal{S}\right) = 0, \mathspace \lambda_{max}\left(\mathcal{S}\right) = \frac{1}{2^{n-1} + 1}, \mathspace \eta = \frac{2^{n- 1}}{2^{n-1} + 1}.\]
    
    Moreover, an eigenbasis of $\mathcal{S}$ is 
    \[\left\{\kket{\widetilde{P}} \suchthat P \in \mathbb{P}_{kn} \setminus \left\{I^{\otimes nk}\right\}\right\}.\]
    
    Finally, splitting $P = \bigotimes_{i=1}^{k} Q_i$ for $Q_i \in \mathbb{P}_n$, each $\kket{\widetilde{P}}$ is associated to the eigenvalue 
    \[\left(\frac{1}{2^{n - 1} + 1}\right)^{w_n\left(P\right)} \prod_{i=1}^{k} 1_{\mathbb{R}}\left(Q_i\right),\]
    where $w_n\left(P\right) \in \left\{1, \ldots, k\right\}$ is the $n$-block weight of $P$ (definition~\ref{def:pauli-block-weight}), and $1_{\mathbb{R}}$ is the real matrix indicator function (definition~\ref{def:real-matrix-indicator-function}).
\end{lemma}
\begin{intuition*}
    Note that the block decomposition of $P$, $P = \bigotimes_{i=1}^{k} Q_i$, already appeared in the case for the Haar measure on the unitary group (lemma~\ref{lma:s-haar-distribution}), in the form of the $n$-block weight function $w_n\left(P\right)$, which counts the number of non-identity blocks. Now, we moreover require that all those blocks are real matrices, with $\prod_{i=1}^{k} 1_{\mathbb{R}}\left(Q_i\right)$.
\end{intuition*}
\begin{proof}[Proof]
    Just like the proof of the similar lemma for the Haar measure on the unitary group (lemma~\ref{lma:s-haar-distribution}), this result directly comes from corollary~\ref{crl:s-for-n-product-pure} and lemma~\ref{lma:m-orthogonal-haar-distribution}. The only difference in the proof is that when we split $P = Q_1 \otimes Q_2$, if $1_{\mathbb{R}}\left(Q_1\right) = 0$ or $1_{\mathbb{R}}\left(Q_2\right) = 0$, then $\kket{\widetilde{Q}_1 \otimes \widetilde{Q}_2}$ is associated to the eigenvalue $0$.  
\end{proof}

The distinction between global and product real-state ensembles depends on the pattern of $Y$ factors in a Pauli string. We separate strings with an odd number of $Y$ factors, a nonzero even number, and none. The largest and smallest risk ratios among these cases give the tight comparison bounds.

\begin{theorem}\label{thm:orthogonal-haar-tight-comparison}
    Let $D_1 = \Haar_n\left(O\right)$ and $D_2 = \Haar_1^{\otimes n}\left(O\right)$. 
    
    If $n \geq 2$, then, for all unitaries $U, V$:
    \[\frac{2^{n-1}}{2^{n-1}+ 1} R_2\left(U, V\right) \leq R_1\left(U, V\right) \leq 2\cdot \frac{2^{n-1}}{2^{n-1} + 1} R_2\left(U, V\right).\]
    
    Let $W = V^{\dagger} U$. If $W \in \mathbb{P}_n$ and it contains a non-zero even number of $Y$ in its tensor product decomposition, then the first inequality is reached with equality. If $W \in \mathbb{P}_n$, none of the Paulis in its tensor product decomposition is $Y$, and $w_1\left(W\right) = 1$, then the second inequality is reached with equality.
\end{theorem}
\begin{remark*}
    Note that the hypothesis $n \geq 2$ is not very restrictive: when $n = 1$, we trivially have $D_1 = D_2$, and thus $R_1\left(U, V\right) = R_2\left(U, V\right)$ for all unitaries $U, V$.
\end{remark*}
\begin{proof}[Proof]
    This proof is very similar to the proof of the similar theorem for the unitary group (theorem~\ref{thm:haar-tight-comparison}). 
    
    Note that, thanks to lemma~\ref{lma:s-orthogonal-haar-distribution}, we have the hypotheses to use the second inequality of theorem~\ref{thm:inequalities}. Moreover, given any $P \in \mathbb{P}_n \setminus \left\{I^{\otimes n}\right\}$, we know that $\kket{\widetilde{P}}$ is an eigenvector of both $\mathcal{S}^{\left(1\right)}$ and $\mathcal{S}^{\left(2\right)}$. There are three possibilities.
    \begin{enumerate}[left=0pt]
    \item If $P$ contains an odd number of $Y$ in its tensor product decomposition (i.e.\ $1_{\mathbb{R}}\left(P\right) = 0$), then $\kket{\widetilde{P}}$ has eigenvalue $0$ for both $\mathcal{S}^{\left(1\right)}$ and $\mathcal{S}^{\left(2\right)}$.
    \item If $P$ contains a non-zero even number of $Y$ in its tensor product decomposition, then $\kket{\widetilde{P}}$ has eigenvalue $\frac{1}{2^{n-1} + 1}$ for $\mathcal{S}^{\left(1\right)}$ and $0$ for $\mathcal{S}^{\left(2\right)}$. Note that this is where the hypothesis that $n \geq 2$ comes in, since this case is impossible when $n = 1$.
    \item If $P$ contains no $Y$ in its tensor product decomposition, then $\kket{\widetilde{P}}$ has eigenvalue $\frac{1}{2^{n-1} + 1}$ for $\mathcal{S}^{\left(1\right)}$ and $\frac{1}{2^j}$ for $\mathcal{S}^{\left(2\right)}$, where $j = w_1\left(P\right) \in \left\{1, \ldots, n\right\}$.
    \end{enumerate}
    
    We now apply twice the second inequality of theorem~\ref{thm:inequalities}, one for each inequality we try to prove.
    \begin{itemize}[left=0pt]
    \item Using usual calculus tools, we can find that: 
    \begin{equation}\begin{aligned}\alpha_{12} &= \max_k \left\{\frac{1 - \lambda_k\left(\mathcal{S}^{\left(1\right)}\right)}{1 - \lambda_k\left(\mathcal{S}^{\left(2\right)}\right)}\right\} \\ &= \max \left\{\frac{1 - 0}{1 - 0}, \frac{1 - \frac{1}{2^{n - 1} + 1}}{1 - 0}, \frac{1 - \frac{1}{2^{n - 1} + 1}}{1 - \frac{1}{2^j}} \suchthat j \in \left\{1, \ldots, n\right\}\right\} \\ &= \max \left\{1, \frac{2^{n-1}}{2^{n-1}+1}, \frac{2^{n - 1}}{2^{n-1} + 1} \frac{2^j}{2^j - 1} \suchthat j \in \left\{1, \ldots, n\right\}\right\} \\ &= \frac{2^{n-1}}{2^{n-1} + 1}\cdot 2.\end{aligned}\end{equation}

    This yields that, by theorem~\ref{thm:inequalities}: 
    \[R_1\left(U, V\right) \leq 2 \cdot \frac{2^{n-1}}{2^{n-1} + 1} R_2\left(U, V\right).\]
    
    Using the second point of theorem~\ref{thm:inequalities}, we wish for $\ket{w\left(U, V\right)} = \kket{\widetilde{P}}$, where $\kket{\widetilde{P}}$ is associated to the eigenvalue $\frac{1}{2^{n-1} + 1}$ in $\mathcal{S}^{\left(1\right)}$ and associated to the eigenvalue $\frac{1}{2}$ in $\mathcal{S}^{\left(2\right)}$ (since it is the eigenvector that maximises $\left(1 - \lambda_{k}\left(\mathcal{S}^{\left(1\right)}\right)\right)/\left(1 - \lambda_{k}\left(\mathcal{S}^{\left(2\right)}\right)\right)$). In other words, we wish for $\ket{w\left(U, V\right)} = \kket{\widetilde{P}}$, where $w_1\left(P\right) = j = 1$ is minimal and $P$ does not contain any $Y$ in its tensor product expansion. We are lucky, there exist solutions, for instance by letting $V = I$ and $U = P$.
    
    \item It is again possible to find that:
    \begin{equation}\begin{aligned}\alpha_{21} &= \max_k\left\{\frac{1 - \lambda_k\left(\mathcal{S}^{\left(2\right)}\right)}{1 - \lambda_k\left(\mathcal{S}^{\left(1\right)}\right)}\right\} \\ &= \max \left\{\frac{1 - 0}{1- 0}, \frac{1 - 0}{1 - \frac{1}{2^{n-1} + 1}}, \frac{1 - \frac{1}{2^j}}{1 - \frac{1}{2^{n-1} + 1}} \suchthat j \in \left\{1, \ldots, n\right\}\right\} \\ &= \max\left\{1, \frac{2^{n-1} + 1}{2^{n-1}}, \frac{2^j - 1}{2^j} \cdot \frac{2^{n-1} + 1}{2^{n-1}} \suchthat j \in \left\{1, \ldots, n\right\}\right\} \\ &= \frac{2^{n-1} + 1}{2^{n-1}}.\end{aligned}\end{equation}
    
    By theorem~\ref{thm:inequalities}, we thus get:  
    \autoeq{R_2\left(U, V\right) \leq \frac{2^{n-1} + 1}{2^{n-1}} R_1\left(U, V\right) \iff R_1\left(U, V\right) \geq \frac{2^{n-1}}{2^{n-1} + 1} R_2\left(U, V\right).}
    
    By the second point of theorem~\ref{thm:inequalities}, we again wish for $\ket{w\left(U, V\right)} = \kket{\widetilde{P}}$, where $\kket{\widetilde{P}}$ is an eigenvector of eigenvalue $\frac{1}{2^{n-1} + 1}$ for $\mathcal{S}^{\left(1\right)}$ but eigenvalue $0$ for $\mathcal{S}^{\left(2\right)}$. This is achieved by any $P$ that contains a non-zero even number of $Y$ in its tensor product decomposition. There moreover again exist solutions, for instance by letting $V = I$ and $U = P$.
    \end{itemize}
\end{proof}

\begin{corollary}\label{crl:orthogonal-haar-zero-mean-generalisation}
    For all mixed-state distributions $D$ over $nk$ qubits of mean $\kket{\mu}$, and for all unitaries $U, V$: 
    \[R_D\left(U, V\right) \leq \frac{2^{n-1} + 1}{2^{n-1}} R_{\Haar_n^{\otimes k}\left(O\right)}\left(U, V\right) + \left\|\kket{\mu}\right\|.\]
\end{corollary}
\begin{remark*}
    Comparing this with the unitary case (corollary~\ref{crl:haar-zero-mean-generalisation}), we notice we get very similar guarantees when training using the distribution $\Haar_n^{\otimes k}\left(U\right)$ or $\Haar_n^{\otimes k}\left(O\right)$.
\end{remark*}
\begin{remark*}
    Leaving $k = 1$, this is precisely the result displayed in corollary~\ref{corol-main-text:random-real} in the main text.
\end{remark*}
\begin{proof}[Proof]
    By lemma~\ref{lma:s-orthogonal-haar-distribution}, $\Haar_n^{\otimes k}\left(O\right)$ is zero mean, and has an expressivity measure $\eta$ given by: 
    \[\mathspace \eta = \frac{2^{n-1}}{2^{n-1} + 1}.\]
    
    But then, by corollary~\ref{crl:zero-mean-generalisation}, given any distribution $D$ over the same number of qubits as $\Haar_n^{\otimes k}\left(O\right)$, of first moment $\kket{\mu}$: 
    \autoeq{R_D\left(U, V\right) \leq \frac{1}{\eta} R_{\Haar_n^{\otimes k}\left(O\right)}\left(U, V\right) + \left\|\kket{\mu}\right\| = \frac{2^{n-1} + 1}{2^{n-1}} R_{\Haar_n^{\otimes k}\left(O\right)}\left(U, V\right) + \left\|\kket{\mu}\right\|.}
\end{proof}

The comparison factor depends on the Haar block size, but not on the number of blocks. In particular, real training states still give a purely multiplicative bound for every zero-mean pure-state target covered by the corollary. A nonzero target mean contributes an additive term.

\suppsection{Analysis of the bias and the expressivity measure}

Bias and expressivity enter the generalization bounds in different ways, but they are not independent. Here we derive constraints relating them and examine how each behaves under tensor products. We then connect these properties to classicality and to the number of states needed in the ensemble.

\begin{remark}
    In this section, we will only work with pure state distributions for simplicity.
\end{remark}

\suppsubsection{Bound linking $\left\|\mu\right\|$ and $\eta$}

The purity of the average state connects the bias to the second moment. We use it to bound the largest second-moment eigenvalue in one direction, and evaluate the second moment along the mean direction to obtain the other. These two estimates give the bounds in Theorem~\ref{thm:mu-eta-bound}.

\begin{lemma}\label{lma:mu-purity}
    Let $D$ be a distribution over $n$-qubit pure states, of mean $\kket{\mu}$. Also, let $\gamma_{\bar{\rho}} = \Tr\left(\exval\left(\rho\right)^2\right)$ be the purity of the average state $\bar{\rho} = \exval\left(\rho\right)$. Then: 
    \[\left\|\mu\right\|^2 = 2^n \gamma_{\bar{\rho}} - 1.\]
\end{lemma}
\begin{remark*}
    Equivalently: 
    \[\gamma_{\bar{\rho}} = \frac{\left\|\mu\right\|^2 + 1}{2^n}.\]
\end{remark*}
\begin{remark*}[Sanity check]
    We know that $\gamma_{\bar{\rho}} \in \left[\frac{1}{2^n}, 1\right]$, by properties of purity. This implies that $\left\|\mu\right\|^2 \in \left[0, 2^n - 1\right]$, which is exactly what we expect.
\end{remark*}
\begin{remark*}
    The purity of the average state $\gamma_{\bar{\rho}} = \Tr\left(\exval\left(\rho\right)^2\right)$ must not be mistaken with the average purity of a mixed state distribution $\bar{\gamma} = \exval\left(\Tr\left(\rho^2\right)\right)$ (definition~\ref{def:average-purity}). Since we work with pure-state distributions here, we have $\bar{\gamma} = 1$.
\end{remark*}
\begin{proof}[Proof]
    We know that $\kket{\mu} = \sqrt{2^n}\kket{\pi\left(\bar{\rho}\right)}$ for pure states. This is then a direct proof, using the fact $\bar{\rho}$ is Hermitian and $\Tr\left(\bar{\rho}\right) = 1$: 
    \begin{equation}\begin{aligned}\left\|\mu\right\|^2 &= \braket{\mu}{\mu} \\ &= 2^n \bbrakket{\pi\left(\bar{\rho}\right)}{\pi\left(\bar{\rho}\right)} \\ &= 2^n \left(\bbra{\bar{\rho}} - \frac{1}{2^n} \bbra{I^{\otimes n}}\right)\left(\kket{\bar{\rho}} - \frac{1}{2^n} \kket{I^{\otimes n}}\right) \\ &= 2^n \bbrakket{\bar{\rho}}{\bar{\rho}} - 2^n\cdot 2\cdot \frac{\bbrakket{I^{\otimes n}}{\bar{\rho}}}{2^n} + 2^n \frac{1}{4^n} \bbrakket{I^{ \otimes n}}{I^{\otimes n}} \\ &= 2^n \Tr\left(\bar{\rho}^2\right) - 1 \\ &= 2^n \gamma_{\bar{\rho}} - 1.\end{aligned}\end{equation}
\end{proof}

\begin{lemma}\label{lma:s-rayleigh-alternate-formulation}
    Let $D$ be a distribution over $n$-qubit pure states, of second moment matrix $\mathcal{S}$. Moreover, let $A \in \mathbb{C}^{2^n \times 2^n}$ be an arbitrary operator. Then: 
    \[\bbra{A} \mathcal{S} \kket{A} = 2^n \exval\left(\left|\Tr\left(A \rho\right) - \frac{1}{2^n} \Tr\left(A\right)\right|^2\right).\]
\end{lemma}
\begin{remark*}
    This lemma is useful to find an upper bound on $\lambda_{max}\left(\mathcal{S}\right)$. Indeed, it suffices to prove that $\bbra{A} \mathcal{S} \kket{A} \leq t$ for all operators $A$ such that $\bbrakket{A}{A} = 1$ to show that $\lambda_{max}\left(\mathcal{S}\right) \leq t$.
\end{remark*}
\begin{proof}[Proof]
    We know that $\mathcal{S} = 2^n \exval\left(\kket{\pi\left(\rho\right)}\bbra{\pi\left(\rho\right)}\right)$. This is thus direct: 
    \autoeq{\bbra{A} \mathcal{S} \kket{A} = 2^n \exval\left(\left|\bbrakket{\pi\left(\rho\right)}{A}\right|^2\right) = 2^n \exval\left(\left|\bbrakket{\rho}{A} - \frac{1}{2^n} \bbrakket{I^{\otimes n}}{A}\right|^2\right) = 2^n \exval\left(\left|\Tr\left(A \rho\right) - \frac{1}{2^n} \Tr\left(A\right)\right|^2\right).}
\end{proof}
\begin{lemma}\label{lma:min-sum-of-squares}
    Let $a \in \mathbb{R}$ be some target value. We moreover define the following function:
    \[\begin{split}
    f: \mathbb{R}^{m} &\longmapsto \mathbb{R} \\
    \left(x_1, \ldots, x_m\right) &\longmapsto \sum_{i=1}^{m} x_i^2
    \end{split}\]
    
    Then, the minimum of $f$ subject to the constraint $\sum_{i=1}^{m} x_i = a$ is $\frac{a^2}{m}$, achieved for $x_1 = \ldots = x_m = \frac{a}{m}$.
\end{lemma}
\begin{proof}[Proof]
    We consider the following function, representing our constraint: 
    \[g\left(x_1, \ldots, x_m\right) = \sum_{i=1}^{m} x_i - a.\]
    
    We prove our affirmation using Lagrange's multipliers. We know that, for any minimum of $f$ subject to the constraint $g\left(x\right) = 0$, there must exist some $\lambda \in \mathbb{R}$ such that: 
    \[\nabla f\left(x\right) = \lambda \nabla g\left(x\right) \iff \forall i,\ 2 x_i = \lambda.\]
    
    Injecting this to the constraint $g\left(x\right) = 0$ yields: 
    \[g\left(x\right) = 0 \iff \sum_{i=1}^{m} x_i = a \iff m\cdot \frac{\lambda}{2} = a \iff \lambda = \frac{2a}{m}.\]
    
    Therefore, $x_i = \frac{\lambda}{2} = \frac{a}{m}$ for all $i$. At this vector, $f$ evaluates to: 
    \[
f(x)=\sum_{i=1}^{m}x_i^2
=m\cdot\frac{a^2}{m^2}
=\frac{a^2}{m}.
\]
    
    Since $f$ is convex and the constraint is affine,
this stationary point is a global minimum, proving the claim.
\end{proof}

\begin{lemma}\label{lma:largest-eigenvalue-mixed-state-bound}
    Let $\rho$ be a mixed state, of purity $\gamma_{\bar{\rho}} = \Tr\left(\rho^2\right)$. Then: 
    \[\lambda_{max}\left(\rho\right) \leq \frac{1 + \sqrt{\left(2^n - 1\right)\left(2^n \gamma_{\bar{\rho}} - 1\right)}}{2^n}.\]
\end{lemma}
\begin{proof}[Proof]

For a fixed largest eigenvalue, the remaining eigenvalues contribute least to the purity when they are equal. This gives a quadratic constraint on the largest eigenvalue, whose upper root yields the bound.

    Let $\lambda_1 \leq \ldots \leq \lambda_N$ be the eigenvalues of $\rho$, with $N = 2^n$. Since $\Tr\left(\rho\right) = 1$, we know that: 
    \[\sum_{i=1}^{N} \lambda_i = 1 \implies \sum_{i=1}^{N-1} \lambda_i = 1 - \lambda_N.\]
    
    By lemma~\ref{lma:min-sum-of-squares}, we know that, subject to the constraint $\sum_{i=1}^{N-1} \lambda_i = 1 - \lambda_N$: 
    \[\sum_{i=1}^{N-1} \lambda_i^2 \geq \frac{\left(1 - \lambda_N\right)^2}{N - 1}.\]
    
    This therefore tells us that: 
    \begin{equation}\begin{aligned}&\gamma_{\bar{\rho}} = \sum_{i=1}^{N} \lambda_i^2 = \lambda_N^2 + \sum_{i=1}^{N-1} \lambda_i^2 \geq \lambda_N^2 + \frac{\left(1 - \lambda_N\right)^2}{N - 1} \\ &\iff \gamma_{\bar{\rho}} \geq \lambda_N^2 + \frac{1 - 2\lambda_N + \lambda_N^2}{N-1} = \frac{N \lambda_N^2 - 2 \lambda_N + 1}{N-1} \\ &\iff N\lambda_N^2 - 2 \lambda_N + \left(1 - \left(N-1\right)\gamma_{\bar{\rho}}\right) \leq 0.\end{aligned}\end{equation}
    
    This is a second degree polynomial in the variable $\lambda_N$, of positive leading coefficient, $N > 0$. For it to be negative, $\lambda_N$ must be between the two roots, $x_- \leq \lambda_N \leq x_+$. We thus compute its roots: 
    \autoeq{x_{\pm} = \frac{2 \pm \sqrt{4 - 4N\left(1 - \left(N-1\right)\gamma_{\bar{\rho}}\right)}}{2N} = \frac{1 \pm \sqrt{1 - N + N\left(N-1\right)\gamma_{\bar{\rho}}}}{N} = \frac{1 \pm \sqrt{\left(N-1\right)\left(N\gamma_{\bar{\rho}} - 1\right)}}{N}.}
    
    Notice that we necessarily have $\lambda_N \geq \frac{1}{N}$ since, otherwise, we would have $\Tr\left(\rho\right) < 1$. This means that $\lambda_N \geq x_-$ is always true. We however get a constraint on the other root: 
    \[\lambda_N \leq x_+ = \frac{1 + \sqrt{\left(N-1\right)\left(N\gamma_{\bar{\rho}} - 1\right)}}{N}.\]
    
    Since we defined $N = 2^n$, this gives our result.
\end{proof}

\begin{lemma}\label{lma:expected-value-purity-bound}
    Let $A \in \mathbb{C}^{2^n \times 2^n}$ be an arbitrary operator, and $D$ be a distribution over $n$-qubit mixed states $\rho$. Moreover, let $\gamma_{\bar{\rho}} = \Tr\left(\exval\left(\rho\right)^2\right)$ be the purity of the average state. Then: 
    \[\exval\left(\left|\Tr\left(A \rho\right)\right|^2\right) \leq \Tr\left(A^{\dagger} A\right) \frac{1 + \sqrt{\left(2^n - 1\right)\left(2^n \gamma_{\bar{\rho}}  -1\right)}}{2^n}.\]
\end{lemma}
\begin{remark*}[Implication]
    Suppose that $A = O$ is an observable. If $\bar{\rho}$ is pure, and hence $\gamma_{\bar{\rho}} = 1$, this yields:
    \[\exval\left(\Tr\left(O \rho\right)^2\right) \leq \Tr\left(O^2\right).\]
    
    If $\bar{\rho}$ is the maximally mixed state, i.e. if $\gamma_{\bar{\rho}} = \frac{1}{2^n}$, then: 
    \[\exval\left(\Tr\left(O \rho\right)^2\right) \leq \frac{1}{2^n} \Tr\left(O^2\right).\]
\end{remark*}
\begin{proof}[Proof]

The available expectation bound applies to Hermitian operators, whereas the lemma allows an arbitrary operator. We extend it by separating the Hermitian and anti-Hermitian parts, then use the purity of the average state to bound its operator norm.

    We use the second-moment bound established in the proof of
Lemma~B1 of
Ref.~\cite{schuster2024polynomial}.
For any Hermitian operator $O$:
    \begin{equation}\label{eq:untight-inequality} 
    \exval\left(\Tr\left(O \rho\right)^2\right) \leq \left\|O\right\|_2^2 \left\|\bar{\rho}\right\|_{\infty},
    \end{equation}
    where $\left\|.\right\|_p$ is the Schatten $p$-norm. In our case, $\left\|O\right\|_2^2 = \Tr\left(O^2\right)$ since $O$ is Hermitian and $\left\|\bar{\rho}\right\|_{\infty} = \lambda_{max}\left(\bar{\rho}\right)$ since $\bar{\rho}$ is Hermitian positive semi-definite.
    
    We split $A = H + i K$ where $H$ and $K$ are Hermitian operators, defined by: 
    \[H = \frac{A + A^{\dagger}}{2}, \mathspace K = \frac{A - A^{\dagger}}{2i}.\]
    
    We notice that: 
    \begin{equation}\begin{aligned}\left|\Tr\left(A \rho\right)\right|^2 &= \Tr\left(A^{\dagger} \rho\right) \Tr\left(A \rho\right) \\ &= \Tr\left(\left(H - iK\right) \rho\right) \Tr\left(\left(H + iK\right) \rho\right) \\ &= \Tr\left(H \rho\right)^2 - i \Tr\left(K \rho\right)\Tr\left(H\rho\right) + i \Tr\left(H\rho\right)\Tr\left(K\rho\right) + \Tr\left(K\rho\right)^2 \\ &= \Tr\left(H\rho\right)^2 + \Tr\left(K \rho\right)^2.\end{aligned}\end{equation}
    
    Applying Eq.~\eqref{eq:untight-inequality} to $H$ and $K$, we obtain: 
    \autoeq{\exval\left(\left|\Tr\left(A \rho\right)\right|^2\right) = \exval\left(\Tr\left(H \rho\right)^2\right) + \exval\left(\Tr\left(K \rho\right)^2\right) \leq \Tr\left(H^2\right) \left\|\bar{\rho}\right\|_{\infty} + \Tr\left(K^2\right) \left\|\bar{\rho}\right\|_{\infty}.}
    
    However, we also notice that:
    \autoeq{\Tr\left(A^{\dagger} A\right) = \Tr\left(\left(H - iK\right)\left(H + iK\right)\right) = \Tr\left(H^2\right) - i\Tr\left(KH\right) + i\Tr\left(HK\right) + \Tr\left(K^2\right) = \Tr\left(H^2\right) + \Tr\left(K^2\right).}
    
    This allows us to simplify our result to:
    \[\exval\left(\left|\Tr\left(A \rho\right)\right|^2\right) \leq \Tr\left(A^{\dagger} A\right) \left\|\bar{\rho}\right\|_{\infty}.\]
    
    Lemma~\ref{lma:largest-eigenvalue-mixed-state-bound} finally allows us to bound $\left\|\bar{\rho}\right\|_{\infty} = \lambda_{max}\left(\bar{\rho}\right)$: 
    \[\lambda_{max}\left(\bar{\rho}\right) \leq \frac{1 + \sqrt{\left(2^n - 1\right)\left(2^n \gamma_{\bar{\rho}} - 1\right)}}{2^n}.\]
    
    This directly gives our result.
\end{proof}

\begin{theorem}\label{thm:mu-eta-bound}
    Let $D$ be a distribution over $n$-qubit pure states, of mean $\kket{\mu}$ and expressivity measure $\eta$. Then: 
    \[- \sqrt{2^n - 1} \left\|\mu\right\| \leq \eta \leq 1 - \left\|\mu\right\|^2.\]
\end{theorem}
\begin{remark*}
    Observe again that this result was only proven for pure-state distributions.
\end{remark*}
\begin{remark*}[Implication]
    This yields multiple results.
    \begin{itemize}
    \item If $\left\|\mu\right\| = 0$, then $\eta \geq 0$.
    \item If $\left\|\mu\right\| \geq 1$, then $\eta \leq 0$.
    \item If $\left\|\mu\right\| = \sqrt{2^n - 1}$ is maximal, then $\eta = 2 - 2^n$ is minimal.
    \end{itemize}
\end{remark*}
\begin{proof}[Proof 1]
    We start by proving that $- \sqrt{2^n - 1} \left\|\mu\right\| \leq \eta$.
    
    Let $\mathcal{S}$ be the second moment of $D$, and let $\ket{x}$ be the normalised eigenvector associated to the largest eigenvalue of $\mathcal{S}$. We know we can find some operator $A$ such that $\kket{A} = \ket{x}$. Since $\ket{x}$ is normalised, we know $\Tr\left(A^{\dagger} A\right) = \bbrakket{A}{A} = \braket{x}{x} = 1$. By lemma~\ref{lma:s-rayleigh-alternate-formulation}, we know that: 
    \[\lambda_{max}\left(\mathcal{S}\right) = \bra{x} \mathcal{S} \ket{x} = \bbra{A} \mathcal{S} \kket{A} = 2^n \exval\left(\left|\Tr\left(A \rho\right) - \frac{1}{2^n} \Tr\left(A\right)\right|^2\right).\]
    
    Let $B = \pi\left(A\right) = A - \Tr\left(A\right) \frac{I^{\otimes n}}{2^n}$ be the traceless part of $A$. We notice that $A = B + \Tr\left(A\right) \frac{I^{\otimes n}}{2^n}$, and hence: 
    \autoeq{\lambda_{max}\left(\mathcal{S}\right) = 2^n \exval\left(\left|\Tr\left(B \rho\right) + \Tr\left(A\right) \frac{\Tr\left(\rho\right)}{2^n} - \frac{1}{2^n} \Tr\left(A\right)\right|^2\right) = 2^n \exval\left(\left|\Tr\left(B \rho\right)\right|^2\right).}
    
    Moreover, the constraint that $\Tr\left(A^{\dagger} A\right) = 1$ gives us that:
    \begin{equation}\begin{aligned}\Tr\left(B^{\dagger} B\right) &= \Tr\left[\left(A^{\dagger} - \Tr\left(A\right)^* \frac{I^{\otimes n}}{2^n}\right)\left(A - \Tr\left(A\right) \frac{I^{\otimes n}}{2^n}\right)\right] \\ &= \Tr\left(A^{\dagger} A\right) - \frac{\Tr\left(A\right)^*}{2^n} \Tr\left(A\right) \fakeequal - \Tr\left(A^{\dagger}\right)\frac{\Tr\left(A\right)}{2^n}  + \Tr\left(A\right)^* \Tr\left(A\right) \frac{\Tr\left(I^{\otimes n}\right)}{4^n} \\ &= 1 - \frac{1}{2^n} \left|\Tr\left(A\right)\right|^2 \\ &\leq 1.\end{aligned}\end{equation}
    
    But then, lemma~\ref{lma:expected-value-purity-bound} tells us: 
    \autoeq{\lambda_{max}\left(\mathcal{S}\right) = 2^n \exval\left(\left|\Tr\left(B \rho\right)\right|^2\right) \leq 2^n \Tr\left(B^{\dagger} B\right) \frac{1 + \sqrt{\left(2^n - 1\right)\left(2^n \gamma_{\bar{\rho}} - 1\right)}}{2^n} \leq 1 + \sqrt{\left(2^n - 1\right)\left(2^n \gamma_{\bar{\rho}} - 1\right)}.}
    
    We recognise thanks to lemma~\ref{lma:mu-purity} that $\left\|\mu\right\|^2 = 2^{n} \gamma_{\bar{\rho}} - 1$, and hence: 
    \[\lambda_{max}\left(\mathcal{S}\right) \leq 1 + \sqrt{\left(2^n - 1\right) \left\|\mu\right\|^2} = 1 + \left\|\mu\right\|\sqrt{2^n - 1}.\]
    
    We can finally use the definition of $\eta$: 
    \[\eta = 1 - \lambda_{max}\left(\mathcal{S}\right) \geq -\left\|\mu\right\| \sqrt{2^n - 1}.\]
\end{proof}
\begin{proof}[Proof 2]
    We now prove that $\eta \leq 1 - \left\|\mu\right\|^2$. If $\kket{\mu} = 0$, this is a trivial bound, we know $\eta \leq 1$ is always true. We thus suppose $\left\|\mu\right\| > 0$.
    
    Consider $\ket{x} = \frac{\kket{\mu}}{\left\|\mu\right\|}$, where $\kket{\mu} = \sqrt{2^n} \kket{\pi\left(\bar{\rho}\right)}$ as usual. Then: 
    \[\lambda_{max}\left(\mathcal{S}\right) \geq \bra{x} \mathcal{S} \ket{x} = 2^n \bra{x} \exval\left(\kket{\pi\left(\rho\right)} \bbra{\pi\left(\rho\right)}\right) \ket{x} = 2^n \exval\left(\left|\brakket{x}{\pi\left(\rho\right)}\right|^2\right).\]
    
    It is well known that the function $f: \mathbb{C} \mapsto \mathbb{R}_+$ defined by $f\left(z\right) = \left|z\right|^2$ is convex. Thus, by Jensen's inequality: 
    \[\lambda_{max}\left(\mathcal{S}\right) \geq 2^n \exval\left(\left|\brakket{x}{\pi\left(\rho\right)}\right|^2\right) \geq 2^n \left|\exval\left(\brakket{x}{\pi\left(\rho\right)}\right)\right|^2.\]
    
    However: 
    \autoeq{\exval\left(\brakket{x}{\pi\left(\rho\right)}\right) = {\frac{\bra{\mu}}{\left\|\mu\right\|} \exval\left(\kket{\pi\left(\rho\right)}\right) = \frac{\bra{\mu}}{\left\|\mu\right\|} \kket{\pi\left(\bar{\rho}\right)}} = {\frac{\bra{\mu}}{\left\|\mu\right\|}\frac{1}{\sqrt{2^n}} \kket{\mu} = \frac{\left\|\mu\right\|^2}{\sqrt{2^n} \left\|\mu\right\|} = \frac{\left\|\mu\right\|}{\sqrt{2^n}}}.}
    
    Injecting it back:
    \[\lambda_{max}\left(\mathcal{S}\right) \geq  2^n \left|\exval\left(\brakket{x}{\pi\left(\rho\right)}\right)\right|^2 = 2^n \frac{\left\|\mu\right\|^2}{2^n} = \left\|\mu\right\|^2.\]
    
    This finally gives our result: 
    \[\eta = 1 - \lambda_{max}\left(\mathcal{S}\right) \leq 1 - \left\|\mu\right\|^2.\]
\end{proof}
\begin{remark*}
    The lower bound is not tight: we know $\eta \geq 2 - 2^n$, but for $\left\|\mu\right\| = \sqrt{2^n - 1}$, the lower bound only tells us $\eta \geq 1 - 2^n$. It is interesting to consider a deterministic distribution, which outputs $\ket{\psi} = \ket{0}$ with probability 1 to find the issue. Almost all inequalities are tight for this distribution, except for Hölder's trace inequality $\exval\left(\Tr\left(O^2 \rho\right)\right) \leq \left\|O^2\right\|_1 \left\|\bar{\rho}\right\|_{\infty}$ used by reference~\cite{schuster2024polynomial} to prove equation~\ref{eq:untight-inequality}. 

\end{remark*}

\suppsubsection{Scaling of $\left\|\mu\right\|_2$}

We compute the bias of a product ensemble from the biases of its components. 

\begin{remark}
    We continue to only work with pure state distributions for simplicity.
\end{remark}

\begin{lemma}\label{lma:mu-2-tensor-product-scaling}
    Let $D_1, D_2$ be independent pure-state distributions over $n_1, n_2$ qubits, of first moments $\mu_1, \mu_2$, respectively. Moreover, let $D = D_1 \otimes D_2$, and let $\mu$ be its first moment.

    Then:
    \[\left\|\mu\right\|_2 = \sqrt{\left(1+\left\|\mu_1\right\|_2^2\right)\left(1+\left\|\mu_2\right\|_2^2\right)-1}.\]
\end{lemma}
\begin{intuition*}
    Reordering the terms, we notice that:
    \[1+\left\|\mu\right\|_2^2 = \left(1+\left\|\mu_1\right\|_2^2\right)\left(1+\left\|\mu_2\right\|_2^2\right).\]

    In other words, the value $1+\left\|\mu\right\|_2^2$ is multiplicative.
\end{intuition*}
\begin{proof}[Proof]
    Let $\bar{\rho}_i = \exval_{\rho_i \followsdistr D_i}\left[\rho_i\right]$, for $i \in \left\{1,2\right\}$, and let $\bar{\rho} = \exval_{\rho \followsdistr D}\left[\rho\right]$. Since the distributions are over pure states, we know $\bar{\rho}_i = \exval_{\rho_i \followsdistr D_i}\left[\rho_i^2\right]$ as well, and hence definition~\ref{def:mean} gives:
    \[\mu_i = \sqrt{2^{n_i}}\left(\bar{\rho}_i - \frac{I^{\otimes n_i}}{2^{n_i}}\right) \implies \sqrt{2^{n_i}}\bar{\rho}_i = \widetilde{I}^{\otimes n_i}+\mu_i.\]

    Moreover, independence tells us that $\bar{\rho} = \bar{\rho}_1 \otimes \bar{\rho}_2$. Therefore:
    \autoeq{\widetilde{I}^{\otimes \left(n_1+n_2\right)}+\mu = \sqrt{2^{n_1+n_2}}\bar{\rho} = \left(\sqrt{2^{n_1}} \bar{\rho}_1\right) \otimes \left(\sqrt{2^{n_2}} \bar{\rho}_2\right) = \left(\widetilde{I}^{\otimes n_1}+\mu_1\right)\otimes\left(\widetilde{I}^{\otimes n_2}+\mu_2\right).}

    Expanding this tensor product and cancelling out the identity term on both sides yields:
    \[\mu = \mu_1 \otimes \widetilde{I}^{\otimes n_2}+\widetilde{I}^{\otimes n_1}\otimes\mu_2+\mu_1\otimes\mu_2.\]

    By definition, $\mu_1$ and $\mu_2$ are traceless. Hence, the three terms on the right-hand side are pairwise orthogonal for the Hilbert-Schmidt inner product, i.e.\ for instance $\Tr\left(\left(\mu_1 \otimes \widetilde{I}^{\otimes n_2}\right)\left(\mu_1 \otimes \mu_2\right)\right) = \frac{1}{\sqrt{2^{n_2}}} \Tr\left(\mu_1^2\right) \Tr\left(\mu_2\right) = \frac{1}{\sqrt{2^{n_2}}} \Tr\left(\mu_1^2\right)\cdot 0 = 0$. Since moreover the normalised identities have norm $1$, this gives:
    \autoeq{\left\|\mu\right\|_2^2 = \Tr\left(\mu^2\right) = \left\|\mu_1\right\|_2^2 \underbrace{\left\|\widetilde{I}^{\otimes n_2}\right\|_2^2}_{=1} + \underbrace{\left\|\widetilde{I}^{\otimes n_1}\right\|_2^2}_{=1} \left\|\mu_2\right\|_2^2 + \left\|\mu_1\right\|_2^2\left\|\mu_2\right\|_2^2 = \left(1+\left\|\mu_1\right\|_2^2\right)\left(1+\left\|\mu_2\right\|_2^2\right)-1.}

    Taking the square root gives our result.
\end{proof}

\begin{lemma}\label{lma:mu-k-tensor-product-scaling}
    Let $C$ be a pure-state distribution over $n$ qubits, of first moment $\mu$. Moreover, let $k \geq 1$ be an integer, $D_k = C^{\otimes k}$, and $\mu_k$ be the first moment of $D_k$.

    Then:
    \[\left\|\mu_k\right\|_2 = \sqrt{\left(1+\left\|\mu\right\|_2^2\right)^k-1}.\]
\end{lemma}
\begin{intuition*}
    Using the intuition of the previous lemma, this is expected:
    \[1+\left\|\mu_k\right\|_2^2 = \left(1+\left\|\mu\right\|_2^2\right)^k.\]
\end{intuition*}
\begin{remark*}
    This is precisely the first result of corollary~\ref{corol-main-text:product-bias-expressivity} in the main text.
\end{remark*}
\begin{proof}[Proof]
    We apply lemma~\ref{lma:mu-2-tensor-product-scaling} successively to the $k$ independent factors. The quantity $1+\left\|\mu\right\|_2^2$ is multiplied by the same factor at every step, giving:
    \[1+\left\|\mu_k\right\|_2^2 = \prod_{\ell=1}^k\left(1+\left\|\mu\right\|_2^2\right) = \left(1+\left\|\mu\right\|_2^2\right)^k.\]

    Rearranging and taking the square root gives our result.
\end{proof}

\begin{theorem}\label{thm:mu-k-tensor-product-below-one}
    Let $C$ be a pure-state distribution of mean $\mu$, and let $\mu_k$ be the mean of $D_k = C^{\otimes k}$, for some integer $k \geq 1$.

    Then:
    \[\left\|\mu_k\right\|_2 < 1 \iff \left\|\mu\right\|_2 < \sqrt{2^{1/k}-1}.\]

    If moreover $\left\|\mu\right\|_2 > 0$, this is equivalent to:
    \[k < \frac{\ln\left(2\right)}{\ln\left(1+\left\|\mu\right\|_2^2\right)}.\]

    If $\mu = 0$, then $\mu_k = 0$ for every $k$.
\end{theorem}
\begin{remark*}
    Observe again that this result was only proven for pure-state distributions.
\end{remark*}
\begin{remark*}[Asymptotics]
    For nonzero $\mu$, notice that $\ln\left(2\right)/\ln\left(1+\left\|\mu\right\|_2^2\right) = \Theta\left(1/ \left\|\mu\right\|_2^2\right)$ whenever $\left\|\mu\right\|_2^2 = O\left(1\right)$, and we can always take $k = 1$ otherwise. Therefore, this theorem states that for $k = \Theta\left(1 + 1/\left\|\mu\right\|_2^2\right)$ sufficiently large, then $\left\|\mu_k\right\|_2^2 \geq 1$. Recall however that for most of our bounds not to be vacuous, we need $\left\|\mu_k\right\|_2 < 1$. Hence, even small non-zero bias for $C$ compounds to make bounds on $D_k$ vacuous for $k$ sufficiently large.
\end{remark*}
\begin{proof}[Proof]
    By lemma~\ref{lma:mu-k-tensor-product-scaling}:
    \autoeq{\left\|\mu_k\right\|_2 < 1 \iff \sqrt{\left(\left\|\mu\right\|_2^2 + 1\right)^k - 1} < 1 \iff \left(1+\left\|\mu\right\|_2^2\right)^k < 2 \iff \left\|\mu\right\|_2^2 < 2^{1/k}-1.}

    This gives the first equivalence. If $\left\|\mu\right\|_2 > 0$, taking logarithms also gives:
    \autoeq{\left(1+\left\|\mu\right\|_2^2\right)^k < 2 \iff k\ln\left(1+\left\|\mu\right\|_2^2\right) < \ln\left(2\right) \iff k < \frac{\ln\left(2\right)}{\ln\left(1+\left\|\mu\right\|_2^2\right)}.}

    Finally, the case $\mu = 0$ follows directly from lemma~\ref{lma:mu-k-tensor-product-scaling}.
\end{proof}

\suppsubsection{Bound on the support size}\label{sec:bound-support-size}

Zero mean is a strong requirement for a pure-state ensemble because its average state must be maximally mixed. A rank argument therefore gives a lower bound on the number of states in its support. We generalise this reasoning to see how bias and expressivity constrain the support size more generally.

\begin{theorem}\label{thm:bound-support-size-mu-eta}
    Let $D$ be a pure-state distribution over $n \geq 1$ qubits, of mean $\mu$ and expressivity measure $\eta$.

    Then:
    \[\left|\supp D\right| \geq \max\left\{\left\lceil\frac{2^n}{\left\|\mu\right\|_2^2+1}\right\rceil, \left\lceil\frac{2^n-1}{1-\eta}\right\rceil\right\}.\]
\end{theorem}
\begin{remark*}
    Observe again that this result was only proven for pure-state distributions.
\end{remark*}
\begin{proof}[Proof]
    The case where the support is infinite is immediate, so suppose that $\supp D = \left\{\rho_1,\ldots,\rho_m\right\}$ is finite. Let $p_i > 0$ be the probability of sampling $\rho_i$, and let $\bar{\rho} = \sum_{i=1}^m p_i\rho_i$ be the average state.

    We first prove the bound involving $\mu$. Since the states are pure, definition~\ref{def:mean} gives:
    \autoeq{\left\|\mu\right\|_2^2 = 2^n\Tr\left(\left(\bar{\rho}-\frac{I^{\otimes n}}{2^n}\right)^2\right) = 2^n\Tr\left(\bar{\rho}^2\right)-1.}

    Moreover, $\Tr\left(\rho_i\rho_j\right) \geq 0$ for every $i,j$, and $\Tr\left(\rho_i^2\right) = 1$. Hence:
    \[\Tr\left(\bar{\rho}^2\right) = \sum_{i,j=1}^m p_i p_j\Tr\left(\rho_i\rho_j\right) \geq \sum_{i=1}^m p_i^2 = \frac{1}{m} \sum_{i=1}^{m} p_i^2 \sum_{j=1}^{m} 1\geq \frac{1}{m} \left(\sum_{i=1}^{m} p_i\cdot 1\right)^2 = \frac{1}{m},\]
    where the last inequality is Cauchy-Schwarz, since $\sum_i p_i = 1$. Combining these expressions yields:
    \[\frac{1+\left\|\mu\right\|_2^2}{2^n} = \Tr\left(\bar{\rho}^2\right)\geq \frac{1}{m} \implies m \geq \frac{2^n}{1+\left\|\mu\right\|_2^2}.\]

    For the second bound, let $\mathcal{S}$ be the second moment of $D$. By lemma~\ref{lma:moment-matrix-equivalence-pure-states}, after applying the $I$-truncation:
    \[\mathcal{S} = 2^n\sum_{i=1}^m p_i\kket{\pi\left(\rho_i\right)}\bbra{\pi\left(\rho_i\right)}.\]
    This is a sum of $m$ matrices of rank at most $1$, so $\operatorname{rank}\left(\mathcal{S}\right) \leq m$. Since $\mathcal{S}$ is positive semi-definite, lemma~\ref{lma:s-properties} gives:
    \autoeq{2^n-1 = \Tr\left(\mathcal{S}\right) \leq \operatorname{rank}\left(\mathcal{S}\right)\lambda_{max}\left(\mathcal{S}\right) \leq m\left(1-\eta\right).}

    Using again positive semi-definiteness of $\mathcal{S}$ and the fact $\mathcal{S} \neq 0$ since $\Tr\left(\mathcal{S}\right) = 2^n - 1 > 0$, we also know $1-\eta = \lambda_{max}\left(\mathcal{S}\right) > 0$. Therefore:
    \[m \geq \frac{2^n-1}{1-\eta}.\]

    Finally, $m$ is an integer, so taking the ceilings and combining both bounds gives our result.
\end{proof}

\suppsection{Generalisation to channels}\label{sec:channels}

Here we extend the comparison from unitaries to channels. 
To do so, we first need to generalise the definition of risk to channels.

\begin{definition}[Risk for channels]\label{def:risk-for-channels}
Let $D$ be an arbitrary quantum state distribution, and $\mathcal{E}_A, \mathcal{E}_B$ be channels. We define the risk to be: 
    \[R_D\left(\mathcal{E}_A, \mathcal{E}_B\right) = \frac{1}{2}\exval_{\rho \followsdistr D}\left(\left\|\mathcal{E}_A\left(\rho\right) - \mathcal{E}_B\left(\rho\right)\right\|_2^2\right).\]
\end{definition}
\begin{remark*}
    This generalises the unitary case (definition~\ref{def:risk}) where $\mathcal{E}_A\left(\rho\right) = U \rho U^{\dagger}$ and $\mathcal{E}_B\left(\rho\right) = V \rho V^{\dagger}$.
\end{remark*}

\suppsubsection{Positive results for Pauli channels}

For Pauli channels, we can separate the role of the channel coefficients from that of the input distribution. We can then derive a bound on how the coefficients control the risk for any ensemble with the moment matrix giving distribution dependence independently (in a manner a bit similar to the unitary case). We then use these two descriptions to compare an arbitrary target risk with the Haar risk.

\begin{definition}[Pauli channel]
    Let $p \in \mathbb{C}^{4^n}$ indexed by Pauli operators be a probability distribution. A quantum channel $\mathcal{E}$ is said to be a Pauli channel if it can be written in the form: 
    \[\mathcal{E}\left(\rho\right) = \sum_{P \in \mathbb{P}_n} p_P P \rho P.\]
\end{definition}
\begin{remark*}
    In our context, we are interested in comparing two channels $\mathcal{E}_A, \mathcal{E}_B$. We will therefore write them as: 
    \[\mathcal{E}_A\left(\rho\right) = \sum_{P \in \mathbb{P}_n} a_P P \rho P, \mathspace \mathcal{E}_B\left(\rho\right) = \sum_{P \in \mathbb{P}_n} b_P P \rho P.\]

    Note that $a, b$ are probability distributions here.
\end{remark*}

\begin{lemma}[General upper-bound for Pauli channels]\label{lma:general-upper-bound-for-pauli-channels}
    Let $D$ be an arbitrary distribution, and let $\mathcal{E}_A, \mathcal{E}_B$ be Pauli channels. Then:
    \[R_D\left(\mathcal{E}_A, \mathcal{E}_B\right) \leq \frac{1}{2} \left\|a - b\right\|_1^2.\]
\end{lemma}

\begin{proof}
    Using the triangle inequality: 
    \[\left\|\mathcal{E}_A\left(\rho\right) - \mathcal{E}_B\left(\rho\right)\right\|_2 = \left\|\sum_{P \in \mathbb{P}_n} \left(a_P - b_P\right) P \rho P\right\|_2 \leq \sum_{P} \left|a_P - b_P\right| \cdot \left\|P \rho P\right\|_2.\]

    Now, observe that for any state $\rho$: 
    \[\left\|P \rho P\right\|_2^2 = \Tr\left(\left(P \rho P\right)^2\right) = \Tr\left(P \rho^2 P\right) = \Tr\left(\rho^2\right) \leq 1,\]
    by properties of purities. We have therefore found that: 
    \[\left\|\mathcal{E}_A\left(\rho\right) - \mathcal{E}_B\left(\rho\right)\right\|_2 \leq \sum_{P} \left|a_P - b_P\right| = \left\|a - b\right\|_1.\]
    
    Therefore:
    \[R_D\left(\mathcal{E}_A, \mathcal{E}_B\right) = \frac{1}{2} \exval_{\rho \followsdistr D}\left[\left\|\mathcal{E}_A\left(\rho\right) - \mathcal{E}_B\left(\rho\right)\right\|_2^2\right] \leq \frac{1}{2} \exval_{\rho \followsdistr D} \left[\left\|a - b\right\|_1^2\right] = \frac{1}{2} \left\|a - b\right\|_1^2.\]
\end{proof}

\begin{remark*}
    If we are given some other distribution $D'$ and manage to bound $R_{D'}\left(\mathcal{E}_A, \mathcal{E}_B\right) \geq \frac{1}{\alpha} \left\|a - b\right\|_1^2$, then this directly gives a multiplicative bound $R_D\left(\mathcal{E}_A, \mathcal{E}_B\right) \leq \frac{1}{2} \alpha R_{D'}\left(\mathcal{E}_A, \mathcal{E}_B\right) $. This is precisely the strategy we will use to produce multiplicative bounds for $D' = \Haar_n$ and $D' = \Haar_1^{\otimes n}$ below.
\end{remark*}

\begin{lemma}\label{lma:risk-pauli-channel-mmd-like}
    Let $\mathcal{E}_A\left(\rho\right) = \sum_{P \in \mathbb{P}_n} a_P P \rho P$ and $\mathcal{E}_B\left(\rho\right) = \sum_{P \in \mathbb{P}_n} b_P P \rho P$ be Pauli channels. Then: 
    \[R_D\left(\mathcal{E}_A, \mathcal{E}_B\right) = \frac{1}{2} \sum_{P, Q} \left(a_P - b_P\right)\left(a_Q - b_Q\right) \bra{w\left(P, Q\right)} M_D \ket{w\left(P, Q\right)},\]
    where $\ket{w\left(P, Q\right)} = \frac{1}{\sqrt{2^n}} \kket{Q^{\dagger} P}$ as defined in definition~\ref{def:w-vector}.
\end{lemma}
\begin{proof}
    We first directly observe that: 
    \[\mathcal{E}_A\left(\rho\right)-  \mathcal{E}_B\left(\rho\right) = \sum_{P} \left(a_P - b_P\right) P \rho P.\]
    
    Therefore: 
    \[R_D\left(\mathcal{E}_A, \mathcal{E}_B\right) = \frac{1}{2} \exval_{\rho \followsdistr D}\left[\Tr\left(\left(\mathcal{E}_A\left(\rho\right) - \mathcal{E}_B\left(\rho\right)\right)^2\right)\right] = \frac{1}{2} \sum_{P, Q} \left(a_P - b_P\right) \left(a_Q - b_Q\right) \exval_{\rho \followsdistr D}\left[\Tr\left(P \rho P Q \rho Q\right)\right].\]

    Note that we can simplify the expected value as follows: 
    \begin{equation}\begin{aligned}\exval_{\rho \followsdistr D}\left[\Tr\left(P \rho P Q \rho Q\right)\right] &= \exval_{\rho \followsdistr D}\left[\Tr\left(Q \rho Q P \rho P\right)\right] \\ &= \exval_{\rho \followsdistr D}\left[\bbrakket{QP}{\rho QP \rho}\right] \\ &= \bbra{QP} \exval_{\rho \followsdistr D} \left[\rho \otimes \rho^T\right] \kket{QP} \\ &= \frac{\bbra{QP}}{\sqrt{2^n}} 2^n \exval_{\rho \followsdistr D}\left[\rho \otimes \rho^T\right] \frac{\kket{QP}}{\sqrt{2^n}} \\ &= \bra{w\left(P, Q\right)} M_D \ket{w\left(P, Q\right)},\end{aligned}\end{equation}
    where $\ket{w\left(P, Q\right)} = \frac{1}{\sqrt{2^n}} \kket{Q^{\dagger} P}$ as defined in definition~\ref{def:w-vector}. The result immediately follows.
\end{proof}

\begin{lemma}[Risk for Pauli channels]\label{lma:risk-pauli-channel}
    Let $D$ be a distribution of moment matrix $M_D = \bar{\gamma} \oplus \mathcal{S} + M_{mean}$. Let $\mathcal{E}_A\left(\rho\right) = \sum_{P \in \mathbb{P}_n} a_P P \rho P$ and $\mathcal{E}_B\left(\rho\right) = \sum_{P \in \mathbb{P}_n} b_P P \rho P$ be Pauli channels. Then: 
    \[R_D\left(\mathcal{E}_A, \mathcal{E}_B\right) = \frac{\bar{\gamma}}{2} \left\|a - b\right\|_2^2 + \frac{1}{2} \sum_{P \neq Q} \left(a_P - b_P\right)\left(a_Q - b_Q\right) \bra{w\left(P, Q\right)} \mathcal{S} \ket{w\left(P, Q\right)},\]
    where $\ket{w\left(P, Q\right)} = \frac{1}{\sqrt{2^n}} \kket{Q^{\dagger} P}$ as defined in definition~\ref{def:w-vector}. In particular, this is independent of the bias $\kket{\mu}$ of $D$.
\end{lemma}

\begin{proof}
    Using lemma~\ref{lma:risk-pauli-channel-mmd-like}, we know $R_D\left(\mathcal{E}_A, \mathcal{E}_B\right) = \frac{1}{2} \sum_{P, Q} \left(a_P - b_P\right) \left(a_Q - b_Q\right) \bra{w\left(P, Q\right)} M_D \ket{w\left(P, Q\right)}$. Note that $\ket{w\left(P, Q\right)} = \frac{1}{\sqrt{2^n}} \kket{QP} \propto \frac{1}{\sqrt{2^n}}\kket{P'}$ for some Pauli $P' \in \mathbb{P}_n$. We are therefore only interested in the diagonal of $M_D$, which corresponds to the diagonal of $M_{core} = \bar{\gamma} \oplus \mathcal{S}$. Since finally $PQ \propto I \iff P = Q$ for Paulis $P, Q$, we can consider the case where $P = Q$ and the case $P \neq Q$ independently: 
    \begin{equation}\begin{aligned}\label{eq:risk-pauli-channels-two-terms} R_D\left(\mathcal{E}_A, \mathcal{E}_B\right) &= \frac{1}{2} \sum_{P} \left(a_P - b_P\right)\left(a_P - b_P\right) \bar{\gamma} + \frac{1}{2} \sum_{P \neq Q} \left(a_P - b_P\right)\left(a_Q - b_Q\right) \bra{w\left(P, Q\right)} \mathcal{S} \ket{w\left(P, Q\right)} \\ &= \frac{\bar{\gamma}}{2} \left\|a - b\right\|_2^2 + \frac{1}{2} \sum_{P \neq Q} \left(a_P - b_P\right)\left(a_Q - b_Q\right) \bra{w\left(P, Q\right)} \mathcal{S} \ket{w\left(P, Q\right)}.\end{aligned}\end{equation}
\end{proof}

\begin{lemma}[Risk of Pauli channels under global Haar]\label{lma:risk-pauli-channels-global-haar}
    For any Pauli channels $\mathcal{E}_A, \mathcal{E}_B$, their risk under the distribution $D = \Haar_n$ is given by: 
    \[R_{\Haar_n}\left(\mathcal{E}_A, \mathcal{E}_B\right) = \frac{2^n}{2^n + 1} \cdot \frac{1}{2}\left\|a - b\right\|_2^2.\]
\end{lemma}
\begin{proof}
    This is direct from lemma~\ref{lma:risk-pauli-channel} and lemma~\ref{lma:s-haar-distribution}:
    \begin{equation}\begin{aligned}R_D\left(\mathcal{E}_A, \mathcal{E}_B\right) &= \frac{1}{2} \left\|a - b\right\|_2^2 + \frac{1}{2} \sum_{P \neq Q} \left(a_P - b_P\right)\left(a_Q - b_Q\right) \frac{1}{2^n + 1} \\ &= \frac{1}{2} \left\|a - b\right\|_2^2 + \frac{1}{2} \cdot  \frac{1}{2^n + 1} \left(\sum_{P} \left(a_P - b_P\right) \sum_{Q} \left(a_Q - b_Q\right) - \sum_{P} \left(a_P - b_P\right)^2\right).\end{aligned}\end{equation}

    Observe that $\sum_{P} \left(a_P - b_P\right) = \sum_{P} a_P - \sum_{P} b_P = 1 -1  = 0$ since they are both probability distributions. Therefore, this simplifies to: 
    \[R_D\left(\mathcal{E}_A, \mathcal{E}_B\right) = \frac{1}{2} \left\|a - b\right\|_2^2 - \frac{1}{2} \cdot  \frac{1}{2^n + 1} \left\|a - b\right\|_2^2 = \frac{1}{2} \left\|a - b\right\|_2^2 \left(1 - \frac{1}{2^n + 1}\right) = \frac{1}{2} \left\|a-b\right\|_2^2 \cdot \frac{2^n}{2^n + 1}.\]
\end{proof}

For Pauli channels, sparsity controls the cost of generalizing from global Haar inputs. The general risk bound depends on the squared $1$-norm of the coefficient difference, whereas the Haar risk measures its squared $2$-norm. Bounding the first by the second gives a comparison factor determined by the number of Paulis in the channels.

\begin{theorem}[Generalisation from global Haar for sparse Pauli channels]\label{thm:sparse-pauli-global-haar}
    Let $D$ be an arbitrary distribution, let $\mathcal{E}_A, \mathcal{E}_B$ be Pauli channels with coefficients $a, b$ and assume that at most $s \geq 1$ coefficients of $a - b$ are non-zero. Then:
    \[R_D\left(\mathcal{E}_A, \mathcal{E}_B\right) \leq s \left(1 + 2^{-n}\right) R_{\Haar_n}\left(\mathcal{E}_A, \mathcal{E}_B\right).\]
\end{theorem}
\begin{remark*}
If each channel is supported on at most $s/2$ Paulis,
then the hypothesis of the theorem holds:
\[
\left|\supp(a-b)\right|
\leq \left|\supp a\right|+\left|\supp b\right|
\leq \frac{s}{2}+\frac{s}{2}
=s.
\]
    In other words, we can just as well take as hypothesis that both channels are each supported on $s'$ many Paulis, turning the bound to:
    \[R_D\left(\mathcal{E}_A, \mathcal{E}_B\right) \leq 2s' \left(1 + 2^{-n}\right) R_{\Haar_n}\left(\mathcal{E}_A, \mathcal{E}_B\right).\]
\end{remark*}
\begin{remark*}
    This Theorem does not require $D$ to be unbiased and can be formulated as follows. For any distribution $D$, and Pauli channels $\mathcal{E}_A, \mathcal{E}_B$ supported on polynomially many Paulis, then there is some $\alpha = \poly\left(n\right)$ such that: 
    \[R_D\left(\mathcal{E}_A, \mathcal{E}_B\right) \leq \alpha R_{\Haar_n}\left(\mathcal{E}_A, \mathcal{E}_B\right).\]
\end{remark*}
\begin{remark*}
    This is the first part of proposition~\ref{thm-main-text:sparse-pauli-channel-generalization} in the main text.
\end{remark*}
\begin{proof}
    By a usual Cauchy-Schwarz argument, it is well-known that a vector of support $\supp x$ is such that $\left\|x\right\|_1^2 \leq \left|\supp x\right| \cdot \left\|x\right\|_2^2$. Therefore, using lemma~\ref{lma:general-upper-bound-for-pauli-channels} and recognising $\frac{2^n+1}{2^n} R_{\Haar_n}\left(\mathcal{E}_A, \mathcal{E}_B\right) = \frac{1}{2}\left\|a - b\right\|_2^2$ by lemma~\ref{lma:risk-pauli-channels-global-haar}, we find:
    \[R_D\left(\mathcal{E}_A, \mathcal{E}_B\right) \leq \frac{1}{2}\left\|a-b\right\|_1^2 \leq \left|\supp\left(a-b\right)\right| \cdot \frac{1}{2} \left\|a-b\right\|_2^2 = \left|\supp \left(a-b\right)\right| \frac{2^n + 1}{2^n} R_{\Haar_n}\left(\mathcal{E}_A, \mathcal{E}_B\right).\]

    Since $\left|\supp(a-b)\right|\leq s$ by assumption, this proves the claimed bound.
\end{proof}

For sparse Pauli channels, the risk on product Haar inputs
controls the risk on any target ensemble. The comparison factor
is quadratic in the number of nonzero coefficients of $a-b$.
To prove this, we bound the target risk from above and the product
Haar risk from below in terms of the same coefficient difference.

\begin{theorem}[Generalisation from product Haar for sparse Pauli channels]
\label{thm:pauli-channel-bound-any-distrib-from-product-haar}
Let $D$ be an arbitrary distribution over $n$-qubit states,
possibly mixed. Let $\mathcal E_A,\mathcal E_B$ be Pauli
channels with coefficients $a,b$, and suppose that at most
$s\ge1$ coefficients of $a-b$ are nonzero. Then
\[
    R_D(\mathcal E_A,\mathcal E_B)
    \le
    s^2R_{\mathrm{Haar}_1^{\otimes n}}
    (\mathcal E_A,\mathcal E_B).
\]
\end{theorem}
\begin{remark*}
    This is the second part of proposition~\ref{thm-main-text:sparse-pauli-channel-generalization} in the main text.
\end{remark*}
\begin{proof}
Set $c:=a-b$, so that $|\operatorname{supp}c|\le s$.
For any real coefficient family
$x=(x_P)_{P\in\mathbb P_n}$, define
\[
    F_n(x)
    :=
    \sum_{P,Q\in\mathbb P_n}
    x_Px_Q\,3^{-|\{i:P_i\ne Q_i\}|}.
\]
For $n=0$, we interpret $\mathbb P_0$ as containing only
the empty Pauli string, so $F_0(x)=x^2$.

We first state three properties of $F_n$, use them to
deduce the theorem, and then prove them. The first
identifies the product-Haar risk, the second gives a
recursive expression, and the third uses that recursion
to bound $F_n$ from below on sparse coefficient families.

\medskip
\noindent
\emph{Claim 1 (Product-Haar risk).}
\begin{equation}
\label{eq:sparse-pauli-product-risk}
    2R_{\mathrm{Haar}_1^{\otimes n}}
    (\mathcal E_A,\mathcal E_B)
    =
    F_n(c).
\end{equation}

\medskip
\noindent
\emph{Claim 2 (Recursive expression).}
For $n\ge1$ and any real coefficient family $x$, partition
its entries according to their first Pauli factor:
\[
    x^{(A)}
    :=
    (x_{A\otimes Q})_{Q\in\mathbb P_{n-1}},
    \qquad A\in\mathbb P_1.
\]
Then
\begin{equation}
\label{eq:sparse-pauli-norm-decomposition}
    F_n(x)
    =
    \frac23\sum_{A\in\mathbb P_1}F_{n-1}(x^{(A)})
    +
    \frac13F_{n-1}\left(
        \sum_{A\in\mathbb P_1}x^{(A)}
    \right).
\end{equation}

\medskip
\noindent
\emph{Claim 3 (Sparse lower bound).}
For every $n\ge0$, every $s\ge1$, and every real coefficient
family $x$ with at most $s$ nonzero entries,
\begin{equation}
\label{eq:sparse-pauli-induction-claim}
    \|x\|_2^2\le sF_n(x).
\end{equation}

\medskip
Assuming these claims, the theorem follows immediately.
Indeed, Lemma~\ref{lma:general-upper-bound-for-pauli-channels},
Cauchy--Schwarz, Claim~3 applied to $x=c$, and Claim~1
give, respectively,
\begin{equation}
\label{eq:sparse-pauli-target-bound}
\begin{aligned}
    R_D(\mathcal E_A,\mathcal E_B)
    &\le \frac12\|c\|_1^2\\
    &\le \frac{s}{2}\|c\|_2^2\\
    &\le \frac{s^2}{2}F_n(c)\\
    &=s^2R_{\mathrm{Haar}_1^{\otimes n}}
      (\mathcal E_A,\mathcal E_B).
\end{aligned}
\end{equation}
It remains to prove the three claims.

\medskip
\noindent
\emph{Proof of Claim 1.}
By Lemma~\ref{lma:risk-pauli-channel-mmd-like},
\[
    2R_{\mathrm{Haar}_1^{\otimes n}}
    (\mathcal E_A,\mathcal E_B)
    =
    \sum_{P,Q\in\mathbb P_n}
    c_Pc_Q\,
    \bra{w(P,Q)}
    M_{\mathrm{Haar}_1^{\otimes n}}
    \ket{w(P,Q)},
\]
where
\[
    \ket{w(P,Q)}=2^{-n/2}\kket{QP}.
\]
We therefore need to evaluate the matrix element associated
with each pair $P,Q$.

Lemma~\ref{lma:s-haar-distribution}, applied to $n$
single-qubit blocks, gives eigenvalue $3^{-w_1(T)}$
in each normalized nonidentity Pauli direction
\[
    \kket{\widetilde T},
    \qquad
    \widetilde T=\frac{T}{\sqrt{2^n}}.
\]
Here $w_1(T)$ counts the nonidentity factors of $T$.
Since the product-Haar ensemble is pure and unbiased,
its moment matrix has the form $1\oplus\mathcal S$.
Thus the identity direction also satisfies this formula,
with eigenvalue $1=3^0$.

For each pair $P,Q$, write $QP=\omega T$, where
$T\in\mathbb P_n$ and $|\omega|=1$.
At each position, $T_i=I$ exactly when $P_i=Q_i$, so
\[
    w_1(T)=|\{i:P_i\ne Q_i\}|,
    \qquad
    \ket{w(P,Q)}=\omega\kket{\widetilde T}.
\]
The phase cancels in the quadratic form, and
$\kket{\widetilde T}$ has unit norm. Hence
\[
\begin{aligned}
    \bra{w(P,Q)}
    M_{\mathrm{Haar}_1^{\otimes n}}
    \ket{w(P,Q)}
    &=
    \bbra{\widetilde T}
    M_{\mathrm{Haar}_1^{\otimes n}}
    \kket{\widetilde T}\\
    &=3^{-w_1(T)}\\
    &=3^{-|\{i:P_i\ne Q_i\}|}.
\end{aligned}
\]
Substituting this expression into the risk formula gives
\[
    2R_{\mathrm{Haar}_1^{\otimes n}}
    (\mathcal E_A,\mathcal E_B)=F_n(c),
\]
proving Claim~1.

\medskip
\noindent
\emph{Proof of Claim 2.}
Write the two Pauli strings as $A\otimes P$ and $B\otimes Q$,
where $A,B\in\mathbb P_1$ and $P,Q\in\mathbb P_{n-1}$.
Their number of differing positions is
\[
    |\{i:(A\otimes P)_i\ne(B\otimes Q)_i\}|
    =
    \mathbf{1}_{A\ne B}+|\{j:P_j\ne Q_j\}|.
\]
Moreover,
\[
    3^{-\mathbf{1}_{A\ne B}}
    =
    \frac23\delta_{AB}+\frac13.
\]

In the following calculation, sums over $A,B$ range over
$\mathbb P_1$, and sums over $P,Q$ range over
$\mathbb P_{n-1}$. Reindexing the defining sum for $F_n$,
we obtain
\begin{align*}
    F_n(x)
    &=
    \sum_{A,B}\sum_{P,Q}
    x_{A\otimes P}x_{B\otimes Q}\,
    3^{-\mathbf{1}_{A\ne B}}\,
    3^{-|\{j:P_j\ne Q_j\}|}\\
    &=
    \sum_{A,B}\sum_{P,Q}
    x_{A\otimes P}x_{B\otimes Q}
    \left(\frac23\delta_{AB}+\frac13\right)
    3^{-|\{j:P_j\ne Q_j\}|}\\
    &=
    \frac23\sum_A\sum_{P,Q}
    x_{A\otimes P}x_{A\otimes Q}\,
    3^{-|\{j:P_j\ne Q_j\}|}\\
    &\quad+
    \frac13\sum_{P,Q}
    \left(\sum_Ax_{A\otimes P}\right)
    \left(\sum_Bx_{B\otimes Q}\right)
    3^{-|\{j:P_j\ne Q_j\}|}\\
    &=
    \frac23\sum_A F_{n-1}(x^{(A)})
    +
    \frac13F_{n-1}\left(\sum_Ax^{(A)}\right).
\end{align*}
The last equality uses $x^{(A)}_P=x_{A\otimes P}$.
The first term contains the quadratic form of each
individual family, while the second contains the
quadratic form of their sum. This proves Claim~2.

\medskip
\noindent
\emph{Proof of Claim 3.}
We proceed by induction on $n$, proving the statement for every sparsity
bound $s\ge1$ at each value of $n$. We allow arbitrary real
coefficients because the subfamilies in Claim~2 need not
be differences of probability distributions.

The base case is $n=0$. Since $F_0(x)=x^2$, the claim
reduces to
\[
    x^2\le sx^2,
\]
which holds for $s\ge1$.

Now assume the claim holds for $n-1$ qubits and every
sparsity bound. Use the families $x^{(A)}$ defined in
Claim~2, and write
\[
    s_A:=|\operatorname{supp}x^{(A)}|,
    \qquad
    r_A:=\|x^{(A)}\|_2.
\]
Because these families partition the original coefficients,
\[
    \sum_A s_A=|\operatorname{supp}x|\le s,
    \qquad
    \sum_A r_A^2=\|x\|_2^2.
\]

We apply the induction hypothesis to the two terms in
the recursive expression from Claim~2.

For each nonzero family $x^{(A)}$, the induction hypothesis
with sparsity bound $s_A$ gives
\[
    F_{n-1}(x^{(A)})\ge\frac{r_A^2}{s_A}.
\]
Weighted Cauchy--Schwarz therefore yields
\begin{equation}
\label{eq:sparse-pauli-separate-families}
\begin{aligned}
    \left(\sum_A r_A\right)^2
    &=
    \left(
        \sum_{A:s_A>0}
        \sqrt{s_A}\,\frac{r_A}{\sqrt{s_A}}
    \right)^2\\
    &\le
    \left(\sum_A s_A\right)
    \left(\sum_{A:s_A>0}\frac{r_A^2}{s_A}\right)\\
    &\le
    s\sum_A F_{n-1}(x^{(A)}).
\end{aligned}
\end{equation}
Zero families contribute nothing and are omitted only
from expressions containing a denominator.

For the summed family, we have
\[
    \operatorname{supp}\left(\sum_Ax^{(A)}\right)
    \subseteq
    \bigcup_A\operatorname{supp}x^{(A)}.
\]
Its support therefore has size at most
$\sum_A s_A\le s$. Applying the induction hypothesis
with sparsity bound $s$ gives
\begin{equation}
\label{eq:sparse-pauli-summed-family}
    \left\|\sum_Ax^{(A)}\right\|_2^2
    \le
    sF_{n-1}\left(\sum_Ax^{(A)}\right).
\end{equation}

Combining \eqref{eq:sparse-pauli-separate-families} and
\eqref{eq:sparse-pauli-summed-family} with the recursion
in Claim~2, we obtain
\begin{equation}
\label{eq:sparse-pauli-combined-bound}
\begin{aligned}
    sF_n(x)
    &=
    \frac{2s}{3}\sum_A F_{n-1}(x^{(A)})
    +
    \frac{s}{3}F_{n-1}\left(\sum_Ax^{(A)}\right)\\
    &\ge
    \frac23\left(\sum_A r_A\right)^2
    +
    \frac13\left\|\sum_Ax^{(A)}\right\|_2^2.
\end{aligned}
\end{equation}

It remains to check that the last expression is at least
$\sum_A r_A^2=\|x\|_2^2$. Expanding both squares gives
\[
\begin{aligned}
    &\frac23\left(\sum_A r_A\right)^2
    +\frac13\left\|\sum_Ax^{(A)}\right\|_2^2
    -\sum_A r_A^2\\
    &\qquad=
    \frac23\sum_{A<B}
    \left(
        2r_A r_B
        +\left\langle x^{(A)},x^{(B)}\right\rangle
    \right).
\end{aligned}
\]
Here $A<B$ refers to any fixed ordering of $\mathbb P_1$.
By Cauchy--Schwarz,
\[
    \left\langle x^{(A)},x^{(B)}\right\rangle
    \ge
    -\|x^{(A)}\|_2\|x^{(B)}\|_2
    =
    -r_A r_B.
\]
Thus every summand is nonnegative:
\[
    2r_A r_B
    +\left\langle x^{(A)},x^{(B)}\right\rangle
    \ge r_A r_B\ge0.
\]
Together with \eqref{eq:sparse-pauli-combined-bound},
this proves
\[
    sF_n(x)\ge\sum_A r_A^2=\|x\|_2^2.
\]
The induction establishes Claim~3, completing the proof
of the theorem.
\end{proof}

\suppsubsection{Negative result for Pauli channels}

Real product inputs can fail to distinguish two Pauli channels even when each has constant support. Our construction uses complete dephasing in the $Y$ basis, which preserves the $Y$ component while erasing the $X$ and $Z$ components. The surviving component vanishes for real inputs but contributes to the risk on complex Haar-random inputs. Applying the construction to one qubit and leaving the others unchanged gives the same separation for product inputs for any number of qubits.

\begin{theorem}[Impossibility of real to complex Haar generalization for Pauli channels]
\label{thm:pauli-channel-real-complex-no-go}
    Let $\alpha, \beta > 0$ be such that, for any Pauli channels $\mathcal{E}_A, \mathcal{E}_B$ of support at most $s$: 
    \[R_{\Haar_1^{\otimes n}}\left(\mathcal{E}_A, \mathcal{E}_B\right) \leq \alpha R_{\Haar_1^{\otimes n}\left(O\right)}\left(\mathcal{E}_A, \mathcal{E}_B\right) + \beta.\]

    If $s \geq 2$, then there exist Pauli channels of constant support that can be efficiently implemented using quantum hardware that require: 
    \[\beta \geq \frac{1}{3}.\]

    In other words, it is impossible to do multiplicative out-of-distribution generalisation from the product Haar distribution on real states to the product Haar distribution on complex states.
\end{theorem}
\begin{remark*}
    This is precisely Proposition~\ref{thm-main-text:pauli-channel-no-go} in the main text.
\end{remark*}
\begin{proof}
    Consider the following two Pauli channels, each supported on $2$ Paulis: 
    \[\mathcal{E}_A\left(\rho\right) = \frac{1}{2} \rho + \frac{1}{2} Y_1 \rho Y_1, \mathspace \mathcal{E}_B\left(\rho\right) = \frac{1}{2} X_1 \rho X_1 + \frac{1}{2} Z_1 \rho Z_1.\]
    
    Note that both distributions are product, and hence we can write $\rho = \rho_1 \otimes \ldots \otimes \rho_n$. We moreover write $\rho_1 = \frac{1}{2}\left(I + x X + y Y + z Z\right)$ in its Bloch form. We observe that: 
    \begin{equation}\begin{aligned}\mathcal{E}_A\left(\rho\right) &= \mathcal{E}_A\left(\rho_1 \otimes \rho_{2, \ldots, n}\right) \\ &= \left[\frac{1}{4}\left(I + x X + yY + zZ\right) + \frac{1}{4} Y \left(I + x X + y Y + z Z\right) Y\right] \otimes \rho_{2, \ldots, n} \\ &= \left[\frac{1}{4}\left(I + x X + yY + zZ\right) + \frac{1}{4} \left(I - x X + y Y - z Z\right)\right] \otimes \rho_{2, \ldots, n} \\ &= \frac{1}{2} \left(I + y Y\right) \otimes \rho_{2, \ldots, n},\end{aligned}\end{equation}
    since $Y Z Y = Y \left(-YZ\right) = -Y^2 Z = -Z$ by commutation properties of Paulis, and similarly for the other terms. By the exact same reasoning:
    \begin{equation}\begin{aligned}\mathcal{E}_B\left(\rho\right) &= \left[\frac{1}{4} X\left(I + x X + yY + zZ\right) X + \frac{1}{4} Z \left(I + x X + y Y + z Z\right) Z\right] \otimes \rho_{2, \ldots, n} \\ &= \left[\frac{1}{4}\left(I + x X - yY - zZ\right) + \frac{1}{4} \left(I - x X - y Y + z Z\right)\right] \otimes \rho_{2, \ldots, n} \\ &= \frac{1}{2} \left(I - y Y\right) \otimes \rho_{2, \ldots, n}.\end{aligned}\end{equation}

    Therefore, for any pure product distribution $D^{\otimes n}$, their risk is given by: 
    \[R_{D^{\otimes n}}\left(\mathcal{E}_A, \mathcal{E}_B\right) = \frac{1}{2} \exval_{\rho \followsdistr D^{\otimes n}}\left[\left\|\mathcal{E}_A\left(\rho\right)-  \mathcal{E}_B\left(\rho\right)\right\|_2^2\right] = \frac{1}{2} \exval_{\rho \followsdistr D^{\otimes n}}\left[\left\|y Y \otimes \rho_{2, \ldots, n}\right\|_2^2\right] = \frac{1}{2} \exval_{\rho \followsdistr D^{\otimes n}}\left[\Tr\left(\left(y Y\right)^2 \otimes \rho_{2, \ldots, n}^2\right)\right].\]
    
    Now, we know that $\Tr\left(A \otimes B\right) = \Tr\left(A\right) \Tr\left(B\right)$, and that $\Tr\left(\rho_{2, \ldots, n}^2\right) = 1$ since both our distributions are pure. Observing moreover that $y = \Tr\left(Y_1 \rho\right)$, this allows us to write:
    \begin{equation}\begin{aligned}R_{D^{\otimes n}}\left(\mathcal{E}_A, \mathcal{E}_B\right) &= \frac{1}{2} \exval_{\rho \followsdistr D^{\otimes n}}\left[y^2 \Tr\left(Y^2\right)\right] \\ &= \exval_{\rho \followsdistr D^{\otimes n}}\left[y^2\right] \\ &= \exval_{\rho \followsdistr D^{\otimes n}}\left[\Tr\left(Y_1 \rho\right)^2\right] \\ &= \bbra{Y_1} \exval_{\rho \followsdistr D^{\otimes n}}\left[\kket{\rho}\bbra{\rho}\right] \kket{Y_1} \\ &= \frac{\bbra{Y_1}}{\sqrt{2^n}} M_{D^{\otimes n}} \frac{\kket{Y_1}}{\sqrt{2^n}},\end{aligned}\end{equation}
    since $M_{D^{\otimes n}} = 2^n\exval_{\rho \followsdistr D^{\otimes n}}\left[\kket{\rho}\bbra{\rho}\right]$ for pure-state distributions (by Lemma~\ref{lma:moment-matrix-equivalence-pure-states}). Since $\frac{\kket{Y_1}}{\sqrt{2^n}}$ is an eigenstate of both $M_{\Haar_1^{\otimes n}}$ and $M_{\Haar_1^{\otimes n}\left(O\right)}$, we can use lemmas~\ref{lma:s-haar-distribution} and~\ref{lma:s-orthogonal-haar-distribution} to find that: 
    \[R_{\Haar_1^{\otimes n}}\left(\mathcal{E}_A, \mathcal{E}_B\right) = \frac{1}{3}, \mathspace R_{\Haar_1^{\otimes n}\left(O\right)}\left(\mathcal{E}_A, \mathcal{E}_B\right) = 0.\]
     
    Therefore, out-of-distribution generalisation on these channels requires:
    \[R_{\Haar_1^{\otimes n}}\left(\mathcal{E}_A, \mathcal{E}_B\right) \leq \alpha R_{\Haar_1^{\otimes n}\left(O\right)}\left(\mathcal{E}_A, \mathcal{E}_B\right) + \beta \implies \frac{1}{3} \leq \alpha\cdot 0 + \beta \implies \beta \geq \frac{1}{3}.\]
\end{proof}

\suppsubsection{Negative result for general channels}\label{sec:negativechannels}

The favorable comparison for Pauli channels for the global Haar measure does not extend to arbitrary channels. We show this via constructing a counter example using a pair of measure-and-prepare channels. The corresponding risk is constant on product Haar inputs but exponentially small on global Haar inputs. This forces the stated trade-off between multiplicative and additive error.

\begin{theorem}[Impossibility of global to product Haar generalization for general channels]\label{thm:general-channel-no-go}
    Let $\alpha, \beta > 0$ be such that, for any channels $\mathcal{E}_A, \mathcal{E}_B$: 
    \[R_{\Haar_1^{\otimes n}}\left(\mathcal{E}_A, \mathcal{E}_B\right) \leq \alpha R_{\Haar_n}\left(\mathcal{E}_A, \mathcal{E}_B\right) + \beta.\]

    Then, there exist channels that can be efficiently implemented using quantum hardware that require: 
    \[\alpha \geq \left(\frac{1}{3} - \beta\right)\left(2^n + 1\right).\]

    In particular, whenever $\beta < \frac{1}{3} - \Theta\left(1\right)$, then we need $\alpha \geq \Omega\left(2^n\right)$.
\end{theorem}
\begin{remark*}
    Setting $\beta = 1/4$, this is precisely proposition~\ref{thm-main-text:general-channel-no-go} in the main text.
\end{remark*}
\begin{proof}
    We consider $\mathcal{E}_A, \mathcal{E}_B$ to be the following measure-and-prepare channels: 
    \[\mathcal{E_A}\left(\rho\right) = \Tr\left(\left(\ket{0}\bra{0} \otimes I^{\otimes n - 1}\right) \rho\right) \ket{0^n}\bra{0^n} + \Tr\left(\left(\ket{1}\bra{1} \otimes I^{\otimes n - 1}\right) \rho\right) \ket{1^n}\bra{1^n},\]
    \[\mathcal{E_B}\left(\rho\right) = \Tr\left(\left(\ket{0}\bra{0} \otimes I^{\otimes n - 1}\right) \rho\right) \ket{1^n}\bra{1^n} + \Tr\left(\left(\ket{1}\bra{1} \otimes I^{\otimes n - 1}\right) \rho\right) \ket{0^n}\bra{0^n}.\]

    In other words, $\mathcal{E}_A\left(\rho\right)$ measures the first qubit of $\rho$ and outputs the measurement $n$ times. $\mathcal{E}_B$ is similar but flips the measurement before outputting it $n$ times. Therefore,
\begin{equation}\begin{aligned}
\mathcal{E}_A\left(\rho\right)-\mathcal{E}_B\left(\rho\right)
&=
\Tr\left(\left(\ket{0}\bra{0}\otimes I^{\otimes(n-1)}\right)\rho\right)
\left(\ket{0^n}\bra{0^n}-\ket{1^n}\bra{1^n}\right)
+
\Tr\left(\left(\ket{1}\bra{1}\otimes I^{\otimes(n-1)}\right)\rho\right)
\left(\ket{1^n}\bra{1^n}-\ket{0^n}\bra{0^n}\right)
\\ &=
\left(\ket{0^n}\bra{0^n}-\ket{1^n}\bra{1^n}\right)
\left[
\Tr\left(\left(\ket{0}\bra{0}\otimes I^{\otimes(n-1)}\right)\rho\right)
-
\Tr\left(\left(\ket{1}\bra{1}\otimes I^{\otimes(n-1)}\right)\rho\right)
\right]
\\ &=
\left(\ket{0^n}\bra{0^n}-\ket{1^n}\bra{1^n}\right)
\Tr\left(Z_1\rho\right).
\end{aligned}\end{equation}
    
    We can compute their risk exactly for any distribution $D$: 
    \begin{equation}\begin{aligned}R_D\left(\mathcal{E}_A, \mathcal{E}_B\right) &= \frac{1}{2} \exval_{\rho \followsdistr D}\left[\Tr\left(\left(\mathcal{E}_A\left(\rho\right) - \mathcal{E}_B\left(\rho\right)\right)^2\right)\right] \\ &= \frac{1}{2} \exval_{\rho \followsdistr D}\left[\Tr\left(Z_1 \rho\right)^2 \underbrace{\Tr\left(\left(\ket{0^n}\bra{0^n} - \ket{1^n}\bra{1^n}\right)^2\right)}_{= 2}\right] \\ &= \frac{\bbra{Z_1}}{\sqrt{2^n}}\cdot 2^n\exval_{\rho \followsdistr D}\left[\kket{\rho}\bbra{\rho}\right] \frac{\kket{Z_1}}{\sqrt{2^n}}.\end{aligned}\end{equation}

    Now, we know that $M_D = 2^n \exval_{\rho \followsdistr D}\left[\kket{\rho}\bbra{\rho}\right]$ for pure-state distributions $D$ by lemma~\ref{lma:moment-matrix-equivalence-pure-states}. Now, note that $\frac{1}{\sqrt{2^n}} \kket{Z_1}$ is a normalised eigenvector of $M_D$ both when $D = \Haar_1^{\otimes n}$ and when $D = \Haar_n$, by lemma~\ref{lma:s-haar-distribution}. Therefore, using this lemma, we find: 
    \[R_{\Haar_1^{\otimes n}}\left(\mathcal{E}_A, \mathcal{E}_B\right)  = \frac{1}{3},\]
    \[R_{\Haar_n}\left(\mathcal{E}_A, \mathcal{E}_B\right) = \frac{1}{2^n + 1}.\]

    We then immediately get our result using our hypothesis: 
    \[R_{\Haar_1^{\otimes n}}\left(\mathcal{E}_A, \mathcal{E}_B\right) \leq \alpha R_{\Haar_n}\left(\mathcal{E}_A, \mathcal{E}_B\right) + \beta \implies \frac{1}{3} \leq \frac{\alpha}{2^n + 1} + \beta \implies \alpha \geq \left(2^n + 1\right)\left(\frac{1}{3} - \beta\right).\]
\end{proof}

\end{document}